\documentclass{amsart}
\usepackage{graphicx,amssymb,epsfig,amsmath}

\usepackage{lipsum}
\usepackage{enumitem}
\usepackage{caption}

\usepackage{hyperref}
\usepackage{xcolor}
\usepackage{amsmath}
\usepackage{amssymb,amsmath}
\usepackage[all]{xy}

\usepackage{tikz}
\usepackage{xkeyval,tkz-base}
\usetikzlibrary{arrows,calc}
\usepackage{tikz,pgfplots}
\pgfplotsset{compat=1.10}
\usepgfplotslibrary{fillbetween}
\usetikzlibrary{patterns}
 \usetikzlibrary{calc,shapes,decorations,decorations.text, mindmap,shadings,patterns,matrix,arrows,intersections,automata,backgrounds}
\usepackage{float}
\usetikzlibrary{arrows}
\usetikzlibrary{decorations.pathmorphing}
 \usetikzlibrary{plotmarks}
        \usetikzlibrary{intersections}
        \usetikzlibrary{shadings}
        \usetikzlibrary{calc}
        \usetikzlibrary{patterns}

\newtheorem{theorem}{Theorem}[section]
\newtheorem{proposition}[theorem]{Proposition}
\newtheorem{lemma}[theorem]{Lemma}
\newtheorem{corollary}[theorem]{Corollary}
\newtheorem{remark}[theorem]{Remark}
\newtheorem{example}[theorem]{Example}

\newtheorem{definition}[theorem]{Definition}

\newtheorem{corollaries}[theorem]{Corollaries}
\newtheorem{theoremf}{Folk Formula}[section]
\newtheorem{maintheorem}{Main Theorem}[section]
\newcommand{\bth}{\begin{theorem}}
\newcommand{\bpr}{\begin{proposition}}
\newcommand{\epr}{\end{proposition}}
\newcommand{\bco}{\begin{corollary}}
\newcommand{\eco}{\end{corollary}}
\newcommand{\ble}{\begin{lemma}}
\newcommand{\ele}{\end{lemma}}
\newcommand{\bde}{\begin{definition}\rm}
\newcommand{\ede}{\end{definition}\rm}
\newcommand{\bre}{\begin{remark}\rm}
\newcommand{\ere}{\end{remark}}
\newcommand{\bex}{\begin{example}\rm}
\newcommand{\eex}{\end{example}}
\newcommand{\bcors}{\begin{corollaries}\rm}
\newcommand{\ecors}{\end{corollaries}}
\newcommand{\bthf}{\begin{theoremf}}
\newcommand{\bmain}{\begin{maintheorem}}
\newcommand{\emain}{\end{maintheorem}}

\def\la#1{\hbox to #1pc{\leftarrowfill}}
\def\ra#1{\hbox to #1pc{\rightarrowfill}}

\def\lrar{{\ra 2}}
\def\tensor{\otimes}
\def\sp#1{\hbox{SP}^{#1}}

\def\bbz{{\mathbb Z}}
\def\bbf{{\mathbb F}}
\def\bbr{{\mathbb R}}
\def\bbc{{\mathbb C}}
\def\bbp{{\mathbb P}}
\def\bbn{{\mathbb N}}
\def\ds{\displaystyle}
\def\conf#1{\hbox{Conf}_{#1}}

\def\bbq{{\mathbb Q}}

\def\G{{\Gamma}}

\def\F{{\mathcal F}}

\newcommand{\boldpi}{\mbox{$\pi$\hspace{-6.5pt}$\pi$}}
\newcommand{\tree}[2]{ \ensuremath{  
 \begin{xy}                          %
   (0,1); (1,0)**\dir{-};            
   (2,1)**\dir{-},                   
   (1,-1); (1,0)**\dir{-},           
   (0,2.2)*{\scriptstyle #1},        %
   (2,2.2)*{\scriptstyle #2},        %
 \end{xy}  } } 
\def\FM{\hbox{\bf FM}}
\AtBeginDocument{\addtocontents{toc}{\protect\setlength{\parskip}{0pt}}}
\begin{document}

\title[Configuration Spaces]{Configuration Spaces of Points:\\  A User's Guide}

\author{Sadok Kallel}\thanks{*American University of Sharjah, UAE, and Laboratoire Painlev\'e, Universit\'e de Lille, France.}
 
\email{sadok.kallel@univ-lille.fr}

\maketitle

\begin{abstract} This user's guide is divided into two parts. The first part is an extensive survey contributed to the Encyclopedia of Mathematical Physics, 2nd edition. It covers many of the main constructions, definitions, and applications of the classical configuration spaces of points. The second part delves into the geometry of chromatic configuration spaces, giving a detailed proof of the remarkable result that the Poincar\'e polynomial of the chromatic configuration spaces of $\bbr^N$, associated to a finite simple graph $\Gamma$, corresponds to the reciprocal of the chromatic polynomial of the graph (with signs). Further applications and a stable splitting are given.  \\
\end{abstract}

\centerline{To Fred Cohen, in fond memory.}

\tableofcontents

\clearpage
.\vskip 200pt
\centerline{\bf\Huge Part I: A User's Guide}

\vskip 30pt
\begin{quote}\normalsize
This extensive survey is an invited contribution to the Encyclopedia of Mathematical Physics, 2nd edition. It covers both classical and more modern aspects of configuration spaces of points on a ``ground space'' $M$. Most results pertain to $M$ a manifold. Configuration spaces of points have become so omnipresent in so many areas of mathematics, physics, and even the applied sciences, that a survey can only cover a selection of topics. We review key ideas, constructions, and results. 
\end{quote}

\clearpage
\part{ A User's Guide}

\section{Introduction}

If configuration spaces were like a magnificent city to discover,  this survey will be a visit to its most important landmarks. Most remarkably in the last two decades, configuration spaces of distinct points have seen an explosion in interest and the number of publications. Applications have long gone beyond algebraic topology to encompass all modern aspects of geometry, analysis, and mathematical physics, further making incursions into the applied science fields. This survey is longer than initially intended for this reason, but also because we strove to explain the motivations and ideas behind many of the results presented, beyond listing them.

Given a topological space $X$, which we will always assume to be path-connected, locally compact, and Hausdorff, we define
\begin{equation}\label{config}
\conf{k}(X):= \{(x_1,\ldots, x_k)\in X^k\ |\ x_i\neq x_j, i\neq j\}
\end{equation}
This is an open subspace of $X^k$ obtained by removing the closed diagonal subspaces
\begin{equation}\label{deltaij}\Delta_{i,j}:= \{(x_1,\ldots, x_k)\in X^k, x_i=x_j\}\end{equation} Each tuple $(x_1,\ldots, x_k)\in \conf{k}(X)$ is called an \textit{ordered configuration}, and each entry $x_i$ is called \textit{a point} of the configuration. 
Very conveniently, we view an element of 
$\conf{k}(X)$ as an ordered collection of $k$ pairwise distinct points or \textit{particles} in $X$, and in fact it can be defined as Emb$([k],X)$, the space of all embeddings of $[k]=\{1,2,\ldots, k\}$ into $X$. When $k=2$, $\conf{2}(X)$ is referred to in the classical literature as the \textit{deleted product} of $X$. Another common notation for $\conf{k}(X)$ is $F(X,k)$\footnote{The author was told that the letter``$F$'' in the notation historically refers to Ed Fadell.}, and there are half a dozen different other notations used, including $\underset{{}^{\{1,\cdots,k\}}
    }{\mathrm{Conf}}(X)$ in the physics literature \cite{ss,ss2}. The mathematical community seems to have finally adopted the notation \eqref{config}.

The symmetric group on $k$-letters $\mathfrak S_k$ acts on $\conf{k}(X)$ by permuting the points of the ordered configuration. More precisely, if $\sigma\in\mathfrak S_k$, and $(x_1,\ldots, x_k)\in\conf{k}(X)$ then $\sigma (x_1,\ldots, x_k) = (x_{\sigma (1)}, \cdots , x_{\sigma (k)})$.  The orbit space is denoted by
$$C_k(X) := \{ \{x_1,\ldots, x_k\}\ |\ x_i\in X, x_i\neq x_j\}$$
The symmetric group action being free and properly discontinuous, the quotient map $\conf{k} (X)\rightarrow C_k(X)$ is a regular covering space of degree $k!=|\mathfrak S_k|$. Here too, numerous other notations for $C_k(X)$ can be found in the literature, in particular $U\conf{k}(X)$,  $X\choose k$ or $B(X,k)$ (``$B$'' standing for ``Braid space'').
Points in $C_k(X)$ are called \textit{indistinguishable} or \textit{identical particles} in the physics literature and we adopt this terminology here as well. 

Configuration spaces of points are very natural and basic objects in physics since they offer the mathematical framework to describe and study particles, distinguishable or not. The space where the particles live is their ``manifold ground state'', the way they travel forms ``worldlines'' which are \textit{braid-like} when the ground state is two-dimensional, and their dynamics involve bundle theory and connections. In the Schrodinger representation, for example, the wave functions are sections of a Hermitian line bundle with connection over configuration spaces. 
In geometric quantization and quantum statistics\footnote{According to \cite{baez}, ``statistics'' refers to ``the behavior of quantum systems under the interchange of identical particles''.} \cite{sawiki}, isomorphism classes of principal $U(1)$-bundles over the space of unordered particles in space (i.e. $C_m(\bbr^3)$) classify these quantum statistics, and there are two of them, the trivial bundle corresponds to Bose statistics and the non-trivial bundle to Fermi statistics \cite{souriau, bbs}. 
Leinaas and Myrheim \cite{leinmyr} studied particles restricted to move in a 2-dimensional quantum material, giving a new form of quantum statistics, called anyons which are now a crucial concept in condensed matter physics. As explained in \cite{myershishurs}, ``Anyon braiding is given by the holonomy of the $su_2$-KZ connection arising as the Gauss-Manin connections on bundles of twisted generalized cohomology groups over configuration spaces of points''. 
In Euclidean field theory, propagators (resp. correlators) can be viewed as functions (resp. differential forms) on the Fulton-MacPherson compactifications of configuration spaces of points (see \S\ref{compactification}).
Configuration spaces also appear crucially in the study of \textit{instantons} and \textit{monopoles} \cite{aj, bhmm}. 
The list goes on, and in such a short account we can only scratch the surface of the relevant physics literature (a  more exhaustive account can be found on nLab).

Configuration spaces permeate most of modern-day geometry and topology. After all, Weirstrass points on an algebraic curve, or marked points on such curves (see \S\ref{geometry}), or critical points of a Morse function on a smooth manifold, or an N-body system (see \S\ref{nbody}) form a point configuration. If $X$ is a quasiprojective variety over a field, then the configuration spaces carry an algebraic structure as well. In this case, $C_k(X)$ is the scheme-theoretic quotient of $\conf{k}(X)$ by the natural action of $\mathfrak S_k$, and both $\conf{k}(X)$ and $C_k(X)$ are quasiprojective varieties. 

This survey attempts to summarize and explain the mathematics surrounding configuration spaces of manifolds, ordered and unordered, and discuss many of their applications. The primary tools, constructions, and viewpoints used are those of algebraic and differential topology. Homology and cohomology are discussed extensively since they are the most accessible algebraic, functorial, and often as well, combinatorial objects that one can use to probe deep into the geometry of spaces. The symmetric group action on $\conf{k}(X)$ extends to an action on its rational cohomology, and the study of the associated representation theory has been an important tool in describing geometric invariants of various moduli spaces (see \S\ref{geometry}).

The case of Euclidean configurations, i.e. when $X=\bbr^n$, is sufficiently rich to fill in hundreds of pages \cite{fh}. We spend a good first half of this survey dissecting both $\conf{k}(\bbr^n)$ and $C_k(\bbr^n)$, with $k\geq 2$, $n\geq 1$, starting with foundational work of Arnold and Cohen (see \S\ref{homology}, \S\ref{cohomology}). 
Observe that in the Euclidean case, $\conf{k}(\bbr^n)$ is the complement of a \textit{subspace arrangement} in $(\bbr^n)^k$ known as the ``braid arrangement'' or the  \textit{Coxeter arrangement of type $A$}. The study of configuration spaces from this point of view has inspired much of the later development of the vast theory of hyperplane and subspace arrangements \cite{orlikterao}.

Remarkably, knowledge of the (co)homology of Euclidean configuration spaces has served well to give, using diverse and often sophisticated ideas and techniques, a good understanding of the (co)homology of configuration spaces of finite-dimensional manifolds. This connection is illustrated throughout the second half of the survey (See \S\ref{labeled}, \S\ref{computational}, \S\ref{unorderedconfigs}). 
A higher level view of this local to global connection is provided by \textit{factorization homology}, a theory we will barely touch upon (see {\bf 4.2.3}), but which provides a means of assembling the rich local structure of configuration spaces across coordinate patches of a general manifold, globalizing the calculation of Arnold and Cohen \cite{af, knudsen2}.

Very early on, it was observed that configuration spaces can be used to give criteria for embedding a manifold in Euclidean space (see \S\ref{embeddingtheory}), and as the theory developed, it helped solve many longstanding problems. The problem of the equipartition of convex polygons was solved using exclusively the homotopy theory of configuration spaces (see \S\ref{coincidences}), while other open problems like Vassiliev knot invariants (see \S\ref{integrals}), the square peg problem or Malle's conjecture in arithmetic galois theory (see \S\ref{geometry}) were satisfactorily tackled using these ideas. In early times, configuration spaces entered homotopy theory in full force through the work of May and Milgram on iterated loop spaces \cite{may1}, providing a great tool for computing the homology of loop spaces (see \S\ref{computational}), and more importantly, being at the origin of operad theory, with immense impact on algebraic topology (see \S\ref{operads}). The idea of decorating the points of a configuration with labels emerged as well in the '90s in connection with the study of poles and zeros of holomorphic maps, and served to model loop spaces and mapping spaces, with deep theorems in the field (see \S\ref{labeled}). Labels can serve to encode interactions between particles, with possible applications to physics (perhaps somewhat under-exploited as of now).

There are numerous variants and extensions of configuration spaces in topology that take the names of chromatic configuration spaces, orbit configuration spaces, colored configuration spaces, cyclic configuration spaces, labeled configuration spaces, generalized configuration spaces, partial configuration spaces, hard-disks configuration spaces, and more. We discuss almost all of those at the end of the survey (see \S\ref{variants}). 

There is one textbook written on configuration spaces \cite{fh} (see \S{\bf 8.5}) and a number of surveys covering various aspects \cite{fred2, fredanthology, idrissi2, knudsen, dev, volic, volic2, westerland}. We have benefited from all these references, and foremost from the original sources. The references \cite{fred2, idrissi2} can serve as excellent introductions to the subject as well. 

\noindent{\sc Acknowledgment}: We thank Najib Idrissi, Paolo Salvatore, Florian Kranhold and Urs Schreiber for reading parts of this survey and making suggestions. We thank Hisham Sati for his invitation to contribute to this volume.

\noindent{\sc Terminology and convention}: A based space $X$ is a space with a preferred basepoint $x_0\in X$. A based map between based spaces $X,Y$ is a continuous map $f$ such that $f(x_0)=y_0$. The $n$-th iterated loop space  $\Omega^nY$ is the space of all based maps from $S^n$ to $Y$. A ``DGA'' means a differential graded algebra, and ``GC'' means graded commutative. The notation $\cong$ means homeomorphism, while $\simeq$ means homotopy equivalence. The letter $\bbf$ refers to (any) field.


\section{Classical theory of ordered  configuration spaces in $\bbr^n$}\label{classical}

\subsection{Early work}\label{early} The first systematic study of the algebraic topology of configuration spaces goes back to Fadell and Neuwirth \cite{fn}, Fox and Neuwirth \cite{fox} and to work of Arnold \cite{arnold} and Cohen \cite{fredbible}. Since we will mostly restrict to $\conf{k}(M)$ when $M$ is a manifold (i.e. a locally Euclidean connected Hausdorff space), we can completely describe the case when $\dim M=1$ and $M$ without boundary. In that case, $M$ is either $\bbr$ or $S^1$. If $M=\bbr$, then $\conf{k}(\bbr)$ consists of $k!$ components, each homeomorphic to $\mathring{\Delta}_{k}$, the interior of the $k$-dimensional simplex $\Delta_{k}$, 
and one sees that $C_k(\bbr)\cong\bbr^k$. If $M=S^1$, then $\conf{k}(S^1)$ is homeomorphic to
$S^1\times\conf{k-1}((0,1))$. This space has $(k-1)!$ components and configurations that can be obtained from each other by a cyclic permutation make up a single component. The description of $\conf{k}(S^1)$ as a $\mathfrak S_k$-space is  more subtle (see \cite{madsenbodig}, \S4). 
In the unordered case, the multiplication map $C_k(S^1)\rightarrow S^1$ is a bundle map over $S^1$ with fiber $\mathring \Delta_{k-1}$. The bundle is trivial only if $n$ is odd \cite{morton}. In particular, $C_k(S^1)\simeq S^1$, $\forall k\geq 1$.

\noindent{\bf 2.1.1.} In their early study, \cite{fn} established the first theorems on the existence of fibrations relating to configuration spaces on manifolds (now called Fadell-Neuwirth fibrations). In its simplest form, this result states that if $M$ is an arbitrary manifold of dimension $n\geq 2$, then the projection onto any of the points of the configuration $\conf{k} (M)\rightarrow M$ is a locally trivial bundle projection with fiber $\conf{k-1} (M-\{p\})$, where $M-\{p\}$ is the once punctured manifold. Here, we convene that $\conf{0}(M)$ is a point. 
The projection $\pi$ is not, in general, a bundle, nor a fibration, if $M$ is not a manifold.
The Fadell-Neuwirth bundle construction extends to the  projection maps
\begin{equation}\label{iteratedbundle}
\pi : \conf{k} (M) \longrightarrow \conf{k-r}(M)\ \ \ ,\ \ \ \ 0<r<k
\end{equation}
where $\pi$ forgets $r$ points of the configuration with fixed coordinates. This is a locally trivial bundle with fiber $\conf{r}(M-Q_{k-r})$, where
$Q_{k-r}$ is a fixed set of $k-r$ distinct points in $M$. As an immediate useful byproduct, and when these projections have sections, one can describe the homotopy groups of $\conf{k}(M)$ in terms of the homotopy groups of the punctured manifold \cite{fh2}. This is the case of $M=\bbr^n$ where one obtains, in the case $k>1$,
$$\pi_*(\conf{k}(\bbr^n))\cong \bigoplus_{r=1}^{k-1}\pi_*(\bbr^n\setminus Q_r)\cong \bigoplus_{r=1}^{k-1}\pi_*((S^{n-1})^{\bigvee r})$$
Here $(S^{n-1})^{\bigvee r}$ is a wedge of $r$ copies of the sphere \cite{fh}. In particular, $\conf{k}(\bbr^n)$ is $(n-2)$-connected, $n\geq 2$. The structure of $\pi_*(\conf{k}(\bbr^n))$ as a graded Lie algebra, with the Whitehead
product providing the multiplication, is however non-trivial and is a good measure
of the twisting of the fibrations in \eqref{iteratedbundle} (see \S\ref{cellularmodels}). 

\noindent{\bf 2.1.2.} (Fundamental Group). When $\dim M=2$, and $M$ a positive genus surface, it is a remarkable fact that the fundamental group determines completely the homotopy type of the configuration spaces (see \S\ref{braidgroup}).
This is no longer true in higher dimensions $\dim M\geq 3$. Note that in this case, $\pi_1(\conf{k}(M))\cong (\pi_1(M))^k$, since removing submanifolds of codimension three or higher in a manifold does not affect fundamental groups. In this case also, $\pi_1 (C_k(M))\cong \pi_1(M)\wr \mathcal S_k$ (the \textit{wreath product} of $\pi_1(M)$ with $\mathfrak S_k$). This is by definition the semidirect product $\pi_1(M)^k\rtimes\mathfrak S_k$, where $\mathfrak S_k$ acts on $\pi_1(M)^k$ by permuting the factors. The split map $\mathfrak S_k\rightarrow\pi_1(C_k(M))$ comes from taking the induced map, at the level of $\pi_1$, of the inclusion $C_k(U)\hookrightarrow C_k(M)$, where $U\cong\bbr^n$ is a chart of $M$ \cite{imbo, kallel1}.

\subsection{The braid groups}\label{braidgroup} This section is about configuration spaces of two dimensional manifolds. Fadell and Neuwirth noticed early on the asphericity property\footnote{An aspherical space $X$ is a space with vanishing higher homotopy groups; $\pi_k(X)=0, k>1$. It is also called a $K(\pi,1)$ space, with the understanding that $\pi=\pi_1(X)$.} of configuration spaces of the plane and most topological surfaces. More precisely, if $M$ is a compact topological surface, then $\conf{k} (M-Q_m)$ are aspherical spaces if $m\geq 1$ and $k\geq 1$.
 If the surface is neither $S^2$ nor $\bbr P^2$, this result remains valid for $m=0$, that is for $\conf{k} (M)$. For instance, configuration spaces of positive genus surfaces $S$ are aspherical, and so are their unordered analogs $C_k(S)$, $g(S)\neq 0$. 
 
\noindent{\bf 2.2.1.} A path in $C_k(S)$ is a $k$-tuple of “worldlines” of identical particles that may move around each other but never coincide (at any given instant of time).  This is like a ``braid'' with $k$ strands in $S$. The \textit{surface braid groups} \cite{fox} are
\begin{equation}\label{braids}
B_k(S) := \pi_1(C_k(S))\ \ ,\ \ PB_k(S) =\pi_1(\conf{k} (S))
\end{equation}
When one further assumes that particles have internal structure, as in quantum statistics (e.g. solitons), one must also keep track of the rotation of the particles as one interchanges them. In this case, it is natural to work with the group $FB_k(S)$ of \textit{framed} braids on $S$ \cite{baez}. In all cases, the asphericity result of \cite{fn} is that $C_{k} (S)$ is a model for the classifying space $BB_k(S)$,
when $S$ has genus $g >0$ (i.e. other than $S^2$ and $\bbr P^2$). 
Since for any finite $K(\pi,1)$, $\pi$ cannot have elements of finite order by a classical result of P.A.Smith, it follows that for positive genus $S$, $B_k(S)$ is torsion-free. This is no longer the case for $\bbr P^2$ since $B_2(\bbr P^2)$ is the generalized quaternion group $Q_{16}$ which is a non-trivial extension of $\bbz_2$ by $D_8$ \cite{wang}. 

The group $B_k:=B_k(\bbr^2)=\pi_1(C_k(\bbr^2))$ is the very well-studied \textit{Artin's braid group} introduced by the named author in 1925 \cite{artin0} who then gave his famous presentation
$$B_k =\left\langle \sigma_1,\ldots, \sigma_k\ |\
\sigma_i\sigma_{i+1}\sigma_i = \sigma_{i+1}\sigma_i\sigma_{i+1}\ \hbox{and}\ \sigma_i\sigma_j=\sigma_j\sigma_i, |i-j|>1\right\rangle$$
The methods of \cite{artin0} were ``entirely intuitive'' (by the words of the author himself), but were later verified rigorously in \cite{artin} and
\cite{bohnenblust}.
A presentation of the torus braid group $B_k(T)$, $T\cong S^1\times S^1$, was given by Zariski \cite{zariski}. The ``simplest'' currently known presentation of surface braid groups for larger genus is probably due to P. Bellingeri \cite{bellingeri}, refining earlier presentations by González-Meneses, Scott, and Kulikov-Shimada.
The ``pure braid groups'', on the other hand, are the fundamental groups of the ordered configuration spaces $\conf{k} (S)$. Artin's pure braid group $PB_k$ is the fundamental group of $\conf{k}(\bbr^2)$ and fits in a group extension $0\rightarrow PB_k\rightarrow B_k\rightarrow\mathfrak S_k\rightarrow 0$, where the map $B_k\rightarrow\mathfrak S_k$ sends a braid to the permutation $\sigma$ of the set $\underline k$ that corresponds to sending the starting point $i$ of a strand to the endpoint of that strand $\sigma (i)$. We can describe the first groups. Here $PB_2=\bbz$ since $\conf{2}(\bbc)\simeq S^1$. For $PB_3$, we can use the following useful result: let $G$ be a Lie group with identity $e$, then the map $(g_1,\ldots, g_n)\mapsto (g_1, g_2g_1^{-1},\ldots , g_ng_1^{-1})$ gives a homeomorphism $\conf{k} (G)\cong G\times\conf{k-1} (G^*)$, where $G^*=G-\{e\}$. Iterating this twice in the case of $G=\bbc$ and $G=\bbc^*$ we obtain that
\begin{equation}\label{conf3}
\conf{3} (\bbc)\simeq S^1\times (S^1\vee S^1)
\end{equation}
and so $PB_3=\pi_1(\conf{3}(\bbc))\cong\bbz\times F_2$, where $F_2$ is the free group on $2$-letters. This group is non-abelian as are all $PB_k$, for $k\geq 3$.  
The homology of the braid group $B_k$ is worked out in {\bf 3.3}.

\noindent{\bf 2.2.2.} The study of the algebraic properties of the braid groups occupies a vast literature. 
It can be noted at this stage that if $G\subset\bbr^2$ is a connected planar graph with \textit{at least one trivalent vertex}, then the natural inclusion
$\conf{k} (G)\hookrightarrow\conf{k} (\bbr^2)$
induces a surjection at the level of fundamental groups, thus giving alternative presentations for Artin's pure braid group \cite{fredanthology}. It turns out that configuration spaces of graphs are aspherical and their fundamental groups $\pi_1$ are torsion-free (see \cite{ghrist1}, Corollary 2.4 and Theorem 2.5). These fundamental groups are the \textit{graph braid groups} which form a very active area of investigation (for a sampling of results see \cite{ghrist1}).

\noindent{\bf 2.2.3.} The representation theory of the braid groups $B_k$ and $PB_k$ is another major field of study that contributes fundamentally to quantum field theory via the work of T. Kohno, V. Drinfel'd, I. Marin, R. Lawrence, D. Krammer, S. Bigelow, A. Varchenko, D. Cohen, A. Suciu, and others (see \cite{abad}). Many such representations arise as \textit{monodromy representations} of bundles over configuration spaces. One  - faithful- such representation 
$B_k \longrightarrow \hbox{Aut}(F_k)$, with $F_k$ the free group on $k$-letters, was given by Artin in the 1924 Hamburg Abhandlung. It sends the generator $\sigma_i$ to the automorphism $\sigma_i(x_i) = x_{i+1}$ and $\sigma_i(x_{i+1}) = x^{-1}_{i+1}x_ix_{i+1}$.  
This can be seen as the monodromy representation of the bundle projection map
$\conf{k+1}(\bbc)/\mathfrak S_k\rightarrow C_k(\bbc)$, where the symmetric group acts freely by permuting the first $k$ points of a configuration in $\conf{k+1} (X)$. The fiber of this projection is $\bbc-Q_k$ and its monodromy homomorphism $\pi_1(C_k(\bbr^2))=B_k\rightarrow \hbox{Aut}(\pi_1(\bbr^2-Q_k)) = \hbox{Aut}(F_k)$ is Artin's representation \cite{fred2}.
Representations of $B_k$ in $\hbox{Aut}(F_k)$ provide group invariants of links \cite{wada}.
Other monodromy representations of the braid group include Hurwitz's representation which comes from Hurwitz covers (see \S\ref{geometry}, also \cite{westerland}),
unitary representations of the braid group that arise via the Hodge theory of cyclic branched coverings of $\bbp^1$ \cite{mcmullen}, or the Burau representation $B_n\rightarrow GL_n(\bbz [t,t^{-1}])$. For $n = 2, 3$ the Burau representation is known to
be faithful, for $n\geq 5$ it has non-trivial kernel (Bigelow), and for $n = 4$ the question of its faithfulness is still open. 

\noindent{\sc Example}: Artin and Magnus identified $B_k$ with the \textit{mapping class group} of all homeomorphisms of the punctured closed disk $D_k:=D-\{p_1,\ldots, p_k\}$, up to isotopies, fixing the boundary $\partial D$ pointwise \cite{meneses}. Its induced action on $\pi_1(D_k)\cong F_k$ is  the aforementioned Artin's representation.

Braid groups enter many areas of mathematics. 
A famous theorem of Alexander states that every classical link is equivalent to the closure of some braid, establishing a key link with knot theory. Important work of Fred Cohen and Jie Wu \cite{fredwu} relates the descending central series in $PB_k$ to the homotopy groups of $S^2$, thus linking with homotopy theory. The way braid groups and unordered configuration spaces intervene in homotopy theory is discussed in \S\ref{homotopytheory}. Finally Braid groups have been important as platforms for various protocols developed in the context of group-based cryptography. This discipline has gained significant relevance recently due to the need to design new crypto schemes that are invulnerable to quantum attacks \cite{ramon}.

\subsection{Homology  and Planetary Systems}\label{homology}
The homology of $\conf{k}(\bbr^n)$ is torsion free and finite dimensional, with non-zero homology in degrees multiple of $n-1$, $n>1$, and vanishing beyond degree $(k-1)(n-1)$.
Its Poincar\'e polynomial listed below gives immediately full knowledge of this graded abelian group 
\begin{equation}\label{poincseries}
P_{\hbox{\tiny Conf}_{k}(\bbr^n)}(t) = \prod_{1\leq j\leq k-1}(1+jt^{n-1})\in \bbn [t]\  \ ,\ \ \ n\geq 2
\end{equation}
The non-zero Betti numbers $\beta_{j(n-1)} = \hbox{rank}H_{j(n-1)}(\conf{k}(\bbr^n),\bbz)$ are the unsigned Stirling numbers of the first kind $\displaystyle \beta_{j(n-1)}= \begin{bmatrix}k\\ k-j\end{bmatrix}$. These are polynomials in $k$ of degree $2j$. In particular
$$\beta_{n-1} = {k(k-1)\over 2}\ \ ,\ \ \beta_{2(n-1)} = {(3k-1)k(k-1)(k-2)\over 24}\ \ ,\ \ 
\beta_{3(n-1)} = {k^2(k-1)^2(k-2)(k-3)\over 48}
$$
The first non-trivial homology group in degree $n-1$ is free of rank $\displaystyle k\choose 2$, while the
top homology is in degree $(n-1)(k-1)$ of rank $\beta_{(n-1)(k-1)} = (k-1)!$. The \textit{total rank} of this graded homology
is $\sum \beta_i = k!$.

\noindent{\bf 2.3.1.} (LCS formula) When $n=2$, $\conf{k}(\bbr^2)=\conf{k}(\bbc)$ is the complement of a complex hyperplane arrangement, and its Poincar\'e series is related to the descending central series of its fundamental group $PB_k$ \cite{falk, kohno2}. Given a finitely generated 
group $G$, we can set $G_1=G$ and inducively
$G_i = [G_1,G_{i-1}]$ the subgroup generated by the commutators of elements in $G$ and $G_{j-1}$. Then $G(i):= G_i/G_{i+1}$ is a finitely generated abelian group, so we denote $\phi (i) = \hbox{rank}\ G(i)$. We let
$G=PB_k$. In (\cite{falk}, Theorem 4.1, \cite{kohno2}), the authors establish the so-called \textit{LCS formula}
$$P_{\hbox{\tiny Conf}_{k}(\bbc)}(-t) = \prod_{j=1}^\infty (1-t^j)^{\phi (j)}$$
Similar results are true for fundamental groups of complements of other ``fiber-type'' arrangements (more next paragraph). The number $\phi(1)$ is the number of hyperplanes in the arrangement, which is ${k\choose 2}$ in our case. Techniques of \cite{falk} enabled the authors to show subsequently that the pure braid group is residually nilpotent, an important result for the theory of knot invariants of finite type (see \S\ref{finitetypeinv}).

\noindent{\bf 2.3.2.} The simplest way to determine 
 $H_*(\conf{k}(\bbr^n);\bbz)$, and its Poincar\'e series \eqref{poincseries}, is through iterated use of the \textit{Serre spectral sequence}\footnote{More on spectral sequences in \S\ref{spectral}, with a glimpse at what they actually are. To every fibration, there is associated a spectral sequence that helps obtain the homology of the total space from that of the base and fiber.} of the Fadell-Neuwirth fibration $\conf{k}(\bbr^n)\rightarrow\conf{k-1}(\bbr^n)$ \eqref{iteratedbundle}, whose fiber has the homotopy type of a bouquet (i.e. one point union) of spheres $\left(S^{n-1}\right)^{\vee {k-1} }$. For $n\geq 3$, the base is simply connected and the Serre spectral sequence degenerates for dimensional reasons (both base and fiber have homology in degrees a multiple of $n-1$, starting with $\conf{2}(\bbr^n)\simeq S^{n-1}$). This yields \eqref{poincseries} immediately and inductively. For $n=2$, the action of $\pi_1$ on the homology of the fiber is trivial, and the spectral sequence degenerates at $E^2$ as well \cite{fredbible}. The existence of these iterated fibrations, with fibers homeomorphic to punctured complex vector spaces is what it means to be a ``fiber-type'' arrangement \cite{falk}. Notice that the Fadell-Neuwirth fibrations \eqref{iteratedbundle} are nontrivial, even though the additive structure of the homology does not see the twisting (but the cohomology ring and the homotopy group Lie algebra do, see \S\ref{cohomology} and \S\ref{loopspace}). Another beautiful way to get to $H_*(\conf{k}(\bbr^n);\bbz)$, without spectral sequences, is via the Goresky-MacPherson formula for the complement of subspace arrangements (see their remarkable book, Part III, Chapter 1 \cite{goreskymacpherson}). An approach to the Goresky-MacPherson computation, using the very elegant theory of posets \cite{wachs}, is detailed in the second part of this user's guide and applied to a more general class of configuration spaces.

\noindent{\bf 2.3.3.} (Planetary System) The additive generators of the homology groups of $\conf{k}(\bbr^n)$ have a beautiful and handy geometric description in terms of \textit{toric} classes. This is to say that there is a set of generators of $H_{j(n-1)}(\conf{k}(\bbr^n);\bbz)$, for all $1\leq j<k$, such that any such generator is the top homology class of an embedded torus $(S^{n-1})^j\hookrightarrow\conf{k}(\bbr^n)$. This description takes the name of \textit{planetary system} \cite{fred2, paolo1, dev}. 
We explain this and indicate some  consequences. 

The starting point are the spherical generators $a_{ij}\in H_{n-1}(\conf{k}(\bbr^n),\bbz)$ which are obtained as follows. Fix $(q_1,\ldots, q_k)\in\conf{k}(\bbr^n)$, with $|q_i-q_j|>1$, and consider for each pair $r,s$ the map 
\begin{eqnarray}\label{aijmap}
a_{rs} : S^{n-1}\lrar \conf{k}(\bbr^n)\ \ ,\ \ 
\zeta\longmapsto (q_1, q_2,\cdots, q_{r-1}, q_s+\zeta,q_{r+1},\ldots, q_{k})
\end{eqnarray}
This is a free map of the sphere, but it can be made based. It describes the point $q_r$ of the configuration as a point orbiting around the point $q_s$, the other points being fixed away from this sphere \cite{fh, paolo1}.
The image of the fundamental class $(a_{rs})_*[S^{n-1}]$ (also denoted by $a_{rs}$ simply) is a non-trivial class. To better understand \eqref{aijmap}, we can replace it up to homotopy by another map $$\phi_{rs} (x) = (q_1,\ldots, q_r=-x,\ldots, q_s=x,\ldots, q_k),\ x\in S^{n-1}$$ 
where the $q_i$'s, $i\neq r,s$ are fixed, pairwise distinct and lying outside the unit sphere. This describes the configuration's $r$-th and $s$ points as antipodal points rotating around a fixed center, thus describing an embedded sphere \cite{fred5, dev}. Since $\phi_{rs} (-x) = (r,s)\phi_{rs} (x)$, where $(r,s)\in\mathfrak S_k$ is the transposition interchanging $r$ and $s$, it follows that in homology,
$a_{rs} = (-1)^na_{sr}$, where $(-1)^n$ is the degree of the antipodal map of $S^{n-1}$. It turns out that the generators 
$a_{rs}$, for $k\geq r>s\geq 1$, form a basis of $H_{n-1}(\conf{k}(\bbr^n),\bbz)$. 

Next, one starts bootstrapping and constructing higher dimensional classes by considering the locus of all configurations $(x_1,\ldots, x_k)$ where some particle $x_i$ in the configuration rotates around another $x_j$, and maybe the cluster $\{x_i,x_j\}$ rotates around $x_s$, etc. Such loci are always products of spheres (i.e. higher dimensional tori). A representation of these points rotating around each other, and generating spheres, is given in terms of arrows in \cite{paolo1} (see Fig. \ref{fig:4T}) or in terms of \textit{rooted planar trees} as done in \cite{dev}. In this latter handy description, the class $a_{12}$, for example, is represented by a labeled rooted binary tree $(\tree{1}{2})$, and the forest with two binary trees $(\tree{1}{2})(\tree{3}{4})$ represents the class of an embedded $S^{n-1}\times S^{n-1}\lrar\conf{4}(\bbr^n)$, describing the locus of particle $1$ rotating around $2$, and $3$ around $4$.
So to each tree $T$, or to a forest, corresponds a submanifold of $\conf{k}(\bbr^n)$ which is the homeomorphic image of a product of spheres $(S^{n-1})^{|T|}\hookrightarrow\conf{k}(\bbr^n)$, where $|T|$ is the number of internal vertices of $T$, and whose top orientation class is a non-zero class in $H_{|T|(n-1)}(\conf{k}(\bbr^n);\bbz )$. Fig. \ref{orbit} represents the torus whose top class is described by the tree 
$T =
\begin{xy}
  (1.5,1.5); (3,3)**\dir{-}, 
  (0,3); (3,0)**\dir{-}; 
  (7.5,4.5)**\dir{-},        
  (3,0); (3,-1.5)**\dir{-}, 
  (3,6); (6,3)**\dir{-},
  (6,6); (4.5,4.5)**\dir{-},
  (0,4.2)*{\scriptstyle 2},
  (3,4.2)*{\scriptstyle 6},
  (3,7.2)*{\scriptstyle 1},
  (6,7.2)*{\scriptstyle 7},
  (7.5,5.7)*{\scriptstyle 3},
\end{xy}$ (as in \cite{dev}, Fig.1).
Pair of points make up a sphere (in this tree notation, on should picture any pair of points as a pair of antipodal points orbiting around a center), and the disposition of points resembles a planetary system indeed.

\begin{figure}[htb]
\begin{center}
\epsfig{file=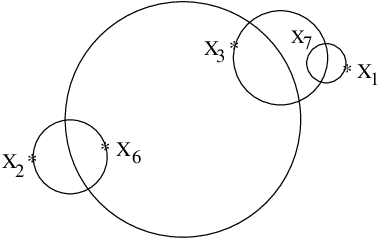,height=1in,width=1.6in,angle=0.0}
\caption{A submanifold homeomorphic to $(S^{n-1})^4$ representing the class of $T$ above.}
\label{orbit}
\end{center}
\end{figure}
Every $k$-forest $F$ (i.e. a forest on $k$ leaves) contributes a homology generator which is the image of the fundamental class $(S^{n-1})^{|F|}\hookrightarrow\conf{k}(\bbr^n)$.
The classes obtained from these rooted binary trees and forests with at most $k$ leaves (where each integer from $1$ to $k$ labels exactly one leaf) exhaust all of the homology of $\conf{k}(\bbr^n)$. 
A most convenient way to express this is to associate to any homology generator, given by a forest on $k$ leaves, an expression in variables $x_1, \ldots, x_k$, so that rotating particles around each other is depicted by bracketing the variables, and the bracket has homological degree $n-1$.  For example, the generator $a_{12}=\tree{1}{2}$ is represented by the class $\left[\tree{1}{2}\right]:=[x_1,x_2]\in H_{n-1}(\conf{2}(\bbr^n))$, 
while the generator $T =
\begin{xy}
  (1.5,1.5); (3,3)**\dir{-}, 
  (0,3); (3,0)**\dir{-}; 
  (7.5,4.5)**\dir{-},        
  (3,0); (3,-1.5)**\dir{-}, 
  (3,6); (6,3)**\dir{-},
  (6,6); (4.5,4.5)**\dir{-},
  (0,4.2)*{\scriptstyle 2},
  (3,4.2)*{\scriptstyle 6},
  (3,7.2)*{\scriptstyle 1},
  (6,7.2)*{\scriptstyle 7},
  (7.5,5.7)*{\scriptstyle 3},
\end{xy}$ gives rise to the class
 $[T]=[[x_2,x_6], [[x_1,x_7],x_3]]$ in $H_{4(n-1)}\conf{7}(\bbr^n)$. This class is an example of an \textit{iterated bracket}.   Concatenating trees comes down to \textit{multiplying} such bracketed expressions. For example
$
\tree{4}{5}
 \begin{xy}
  (1.5,1.5); (3,3)**\dir{-}, 
  (0,3); (3,0)**\dir{-}; 
  (7.5,4.5)**\dir{-},        
  (3,0); (3,-1.5)**\dir{-}, 
  (3,6); (6,3)**\dir{-},
  (6,6); (4.5,4.5)**\dir{-},
  (0,4.2)*{\scriptstyle 2},
  (3,4.2)*{\scriptstyle 6},
  (3,7.2)*{\scriptstyle 1},
  (6,7.2)*{\scriptstyle 7},
  (7.5,5.7)*{\scriptstyle 3},
\end{xy}$ corresponds to the class $[x_4,x_5]\cdot [[x_2,x_6], [[x_1,x_7],x_3]]\in H_{5(n-1)}(\conf{7}(\bbr^n),\bbz)$. 
This multiplication is defined for factors with different labels, and it is graded commutative, in the sense that if $[T_1\cdot T_2]$ is a toric class in $H_*(\conf{k}(\bbr^n))$, then
$[T_1\cdot T_2]=[T_1]\cdot [T_2] = (-1)^{|T_1||T_2|}[T_2]\cdot [T_1]$.
The homology can be stated succinctly as follows.

\begin{quote}\textit{{\bf Theorem 2.1}: $H_*(\conf{k}(\bbr^n);\bbz)$ is the free graded $\bbz$-module spanned by products of iterated brackets on
the $k$ variables $x_1,\ldots, x_k$ (of degree $0$), each appearing
exactly once,
subject to the
Jacobi relations, antisymmetry relations $[x_i,x_j]= (-1)^n[x_j,x_i]$, and to graded commutativity. In particular, we have the relation
$$\hbox{(Jacobi)}\ \ \ \ \  \ \ \ \ \ [[x_i,x_j], x_k]+ [[x_k,x_i], x_j]+ [[x_j,x_k], x_i]=0\ ,\ i\neq j\neq k\neq i$$ 
The homological degree of a homogeneous term with $j$ left (or right) brackets is $j(n-1)$.}
\end{quote}
\vskip 5pt

In this theorem, no multiplication occurs inside an iterated bracket, and no $x_i$ can be repeated.  The displayed Jacobi relation derives from the geometric fact that there is an explicit manifold bounding the three tori
 \begin{xy}   
   (1.5,1.5); (3,3)**\dir{-}, 
   (0,3); (3,0)**\dir{-};
   (6,3)**\dir{-},   
   (3,0); (3,-1.5)**\dir{-}, 
   (-.4,4.2)*{\scriptstyle i}, 
   (3.2,4.2)*{\scriptstyle j},
   (6.8,4.2)*{\scriptstyle k}
 \end{xy}, \begin{xy}   
   (1.5,1.5); (3,3)**\dir{-}, 
   (0,3); (3,0)**\dir{-};
   (6,3)**\dir{-},   
   (3,0); (3,-1.5)**\dir{-}, 
   (-.4,4.2)*{\scriptstyle j}, 
   (3.2,4.2)*{\scriptstyle k},
   (6.8,4.2)*{\scriptstyle i}
 \end{xy} and \begin{xy}   
   (1.5,1.5); (3,3)**\dir{-}, 
   (0,3); (3,0)**\dir{-};
   (6,3)**\dir{-},   
   (3,0); (3,-1.5)**\dir{-}, 
   (-.4,4.2)*{\scriptstyle k}, 
   (3.2,4.2)*{\scriptstyle i},
   (6.8,4.2)*{\scriptstyle j}
 \end{xy}
\cite{dev}. It is direct to write a homology basis (see \cite{paolo1}, Lemma 6). We write the dual basis in \eqref{admissible}. In Part II of this work, we give an alterate derivation of Theorem {\bf 2.1} using poset topology, and explain why a basis of generators is given by all spanning forests in the complete graph $K_k$ with "no broken circuits" (see Theorem \ref{generators}, Part II).

\noindent{\sc Example}: $H_*(\conf{2}(\bbr^n);\bbz)$ has generator $x_1x_2$ in degree $0$ and $[x_1,x_2]$ in degree $n-1$. The generators of $H_*(\conf{3}(\bbr^n))$ are $x_1x_2x_3$ in degree $0$, $[x_1,x_2], [x_1,x_3]$ and $x_1[x_2,x_3]$ in degree $(n-1)$, and $[x_1, [x_2,x_3]]$, $[x_3,[x_1,x_2]]$ in degree $2(n-1)$.

When assembled together over all non-negative integers $k$, the homology groups $H_*(\conf{k}(\bbr^n),\bbz)$ in Theorem 2.1 form an \textit{algebraic operad}, and concatenating trees becomes an actual associative product (see \S\ref{operads} next). Under this product, $x_{k+1}\cdot H_*(\conf{k}(\bbr^n))$ is the image of the split embedding $H_*(\conf{k}(\bbr^n))\hookrightarrow H_*(\conf{k+1}(\bbr^n))$.
 
\noindent{\sc Example}: The top homology $H_{(k-1)(n-1)}(\conf{k}(\bbr^n),\bbz )$ is generated by trees with $k$ leaves, or by the classes $[[\cdots [x_{\sigma (1)},x_{\sigma (2)}] \cdots ] x_{\sigma (k)}]$, where $\sigma$ runs over the elements of the symmetric group $\mathfrak S_k$. This is by definition, the abelian subgroup Lie$(k)$ of the free Lie algebra $L(V_k)$ on $V_k=\bigoplus_k\bbz$, the free abelian group with basis $\{x_1,\ldots, x_k\}$. Lie$(k)$ is naturally a module over the group ring $\bbz [\mathfrak S_k]$ and, as an abelian group, it is isomorphic to $\bigoplus_{(k-1)!}\bbz$. It is well-known (see \cite{fred5}, Theorem 6.1) that \textit{as $\bbz[\mathfrak S_k]$-modules}: if $n$ is odd, $H_{(k-1)(n-1)}(\conf{k}(\bbr^n),\bbz)$ is isomorphic to Lie$(k)$, and if $n>0$, $n$ even, it is isomorphic to Lie$(k)\otimes_\bbz\bbz(-1)$. It turns out that these modules Lie$(j)$ for
$j\leq k$, induced up, describe the entire structure of $H_*(\conf{k}(\bbr^n))$ as a graded $\mathfrak S_k$- representation (see \S\ref{stability}). 

\noindent{\bf 2.3.4.} (Stable Splitting) The geometric fact that the homology of $\conf{k}(\bbr^n)$ is generated by toric classes (or fundamental classes of embedded tori) implies that $\conf{k}(\bbr^n)$ must split into a product of wedges of spheres after a single suspension.  This splitting takes the form\begin{equation}\label{singlesplit}
\Sigma \conf{k}(\bbr^n)\simeq  \bigvee_{j=1}^{k-1}
\left(S^{j(n-1)+1}\right)^{\vee \left[{k\atop k-j}\right]}\ \ \ \ ,\ \ \ n\geq 2
\end{equation}
The term $\left(S^{r}\right)^{\vee \left[{k\atop k-j}\right]}$ means a wedge of $\left[{k\atop k-j}\right]$ (stirling number) copies of this sphere $S^r$.
A generalization of this splitting to the case of no-$\ell$ equal configuration is in \cite{dobtur}, and to the case of \textit{chromatic configuration spaces}, as defined in \S\ref{variants}, is worked out in part II of this user's guide, where the number of spheres in the splitting is given in terms of the chromatic numbers of the graph (see Part II, Theorem \ref{main}).
Note that more generally, $\conf{k}(M\times\bbr)$, for $M$ a connected manifold, also splits after one suspension as a wedge of suspended Thom spaces \cite{taylor}. Manifolds of the form $M\times\bbr$ are special, so-called ``i-acyclic'', and their configuration spaces are relatively well-understood (see \S\ref{compactsupport}).

\subsection{Operads and homology operations}\label{operads}

An \textit{operad} in any symmetric monoidal category is, by definition, a collection of objects $\mathcal O:=\{\mathcal O(k)\}_{k\geq 0}$, together with ``structure morphisms'' \begin{equation}\label{structuremaps}\mathcal O_n(r)\times \mathcal O_n(k_1)\times\cdots\times\mathcal O_n(k_r)\rightarrow\mathcal O_n(k_1+\cdots+k_r),
\end{equation}
which satisfy a short list of associativity and equivariance relations (see \cite{fresse} or any reference on operads).
An action of an operad on an object $A$ consists of a family of maps
$\mathcal O(k)\times A^k\lrar A$ which again verify certain compatibility conditions with the structure morphisms \eqref{structuremaps} \cite{fresse, idrissi2}. When such an action exists, we say that $X$ is a $\mathcal O$-space (or an \textit{$\mathcal O$-algebra}). As it is for groups, an operad is understood by the way it acts on objects.

\noindent{\bf 2.4.1.} Historically, the starting point of operads begins with Euclidean configuration spaces and loop spaces. More precisely, it was a simple but impactful idea of Boardman and Vogt to thicken the points of configurations to closed balls in $\bbr^n$, whose interiors are pairwise disjoint, and obtain this way a collection ${\mathcal D}_n = \{\mathcal D_n(k)\}_{k\geq 1}$ of ``little disks'' in $\bbr^n$, where $\mathcal D_n(k)$ is the collection of all $k$ such disks; a space homotopy equivalent to $\conf{k}(\bbr^n)$. The little disks can be composed by shrinking their size and inserting them into other disks (an operation that cannot obviously be done with points).  Insertions of little disks into other little disks, after resizing, give rise to  the operad structure maps
$$\mathcal D_n(r)\times \mathcal D_n(k_1)\times \cdots \times \mathcal D_n(k_r)\longrightarrow \mathcal D_n(k_1+\cdots +k_r)$$

By passing to homology, we obtain the collection
 $\{H_*(\conf{k}(\bbr^n)\}_{k\geq 0}$ which is an operad in the category of  
 $\bbz$-graded abelian groups (this is the case because the homology is torsion-free). We will describe this operad, which is also written $e_n:=H_*(\mathcal D_n)$ in the literature, by describing the algebras over it. The generators of $H_0\mathcal D_n(0)$,
$H_0\mathcal D_n(2)$, and of $H_{n-1}(\mathcal D_n(2))$ are, respectively, the elements $1$, $x_1\cdot x_2$, and $[x_1, x_2]$ of the operad. If $A$ is an algebra over $e_n$, then the structure map
$H_0(\mathcal D_n(2))\otimes A\otimes A\rightarrow A$ sends $(x_1\cdot x_2 ,\alpha,\beta)\mapsto \alpha\cdot\beta$ (the multiplication), and sends $([x_1, x_2],\alpha,\beta)\mapsto [\alpha,\beta]$. This operation of replacing entries is called ``grafting''. We thus see that
$A$ is a graded associative algebra with bilinear pairing $\cdot : A_p\otimes A_q\rightarrow A_{p+q}$ of degree $0$ and a Lie bracket $[-]: A_p\otimes A_q\rightarrow A_{p+q+n-1}$, of degree $n-1$, satisfying the following relations
\begin{eqnarray*}
&&xy = (-1)^{|x||y|}yx\\&&[x,y] = -(-1)^{(|x|+n-1)(|y|+n-1)}[y,x]\\
&&[x,yz] = [x,y]z + (-1)^{(|x|+n-1)|y|}y[x,z]\ \ \ \hbox{Leibniz rule}
\end{eqnarray*}
Bracket expressions which may include multiplications $\cdot$ within brackets, such as $[x_1,[x_2,x_3\cdot x_4]]$ for example, can always be reduced to expressions associated to forests after Leibniz rule. The reason of why this rule exists in $e_n$ is simple to explain: $[x_1, x_2\cdot x_3]$ means, in our planetary system description, that planet $1$ rotates around $2$ and $3$ simultaneously, thus implying that the class of such a map at the level of homology or homotopy group, is the class of $1$ rotating around $2$ summed up with the class obtained by rotating $1$ around $3$ (see Fig. \ref{fig:4T} and explanation therein).
The algebra $A$ we just presented is named \textit{a Poisson $n$-algebra} (or also a Gerstenhaber $n$-algebra \cite{berglund}).
Poisson $2$-algebras are the \textit{Gerstenhaber algebras}\footnote{In physics references, the degree is taken in negative cohomological degree, so the Poisson bracket has degree $1-n$, and this is $-1$ for Gerstenhaber \cite{kontsevich}.}.

\begin{quote}\textit{{\bf Theorem} 2.2 \cite{fredbible}. Algebras over the operad $e_n=\bigoplus_{k\geq 0}H_*(\mathcal D_n(k))$,  $n>1$, are $n$-Poisson algebras. An algebra over $e_1$ is just an associative algebra. }
\end{quote}
\vskip 5pt

In particular, $e_n$ itself is $n$-Poisson
and $e_2$ is Gerstenhaber.
Poisson algebras appear naturally in Hamiltonian mechanics and in deformation quantization \cite{kontsevich}.

\noindent{\bf 2.4.2.}  A paramount example of a $\mathcal D_n$-space (equivalently a $\mathcal D_n$-algebra) is the n-fold loop space $\Omega^nX$ of a based space $X$.
 The converse is almost true. A major \textit{recognition principle} of Peter May asserts that a connected $\mathcal D_n$-space whose $\pi_0$ is a group is weakly homotopy equivalent to an $n$-fold loop space \cite{may1}. By passing to homology \textit{with field coefficients}, we see that $e_n$ acts on $H_*(\Omega^nX)$, and thus parameterizes homology operations on these loop spaces. For example, the bracket $[x_1,x_2]$ generates an operation
$H_i(\Omega^nX)\otimes H_j(\Omega^nX)\rightarrow H_{i+j+n-1}(\Omega^nX)$; the Poisson bracket which is known in this context as the \textit{Browder bracket}.  The structure of the operations with mod-$2$ coefficients was first studied by Kudo and Araki, and for odd $p$ by Dyer and Lashoff. Browder used them effectively to compute the homology of $n$-fold loop spaces of spheres $\Omega^nS^{n+k}$ modulo $2$ \cite{browder}. Rationally, there is a very aesthetic answer for $n$-connected spaces spelled out in \cite{fredbible, berglund}. It takes the form of an isomorphism of $n$-Poisson algebras
$$H_*(\Omega^nX,\bbq) = \Lambda (s^{-n}\pi_*(X)\otimes\bbq)$$
where $s^{-n}\pi_*(X)\cong \pi_*(\Omega^nX)$ is the $n$-fold desuspension which shifts down the degree by $n$, and the Lie bracket on the right hand side is induced by Whitehead products on the homotopy groups of $X$. Here $\Lambda V$ means the free commutative graded algebra on the graded vector space $V$ (see {\bf 7.0.1}).

\noindent{\bf 2.4.3.} The next development (Deligne conjecture) is very well explained in \cite{kontsevich}, with major expansions in \cite{fresse}. It is about operadic structures on chains of $\mathcal D_n$ and their algebras. If $\mathcal O=\{O(k)\}_{k\geq 0}$ is a topological operad, then the collection of complexes $\{\hbox{Chains}(O(k))\}_{k\geq 0}$ has a natural operad structure in the category of complexes of abelian groups. The compositions in this operad are defined using the external tensor product of \textit{cubical chains}. A main result of Gerstenhaber is that for an associative algebra $A$, the Hochschild cohomology $H^*(A,A)$ over a field carries the structure of Poisson $2$-algebra (i.e. a Gerstenhaber algebra), so it is an $e_2$ algebra.  In 1993, Pierre Deligne conjectured that there is an algebra structure at the level of cochains already; i.e. there exists a natural action of the operad $\hbox{Chain}(\mathcal D_2)$ on the Hochschild complex $C^*(A,A)$ for any associative algebra $A$. As stated in \cite{kontsevich}, ``The story of this conjecture is quite dramatic'' with several proofs (and corrected proofs) in the literature.  



\subsection{The Arnold-Cohen computation}\label{cohomology}
The cohomology \textit{ring structure} of Euclidean configuration spaces is  undoubtedly one of the most fundamental and recognizable results in the field. V.A. Arnold \cite{arnold} was first to compute this cohomology ring in the planar case. 
In 1970,  one year after Arnold and independently, Fred Cohen computed in his thesis \cite{fredbible} the cohomology ring $H^*(\conf{k}(\bbr^n),\bbz)$ for all $n\geq 2$. The Arnold-Cohen result takes the following form. First, we identify the generators of this cohomology algebra. For each pair $(i,j), i\neq j$, consider the map \begin{equation}\label{alphaij}
\alpha_{ij}: \conf{k}(\bbr^n)\rightarrow S^{n-1}\ \ \hbox{sending}\ \ \displaystyle (x_1,\ldots, x_k)\mapsto {x_i-x_j\over |x_i-x_j|}
\end{equation}
The pullback of the orientation class in degree $n-1$ is a class also denoted $\alpha_{ij}\in H^{n-1}(\conf{k}(\bbr^n),\bbz)$. One can check that the composite $\xymatrix{ S^{n-1}\ar[r]^{a_{ij}\ \ \ }&\conf{k}(\bbr^n)\ar[r]^{\ \ \ \alpha_{rs}}&S^{n-1}}$, for $i>j$, $r>s$, has degree $1$ if $i=r,j=s$, and is otherwise of degree $0$. This means that the class $\alpha_{ij}$ is dual to the basis generator $a_{ij}$ defined in \eqref{aijmap}.  Remarkably, these classes turn out to generate all of the cohomology ring under cup product. As a graded additive group, one starts by showing that $H^*(\conf{k}(\bbr^n),\bbz)$ is the free graded abelian group generated by \textit{admissible} monomials 
\begin{equation}\label{admissible}\alpha_{i_1j_1}\alpha_{i_2j_2}\cdots\alpha_{i_rj_r}\ \ \ ,\ \ i_s>j_s\ \ \hbox{for}\ \ 1\leq s\leq r \ \  \hbox{and}\ i_1<\cdots < i_r
\end{equation}
The ring structure has obvious relations. Since we are pulling back the class of a sphere, $\alpha_{ij}^2=0$ for all $i\neq j$, and the same orientation argument as in homology gives that $\alpha_{ij} = (-1)^{n}\alpha_{ji}$. The last key set of relations consists of three-term relations of the form $\alpha_{rt}\alpha_{rs} = \alpha_{st}(\alpha_{rs} - \alpha_{rt})$ for $r>s>t$.
This three-term relation (or ``Arnold-Cohen relation'') can be rewritten in the following \textit{cyclic} notation
\begin{equation}\label{arnold}
\alpha_{ij}\alpha_{jk}+\alpha_{ki}\alpha_{ij} + \alpha_{jk}\alpha_{ki} = 0\ \ \ \ \hbox{for}\ \ i\neq j\neq k\neq i
\end{equation}
It is a beautiful fact that the dual to the Arnold-Cohen relation given in cohomology is the Jacobi relation in homology. 
The main theorem can now be stated below:

\begin{quote}\textit{{\bf Theorem 2.3} \cite{fredbible}: As an algebra, $H^*(\conf{k}(\bbr^n),\bbz)$ is generated by the $\alpha_{ij}$, $1\leq j<i\leq k$, $\deg\alpha_{ij}=n-1$, subject to the following complete set of relations: associativity, graded commutativity, and the quadratic relations
\begin{equation}\label{quadratic1}
\left\{
\begin{array}{l}
\alpha_{ij}^2=0\\
\alpha_{ij}\alpha_{jk}+\alpha_{jk}\alpha_{ki}+\alpha_{ki}\alpha_{ij}=0 \ \ \ \ \ \ \ \hbox{(Arnold-Cohen relation)}
\end{array}
\right.
\end{equation}
The action of the symmetric group $\mathfrak S_k$ on this cohomology is generated by $$\sigma(\alpha_{ij}) = \alpha_{\sigma (i)\sigma (j)}\ 
\hbox{if}\ \sigma(i) > \sigma (j)\ \ \hbox{and}\ \
\sigma(\alpha_{ij}) =(-1)^n\alpha_{\sigma(j)\sigma(i)}\ \hbox{if}\ \sigma (i) < \sigma (j)$$
Finally, $H^*(\conf{k}(\bbr^n);\bbq)=H^*(\conf{k}(\bbr^n);\bbz)\tensor\bbq$ is a Koszul algebra \cite{priddy}.}
\end{quote}
\vskip 5pt

We explain the last statement: a graded commutative algebra $A$ is \textit{Koszul} if it is a quadratic algebra (that is generated by elements $x_i$ modulo certain quadratic
relations $\sum_{i,j}c_{ij} x_ix_j = 0$) such that $Tor^A_{i,j}(\bbq,\bbq)=0$ for $i\neq j$ \cite{berglund, priddy}. Because $H^*(\conf{k}(\bbr^n);\bbq)$ is quadratic and has a PBW-basis consisting of all monomials $\alpha_{i_1j_1}\cdots 
\alpha_{i_rj_r}$ where $i_1<\ldots < i_r$ and $i_p < j_p$, 
for all $p$, it is Koszul (\cite{priddy}, Theorem 5.3). For interesting consequences, see \S\ref{loopspace}. 

The generators $\alpha_{ij}$ are referred to as ``tautological classes'' in the physics literature  \cite{bott}. Based on Theorem 2.1, the Arnold-Cohen relation \eqref{quadratic1} is quickly derived as follows. The homology being torsion-free \S\ref{homology}, its Hom-dual is $H^*(\conf{k}(\bbr^n);\bbz )$.  Assume wlog that $n=2$ with one dimensional generators $\alpha_{21}, \alpha_{31}$ and $\alpha_{32}$. By \eqref{conf3}, $H_2(\conf{3}(\bbr^2))\cong\bbz^2$ has only two generators, and so does the cohomology. This means that
$\alpha_{31}\alpha_{32} = a\alpha_{21}\alpha_{31}+b\alpha_{21}\alpha_{32}$, for some $a,b\in\bbz$. But there is exactly one choice of $a,b$ which is consistent with the $\mathfrak S_3$-action: $a=-1$ and $b=1$. 

Note that the action of the Steenrod algebra on $H^*(\conf{k}(\bbr^n),\bbf_p)$ is trivial (\cite{fredbible}, Proposition 7.8). 
We record as well that $\bigoplus_{k\geq 0} H^*(\conf{k}(\bbr^n))=H^*(\mathcal D_n)$ has a graphical interpretation from which the structure of a \textit{cooeperad} with cocomposition maps can be made explicit (\cite{idrissi}, Theorem 5.69). 

\noindent{\sc Example}: As a simple illustrative application of Theorem 2.3, show that the Fadell-Neuwirth fibration $\conf{3}(\bbr^n)\rightarrow\conf{2}(\bbr^n)$, with fiber $S^{n-1}\vee S^{n-1}$ cannot be trivial if $n$ is odd. Indeed, the two \textit{algebras} given by
$H^*(S^{n-1}\times (S^{n-1}\vee S^{n-1});\bbz)=\bbz[a,b,c]/_{a^2,b^2,c^2,bc}$ and $H^*(\conf{3}(\bbr^n);\bbz) = \bbz[a,b,c]/_{a^2,b^2,c^2,ab+bc+ca}
$, $|a|=|b|=|c|=n-1$, are not isomorphic if $n$ odd (algebra exercise). A much harder result of Massey \cite{massey1} shows that this bundle is trivial, i.e. 
$\conf{3}(\bbr^n)\simeq S^{n-1}\times (S^{n-1}\vee S^{n-1})$, if and only if $n=1,2, 4$ or $8$ (see also \cite{fh}).

\subsection{Formality}\label{formality} There is a very important nuance between the computations of Cohen and those of Arnold of the cohomology ring of Euclidean configuration spaces. Arnold's computation in the case of $\bbr^2$ has deeper implications since it works at the chain level and constructs a quasi-isomorphism between $H^*(\conf{k}(\bbc),\bbc))$ and the algebra of complex differential forms $\Omega^*_\bbc (\conf{k}(\bbc))$ sending $a_{ij}$ to $d(z_i-z_j)/(z_i-z_j)$. 
This result was extended to complements $\mathcal M(\mathcal A)$ of any complex hyperplane arrangement $\mathcal A$ by Brieskorn who gives an embedding of $H^*(\mathcal M(\mathcal A))$ into the de Rham complex of $\mathcal M(\mathcal A)$.

Constructing a quasi-isomorphism of differential graded algebras (dga's) between cochains and cohomology (with zero differential) is a very strong property and is referred to as (strong) formality.
In general, one defines formality over any coefficient ring $R$  by requiring the algebra of $R$-valued singular cochains $C^*(X, R)$ to be connected to its cohomology $H^*(X, R)$ (with trivial differential) by a zig-zag of homomorphisms of differential graded associative $R$-algebras 
$$H^*(X, R)\longleftarrow \cdots \longrightarrow C^*(X,R) $$
inducing isomorphisms in cohomology (i.e. quasi-isomorphisms). If $R=\bbf$ is a field then this property depends only on the characteristic of $k$. For spaces of finite type with torsion-free homology, $\bbz$-formality is universal as
it implies $R$-formality for any ring $R$ \cite{paolo2}. 

\noindent{\bf 2.6.1.} In characteristic $0$, a manifold $X$ is formal if there exists \textit{a zig-zag } of quasi-isomorphisms between the algebra of de-Rham forms $\Omega^*_{dR}(X)$ and its real cohomology $H^*(X, \bbr )$ with trivial differential. For a general (nilpotent) topological space $X$, rational formality means a zigzag of quasi-isomorphisms with a commutative differential algebra model of $X$, like the Sullivan-deRham algebra $A_{PL}(X)$ or $\Omega_{PL}(X)$ of \textit{piecewise linear differential forms} on $X$, which one can view as the ``de Rham algebra for non-manifolds'' (also over $\bbq$ instead of $\bbr$). Rational formality implies that it is possible
to extract any rational homotopy invariant, such as the rational homotopy groups,
from the cohomology algebra. However, doing this in practice entails the non-trivial algebraic problem of constructing a minimal model for the cohomology (see \S\ref{models}, and \cite{voronov, idrissi}). 

\noindent{\sc Example}: Spheres $S^n$ are formal (\cite{idrissi2}, Example 2.80). If $n$ is odd, choose a closed representative $\hbox{vol}_{S^n}\in\Omega^n_{PL}(S^n)$ of the volume form of $S^n$. The square of this form is $0$, because $n$ is odd, so the map $H^*(S^n)\rightarrow\Omega^n_{PL}(S^n)$ which maps $[S^n]^*\mapsto \hbox{vol}_{S^n}$ is a quasi-isomorphism. If $n$ is even, the square of $\hbox{vol}_{S^n}$ is non-zero. It must however be a boundary, since $H^{2n}(S^n)=0$, so $\hbox{vol}_{S^n}=d\alpha$. One then builds a zigzag with $A=(\Lambda (x_n,y_{2n-1}), dy=x^2)$ (for $\Lambda$ notation see {\bf 7.0.1}). Here 
$A\rightarrow\Omega^*_{PL}(S^n)$ sends $x\mapsto \hbox{vol}_{S^n}$ and $y\mapsto \alpha$. On the other hand $A\rightarrow H^*(S^n)$ sends $x\mapsto [S^n]^*$ and $y$ to $0$. Both maps are CDGA maps and quasi-isomorphisms. Note that the CDGA $A$ is a \textit{minimal model} for $S^n$, $n$ even.

\noindent{\bf 2.6.2.} It turns out that configuration spaces $\conf{k}(\bbr^n)$ are formal under some (weak) conditions, and these are hard theorems to prove. Kontsevich \cite{kontsevich} and Lambrechts-Volic \cite{pascal}  proved that the configuration
space $\conf{k}(\bbr^n)$ is formal over $\bbr$ for any $k$ and $n$, by using graph complexes (see \S\ref{graphhom}).
Over the integers, the situation is much more complicated. Salvatore in \cite{paolo2} shows that $\conf{k}(\bbr^n)$ is formal over $\bbz$ as long as the dimension is larger than the number of points $n\geq k$.  Surprisingly however, he shows that $\conf{k}(\bbr^2)$ is not formal over $\bbz_2$ for $k\geq 4$, and therefore it is not formal over $\bbz$. 
Formality fails globally for configuration spaces other than Euclidean space. Bezrukvnikov \cite{bezru} shows that if $X$ is a smooth complex projective variety, then $\conf{n}(X)$ is not formal as soon as $n>2$. 

Formality has a rare occurrence among spaces. It is a strong property that can dramatically simplify complex calculations with chain complexes (bar, cobar, and (co)Hochschild complexes), e.g. \S\ref{nbody}. 
The formality of configuration spaces can further be enhanced to the formality of the little disks \textit{as an operad} (over $\bbr$ or over $\bbq$). This was shown to be the case through the work of Kontsevich, Tamarkin, Lambrechts and Voli\'c,
Fresse-Willwacher, Boavida de Brito and more (for definitions and references, see \cite{idrissi}, chapter 5).


\subsection{Homotopy, Koszulness and the loop space}\label{loopspace}
A fundamental fact in this section is that relations in the homotopy groups of $\conf{k}(\bbr^n)$ are ``orthogonal'' to the Arnold-Cohen relations in cohomology. This leads to  interesting consequences. 

\noindent{\bf 2.7.1.} (The Homotopy Groups) The bottom spherical generators $a_{rs}$ in the homology of $\conf{k}(\bbr^n)$ \eqref{aijmap} give rise to homotopy classes $a_{rs}$ (also of the same name) in $\pi_{n-1}(\conf{k}(\bbr^n))$, $k\geq 2$. These classes $a_{rs}$, for $1\leq s<r\leq k$, generate the first non-zero homotopy group in dimension $n-1$, and the rest of the rational homotopy groups are generated systematically as follows. Recall that the Whitehead product $[-]$, of degree $-1$, turns the graded homotopy groups into a graded Lie algebra. In $\conf{k}(\bbr^n)$ one has the following relations (\cite{fh}, Theorem 3.1, ChapterII)
$$
\begin{array}{cl}
& \,\,\quad a_{ji}=(-1)^n a_{ij}\\
\begin{array}{c}
\text{quadratic}\\
\text{relations}
\end{array}
&
\left\{
 \begin{array}{l}
 \left[ a_{ij},a_{kl}\right]=0, \quad\text{ if  $\#\{i,j,k,l\}=4$}\\
 \left[ a_{ij},a_{ik}+a_{jk}\right]=0
 \end{array}
\right.
\end{array}
$$

It is easy and instructive to see how these relations come about. They are a direct consequence of the planetary generators described earlier, and of the elementary fact that the Whitehead product of two classes $\alpha,\beta: S^m\rightarrow X$ vanishes if and only if the map $\alpha\vee \beta : S^m\vee S^m\rightarrow X$ extends to a map $S^m\times S^m\rightarrow X$. Figure \ref{fig:4T} explains it clearly: by $1,2$ and $3$ we mean the first, second and third entry of a configuration in $\conf{3}(\bbr^n)$. An arrow between $i$ and $j$ means that $j$ rotates around $i$. The figure on the left is a snapshot of the image of $f_1: S^{n-1}\times S^{n-1}\rightarrow \conf{3}(\bbr^n)$, $(u,v)\mapsto (0,u, u+{v\over 2})$, and the part on the right is a similar snapshot of the image of $f_2(u,v)= (0,u,2v)$. The solid sphere is the restriction of the map to the first factor, and the dashed sphere is its restriction to the second factor. Notice that $f_2$ restricted to the second sphere factor (the dashed sphere on the right) is the map describing $3$ rotating around $2$ and $1$ simultaneously, so its homotopy class is $a_{31}+a_{32}$. This map $f_2$ restricted to the first copy is obviously $a_{21}$ ($2$ rotating around $1$). Since $f_2$ extends these two homotopy classes to the product, necessarily $[a_{21}, a_{31}+a_{31}]=0$, which is one of the claimed relations. Same for $f_1$. Notice in this case, that the restriction of $f_1$ to the first factor (solid sphere) is homotopically $a_{21}+a_{31}$;
the moral being that when $2$ rotates around $1$, $3$ also does!

\begin{figure}[htb]
\centering
\includegraphics[width=3.5in]{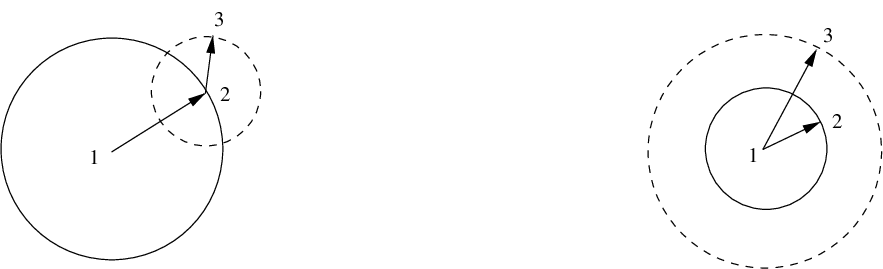}
\caption{The homotopy 4T relations illustrated. Each figure represents a map of a torus $S^{n-1}\times S^{n-1}$ into configuration space. In the figure on the left, the solid sphere represents the map $a_{21}+a_{31}$, and the dotted sphere is the map $a_{32}$. The corresponding relation is $[a_{32}, a_{21} + a_{31}]=0$.}
\label{fig:4T}
\end{figure}

\noindent{\bf 2.7.2.} (The Loop Space) Write $B_{ij}= ad_*(a_{ij})$, where 
$ad(a_{ij}): S^{n-2}\rightarrow \Omega\conf{k}(\bbr^n)$ is the adjoint map to $a_{ij}$, $n>2$.
By a Theorem of Samelson \cite{samelson}, $ad_*([a_{ij}, a_{rs}]) = B_{ij}B_{rs}- (-1)^{n}B_{rs}B_{ij}$. Consider the graded free Lie algebra $L_k(n)$ over $\bbz$ generated by the classes $B_{ij}$, for $1\leq i\neq j\leq k$, and let $\mathcal I$ denote the ideal of relations generated by
\begin{enumerate}
    \item $B_{ij} = (-1)^nB_{ji}$
    \item $[B_{ij}, B_{ik} + B_{jk}]$ for $i< j< k$ (or equivalently 
    $[B_{ij}, B_{jk}] = [B_{ki}, B_{ij}]$)
    \item $[B_{ij}, B_{rt}]=0$ when $i,j,r,t$\ distinct. 
\end{enumerate}
with bracket taken in the graded sense. Define the \textit{Yang-Baxter} algebra $YB_k^{(n)}$ to be the \textit{enveloping algebra} of the Lie algebra $L_k(n)/\mathcal I$. This is also called the \textit{infinitesimal braid Lie algebra}.

\begin{quote}\textit{{\bf Theorem 2.4} \cite{{fh, cg, kohno1}}: (i) If $n>2$, there is an isomorphism Prim$H_*(\Omega\conf{k}(\bbr^n))\cong L_k(n)/I$ as Lie algebras, where Prim are the primitives. Furthermore, (ii) there is an isomorphism of
graded Hopf algebras
\begin{equation*}\label{ybkn}
YB_k^{(n)}\cong H_*(\Omega\conf{k}(\bbr^n);\bbz) \ \ \ 
\end{equation*} (iii) This loop algebra is torsion-free and has a graded basis
$\{B_{i_1j_1}\cdots B_{i_\ell j_\ell}\}$ with $\ j_1< \cdots < j_\ell, i_t<j_t\ , \forall\ t$. (iv) The Poincar\'e series for the based loop space for $n>2$ is
\begin{equation*}\label{poincloop}
P_{\Omega\hbox{\tiny Conf}_{k}(\bbr^n)}(t) = \prod_{j=1}^{k-1}(1-jt^{n-1})^{-1}
\end{equation*}
Finally, \hbox{(v)} $H_*(\Omega\conf{k}(\bbr^n))$ is Koszul-dual  to $H^*(\conf{k}(\bbr^n))$ (in the sense of \cite{priddy}).}
\end{quote}
 
\noindent{\bf 2.7.3.} (Koszul Duality) 
We explain the last statement of Theorem 2.4, and the various steps to prove the theorem. 
The Lie algebra
$L_k:=\pi_*(\conf{k}(\bbr^n))\otimes\bbq$ is \textit{ Koszul dual} to $A_k:=H^*(\conf{k}(\bbr^n);\bbq)$. This means that the ${k\choose 2}$-dimensional space $V_k$ of generators of $A_k$ is dual to the space $V^k$  of generators of $L_k$, and the space $R_k\subset S^2V_k$ spanned
by the quadratic relations \eqref{quadratic1} is orthogonal to the space $R^k\subset S^2L_k$ of quadratic relations of $L_k$. Why such duality holds is clarified by the following general formalism in \cite{berglund}:
\begin{itemize}
\item In the graded world,
if $A$ is a graded commutative Koszul algebra, its \emph{Lie Koszul dual} $A^{!L}$ is the free graded Lie algebra on dual generators desuspended once, modulo the \textit{orthogonal relations}.
\item If a space is formal and $H^*(X,\bbq)$ is Koszul, then homotopy and cohomology are Koszul dual in the sense that there is an isomorphism of graded Lie algebras
$\pi_*(\Omega X)\tensor \bbq = H^*(X;\bbq)^{!L}$, where the left-hand side is
equipped with the Samelson product (assuming $\pi_1(X)=0$)."
\end{itemize}
Theorem 2.4 (ii) is the consequence of the following facts: $\conf{k}(\bbr^n)$ is rationally formal, $H^*(\conf{k}(\bbr^n))$ is torsion-free, it is Koszul by Theorem 2.3, 
$L_k(n)/\mathcal I\tensor\bbq=\pi_*(\Omega\conf{k}(\bbr^n))\tensor_\bbz\bbq$
and the rational loop space homology is the enveloping algebra of its loop lie algebra (Milnor-Moore Theorem). 

When $n=2$, $\conf{k}(\bbc)$ is the classifying space of the pure braid group $PB_k$, and $H^*(\conf{k}(\bbc ))$ is a Koszul quadratic algebra. Its Koszul dual $YB_k^{(2)}$ can be identified with the enveloping algebra of the Lie algebra $\mathcal L_k$  obtained from the descending central series of $PB_k$, which we recall, is obtained by setting $G_1^k=PB_k$, $G_i^k = [G_1^k, G_{i-1}^k]$ and  $\displaystyle \mathcal L_k = \bigoplus_i(G_i^k/G_{i+1}^k)$ with bracket induced by taking commutators \cite{falk, kohno2}. An extension of these results to an arbitrary Riemann surface is in \cite{bezru}.

\subsection{The loop space and finite type invariants}\label{finitetypeinv}

It turns out that the universal enveloping of the infinitesimal braid Lie algebra has a faithful description in terms of \textit{horizontal chord diagrams}, and the ``Yang-Baxter relations'' become equivalently 2T and 4T relations.
This is quite relevant to Physics since it provides a direct link between $\Omega\conf{k}(\bbr^n)$ and \textit{finite type invariants of knots} \cite{kohno1} or \textit{higher observables on brane configurations} \cite{ss}. In this context [of physics/branes], based loop spaces of configurations appear as spaces of 'vaccum scattering processes' of 'topological solitons' (centered at the points) and as such connect the mathematical representation theory of braid groups, and of higher homotopy groups of configuration spaces, to the physical dynamics of such topological physical objects (still largely hypothetical, but at the focus of much attention in contemporary quantum materials research).

Define $D_n^k$ to be the set of all $n$ vertical directed lines (‘strands’) with $k$ horizontal lines (‘chords’) connecting them.
The {\it monoid of horizontal chord diagrams} on $n$-strands is the free monoid
on the set of pairs of distinct strands
\begin{equation}
\label{freemonoid}
\mathcal{D}^{{}^{\mathrm{pb}}}_{\!n}
  \;:=\;
  \mathrm{FreeMonoid}
  \Big(
    \big\{
      t^{ij}
      \,\vert\,
      1 \leq i \neq j \leq n
    \big\}
  \Big)
  \,= \bigoplus_{k\geq 0} D_n^k,
\end{equation}
where the generator $t^{ij}=t^{ji}$ is called the
{\it chord} connecting the $i$th and $j$th strand.
Hence a general horizontal chord diagram is a finite list of chords
$t^{i_1 j_1}r^{i_2 j_2}\cdots t^{i_d j_d}$, possibly empty, and the product operation in $\mathcal D_n^{{}^{\mathrm{pb}}}$ is the concatenation of these lists, the neutral element is given by the empty list. 

\begin{center}
\begin{tikzpicture}
 \draw (-.23, 0.35) node {\tiny $i$};
 \draw (-.23+1.2, 0.35) node {\tiny $j$};
 \draw (-.23+2.4, 0.35) node {\tiny $n$};
 \draw (-.23-0.9, 0.35) node {\tiny $1$};
 \clip (-1,+.5-1.5) rectangle (3.5,+.5);
 \draw[thick]
   (0,+1)
   to
   node{\colorbox{white}{\hspace{-2pt}}}
   (2.4,+1);
 \begin{scope}[shift={(0,+1)}]
   \clip (0,-.1) rectangle (.1,+.1);
   \draw (0,0) circle (.07);
 \end{scope}
 \begin{scope}[shift={(2.4,+1)}]
   \clip (-.1,-.1) rectangle (0,+.1);
   \draw (0,0) circle (.07);
 \end{scope}

 \draw[thick] (0,0) to (1.2,0);
 \begin{scope}[shift={(0,0)}]
   \clip (0,-.1) rectangle (.1,+.1);
   \draw (0,0) circle (.07);
 \end{scope}
 \begin{scope}[shift={(1.2,0)}]
   \clip (-.1,-.1) rectangle (0,+.1);
   \draw (0,0) circle (.07);
 \end{scope}
 \draw[thick,blue] (0,2) to (0,-1);
 \draw[thick,blue] (1.2,2) to (1.2,-1);
 \draw[thick,blue] (2.4,2) to (2.4,-1);\draw[thick,blue] (-.9,2) to (-.9,-1);
 \draw (-.6,1.5) node {$\cdots$};
 \draw (0.6,1.5) node {$\cdots$};
 \draw (1.2+0.6,1.5) node {$\cdots$};
 \draw (2.4+0.6,1.5) node {$\cdots$};

 \draw (-.6,-.5) node {$\cdots$};
 \draw (0.6,-.5) node {$\cdots$};
 \draw (1.2+0.6,-.5) node {$\cdots$};
\end{tikzpicture}
\captionof{figure}{The degree $1$ generator $t^{ij}$. This is an element of $D_n^1$, and in this representation,  the product is stacking vertically chord diagrams.} 
\end{center}

The {\it algebra of horizontal chord diagrams} 
$\mathcal{A}^{{}^{\mathrm{pb}}}_{n}
  \;:=\;
  \mathbb{Z}[\mathcal{D}^{{}^{\mathrm{pb}}}_{\!n}]/(\mathrm{2T}, \mathrm{4T})
$
is the associative unital algebra (non-commutative), graded by number of chords,
which is spanned by the monoid of horizontal chord diagrams \eqref{freemonoid}
and then quotiented by the degree $n$ piece of the double-sided ideal denoted $(2T,4T)$ generated by the relations
\begin{eqnarray*}
 t^{ij}t^{kl}=t^{kl}t^{ij}	& \text{when} &	|\{i,j,k,l\}|=4,
    \label{ijkl}	\\
  \left[t^{ik}+t^{jk},t^{ij}\right]=0	& \text{when} &	|\{i,j,k\}|=3.
    \label{ijk}
\end{eqnarray*}
Note that when $m=3$,
$H_*(\Omega\conf{n}(\bbr^m))$ is
canonically isomorphic, as a graded associative algebra, to $\mathcal{A}^{{}^{\mathrm{pb}}}_{n}$
by Theorem {2.3}. This is now related to finite type invariants of the pure braid groups, and to weight systems by taking the duals \cite{kohno1, dror, ss}.
To explain, start with $P_n:=\pi_1(\conf{n}(\bbr^2))$ the pure braid group on $n$-strands, and $\mathcal I$ the augmentation ideal of the group ring $\bbz P_n$. An invariant $v: P_n\rightarrow\bbz$ is said to be of order $k$ if the induced map $v:\bbz P_n\rightarrow\bbz$ factors through $\mathcal I^{k+1}$. The set of order $k$ invariants for $P_n$ with values in $\bbz$ has  structure of a $\bbz$-module and is identified with $$W_k(P_n) := \hbox{Hom}_\bbz (\bbz P_n/\mathcal I^{k+1},\bbz )$$
There are natural inclusions
$W_k(P_n)\subset W_{k+1}(P_n)$, and one sets
$W(P_n) = \bigcup_{k\geq 0}W_k(P_n)$. This is referred to as the space of \textit{finite type
invariants} for $P_n$ with values in $\bbz$.  There is also a natural ${\bf Z}$-module homomorphism
$$
w : W_k(P_n) \rightarrow {\rm Hom}_{\bf Z}(D_n^k, {\bf Z})
$$
where
$w(v)$ is called the weight system for $v \in W_k(P_n)$.
It is shown in \cite{kohno1} (and references therein)  that $w(v)$ vanishes on the ideal
$(2T,4T)$ in $\mathcal D^{pb}_n$, and that $w$ induces (formally) an isomorphism of groups for $m \geq 3$ 
$$
H^{k(m-2)}(\Omega {\rm Conf}_n({\bf R}^m); {\bf Z})
\cong W_k(P_n) / W_{k-1}(P_n).
$$
When $m$ is even, there is an isomorphism of Hopf
algebras $
H^*(\Omega {\rm Conf}_n({\bf R}^m); {\bf Z})
\cong W(P_n)
$ (\cite{kohno1}, Theorem 4.1).


\subsection{Cellular models}\label{cellularmodels}

The space $\conf{k}(\bbr^n)$ is not a CW-complex but it affords various homotopy equivalent CW models, some being deformation retracts. We describe the first such which stems from the existence of these tori inside the configuration space that generate its homology \S\ref{homology}.  The existence of such tori produces a \textit{small CW-model} (see \cite{hatcher}, 4.C) for $\conf{k}(\bbr^n)$ which is a union of products of spheres which can be mapped skeleton by skeleton to the configuration space.

\begin{quote}\textit{{\bf Theorem 2.5} \cite{fh, paolo1}: 
Let $n,k\geq 2$. Then $\conf{k}(\bbr^n)$ is homotopy equivalent to a minimal CW-complex
$Y$ with cells only in dimensions $q(n-1)$ for $q\in\{0, 1, \ldots , k-1\}$. The cells are attached
via generalized Whitehead products\footnote{The attaching map of the top cell of a product of spheres represents the higher order Whitehead product of
the embeddings of the factors}.}
\end{quote}
\vskip 5pt

In the case of three points, this is an earlier result of Massey \cite{massey1} who constructs a CW-deformation retract $Y$ of $\conf{3}(\bbr^n)$ with one cell in dimension $0$, cells $a_{1}, a_{2}, a_{3}$ in dimension $n-1$ and cells $e,f$ in dimension $2n-2$
with $\partial e = [\iota_1,\iota_2-\iota_1]$ and $\partial f = [\iota_3,,\iota_2-\iota_3]$. The class $\iota_j$ is the homotopy class of the $j$-th sphere in the $n-1$ skeleton, $j=1,2,3$. Massey deduces an explicit description of the homotopy type: this is the union of two copies of $S^{n-1}\times S^{n-1}$ along their diagonal (\cite{massey1}, Theorem III).

\noindent{\bf 2.9.1.} (FN complex) The earliest known ``cell complex'' for $\conf{k}(\bbr^n)$ is due to Fox and Neuwirth \cite{fn} for $n=2$, and it is a cell decomposition of the one-point compactification of $\conf{k}(\bbr^2)$. In a nutshell, 
this cell complex, dubbed ``FN-complex'', is based on the simple idea of applying the projection $(x,y)\mapsto x$ to configurations in $\bbr^2$, then ``filtering vertically'' by taking preimages. This gives a so-called \textit{stratification} of $\conf{k}(\bbc)$ by open cells. This is not a CW-decomposition, but instead, this gives the 1-point compactification of $\conf{k}(\bbc)^+\subset S^{2k}$ the structure of a CW complex where the unattached boundaries are now glued to the point at infinity. 
This cell structure has been extended to $\conf{k}(\bbr^n)^+$ for $n\geq 2$ by Nakamura \cite{nakamura} and Vassiliev \cite{vassiliev}. Vassiliev gave an explicit description of the chain boundaries mod-$2$, while a thorough treatment integrally is given in \cite{gs}. 
The cell structure is based on the lexicographic ordering of points in $\bbr^m$ using standard coordinates. This ordering gives rise to an ordering of points in a configuration. The resulting FN cell complex is equivariant with respect to the action of the symmetric group and has dimension $(n-1)(k-1)$. It induces a CW-decomposition of the one-point compactification of the unordered configuration space $C_k(\bbr^n)$ \cite{gs}. As a direct consequence of the existence of the FN equivariant complex, the homology of both $\conf{k}(\bbr^n)$ and $C_k(\bbr^n)$ vanishes beyond the bound  $(n-1)(k-1)$  (as already noted in \S\ref{homology}). The FN complex is used directly to settle the Nandakumar-Ramana-Rao conjecture in \cite{kha} (see \S\ref{coincidences}). An explicit description of this complex for $n=2$ is described in detail in \S\ref{cellunordered} in the unordered case. 

A convenient and attractive look at the Fox-Neuwirth cell decomposition in the case $n\geq 2$ is through the lens of posets \cite{ah}. We pointed out the FN decomposition of $\conf{k}(\bbr^n)$ (without joining the point at $\infty$) is not a CW-decomposition, but still, the topological boundary of each cell
meets only cells of lower dimension. This gives a stratification having the \textit{frontier condition}, meaning if a stratum $S_\alpha$ meets the closure
$\overline{S_\beta}$, then $S_\alpha\subset\overline{S_\beta}$. When this happens, one gets a partial ordering on strata $\alpha\leq\beta$. Consequently, the cells/strata of the Fox-Neuwirth decomposition form a poset which is indexed by planar rooted trees. The \textit{order complex} of this poset, or equivalently the realization of this poset (viewed as a category in which for each pair of objects $c, d$ there is at most one morphism $c\rightarrow d$, and in which the only isomorphisms are the
identity morphisms) is homotopy-equivalent to $\conf{k}(\bbr^n)$. A good description of the poset and a short proof of this realization result is in \cite{ah} and \cite{bfsv}. There is as well a clear description of this decomposition
in terms of trees in (\cite{kha}, \S5).

\noindent{\bf 2.9.2.} Other known CW models for configuration spaces are listed below: 
\begin{itemize}
\item From the theory of complements of subspace arrangements, the \textit{Salvetti complex} gives a CW model for complements of complexified arrangements. It has been used to give extensive computations of the homology of braid groups with coefficients \cite{callegaro}. 
\item The Bjorner-Ziegler complex gives models for more general subspace arrangements, and the special case of $\conf{k}(\bbr^n)$ is worked out explicitly in \cite{blz} who describe in great detail a regular CW complex model sitting inside the configuration space as a deformation retract.  This complex was better suited for the proof of the NRR-conjecture (see \S\ref{coincidences}), as opposed to the use of the FN-complex in \cite{kha}.
\item D. Tamaki \cite{tamaki} develops CW models for complements of subspace arrangements, like a real version of the Salvetti complex, given in terms of the ordered complex of suitable posets.
\item An attractive (yet unexploited) approach using Morse theory is discussed in (\cite{ryan}, chapter 3). It is explained that the ``electrostatic potential'' function $V: C_k(M)\rightarrow \bbr$, for $M$ smooth and closed, induced from the $\mathfrak S_k$-equivariant map $\conf{k}(M)\rightarrow \bbr$
$$V(x_1,\ldots, x_k) = \sum_{1\leq i<j\leq k}{1\over |x_i-x_j|^2}$$  
is a proper map, and its critical values form a bounded subset of $\bbr$. Consequently, $C_k(M)$ has a compact deformation retract given by the flow of $-\nabla V$ (see also \cite{carlsson}).
\item In \cite{wilshire}, the author constructs an explicit simplicial
complex structure for $\conf{k}(X)$, when $X$ is a finite simplicial complex to start with. Simplicial and multisimplicial models that work in the operadic context are given in \cite{berger, kashi} and \cite{paolomedina} (see references therein). This is a long story whose starting point is {\bf 2.4.3}.
\end{itemize}


\section{Identical particles in Euclidean space}\label{cellunordered}

We discuss unordered configuration spaces of $\bbr^n$ here, and in \S\ref{unorderedconfigs} we extend to general manifods.

\noindent{\bf 3.1.} The first useful observation about the $C_k(\bbr^n)$ is that they classify finite covering spaces over finite CW complexes. Let $X$ be any such complex, and embed it in $\bbr^N$ for some $N$. Let $\pi : E\rightarrow X$ be a degree $k$ covering space.
By viewing $X\subset\bbr^N$, we  associate to $x\in X$ the set of preimages $\pi^{-1}(x)\in C_k(\bbr^N)$. This gives a continuous map $X\rightarrow C_k(\bbr^N)$. There are inclusions $C_k(\bbr^N)\hookrightarrow C_{k}(\bbr^{N+1})$ and the direct limit is denoted by $C_k(\bbr^\infty)$. The homotopy class of the composite map $X\rightarrow C_k(\bbr^\infty)$ determines the isomorphism class of the covering over $X$. This gives a meaning to the statement that $C_k(\bbr^\infty)$ is the classifying space $B\mathfrak S_k$.

\noindent{\bf 3.1.} The (co)homology of $C_k(\bbr^n)$ is of paramount importance in various physical applications as discussed in the introduction. In the particular case $k=n=3$, an amusing computation of \cite{bbs} shows that
$$H^*(C_3(\bbr^3),\bbz) = \bbz, 0,\bbz_2, 0,\bbz_3\ ,\ \ \hbox{for $*=0$ to $4$ respectively, and zero for $*>4$}$$
The authors obtained this result in their study of systems of identical spinless particles moving in $\bbr^3$ possessing an $SU(n)$ gauge symmetry. They give a fairly picturesque description of the dual homology class in dimension $3$ which we want to explain, despite the narrow margin we have in this survey. First, by the universal coefficients theorem, one can show that $H_3(C_3(\bbr^3),\bbz)=\bbz_3$ and $H_q=0$ for $q>3$. This cyclic generator in degree $3$ exists as follows: view $C_3(\bbr^3)$ as the space of all triangles in $\bbr^3$, including collinear ones. One can then realize the generator of $H_3$ as the $3$-manifold of
equilateral triangles of unit side and fixed centroid, $0$ say. Let $W$ is the $4$-manifold with boundary, consisting
of isosceles triangles of unit base,  centroid $0$, and height $h, 0 \leq h \leq \sqrt{3} / 2$,
then as $h$ tends to $\sqrt{3} / 2$, the three isosceles triangles lying in the same plane and having bases at $60^o$ to each other all approach the same equilateral triangle (see figure) and so, in homology, $\partial W=3D$. In fact, \cite{bbs} give an explicit strong deformation retraction of $C_3(\bbr^3)$ onto $W\cup D$, a CW-complex. Note that \cite{aj} analyze $H^*(C_3(\bbr^n);\bbf_p)$ directly through the $3!$-covering
$\conf{3}(\bbr^n)\rightarrow C_3(\bbr^n)$. This method however has major limitations for more points. 

\begin{figure}[htb]
\centering
\includegraphics[width=2.5in]{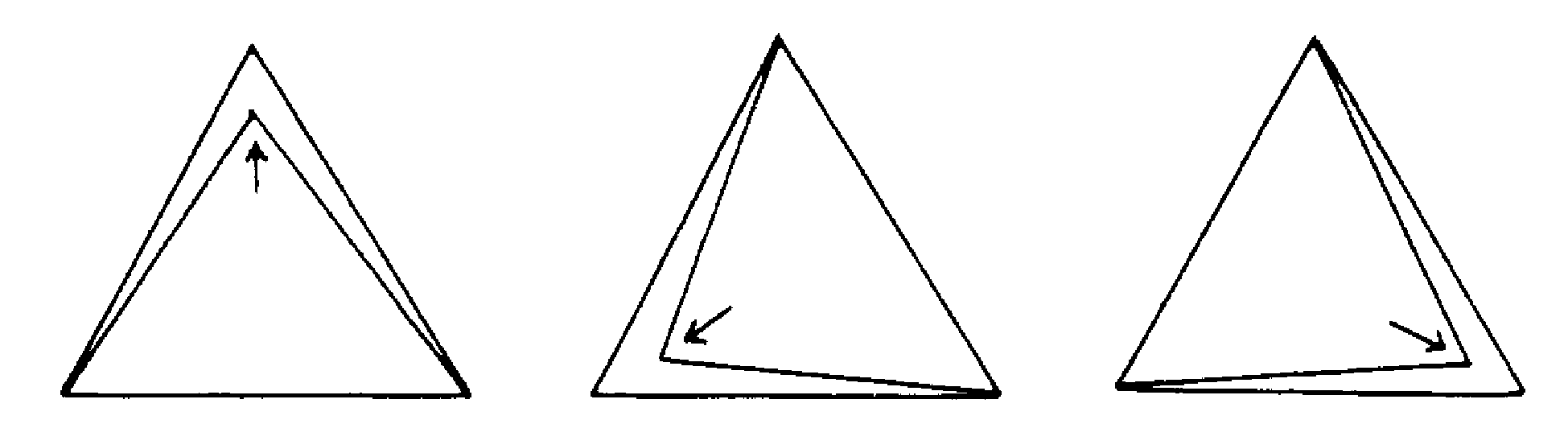}
\caption{Three isosceles triangles approach the same equilateral triangle \cite{bbs}.}
\label{triangle}
\end{figure}

\noindent{\bf 3.2.}  Whereas $H_*(\conf{k}(\bbr^n);\bbz)$ is torsion-free, $H_*(C_k(\bbr^n);\bbz)$ is almost all torsion! One proves
$$H_i(C_k(\bbr^n);\bbq) = \begin{cases} 
\bbq& \hbox{if}\ i=0\ \hbox{or}\ i=n-1\ \hbox{is odd}\\
0&\hbox{otherwise}
\end{cases}$$
This can be seen using Theorem 2.1 and formula \eqref{invariants}: the non-torsion class when $n$ is even is the image of the (only) $\mathfrak S_k$-invariant class $\sum_{i<j} [x_i,x_j]\in H_{n-1}(\conf{k}(\bbr^n))$. Geometrically, this class is the image of the orientation class of $C_2(\bbr^n)\simeq\bbr P^{n-1}$ embedded in $H_{n-1}(C_{k}(\bbr^n))$ for $k\geq 2$, $n$ even. 

For more general field coefficients,
the optimal approach to the homology of $C_k(\bbr^n)$ is to observe that $\coprod_{k\geq 0}C_k(\bbr^n)$ is an algebra over the operad of the little disks $\mathcal D_n$.
It follows that for any field coefficients,
$\displaystyle H_*\left(\coprod C_{k}(\bbr^n);\bbf \right)= \displaystyle \bigoplus_{k\geq 1} H_*(C_k(\bbr^n);\bbf)$ is an algebra over the Poisson operad $e_n=H_*(\mathcal D_n)$. It is in particular a bigraded ring. The product map, which is induced from $x_1x_2$, as explained in \S\ref{operads}, is the same as the induced map in homology of the standard concatenation product
\begin{equation}\label{product}
C_{k_1}(\bbr^n)\times C_{k_2}(\bbr^n)\longrightarrow C_{k_1+k_2}(\bbr^n)  
\end{equation}
which pushes the first configuration into the half-space $\bbr^n_+:=\{(x_1,\ldots, x_n)\in\bbr^n, x_n>0\}$, the second configuration into $\bbr^n_-$, and then concatenates them. Further to this multiplicative structure, there are homology operations (see \S\ref{operads}), which are enough to compute these groups. A full answer with mod-$p$ coefficients is in \cite{fredbible}. 

\noindent{\bf 3.3.} We discuss $H_*(C_k(\bbr^2);\bbf)$ in detail (i.e. the case $n=2$), as this is the homology of Artin's braid group $B_k$. This homology informs in particular and completely on the homology of second fold loop spaces of spheres (see Example {\bf 5.5.1}). An effective method to compute these groups is by means of a symmetrized version of the FN (Fox-Neuwirth) cell complex of \S\ref{cellularmodels} given by Fuks. The starting point is the projection $\bbr^2\rightarrow\bbr$ onto the first coordinate which maps any configuration in $C_n(\bbr^2)$ to a finite set of points in $\bbr$. The points in $\bbr$ are then ordered, and by counting the number of preimages of each of these points we get an \textit{ordered partition} (or \textit{composition}) of
$n$ (see Fig. \ref{confpicture}). 
\begin{figure}[htb]
\[ \vcenter{	\scalebox{0.8}{
	\xy
		(-30,14)*{\bullet}; (-30,-8)*{\bullet}; (-30,-4)*{x_1}; (-30,0)*{|};
		(-10,12)*{\bullet}; (-10,-17)*{\bullet}; (-10,17)*{\bullet}; (-10,-4)*{x_2}; (-10,0)*+{|};
		(5,10)*{\bullet}; (5,-10)*{\bullet}; (5,15)*{\bullet}; (5,-14)*{\bullet}; (5,5)*{\bullet}; (5,-4)*{x_3}; (5,0)*+{|};
		(25,10)*{\bullet}; (25,-10)*{\bullet}; (25,17)*{\bullet}; (25,-4)*{x_4}; (25,0)*+{|};
		{\ar@{-} (-50,0)*{}; (50,0)*{}};
		{\ar@{-} (0,-20)*{}; (0,20)*{}};
		(40,18)*{\mathbb{R}^2};
	\endxy}
	}
\]
	\caption{An element of $C_{13}(\bbr^2)$ corresponding to the composition $[2,3,5,3]$. The set of all such configurations, with the same composition, is a cell of dimension $13+4=17$.} \label{confpicture}
\end{figure}
The set of all points in $C_n(\bbr^2)$ mapping to the same composition $[n_1, \ldots, n_s]$, $n = n_1 + \ldots + n_s$, together with the point at $\infty$, is an $n + s$-dimensional cell in the one-point compactification $C_n(\bbr^2)^+$. This cell is also denoted by $[n_1, \ldots, n_s]$ \cite{schiessl}. All such cells, together with the vertex at $\infty$, make up a cellular decomposition of $C_n(\bbr^2)^+$. Using Poincar\'e-Lefschetz duality 
$H_i(C_n(\bbr^2)) = \tilde H^{2n-i}(C_n(\bbr^2)^+)$, this cell complex can be used to compute the cohomology of the unordered configuration spaces.
This was precisely the method used successfully by Fuks   \cite{fuks} to compute $H^*(C_n(\bbr^2),\bbz_2)$ (mod-$2$ coefficients), and immediately afterwards by Vainshtein with mod-$p$ coefficients \cite{vainshtein}. 
The nicest way to express that computation is to describe the bigraded ring $\bigoplus_{d,n}H_d(C_n(\bbr^2);\bbf )$. The product \eqref{product} in this case can be seen to be homotopic to the map induced at the level of classifying spaces from the map $B_{k_1}\times B_{k_2}\rightarrow B_{k_1+k_2}$ which juxtaposes braids. The computation of Fuks can now be stated as follows \cite{filippo, fredbible} (with a somewhat shorter proof in \cite{callegaro}).

\begin{quote}\textit{{\bf Theorem 3.1} \cite{fuks}:
$\bigoplus_{d,k} H_d(C_k (\bbr^2),\bbf_2)$ is a polynomial ring on generators $x_i$ of homological degree $\deg x_i = 2^i-1$ and internal degree $\dim x_i=2^i$, $i\geq 0$. In particular
$$H_d(C_{k}(\bbr^2);\bbf_2)\cong \bbf_2[x_0,x_1, x_2,\ldots ]_{\dim = k, \deg =d}$$}
\end{quote}
\vskip 5pt

This computation has numerous applications to geometry at large, including geometric analysis (see \cite{malchiodi}).
The mod-$p$ computation \cite{vainshtein}, for odd $p$, takes the form
$$H_d(C_{k}(\bbr^2);\bbf_p)\cong \left(\bbf_p[h,y_1,y_2,\ldots, ] \otimes \Lambda  (x_0,x_1, x_2,\ldots )\right)_{\dim = k, \deg =d}$$
where the second factor in the tensor product is the exterior algebra over the
field $\bbf_p$ with generators $x_i$, $i\geq 0$. The generator $h$ has $\dim h = 1$ and homological degree $\deg h = 0$. The generator $y_i$ has $\dim y_i = 2p^i$
and homological $\deg y_i = 2p^i - 2$, and finally
$\dim x_i = 2p^i$ and $\deg x_i = 2p^i - 1$. Over a field $\bbf$ in general, the following concise form is stated and proved in (\cite{etw}, Corollary 3.9): there is an isomorphism of $\bbf$-algebras
$$H_*\left(\coprod_k C_k(\bbr^2),\bbf \right)\cong Ext^*_{\Lambda (x_1)\otimes\Gamma (y_2)}(\bbf ,\bbf )$$
where $\Lambda (x_1)$ is an exterior algebra on a $1$-dimensional generator, and $\Gamma (y_2)$ is a divided power algebra on a $2$ dimensional generator. One recovers Theorem 3.1 this way \cite{etw} and  $H_*(\Omega^2S^2,\bbf_2)$ (see {\bf 4.2.3}). 

\noindent{\bf 3.3.} The FN complex, in the unordered case, is used to determine the explicit cohomology groups for configuration spaces of spheres for example \cite{schiessl, napolitano}. Furthermore, this FN-complex offers an effective way to compute the cohomology of the symmetric group $\mathfrak S_k$. By taking inverse limits, \cite{gs, gss} construct an FN complex for $C_k(\bbr^\infty)\simeq B\mathfrak S_k$, and
develop the theory of ``basic Fox-Neuwirth classes'' whose cup and transfer products yield all mod-two cohomology of symmetric groups. More precisely, they succeed in describing $H^*(\coprod B\mathfrak S_n,\bbf_2)$ as a \textit{Hopf ring}, a notion introduced in \cite{st}.


\section{Particles and labels}
\label{labeled} 

The study of (unordered) configuration spaces in topology is motivated in large part by the fact that they appear as approximations of certain function spaces with source a manifold. Two ingenious ideas that happen to work very well in tandem consist of adding \textit{labels} to the point configurations and then \textit{scanning}\footnote{Terminology due to G. Segal \cite{segalacta}.} the manifold to obtain sections of bundles of ``local configurations'' over the manifold. This ``local to global'' construction produces \textit{scanning type} maps from configuration spaces with labels to section spaces over $M$. Remarkably such maps are almost always homology equivalences through a range (in other words, they induce isomorphisms on homology groups up to a maximal degree that depends on the number of points in configurations and the dimension of the manifold). The germ of this idea goes back to the work of R.J. Milgram and P. May \cite{may1} who show that the unordered configuration spaces of $\bbr^n$ with labels in a pointed connected space $X$ is a model for the iterated loop space $\Omega^n\Sigma^nX$  (more below).  It is no exaggeration to say that this construction and this theorem have been among the most impactful in algebraic topology. Extension of the work of May and Milgram was given by Segal primarily, McDuff and B\"odigheimer \cite{segal, mcduff, bodig2} who related configuration spaces on general manifolds to section spaces of some bundles associated to their tangent spaces (see \S\ref{scanning}).

Taking more exotic configuration spaces led to defining spaces of particles on manifolds, where each particle carries a label in an abelian (or homotopy abelian) \textit{partial monoid}. The structure of the labels characterizes the topology of the space in the following descriptive sense: two particles are allowed to coalesce and take the same position if and only if their labels are allowed to add up in the partial monoid. Scanning such labeled configuration spaces leads to more general section spaces over $M$ of bundles with fibers dubbed the ``classifying space of the monoid'' (more below), and the same homology equivalences are again found to hold (special care is needed when the monoid is disconnected, see {\bf 4.2.3}). From this point of view, the theory gives a mathematical framework for studying interacting particles, a situation relevant to physics. The germ of this idea can be found in \cite{segalacta, guest} and was essential in understanding the topology of some spaces of rational maps from Riemann surfaces. A first global treatment appeared in \cite{kallel2}, with a connection to Poincar\'e duality, and in \cite{paolo0}, with a connection to operads. Both points of view got integrated into the \textit{topological chiral homology} and non-abelian Poincar\'e duality of Jacob Lurie (\cite{lurie}, Chapter 5).
Lurie's construction is general and categorical, and is a homotopy-theoretic analogue of the Beilinson-Drinfeld theory of chiral homology for chiral algebras \cite{beidrin}. The Lurie-Salvatore theory now takes the name of \textit{factorization algebra} as developed by Ayala and Francis \cite{af} (see references therein), and exploited in \cite{miller}. Further development  in the equivariant setting or in the stable category is in work of  Zou et al (eg. \cite{zou}).

\subsection{Scanning maps and loop spaces}\label{scanning} To the following question the author asked Graeme Segal one day at the IHES: "What is a favorite result of yours?'', Segal answered: ``Possibly the scanning map!''.  This is discussed in {\bf 4.2.2}, but an ancestor of this construction is the \textit{electric field} map \cite{segal, aj} which takes the form
\begin{equation}\label{scanningmap}
\phi^n_k : C_k(\bbr^n)\lrar\Omega^n_kS^n
\end{equation}
where $\Omega^n_kS^n$ is the space of all based maps
$f: S^n\rightarrow S^n$ having topological (Brouwer) degree $k\in\bbz$. The space $\Omega^n_kS^n$ is a connected component of the space of all maps $\Omega^nS^n$, and $\Omega_0^nS^n$ consists of all maps homotopic to the constant map. The construction of $\phi^n_k$ in \eqref{scanningmap} goes as follows \cite{aj}: 
view a configuration $\{x_1,\ldots, x_k\}$ of $k$ points in $\bbr^n$ as $k$ ``electrically charged'' particles and 
associate to them the corresponding ``electric field'' $E$ which is a function on $\bbr^n$
taking values in the one-point compactification $\bbr^n\cup\{\infty\}=S^n$. Since $|E|\rightarrow 0$, at $\infty$ the electric field extends to a map between the compactification
 $\overline{E}: S^n \rightarrow S^n$ with $\overline{E}(\infty) = 0$, so it is based.
It has degree $k$ because the preimage of $\infty$ (a regular value) consists of precisely the points of the configuration $\{x_1,\ldots, x_k\}$.
In this construction it is immaterial
(up to homotopy) what law of force or potential one takes as long as it is linear and
the field of a single charge has the properties: (i) $E\rightarrow 0$ as $\infty$, $E\rightarrow\infty$ at the source, and $E$ is spherically symmetric. 
Thus we could take the potential to be $\log r$, of $1\over r^m$, $m\geq 1$.
If $n= 1, 2$ or $4$ so that $\bbr^n$ can be identified with one of the basic division algebras,
the scanning map can be defined by the function $\displaystyle \phi^n_k(x) = \sum_{i=1}^k{1\over x-x_i}$,
where $(x-x_i)^{-1}$ is the inverse in the appropriate field. 
The map \eqref{scanningmap} is referred to as the \textit{cohomotopy charge map} in the physics literature \cite{ss}.
Its topological significance is made clear by the following fundamental result of G. Segal \cite{segal}.

\begin{quote}\textit{{\bf Theorem 4.1} \cite{segal}: The map $\phi_k^n: C_k(\bbr^n)\rightarrow\Omega^n_kS^n$ is an embedding in homology, and induces an isomorphism in homology in degrees $q$, where $0\leq q\leq {k\over 2}$.}\end{quote}
\vskip 5pt

Note that all connected components of $\Omega^nS^n$ are homotopy equivalent. Let's check the theorem when $k=2$ and $n> 2$: in that case, $C_2(\bbr^n)\simeq\bbr P^{n-1}$ with $H_1=\bbz_2$, and $H_1(\Omega^n_0S^n,\bbz)=\pi_1(\Omega^n_0S^n)\cong \pi_{n+1}S^n\cong\bbz_2$ as well since $n> 2$. Observe that since we have inclusions $C_k(\bbr^n)\rightarrow C_{k+1}(\bbr^n)$, the direct limit $C_\infty(\bbr^n)$ has the homology of $\Omega^n_0S^n$. In the case $n=2$, $\Omega^2_0S^2\simeq\Omega^2S^3$, $\displaystyle BB_\infty = \lim_{\rightarrow}BB_n$ and $H_*(BB_\infty)\cong H_*(\Omega^2S^3)$; this latter isomorphism being a famous instance of a theorem of Barratt and Priddy. A notable restriction of Segal's theorem is that one cannot replace homology by homotopy. This is clear by looking at fundamental groups, since
$\pi_1(C_k(\bbr^n))\cong\mathfrak S_k$  while $\pi_1(\Omega^n_kS^n)\cong\bbz_2$ for $k\geq 3$ and $n> 2$ (more in {\bf 4.2.3}).
The physics significance of the electric field map was investigated by \cite{aj} in relation to some spaces of selfdual connections on principal $SU (2)$-bundles on $S^4$ (the ``instantons''). 

\subsection{Labels and interactions}\label{maymilgram}
Let $M$ be a background space and $X$ a CW-complex with basepoint $*$. The
\textit{labelled configuration} space $C(M;X)$ is defined as the quotient space
\begin{equation}\label{maymilgrammodel1}
C(M;X) = \coprod_{k=0}^\infty \conf{k}(M)\times_{\mathfrak S_k}X^k/_\sim
\end{equation}
where the symmetric group on $k$-letters, $\mathfrak S_k$ acts diagonally on both factors, and
where the relation $\sim$ is generated by
$$(m_1, \cdots ,m_j)\times_{\mathfrak S_j}(x_1,\cdots, x_j) \sim (m_1, \cdots ,m_{j-1})\times_{\mathfrak S_{j-1}} (x_1,\cdots, x_{j-1}) $$
if $x_j=*$. Here $k=0$ corresponds to point. For example, if $X=S^0$, then $C(M;S^0)= \coprod_{k\geq 0}C_k(M)$. This space is filtered by the $C_n(M;X)$, where $n$ is the maximum size of a configuration, with $C_0(M;X)$ being basepoint, and $C_1(M;X) = M\times X/M\times *$ (this is the half-smash product often denoted $M\ltimes X$). Its homotopy type depends only on the proper homotopy type of $M$ and the homotopy type of $X$. In practice, the background space $M$ is always assumed to be a smooth (connected) manifold, and $X$ a based CW complex. 
We now consider the fiberwise one-point compactification $\dot{TM}$ of the tangent bundle $TM$ of $M$ and we write $\dot{TM}\wedge_M X$ the \textit{fiberwise smash product} of $\dot{TM}$ with $X$. There is a bundle projection $\pi: \dot{TM}\wedge_M X\rightarrow M$, with fiber the $n$-th suspension $\Sigma^nX=S^n\wedge X$ where $\dim M=n$. Note that this bundle has a preferred section ``at $\infty$'' $s_\infty$ sending $x\in M$ to the compactification point at $\infty$ of its fiber. Write $\Gamma^c_M( \dot{TM}\wedge_M X)$ the subspace of all \textit{compactly supported sections} of this bundle  consisting, by definition, of all sections which, outside a compact set of $M$,  agree with $s_\infty$. The following result was a real breakthrough.

\begin{quote}\textit{{\bf Theorem 4.2} (Approximation Theorem) \cite{mcduff, bodig}: Let $M$ be a connected smooth manifold $M$, and let $X$ be a pointed path-connected CW complex. Then there is a weak
homotopy equivalence
$$C(M; X)\lrar \Gamma^c_M({\dot TM}\wedge_MX)$$
}\end{quote}
\vskip 5pt

This is referred to as an \textit{approximation} because configuration spaces provide a ``small combinatorial'' model for the homotopy type of section (or mapping) spaces, and have been used as such to understand their homology and stable homotopy type (see \S\ref{stability}).
A very useful situation is when  $M$ is the interior of a manifold with boundary $\overline{M}$, which is assumed to be \textit{parallelizable} (i.e. to have a trivial tangent bundle). In that case $\dot{T M}\wedge_MX$ is the trivial bundle over $M$ with fiber $\Sigma^nX$, and sections become maps into the fiber. In that case, Theorem 4.2 yields a homotopy equivalence
\begin{equation}\label{maymilgram2}\xymatrix{C(M;X)\simeq C(\overline{M}; X)\ar[r]^-\simeq &\hbox{Map}(\overline{M}/\partial \overline{M},\Sigma^nX)}\ \ ,\ \dim M=n\ ,\ X\ \hbox{connected CW-complex}
\end{equation}
The special case of the open disk (homeomorphic to Euclidean space), with a boundary sphere, takes the form $C(\bbr^n; X)\simeq \Omega^n\Sigma^nX$ and is known as the \textit{May-Milgram model for iterated loop spaces}. 

\noindent{\sc Example} (from physics):  For any natural numbers $d< p\in\bbn$, the labeled configuration space $C(\bbr^d , D^{p-d})$ of points in $\bbr^d$ with labels in the disk $D^{p-d}$ has the homotopy type of the \textit{$p$-cohomotopy cocycle} space
$\boldpi^p((\bbr^d)^{cpt})$ of the one-point compactified $d$-dimensional Euclidean space $(\bbr^d)^{cpt}$ \cite{ss}).

Well-known applications of Theorem 4.2 and \eqref{maymilgram} in topology include the stable splitting of mapping spaces \cite{bodig} (see {\bf 5.5.1}), or Fred Cohen's computation of the homology of iterated loop spaces $\Omega^n\Sigma^nX$ mod-$p$, for $X$ a Hausdorff, locally compact, well-based space \cite{fredbible}. Huỳnh Mùi carried out similar computations by means of Nakamura's decomposition (see \S\ref{cellularmodels})  \cite{muiagt}. Applications to the cohomology of the varieties of regular semisimple elements in the simple complex Lie algebras of classical type are in \cite{lehseg}. Theorem 4.2 is also used to compute the Pontryagin ring structure of loop spaces on polyhedral products \cite{natalya}, and the homology of \textit{free loop spaces} of spheres \cite{paolo3}. These models provide key geometrical information that traditional methods to compute these loop invariants (cobar construction, Hochschild cohomology) often lack.

\noindent{\bf 4.2.1.} (Topology of Interactions). A natural way to extend the above constructions and results is to label points of a configuration not simply by a space $X$, but by a \textit{ partial abelian monoid} $P$ or ``PAM'' (generally disconnected or discrete). One should think of the label space as \textit{internal states} that can interact if particles come together. To be more precise, define a topological abelian partial monoid, or PAM $P$, to be any subset of a topological abelian monoid. As in \cite{segal} where the notion of a PAM is introduced, $P$ is always assumed to have a $0$. To $P$, and any closed pair $(M,M_0)$, we can associate the space $C(M, M_0; P)$ of configurations of distinct points in $M$ with labels in $P$ and \textit{vanishing} on $M_0$. The way to view an element of this space is as a finite set of unordered particles on $M$, carrying labels, thus of the form $\{(x_1,p_1),\ldots, (x_k,p_k)\}, x_i\in M, p_i\in P$, with some identifications. The resulting topology on this set of labeled configurations is such that if a labeled particle in a configuration moves into $M_0$, it is suppressed, and when several particles with labels $p_1, . . . , p_k$ approach one point in $M$, they either collide and form one particle with label $p_1+p_2+\cdots +p_k$ if $(p_1, \cdots, p_k)$ are composable in $P$, or never reach this point simultaneously. Additionally, points of the form $(x,0)$ are identified (i.e. $0$ is a \textit{sink}). Here are some main examples of $P$ and $C(M,P)$  \cite{cnossen, dobrinskaya, kallel2, mcduff, segalacta, shimakawa}
\begin{itemize}
    \item We can view a based space $X=P$ as a trivial partial monoid (no compositions but with $0$=basepoint); $C(M,P)=C(M,X)$ is the space \eqref{maymilgrammodel1}.
\item If $P$ is taken to be the abelian monoid of positive integers $\bbn$, then $C(M,P)=\bigsqcup_{n\geq 0}\sp{n}(M)$, the abelian monoid  generated by the points of $M_+$, with $+$ being an added basepoint.
\item If $P = \{0,1\}=S^0$, then
$C(M;P) = \bigsqcup_{k\geq 0}C_k(M)$. 
\item Rational maps of the Riemann sphere $f: \bbp^1\rightarrow\bbp^1$, sending $\infty\in\bbp^1=\bbc\cup\{\infty\}$ to $1$, are  of the form $f(z) = {p(z)\over q(z)}$, where $p,q$ are monic relatively prime polynomial of same degree. Any such maps is then entirely determined by the zero set of $p$ (zeros) and the zero set of $q$ (poles). Such a map can then be viewed as an element in $C(\bbc;\bbn\vee\bbn)$, where $P=\bbn\vee\bbn = \{(n,m)\in\bbn\times\bbn\ |\ nm=0\}$ \cite{segalacta}. In this $P$, the left copy of $\bbn$ keeps track of multiplicities of poles, while the right copy of $\bbn$ keeps track of multiplicities of zeros, with the understanding  that zeros can coalesce, poles can do as well, but a pole and a zero cannot coincide, so $(n,0)$ and $(0,m)$ cannot add up in $P$ if $nm\neq 0$. 
\end{itemize}

Care has to be taken in defining partial monoids since different papers have different conditions \cite{cnossen, okuyama, shimakawa}. A common requirement is the following: $(a+b)+c$ exists in $P$ if and only if $a+(b+c)$ exits, and they sum to the common $a+b+c$. The category of PAM's is discussed in the aforementioned references, and the construction $C(M,P)$ is preserved under homotopies in the category of PAMs (\cite{okuyama}, Lemma 2.5). More on this in {\bf 4.2.3}.

\noindent{\bf 4.2.2.}\ (Scanning). One now relates these labeled configurations to section or mapping spaces via scanning, as in Theorem 4.2. We describe the main idea behind this construction \cite{segalacta, kallel2} (see the appendix of \cite{zou} for a thorough treatment, including the equivariant case). Start with a configuration $\zeta\in C(M,P)$. Scanning means that at every $x\in M$, we zoom in to see a disk-like neighborhood $U_x\cong D^n$, $n=\dim M$, and part of the configuration of $\zeta$ (a subconfiguration) that lies in this disk. As these disk neighborhoods change (as if scanning with a microscope!), points in subconfigurations vanish or appear at the boundaries of the disks. Thus, with a fixed $\zeta\in C(M;P)$, and a variable $x$ with a disk neighborhood $U_x\cong D^n$, we associate the subconfiguration $\zeta\cap D_x$, viewed as an element of $C(D^n,\partial D^n,P)$. This latter space is also denoted by $B_n(P)$; the \textit{$n$-fold deloop} of $P$ \cite{dobrinskaya}. As a result, we produce a map from $C(M,P)$ to the space of sections of a bundle associated to $TM$ over $M$, with fiber $B_n(P)$. The crux of the theory asserts that this correspondence is a (weak) homotopy equivalence under connectedness conditions, or a homology equivalence after ``group completion''. More next. 

\noindent{\bf 4.2.3.} (Approximation and Group Completion). We discuss the case of a (partial) abelian monoid $P$. When $P=X$ is the trivial monoid, $0=x_0$ the basepoint, then Theorem 4.2 is the optimal result. When $P$ is non-trivial connected, this is treated in Theorem 4.3 next. The interesting case here is when $P$ is not connected, eg. discrete (see {\bf 4.2.1}). In this case, $C(M,P)$ falls generally into components, of different homotopy types (unless $P$ has homotopy inverses), and an analog of Theorem 4.2 is not true as is. One can approach the problem by stabilizing or by the ``group completion'' theorem. The illustrative fundamental example is $A:=C(\bbr^n,S^0)=\coprod_{k\geq 0} C_k(\bbr^n)$, i.e.  $M=\bbr^n, P=S^0$ \cite{segal}. The scanning map
$\displaystyle A\rightarrow\Omega^nS^n$ is not a homotopy equivalence since $\pi_0(A)=\bbn$ and $\pi_0(\Omega^nS^n)=\bbz$. At the level of $\pi_0$, this is the inclusion of a commutative monoid $\bbn$ into its (Grothendieck) group completion $\bbz$. Segal's theorem 4.1 indicates that scanning after stabilizing, $st: \bbz\times \displaystyle C_\infty (\bbr^n):= \bbz\times\underset{{}^{\longrightarrow k}}{\mathrm{\smaller colim}} C_k(\bbr^n)\rightarrow\Omega^nS^n$ is a homology equivalence, and in fact we see that, in this case, we have a localization
\begin{equation}\label{firsttgc}
H_*(A)[\pi^{-1}]\cong H_*(\bbz\times C_\infty(\bbr^n))\cong  H_*(\Omega^nS^n,\bbz)\ \ \ \ ,\ \ \ \ \pi=\pi_0(A)
\end{equation}
This is an instance of what is called a ``topological group completion" (TGC for short) summarized as follows (\cite{may2}, \S4): A TGC of an $H$-space $Y$ (eg. topological monoid) is an $H$-space $Z$ and a $H$-map $g:  Y\rightarrow Z$ such that $\pi_0Z$ is a group and $g_*: H_*(Y; k)\rightarrow H_*(Z; k)$ is a localization of the Pontryagin ring $H_*(Y; k)$ at its multiplicative submonoid $\pi_0(Y)$ for any commutative coefficient ring $k$. Intuitively, $Z$ is obtained from $Y$ by adjoining inverses to components. 

\noindent{\sc Example}: By Theorem 3.1 and \eqref{firsttgc}, we see that 
$H_*(\Omega^2S^2; \bbf_2) \cong \bbf_2[x_0, x_2, . . . , x_{2^i},\cdots ][x_0^{-1}]$, $\deg x_i=2^i-1$.
According to \cite{fredbible}, $x_{2^i} = Q_1(x_{2^i-1} )$, where $Q_1$ and its iterates are the (only) Kudo-Araki operations applied to $x_0$ (see {\bf 2.4.2}).

The fact that $\bbz\times C_\infty(\bbr^n)$ is a TGC for $\bigsqcup_{k\geq 0} C_k(\bbr^n)$ has an explanation and a generalization in \cite{segaldusa}. Recall that if $A$ is a monoid, it has a classifying space $BA$ and there is always a natural map $\psi : A\rightarrow\Omega BA$. McDuff and Segal prove that if $\pi=\pi_0(A)$ in the center of $H_*(A)$, then $\psi$ is a TGC. Moreover, if $\{a_i\}$ are representatives of generators of $\pi_0(A)$, $m_i: A\rightarrow A$ multiplication by $a_i$ maps, then there is an induced map $\hbox{hocolim}_{m_i}A\rightarrow \Omega BA$ which is a homology equivalence.
This can be applied to $A=C(M,P)$ when $M=\bbr^n$ since this labeled configuration space tends to have an action of the little disk operad $\mathcal D_n$, therefore it is a monoid up to homotopy (in fact it is much more, an $E_n$-space). Here too, the group completion of $C(\bbr^n,P)$ is an $n$-fold loop space on an explicit space depending on $P$. When $M=\bbr^\infty$, one obtains infinite loop spaces, and thus homology theories \cite{shimakawa, shimakawa2}.

For more general non-compact manifolds $M$ of dimension $n$, stability isomorphisms can be stated and proven (see {\bf 4.2.4}, also \cite{cnossen}). We illustrate with the case $P=S^0$, $C(M,S^0)=\bigsqcup_{k\geq 0} C_k(M)$, and $M$ parallelizable, which is the interior of a manifold with boundary. Scanning yields a map $C_k(M)\rightarrow \hbox{Map}_k^*({\overline{M}/\partial \overline{M}},S^n)$, where the right hand side consists of all \textit{based, degree $k$} maps into the sphere $S^n$, $n=\dim M$. This induces a homology isomorphism in a range. As explained in \S\ref{stability}, there are \textit{stabilization} maps $C_k(M)\rightarrow C_{k+1}(M)$, and the main result is that the homology of $\displaystyle \underset{{}^{\longrightarrow k} }{\mathrm{\smaller colim}} C_k(M) =: C_\infty (M)$ 
is the homology of (any) component of $\hbox{Map}^*({\overline{M}/\partial \overline{M}},S^n)$. Similar results hold for other discrete labels $P$ like wedges of $S^0$ or for $P_m=\{0,1,\ldots, m\}$ (the sum is partial since it does not exist in $P_m$ if it is bigger than $m$), see \cite{dobrinskaya, kallel2}. Interesting applications to loop space homology is in \cite{dobrinskaya}.

\noindent
{\bf 4.2.4} (Non-abelian case, factorization homology). When $P$ is abelian, two labeled particles  can approach each other in any direction, the end result, if they can collide, is a new particle with the labels summing up. In the non-abelian case, working on a manifold, one must keep track of all ways particles can collide and the ways labels can combine. The different ways of colliding are homotopic, but there can be several essentially different homotopies between two ways of colliding. Keeping track of this hierarchy of homotopies is the essence of operads. Involving the operad of compactified configuration spaces gives also an organized way to describe collisions of more than two particles. This discovery by Salvatore, and later independently by Lurie, takes nowadays the name of \textit{factorization homology}. 
We need define $C(M,P)$ in this context. A framed $n$-monoid $P$ means precisely an algebra over the Fulton-MacPherson operad $\FM_n$ defined and discussed in \S\ref{compactification}. A $1$-monoid turns out to coincide with an $A_\infty$-space, and $B_1(P)$ is Stasheff's classifying space of $P$. If $M$ is a closed $n$-dimensional framed manifold, then $\FM_M:=\{\conf{k}[M]\}_{k\geq 0}$ form a right module over the operad $\FM_n$.
If $A$ is an $\FM_n$-algebra, then the \textit{factorization homology of $M$ with coefficients in $A$} is the topological space given by a suitable tensor product in this category (\cite{af, miller, paolo0}, and \cite{ginot, idrissi} for expository presentation)
$$C(M,P)=\int_MP := \FM_M\circ_{\tiny \FM_n} P$$
If $P$ is discrete \textit{abelian} monoid, we recover our earlier definition of $C(M,P)$ in {\bf 4.2.1}, up to homotopy. 
The first general and elegant result in the theory is proved in \cite{paolo0} (Theorem 7.6).

\begin{quote}\textit{{\bf Theorem 4.3} \cite{paolo0}: If $M$ is a compact closed parallelizable $n$-manifold and $P$ is a path connected partial framed $n$-monoid, then there is a weak equivalence $C(M; P)\rightarrow Map(M; B_n(P))$.}\end{quote}

The disconnected case is treated in \cite{miller}. Stabilization maps are defined in \cite{miller}, Definition 3.17), and their homotopy colimit is shown to be homology equivalent to the space of compactly supported sections of a bundle over $M$. 

In her widely cited paper \cite{mcduff}, McDuff lays out the fundamental quasi-fibration techniques needed for proving most of the results listed above. The approximation theorems discussed above are considered to be a manifestation of the $h$-principle (Gromov), and are proved using similar techniques.


\section{Topological and algebraic invariants}

\subsection{Homotopy Invariance}\label{invariance} Let $M$ be a connected closed smooth manifold. It was conjectured for many years that the homotopy type of $\conf{k}(M)$ depended only on the homotopy type of $M$. Initial meaningful results in this direction have been obtained by \cite{levitt} who showed that the homotopy type of the based loop space $\Omega\conf{k}(M)$ only depended on the homotopy type of $(M,\partial M)$, for connected compact manifolds. However, in 2004, Longoni and Salvatore \cite{ls} found a striking counterexample using the non-homeomorphic but homotopy equivalent Lens spaces $L(7,1)$ and $L(7,2)$.  
They proved that the configuration
spaces $\conf{k}(L(7; 1))$ and $\conf{k}(L(7; 2))$ are not homotopy equivalent for any $k\geq 2$, by
showing that their universal covers have different Massey products: all of
the Massey products vanish for the former but not for the latter. The failure of the conjecture also applies to unordered configuration spaces on these Lens spaces.

At this time, the homotopy invariance of configuration spaces is not known for simply connected spaces. It is not known as well if there are non-homeomorphic but homotopy equivalent closed manifolds $M,N$ for which the conjecture holds; that is $\conf{k}(M)\simeq \conf{k}(N)$ for all $k\geq 1$. Note that \textit{stable homotopy} invariance is known, that is invariance after a certain number of suspensions \cite{ak, malin}. Cohen and Taylor (unpublished manuscript) prove early on that the configuration spaces of \textit{smooth} closed manifolds are stable homotopy invariant.

\subsection{Algebraic Invariance} On can ask: how does the (co)homology of $M$ affect the (co)homology of $\conf{k}(M)$ for $k>1$? Several important results in this direction are known and are listed below.
\begin{enumerate}
\item Fulton and MacPherson \cite{fm} show that the Betti numbers of $\conf{3}(M)$ are not determined by the Betti numbers of $M$ by checking this on the particular pair of closed manifolds $M_1 = \bbc P^1\times\bbc P^2$ and $M_2$ the nontrivial $\bbc P^1$-bundle $P(O(1)\oplus O(-1))$ over $\bbc P^2$. For this $M_1$ and $M_2$, Totaro proves that the same is true for the unordered configuration spaces using his spectral sequence (see \S\ref{spectral}), namely $M_1$ and $M_2$ have the same betti numbers, but $C_3(M_1)$ and $C_3(M_2)$ do not. 
\item The homology groups of $C_k(M)$ and $\conf{k}(M)$ are homotopy invariants of closed oriented manifolds, by work of Bendersky and Gitler \cite{bg} (see \S\ref{spectral}, and also \S\ref{compactsupport}). For rationally formal manifolds $M$ such as smooth complex projective varieties, they compute the rational cohomology groups of $\conf{k}(M)$ and $C_k(M)$ explicitly in terms of the rational cohomology ring of $M$. 
\item Over a field of characteristic zero, if $M$ is a smooth projective complex manifold, then the rational homotopy type of $\conf{k}(M)$ only depends on the one of $M$ and an explicit model is given in \cite{kriz} (see \S\ref{models}). The same result holds for $k=2$ when $M$ is closed manifold which is either $2$-connected \cite{ls} or simply connected and even dimensional \cite{cordova}. 
    \item The real homotopy type of $M$ when the manifold is closed, smooth, simply connected, and of dimension at least $4$ only depends on the real homotopy type of $M$ \cite{idrissi}
and \cite{camposwill} (see \S\ref{models}). 
\item With mod-$2$ coefficients, the
ranks of $H_*(C_k(M),\bbf_2)$ are fully determined by the $\bbf_2$-Betti numbers of $M$, the dimension of $M$, and the integer k \cite{ml}. In the case $k=2$ and $M$ compact,
Totaro describes explicitly the
generators of the $\bbf_2$-cohomology of
$C_2(M)$ in terms of the cohomology
generators of $M$ \cite{totaro2}.
\item The same result holds mod-$p$, for $p$ odd: the $\bbf_p$ homology of $C_k(M)$, $M$ smooth and compact, depends only on the $\bbf_p$ homology of $M$ if $n$ is odd \cite{bct} (compare with Totaro's computation in first bullet point).
\item By contrast, the result in (6) is no longer true for even dimensional manifolds as shown in various places \cite{bc, bcm, ft2, zhang, totaro}, one also needs the cup product structure on $H^*(M)$.
F\'elix and Thomas \cite{ft2} prove that the rational Betti numbers of $\conf{k}(M)$ are determined by the graded algebra $H^*(M,\bbq)$. They use Theorem 4.2 of \S\ref{maymilgram} and the Haefliger rational model for section spaces to derive this result. 
\item Relating the cohomology of configuration spaces on punctured manifolds to the configurations of the manifold themselves possibly goes back to Gorjunov and Fuks \cite{gorjunov}. This point of view is developed in \cite{kallel1}, in \cite{napolitano} for surfaces, and in \cite{huang} in the general motivic context. For example, the following is established in \cite{kallel1}, and again in \cite{huang} for $X$ 
obtained from removing $r \geq 1$ points
from a connected closed orientable manifold of even dimension $2n$
$$\sum_{i,k\geq 0} \beta_i(C_k(X-p))(-u)^it^k = {1\over
1 + u^{2n-1}t}\sum_{i,k\geq 0}
\beta_i(C_k(X))(-u)^it^k.$$
where $\beta_i(X)$ is the i-th (rational) Betti number.
\item In this long list of invariants, it is fitting to include the calculation of the \textit{Lusternik-Schnirelmann category} of the unordered configuration spaces, proved in \cite{blz} and conjectured in \cite{roth}: $$cat(C_n(\bbr^d)) = (d-1)(n-1)\ \ , \ p\ \hbox{odd prime and}\ n = p^k\  \hbox{for some}\ k\geq 1$$
For spheres, there are bounds: $(d-1)(n-1)\leq cat(C_n(S^d)) \leq (d-1)(n-1) + 1$.
\end{enumerate}

The next sections explain how most of these results are obtained, and what the techniques involved are. We  point out that \textit{integral cohomology ring} computations for $\conf{k}(X)$ are very few and hard to obtain, even more so in the unordered case $C_k(X)$. The case of $S^n$ is worked out in \cite{fz}, and \cite{dominguez} give the full computation for the two point configuration spaces of real projective space $X=\bbr P^n$.



\subsection{Euler characteristic} The most basic homological invariant of any reasonable space is of course its Euler characteristic $\chi$. We have already given the formula for Euclidean configurations, corresponding to $P_t(-1)$ in the formula \eqref{poincseries}. Generally, for even dimensional manifolds, the following formula was obtained in 2000 by \cite{ft}
\begin{equation}\label{eulerchar}
1+\sum_{n=1}^\infty \chi (C_n(M)) t^n = (1+t)^{\chi (M)}, \ \ \dim (M)\ \hbox{is even}
\end{equation}
Using compactly supported Euler characteristic, Getzler \cite{getzler} establishes a formula for  the ordered configuration spaces of
any complex quasi-projective variety, which is in fact valid for any locally compact Hausdorff space (see \cite{hainaut} or \cite{walid})
$$\sum_{n=0}^\infty\chi_c(\conf{n}(X))\cdot {t^n\over n!} = (1+t)^{\chi_c(M)}$$
This specializes to the formula in \eqref{eulerchar} in this case since $\chi (C_k(M)) = {1\over n!}\chi (\conf{k}(M))$ and for even-dimensional
manifolds, $\chi_c(M) = \chi(M)$. The idea of using $\chi_c$ to get to $\chi$ of configuration spaces and their variants is an idea that has been expanded in all of \cite{arabia, barysh, walid}.
In the case $X$ is a finite simplicial complex, not necessarily a manifold, the following formula (Gal's formula) can be found in \cite{gal}
$$\sum_{n=0}^\infty\chi (\conf{n}(X))\cdot {t^n\over n!} = \prod_\sigma (1+(-1)^{\dim (\sigma)}(1-\chi (L_\sigma))t)^{(-1)^{\dim (\sigma)}}$$
where the product runs over all the cells $\sigma$ of $X$ and $L_\sigma$ is the normal link of 
the closed cell $\sigma$. The righthand term is known as the Euler-Gal series $\mathfrak{eu} (X)$. It has a pleasant form when $X$ is a graph \cite{bf}. A generalization of Gal's formula is given in \cite{hainaut} who computes the generating function of $\chi (\conf{n}(X))$ for a topologically stratified space $X=\bigsqcup X_\alpha$, such as a pseudo-manifold. Here the collection of strata is finite and every stratum is (homeomorphic to) the
interior of a compact manifold with boundary


\subsection{Real and rational models}\label{models}
Sullivan's minimal model theory states that one can entirely capture the rational homotopy type of nilpotent spaces (that is the non-torsion part of both homology and homotopy groups of the space) via algebraic models which are differential graded commutative $\bbf$-algebras (cdga), where $\bbf$ is a field of characteristic zero. A \textit{rational model} of $X$
is any cdga quasi-isomorphic 
to the Sullivan-deRham cdga $A_{PL}(X)$ (see \S\ref{formality}) or
to its minimal model. Note that with this definition, $H^*(X)$ being a model for $X$ is equivalent to the formality of $X$ (see \S\ref{formality}). If given a (real or rational) model of $M$, one can construct a small model
of $\conf{k}(M)$ which only depends on the model of $M$, then this will establish the (real or rational) homotopy invariance of $\conf{k}(M)$ on that class of manifolds. Such a model can be used to perform computations, e.g. the cohomology ring of $\conf{k}(M )$ for those coefficients. Recommendable introductions to this theory are \cite{idrissi2} or \cite{voronov} in this volume.

\noindent{\sc Example:} A clear and pedagogical step-by-step construction of a model for two point configurations on a Lens space is given in \cite{calimici}.
 
\noindent{\bf 5.4.1.} The configuration space $\conf{k}(M)$, for any $k\geq 2$ and $M$ a smooth complex projective manifold, has a remarkable rational model discovered by I. Kriz \cite{kriz} as a simplified version of a more complicated model introduced earlier by Fulton-MacPherson \cite{fm}.
The model depends only on $k$, the graded algebra $H^*(M,\bbq)$
together with its canonical orientation class. More precisely, for this class of spaces,  $E_k(M)$ as an algebra is a free graded commutative $H^*(M^k)=H^*(M)^{\otimes k}$-algebra with generators $G_{ij}$ of degree $2\dim_\bbc M-1$, $1\leq j<i\leq \ell$, divided out by the ideal generated by the following relations
\begin{eqnarray*}
    G_{ij}\cdot G_{i\ell } &=& G_{j\ell}\cdot (G_{i\ell }-G_{ij})\ \ i > j > \ell\\
    G_{ij}\cdot p_i^*(x) &=& G_{ij}\cdot p_j^*(x)
\end{eqnarray*}
where $x\in H^*(M,\bbq)$ and $p_i: M^k\rightarrow M$. The differential $d$ on $E_k(X)$ is given by $dp_i^*(x) = 0$ and $dG_{ij}=\Delta_{ij}$; where $\Delta_{ij}\in H^2(X^k)$ is the class of the diagonal $x_i=x_j$. We should explain this class. 
Suppose $M$ is closed oriented, $\dim_\bbc M=m$. Then $\Delta\in H^{2m}(M\times M)$ is a class Poincar\'e dual to the embedded diagonal $M$. If $\omega= [M]^*$ is a fixed generator of $H^{2m}(M)$, then 
$\Delta$ has the expression 
$\Delta = \sum (-1)^{|\beta_i|}\beta_i\times\beta_i^*$, where $\beta_i$ and $\beta_i^*$ are dual bases for $H^*(M)$ satisfying 
$\beta_i\cup \beta_j^* = \delta_{i,j}\omega$,
$\delta_{i,j}$ being Kronecker symbol. Now the class $\Delta_{ij}=p_{ij}^*(\Delta)$, where $p_{ij}: M^k\rightarrow M^2$ is projection onto the $i$-th and $j$-th components. Kriz proves his result by constructing a quasi-isomorphism into the model constructed by Fulton and MacPherson, who in turn obtain their model by constructing a quasi-isomorphism with an earlier more general rational model given by Morgan for any complement of a union of smooth divisors with normal crossings in a smooth, compact, complex, algebraic variety \cite{morgan} (see \S\ref{spectral}).

The simplicity of the model has prompted Kriz to write: ``What happens could be described by saying that $H^*(\conf{k}(M);\bbq)$ for a smooth projective variety $M$ of dimension $m$ is
related to $H^*(\conf{n}(\bbr^{2m});\bbq)$ in the ``simplest possible way".'' 
The reason this model exists for smooth complex projective varieties, and not for general manifolds, is because these spaces are rationally formal (see \S\ref{formality}). The situation is in general more complicated even for oriented closed manifolds. The cdga $E_k(M)$ appeared earlier (not as a model) as the $E^1$-term of the Cohen-Taylor-Totaro spectral sequence (see \S\ref{spectral}). .  

\noindent{\sc Example}: Using Kriz's model, we compute $H^*(\conf{2}(\bbp^m),\bbq)$, where $\bbp^m$ is complex projective space \cite{sohail}. Conveniently write $E^*[k]\subset E_2(X)$ the homogeneous component of total degree $*$ and degree $k$ in the exterior generators $G_{ij}$. One has $E^*[0]=H^{\otimes n}$ and $E^{2m-1}[1]$ is generated by the $G_{ij}$. A canonical basis of $H^{2i}(\bbp^m\times \bbp^m)$ is $x^i\otimes 1$, $x^{i-1}\otimes x, \ldots, 1\otimes x^i$, where $x$ is the algebra generator in $H^2(\bbp^m)$. For $\conf{2}(\bbp^m)$, there is only one $G_{12}$, and 
$dG_{12}= x^m\otimes 1 + x^{m-1}\otimes x + \cdots + 1\otimes x^m$. Clearly $H^i(\conf{2}(\bbp^m))\cong E^{i}[0]$ for $0\leq i\leq m-1$, so $b_{odd}=0$ and $b_{2i}=i+1$. On the other hand, there is only one non-trivial differential in the model $d: E^{2i+2m-1}[1]\rightarrow E^{2i+2m}[0]$ given by 
$d((x^i\otimes 1)G_{12}) = d((1\otimes x^i)G_{12}) = x^m\otimes x^i+x^{m-1}\otimes x^{i+1}+\cdots + x^i\otimes x^m$. This differential is injective, which implies that for degrees $\geq 2m-1$, the betti numbers are
$b_{odd}=0$ and $b_{2m+2i}= \dim E^{2m+2i}-1=\dim H^{2m+2i}(\bbp^m\times\bbp^m)-1= m-i$. Adding it all up and factoring gives the Poincar\'e polynomial: $P_{\tiny \conf{2}(\bbp^m)}(t) = (1+t^2+\cdots + t^{2(m-1)})(1+t^2+\cdots + t^{2(m-1)} + t^{2m})$. Note that $\conf{2}(\bbp^m)$ is  homogeneous space, and its homology mod-$2$ is computed in \cite{handel}. More generally, the integral cohomology \textit{ring} of $\conf{2}(\bbp^m)$ is determined in (\cite{yasui}, Theorem 3.1).


\noindent{\bf 5.4.2.} Real models for configuration spaces on manifolds have attracted much attention and work in the last few years. 
If $M$ is a smooth manifold, then a real model $A$ of $M$ is a cdga which is quasi-isomorphic to the cdga of de Rham forms $\Omega^*_{dR}(M)$ or to Sullivan's $A_{PL}(M)$ (see \S\ref{formality}). 
It turns out that closed, smooth and simply connected manifolds $M$ have convenient models $A$ satisfying a chain-level Poincar\'e duality. Lambrechts and Stanley \cite{lamstan} construct a cdga $\mathcal G_A(k)$ which is a dg (only) model for $\conf{k}(M)$. If we
view $H^*(\conf{k}(\bbr^n))$ as spanned by graphs modulo Arnold relations (see Theorem 1.1), then $\mathcal G_A(k)$ consists of similar graphs with connected components labeled by $A$, and the differential ``splits edges''. From this point of view, Kriz's model discussed at the beginning of the section corresponds to $\mathcal G_{H^*(M )}(k)$. Lambrechts and Stanley conjectured that their model is a rational model, and Idrissi \cite{idrissi} proved it is a real model for manifolds as
above of dimension
at least $4$, thus extending in the real case work of \cite{kriz, totaro, cohtay, bg} (see \S\ref{compactification}). A similar model is constructed in \cite{camposwill}, although the Idrissi-Lambrechts-Stanley model has the advantage of being finite-dimensional
and more computable.  


\subsection{Stability and Splittings}\label{stability}
A remarkable property of the homology of \textit{unordered} configuration spaces on non-compact manifolds $M$ is its \textit{homological stability}. For those manifolds, there are stabilization maps $st: C_k(M)\rightarrow C_{k+1}(M)$ which consist of ``pushing an additional point near infinity'' \cite{church}. For example, if $M=\bbr^n$, $st(x_1,\ldots, x_k)=(x_1,\ldots, x_k, 1+\sum |x_i|)$. 
It can be shown, by transfer arguments of Dold \cite{dold}, that $st_*$ embeds $H_*(C_k(M),R)$ into $H_*(C_{k+1}(M),R)$, for any commutative ring $R$, and thus we can view these homology groups as part of a ``stable'' group, which by \S\ref{labeled} must be the homology of a section space over $M$. The slogan is that by adding points to the configurations, we stabilize the homology. This phenomenon was already detected by Arnold who showed that, once we fix the homological degree $*$, the maps
$st_*: H_*(B_k;\bbz)\rightarrow H_*(B_{k+1};\bbz)$ induced by the inclusions of braid groups $B_k \subset B_{k+1}$ are isomorphisms for all $*\leq k/2$. G. Segal extended this theorem from $M = \mathbb R^2$ to non-compact connected manifolds $M$ and showed that stablization $st_*$ is also an isomorphism for $*\leq f(k)=k/2$ (\cite{segalacta}, appendix). This $f(k)$ is called the \textit{stable range}. Randal-Williams proves that this stable range is not optimal with rational coefficients and that for manifolds of dimension at least $3$, $st_*$ is an isomorphism for $*\leq f(k)=k$. 
Arabia \cite{arabia} extends Church's and Randal-Williams' stability results to other configuration spaces (see \S\ref{variants}) and to complex algebraic varieties that are not necessarily smooth (see \S\ref{compactsupport}).

\noindent{\bf 5.5.1.} The homological stability result can be upgraded to a stable homotopy splitting known as the \textit{Snaith splitting}. Recall the definition of $C(M,X)$ in (\S\ref{labeled}, \eqref{maymilgrammodel1}).
There is a natural filtration $C_k(M,X)$ of $C(M,X)$, and we write
$D_k(M,X):= C_k(M,X)/C_{k+1}(M,X)\simeq \conf{k}(M)\ltimes_{\Sigma_k}X^{\wedge k}$, where $\ltimes$ means the half smash product of based spaces, i.e. $X\ltimes Y= X\times Y/X\times y_0$, and where
$X^{\wedge k}$ is the smash product of $X$ with itself $k$ times. This is called a \textit{Snaith summand}. Theorem 5.1 below has a proof for $N=\infty$ and connected $X$ in \cite{snaith, ralph}. In \cite{fredsplitting}, the number of suspensions was shown to be finite (see also \cite{stylian}). 

\begin{quote}\textit{{\bf Theorem 5.1} (Stable splitting). Let $M$ be a manifold. Suppose  $X$ is a connected finite CW complex. Then there is $N>0$ such that
$\displaystyle\Sigma^{N} C_k(M,X)\simeq \Sigma^{N}\bigvee_{j=1}^k D_j(M,X)$.}\end{quote}
\vskip 5pt

When $X=S^0$, and $M$ is non-compact, we can let $D_j(M)$ be the cofiber of the ``adding a point'' stabilization map $C_{j-1}(M))\hookrightarrow C_j(M)$. Then a similar stable splitting holds $\displaystyle\Sigma^{N} C_k(M)\simeq \Sigma^{N}\bigvee_{j=1}^k D_j(M)$ for some $N>0$, by \cite{fredsplitting} (appendix) and \cite{fredmaytaylor}.
The splitting of the configuration spaces of non-compact manifolds yields, by the approximation theorem,  stable splittings for mapping spaces from $M$ into spheres or suspensions \cite{bodig, westerland}. 

\noindent{\sc Example}: For $n > 2$, one has the stable homotopy equivalence
$\displaystyle\Omega^2S^{2n}_+\simeq_s\bigvee_{k\geq 0}\Sigma^{(2n-2)k}C_k(\bbr^2)_+$ (an attractive result!).
The reasons why this is true are: for $n>1$, $C(\bbr^2,S^{2n-2})\simeq \Omega^2S^{2n}$ by the Approximation Theorem 4.2. Now $C(\bbr^2,S^{2n-2})$ splits into  wedge terms $D_k(S^{2n-2})$ which are Thom spaces of certain vector bundles over the configuration
spaces of $\bbr^2$. These bundles are discussed in {\bf 7.0.2}.

Stability  results are not generally true for closed manifolds. A simple example is provided by $M=S^2$, with $H_1(C_k(S^2);\bbz) = \bbz/(2k-2)\bbz$, which means that $H_1(C_k(S^2);\bbz)$ is never isomorphic to $H_1(C_{k+1}(S^2);\bbz)$. Nonetheless, some modified stability can be stated for closed manifolds, using \textit{replication maps} (see \cite{martin} and references therein). 

\noindent{\bf 5.5.2.} (Representation Stability). No homological stability exists for ordered configuration spaces in general. For example, the homology groups of $\{\conf{k}(\bbr^2)\}_{k\geq 1}$ never stabilize since $H_1(\conf{k}(\bbr^2)) \cong\bbz^{k\choose 2}$. 
There are however projections $\pi_{k}: \conf{k+1}(M)\rightarrow\conf{k}(M)$ which forget the last point, and so maps in cohomology $\pi^*_k: H^*(\conf{k}(M))\rightarrow H^*(\conf{k+1}(M))$. These are never isomorphisms with any coefficients in general, even for $k>>*$. Nonetheless, Church and Farb uncovered a new form of stability, called
``representation stability'' or (RS) \cite{churchfarb} which is about the pattern of the decomposition of $H^*(\conf{k}(M);\bbq)$, as a $\bbq\mathfrak S_k$-module, into irreducible $\mathfrak S_k$-representations, as $k$ grows. The idea is this: each 
$V_k:=H^*(\conf{k}(M);\bbq)$ decomposes into irreducible $\mathcal S_k$-representations  $V_k = \bigoplus_\lambda c_{\lambda,k}V(\lambda)_k$, where $V(\lambda)_k$ is the irreducible representation of $\mathfrak S_k$ corresponding to the partition $\lambda$, and $c_{\lambda,k}$ is its multiplicity\footnote{The irreducible representations
of $\mathcal S_k$ are in bijection with the partitions of $k$.}. The maps $\pi_k^*: V_k\rightarrow V_{k+1}$ are injections and they are $\mathfrak S_k$-equivariant. They induce morphisms of $\bbq\mathfrak S_k$-modules. Representation stability of this sequence of maps $\pi_k^*$, with stable range $N$, means that for all $k\geq N$, and each partition $\lambda$, the multiplicities $c_{\lambda,k}$ are independent of $k$ for all $k\geq N$. This is discussed, with many examples, in \cite{farber}. 
The main theorem of \cite{church} is that the cohomology groups $\{H^i(\conf{k}(M);\bbq)\}$ are representation stable, with stable range $k\geq 2i$, if $M$ is connected orientable with $\dim M\geq 3$. 
A very pedestrian account of this theory, and its explosive development is in \cite{wilson}.
The case of Euclidean configuration spaces is treated in details in \cite{farb}.

\noindent{\bf 5.5.3.}  The
representation theory of the symmetric group $\mathfrak S_k$ acting on $H^*(\hbox{Conf}_{k}({\mathbb R}^n))$  appears in several references, with integral or field coefficients. The theory has qualitative differences between the case when $n$ even and the case when $n$ is odd. It was noticed early on by Lehrer,  F. Cohen and L. Taylor that in the case $n$ is odd, the \textit{total} cohomology $H^*(\hbox{Conf}_{k}({\mathbb R}^n),{\mathbb C})$ is the regular representation ${\mathbb C}(\mathfrak S_k)$. In the case $n$ even, $H^*(\hbox{Conf}_{k}({\mathbb R}^n),{\mathbb C}) = 2\hbox{Ind}_{\mathfrak S_2}^{\mathfrak S_k}1$ since $H^*(\hbox{Conf}_{2}({\mathbb R}^n),{\mathbb C})$ consists of two copies of the trivial representation of $\mathfrak S_2$ (\cite{lehrer} for the case $n=2$, see \cite{gaiffi} for the general case $n$ even). 
When working over the integers, Cohen and Taylor \cite{fredrep} decompose the $\mathfrak S_k$-module structure of $H^*(\hbox{Conf}_{k}({\mathbb R}^n);{\mathbb Z})$ into sums of representations induced from Young subgroups; i.e. subgroups of the form
$\mathfrak S_{k_1}\times\cdots\times\mathfrak S_{k_r}\subset\mathfrak S_{k}$, where $(k_1,\ldots, k_r)$ is an integer partition of $k$. The following statement in (\cite{aj}, Lemma 5.2) gives perhaps the most streamlined proof of this fact. Set the $\mathfrak S_k$-module $M_k:=H^{(k-1)(n-1)}(\hbox{Conf}_{k}({\mathbb R}^n);{\mathbb Z})$ to be the top cohomology group (see {\bf 2.3.3}), and let $\alpha$ be an ordered \textit{integer partition} of the integer $k$, that is  
$\alpha = (k_1,\ldots, k_r)$, $k_1 + \ldots +k_r = k$. Define $M_\alpha$ to be $M_{k_1}\otimes\cdots\otimes M_{k_r}$. Then $M_\alpha$ is a $\mathfrak S_\alpha := \mathfrak S_{k_1}\times\cdots\times\mathfrak S_{k_r}$-module. Write
$\hbox{Ind}_{\mathfrak S_\alpha}^{\mathfrak S_k}M_\alpha = M_\alpha\otimes_{{\mathbb Z} [\mathfrak S_\alpha]}  {\mathbb Z}[\mathfrak S_k] $ the induced module. There is a decomposition as $\mathfrak S_k$-modules (Atiyah-Jones 1993, $n=3$)
$$H^*(\hbox{Conf}_{k}({\mathbb R}^n),{\mathbb Z}) = \sum_\alpha \hbox{Ind}_{\mathfrak S_\alpha}^{\mathfrak S_k}M_\alpha\ \ \ \ \ \  \ (\hbox{the sum is over the integer partitions $\alpha$ of $k$})$$
More delicate analysis is needed to understand the $\mathfrak S_k$ representations \textit{in each degree}, i.e. for each homogeneous term $H^{j(n-1)}(\hbox{Conf}_{k}({\mathbb R}^n))$.
For example, for $k\geq 3$ and $n$ odd, there is exactly one copy of the standard representation in degree $j(n-1)$, for each $1\leq j\leq n-1$ \cite{gaiffi}. On the other hand, and as pointed out in \S{\bf 3.2}, the trivial representation only appears in degree $0$ if $n$ is odd, and appears in higher degrees when $n$ is even.
A general analysis of the representation theory of $\mathfrak S_k$ on the homogeneous terms appears in \cite{blz} in terms of the lower intervals of the partition lattice $\Pi_k$ of $[k]=\{1,\ldots, k\}$, and in \cite{hersh}, with rational coefficients, in terms of \textit{higher Lie characters} Lie$_\lambda$ when $n$ is odd, and the \textit{Whitney homology} of $\Pi_k$ when $n$ is even. This literature is not exhaustive and one can pose many questions on the representation theory of $\mathfrak S_k$ acting on $H^*(\hbox{Conf}_{k}(M),{\mathbb Q})$ for more general manifolds $M$ and general coefficients $R$ (more in \cite{farb}).

A surprising link between the $\mathfrak S_n$-character of
$H^*(\conf{k}(\bbr^3))$ and counts of polynomials over finite fields was discovered in
\cite{hyde}. 


\section{Computational tools}
\label{computational}

\subsection{Spectral sequences}\label{spectral} Computing $H^*(\conf{k}(X);\bbf)$ and $H^*(C_k(X);\bbf)$ for general $X$ is difficult but this has been done with relative success for manifolds and graphs. We only discuss the case of manifolds. This section deals with the ordered case. In this situation, several naturally associated spectral sequences can help with the task  \cite{bg, radmila, cohtay, totaro, petersen}, and our objective here is to sort them out. We recall that a (cohomological) spectral sequence is the
data consisting of ``pages'' one for each $r=1,2,...$, each page is a differential graded algebra $(E_r,d_r)$, and $E_{r+1} = H_*(E_r,d_r)$. The spectral sequence \textit{converges} to $H^*(\conf{k}(M);\bbf)$ means that
$E_\infty=H^*(\conf{k}(M);\bbf)$ as vector spaces (here we are working over a field $\bbf$). In the case of the configuration spaces at hand, $E_\infty = E_N$ for some finite $N$.

\noindent{\bf 6.1.1.} (Totaro and Cohen-Taylor)
The most well-known spectral sequence in the field is the one due to Cohen-Taylor \cite{cohtay} and to Totaro \cite{totaro}. Both describe the same spectral sequence, conveniently labeled CTT-spectral sequence, but they derive it using very different techniques. 
Totaro's approach turns out to be most impactful and succinct. He considers the Leray
spectral sequence for the inclusion $\conf{k}(M)\hookrightarrow M^k$ which converges as an algebra to
$H^*(\conf{k}(M))$, for any $M$. Its $E_2$-term and first non-trivial differential $d_m$, $m=\dim M$, can be described for any real oriented manifold. The DGA $(E_m,d_m)$ is precisely the Kriz complex described in \S\ref{models}.
Totaro proves that for a smooth complex projective variety $M$, the differential $d_{2m}$, with $m=\dim_\bbc(M)$, of his spectral sequence is the only nontrivial differential and that $H(E_{2m}, d_{2m})\cong H^*(\conf{k}(M),\bbf)$ as an algebra, with $\hbox{char}\ \bbf = 0$. This collapse result came as a surprise, and it is still not clear how to deduce it using the Cohen-Taylor approach. Beyond this result, Totaro's construction via the Leray spectral sequence has provided a very successful approach to the cohomology of various types of configuration spaces  (eg. \cite{farber}, see \S\ref{variants}). Finally we should point out a difference between the spectral sequence constructed by Morgan \cite{morgan}, which also uses the Leray spectral sequence, and that of Totaro. Morgan's result applies to spaces that are complements of normal crossing divisors. This is not the case of $\conf{k}(X)\subset X^n$, unless $\dim_{\bbc} X=1$ \cite{totaro}. For good $X$, the collapse of Morgan's spectral sequence occurs at $E_3$ (a result of Deligne), while the collapse of Totaro is at $E_{2m+1}$.

\noindent{\bf 6.1.2.} The second widely used spectral sequence is the Bendersky-Gitler spectral sequence which is obtained by considering the nerve of the covering of the fat diagonal
$\displaystyle \Delta_{fat}^k(M)\subset M^k$, which is the subspace where any two entries coincide. This defines a double cochain complex
and the resulting spectral sequence of the bicomplex is a spectral sequence 
converging to the relative group $H^*(M^k, \Delta_{fat}^k)$. By Lefschetz duality, we have $H^*(M^k, \Delta_{fat}^k,R)\cong H_{mk-*}(\conf{k}(M); R)$, $m=\dim M$, $M$ compact $R$-oriented. 
It is no surprise that the BG-spectral sequence has isomorphic $E_2$-term, after $mk$-suspension, with the $E_2$-term of the CTT-spectral sequence \cite{ft2} (i.e. ``they agree up to suspension at the $E_2$-page''). This is also a result obtained in \cite{bermarpap}. If one spectral sequence collapses at $E_2$, then so does the other. Bendersky and Gitler conjecture that higher differentials ($d_2,d_3,\ldots$) are determined by higher-order Massey products. Although \cite{ft2} fail to verify this conjecture, they nonetheless find higher differentials using Massey products indeed, and give interesting examples when both spectral sequences fail to collapse at $E_2$. This is the case for the spectral sequence of $\conf{4}(T(S^2\times S^2))$, the tangent bundle of $S^2\times S^2$ (\cite{ft2}, section 5). Note that in practice, intersection theory
has also been used to evaluate Massey products on classical configuration spaces \cite{ls}.

\noindent{\bf 6.1.3.} Baranovsky and Sazdanovic \cite{radmila} introduce their ``BS''-spectral sequence which converges to the homology of the so-called \textit{chromatic configuration spaces} $\conf{G}(M)\subset M^{|V|}$ of a manifold $M$, $\bbf$-oriented if the homology is with $\bbf$-coefficients, and $G=(V,E)$ is a simple graph on a finite vertex set $V$ and edge set $E$ (for a definition, see {\bf 10.2}). In the case $G$ is the complete graph, one recovers the standard configuration space of points $\conf{G}(M)=\conf{n}(M), n=|V|$. The construction and approach of \cite{radmila} is reminiscent of Kontsevich's construction of his graph complex. The $E_1$ term of the BS-spectral sequence is  given by a \textit{graph complex}
$\mathcal G_A(G)$ of the graded commutative algebra $A = H^*(M, \bbf)$. This is a bigraded complex and the authors verify that its higher differentials
are obtained from the Massey products of $A$, as conjectured by Bendersky and Gitler for the special case of the classical configuration space, thus verifying their conjecture. Note that, by construction, the BS-spectral sequence converges to the cohomology of the pair $(M^{|V|},\Delta)$, where $\Delta = \bigcup\Delta_{i,j}, \{i,j\}\in E$, see \eqref{deltaij}, thus converges to the cohomology with compact supports of $\conf{G}(M)$, with suitable coefficients, as explained in \S\ref{compactsupport} next. Petersen \cite{petersen} expands on this observation and uses it to compute models for compactly supported cohomology for these configuration spaces.

\subsection{Compactly supported cohomology}\label{compactsupport}
Cohomology with compact supports is intrinsically related to duality and leads to interesting developments \cite{arabia, petersen}. Arabia in his monograph \cite{arabia} makes the justified claim that cohomology with compact supports is a most suitable theory to study configuration spaces. 
This was already apparent in an early appendix by Segal \cite{segalacta} who uses $H^*_c$ to prove the stability of the homology of configuration spaces (see \S\ref{stability}). 
The starting point is the isomorphism valid for any compact Hausdorff space $M$
\begin{equation}\label{dualgroup}
H^*_c(\conf{k}(M);\bbf)\cong H^*(M^k, \Delta_{fat}(M);\bbf)
\end{equation}
The above can be taken as a definition of $H^*_c$ on this category of spaces, and it is related via Poincar\'e duality to the homology of $\conf{n}(M)$ (this is discussed in \S\ref{duality}).
Generally, $H^*_c$ is defined \textit{sheaf-theoretically} on spaces that are paracompact locally compact Hausdorff spaces. The discussion below applies to those spaces, and we do not need to restrict to manifolds anymore. In this case,
a proper homotopy equivalence $f: X\rightarrow Y$ induces an isomorphism $H^*_c(\conf{k}(Y),\bbf )\rightarrow H^*_c(\conf{k}(X),\bbf )$ in each degree, and if $X$ and $Y$ are compact, the
same follows for any homotopy equivalence \cite{arabia, gadishhainaut}.

A main property enjoyed by $H^*_c$ on configuration spaces of a class of spaces is that it \textit{splits}. More precisely, we can use the following example from \cite{petersen} for $k=2$ to explain the phenomenon (this is expanded in \cite{arabia}, Chapter 3). There is a Gysin long exact sequence
$$ \ldots \to H^k_c(\conf{2}(X)) \to H^k_c(X^2) \to H^k_c(X) \to H^{k+1}_c(\conf{2}(X)) \to \ldots $$
associated to the closed diagonal inclusion of $X$ into $X^2$. If we work with field coefficients, 
then the map $H^k_c(X^2) \to H^k_c(X)$ is given by multiplication in the cohomology ring $H^*_c(X)$; consequently, the compactly supported cohomology groups of $\conf{2}(X)$ are completely determined by the compactly supported cohomology of $X$, with its ring structure. No such simple statement is true for the usual cohomology or homology. For $k > 2$, it is no longer true that $H^*_c(\conf{n}(X))$ depends only on the ring structure on $H^*_c(X)$, not even in cases where the ring structure is
identically zero. However, Arabia
has introduced the class of \textit{i-acyclic
spaces} for which the cup-product
on $H^*_c(X)$ vanishes in a very
strong sense, and for i-acyclic spaces
$H^*_c(\conf{n}(X))$ depends only on
$H^*_c(X)$. A topological space $X$ is $i$-acyclic\footnote{Here $i$ stands for ``interior cohomology'', a term presumably due to Karoubi.} over a commutative ring $R$ if $H^k_c(X,R) \to H^k(X,R)$ is the zero map for all $k$. This condition is in fact satisfied in many cases of interest: for example, any open subset of Euclidean space is $i$-acyclic, and the product of any space with an $i$-acyclic space is $i$-acyclic. As pointed out in \cite{petersen}, the remarkable fact about $i$-acyclicity is that it is exactly the right hypothesis to ensure that the compactly supported cohomology of configuration spaces of points on $X$ depends in the simplest possible way on the compactly supported cohomology of $X$ itself. This fact was already observed by Cohen-Taylor and Jie Wu \cite{wu} who obtained near-complete calculations of $H_*(\conf{k}(M\times\bbr))$, for $M$ a manifold. 

\begin{quote}\textit{{\bf Theorem 6.1} \cite{arabia}: Let $X$ be an $i$-acyclic paracompact locally compact Hausdorff space over a field $\bbf$. Then $H^*_c(\conf{k}(X),\bbf )$ depends only on the graded vector space $H^*_c(X,\bbf )$.}\end{quote}
\vskip 5pt

Petersen \cite{petersen} shows that $H^*_c(\conf{n}(X),\bbq)$ depends only on the choice of a cdga model for the compactly supported cochains $C_c^*(X,\bbq)$, thus refining Theorem 6.1. We note finally that when $X$ is a simplicial set with finitely many non-degenerate simplices, and geometric realization $|X|$, then $H^*_c (\conf{n}(|X|, n), \bbq)$ can be computed from the \textit{Hochschild–Pirashvili} cohomology of $X_+ = X\sqcup \{*\}$ with coefficients in an explicit $\bbq$ coalgebra (definition and detail in 
\cite{gadishhainaut},\S3).


\section{Identical particles on manifolds}\label{unorderedconfigs}

\noindent{\bf 7.0.1.} We summarize what is known about the (co)homology of $C_k(M)$ for manifolds $M$, and various ring coefficients. The starting point is the isomorphism valid in characteristic zero
\begin{equation}\label{invariants}
H^*(C_k(M);\bbq)\cong H^*(\conf{k}(M);\bbq)^{\mathfrak S_k}
\end{equation}
where the righthand term is the subring of invariant classes under the action of the symmetric group.
This identity is used for example in \cite{bc} to compute the Betti numbers of $C_k(X)$ for $X$ a Riemann surface. More importantly, this is the method used by F\'elix and Tanr\'e \cite{feltan} and Church \cite{church} to obtain the following general result. We write $Sym$ for the symmetric algebra,  $\Lambda$ for the exterior algebra, and respectively $\hbox{Sym}^i$, $\Lambda^i$ the submodules generated by the monomials of length $i$.

\begin{quote}\textit{{\bf Theorem 7.1} \cite{bct, feltan, church}: Let $M$ be a compact odd-dimensional manifold, and let $\bbf=\bbq$ or $\bbf_p$, with $p>k$. Then 
$$H^*(C_k(M);\bbf)\cong (H^*(M)^{\otimes k};\bbf)^{\mathfrak S_k}\cong \bigoplus_{i+j=k}Sym^i H^{even}(M)\otimes\Lambda^j H^{odd}(M)$$}
\end{quote}

In the world of \textit{graded vector spaces},  one writes $\wedge V=\hbox{Sym}(V) = \hbox{Poly}(V^{\tiny even})\otimes\hbox{Ext}(V^{\tiny odd})$, 
the free commutative graded algebra on the graded vector space $V$. The 
right-hand side of Theorem 7.1 becomes $\hbox{Sym}^k H^*(M)$,  and  succinctly, the Theorem takes the form
$$\bigoplus_{k\geq 0}H_*(C_k(M);\mathbb Q)\cong \hbox{Sym}(H_*(M;\mathbb Q))\ \ \ ,\ \ \ \hbox{$\dim M$ odd}$$
To prove this result, the method of \cite{bct} is explained in the next paragraph. The method used by  \cite{feltan, church} is to observe that the Cohen-Taylor-Totaro spectral sequence (see \S\ref{spectral}) has an action of  $\mathfrak S_k$, 
so it is in fact a spectral sequence of $\mathfrak S_k$-algebras, and consequently it induces a spectral sequence $E_r^{\mathfrak S_k}$ of differential graded commutative algebras
consisting of the $\mathfrak S_k$-invariants of each stage. 
Surprisingly, this spectral sequence of  $\mathfrak S_k$-invariants degenerates at $E_2$ for any compact manifold $M$ with rational or $\bbf_p$ coefficients, if $p>k$. This is in contrast with \cite{ft} which produces non-zero higher differentials in the Cohen-Taylor spectral sequence (i.e. before passing to invariants).
For even-dimensional manifolds, the answer is not as simple but is very explicit as well (\cite{feltan}, Theorem 1).

\noindent{\bf 7.0.2.} For general field coefficients, one can use an alternative potent observation which is that $H_*(C_k(M))$ can be read off directly from the homology of the much larger
space $C(M;X)$ introduced in \S\ref{maymilgram}, with $X$ a sphere \cite{bct}. This approach enables complete calculations of the homology of the unordered configuration spaces with mod-$2$ and mod-$p$ ($p$ odd) coefficients, albeit it is not suitable for either the product structure or the Steenrod operations.  The following is the main calculation of \cite{bct}. Again, $\Omega^mX$  denotes the $m$-fold based loop space of $X$.

\begin{quote}\textit{{\bf Theorem 7.2} \cite{bct}: 
Assume $n\geq 2$ and $n+m$ odd if $\bbf\neq\bbf_2$. There is an isomorphism of graded vector spaces
$\displaystyle H_*(C(M,S^n);\bbf)\cong \bigotimes_{q=0}^m H_*(\Omega^{m-q}S^{m+n},\bbf)^{\otimes\beta_q}$,
where $\beta_q=dim H_q(M;\bbf)$ is q-th Betti number of $M$.}
\end{quote}
\vskip 5pt

To see how to read off the homology of the configuration space from here, one considers the cofiber of the inclusion of $C_{k-1}(M;X)$ into $C_{k}(M;X)$, this is $D_k(M;X):= \conf{k}(M)\rtimes_{\mathfrak S_k}X^{\wedge k}$, where $X^{\wedge k}$ means the $k$-fold smash product of $X$. The homology of $D_k(M;X)$ \textit{embeds} in that of $C(M, X)$ \cite{snaith}. On the other hand, each factor $H_*(\Omega^{m-q}S^{m+n})$ is an algebra with weights associated to its generators, and the homology of $D_k(M, S^n)$ is precisely the vector subspace generated by the elements of weight $k$. Since one understands the homology of this ``Snaith summand'' $D_k(M;S^n)$, it remains to relate it to that of $C_k(M)$, and one main observation in \cite{bct} is that $D_k(M;S^n)$ is the Thom space of the $n$-fold Whitney sum $n\eta_k$ of the bundle
$\eta_k: \bbr^k\rightarrow \conf{k}(M)\times_{\mathfrak S_k}\bbr^k\longrightarrow C_k(M)$. By the Thom isomorphism, its homology, is up to a shift, that of the unordered configuration spaces (with twisted coefficients if $n$ is odd).
As a consequence of Theorem 7.2,  one proves bullet point (6) in \S\ref{invariance} (\cite{bct}, Theorem C).

\subsection{Duality}\label{duality}
This covers primarily work in \cite{bcm, kallel1, oscar}. This approach has been most successful for computations. The starting point is the observation that the one-point compactification $\conf{k}(M)^+$ is homeomorphic to $(M^+)^{\wedge k}/\Delta_{fat}((M ^+)^{\wedge k})$, where again $\Delta_{fat}$ is the fat diagonal where two entries coincide. The symmetric group $\mathfrak S_k$ acts on $M^k$, on $+$ and $\Delta_{fat}$, so acts on the quotient. Since the ``singular set'' of the action (where the action is not free) is contained in $\Delta_{fat}$, Poincar\'e-Lefshetz duality isomorphism holds, and  tells us that, for oriented $M$, the reduced homology of $\conf{k}(M)$ is the cohomology of $(M^+)^{\wedge k}/\Delta_{fat}((M ^+)^{\wedge k})$, up to a shift of $nk$ (as already discussed for the BG-spectral sequence in the compact case \S\ref{spectral}).   
In \cite{kallel1}, the following general homological dimension result was established. Let $M$ be compact, $r$ its connectivity if $\partial M=\emptyset$, or the connectivity of $M/\partial M$ otherwise. Suppose $M$ is even dimensional orientable, $0\leq r<\infty$ and $k\geq 2$. Then 
$$cohdim_{\bbz} (C_k(M)) \leq  \begin{cases} (d-1)k - r+1, &\hbox{if} \ \partial M=\emptyset\\
(d-1)k-r, &\hbox{if}\ \partial M\neq\emptyset
\end{cases}
$$
A similar statement is valid for $M$ odd-dimensional with mod-$2$ coefficients. 

In both \cite{bcm, oscar}, the authors make a strong connection with homological algebra. Here, the observation is that the spaces $C_k(M)^+$ assemble into an
associative, commutative, and unital monoid, with operation
$C_n(M )^+\wedge C_m (M )^+\rightarrow C_{n+m} (M )^+$ whose homology is that of a derived relative tensor product which may be computed by the two-sided
bar construction. By duality, one recovers the homology of $C_k(M)$. Good applications follow, and extensions to more general configuration spaces are possible. For instance \cite{oscar} recovers the computations of Farb, Wolfson and Wood for spaces of ``$0$-cycles'' and the present author's spectral sequence for ``divisor spaces''. 

\noindent{\bf 7.1.1.}  An important earlier result was given by Knudsen who describes the rational homology of the unordered configuration spaces of an arbitrary manifold $M$, possibly with boundary, as the homology of a Lie algebra constructed from the compactly supported cohomology of $M$. By locating the homology of each configuration space within the ``Chevalley-Eilenberg complex'' of this Lie algebra, he was able to extend the theorems of Bodigheimer-Cohen-Taylor and F\'elix-Thomas, presented in \S\ref{labeled}, reproducing as well the homological stability results of Church (see \S\ref{stability}) and Randal-Williams.

\begin{quote}\textit{{\bf Theorem 7.3} \cite{knudsen2}: 
Let M be an orientable $n$-manifold. There is an isomorphism
$$\bigoplus_{k\geq 0}H_*(C_k(M );\bbq )\cong H^{Lie}_*(\mathfrak g_M )$$
of bigraded vector spaces, where  
$\mathfrak g_M: H^{-*}_c
(M )\otimes \mathcal L(v_{n-1},1)$ is a Lie algebra with bracket determined up to sign by the cup product, and
$H^{Lie}_*(\mathfrak g_M )$ its Lie algebra homology.}
\end{quote}
\vskip 5pt

It takes room to decipher this theorem and we refer to \cite{knudsen} (\S9) for a clear explanation. Here 
$H^{Lie}(\mathfrak g_M)=Tor_*^{U(\mathfrak g_M)}(\bbq,\bbq)$, where $U$ stands for universal enveloping algebra, and $H^*_c$ stands for cohomology with compact supports concentrated in negative degrees and weight $0$. When $n$ is odd, the Lie bracket on $\mathfrak g_M$ vanishes, and one recovers immediately Theorem 7.1.
Theorem 7.3 is recovered in \cite{idrissi2} and more succinctly in \cite{oscar}. Note that the appearance of compactly supported cohomology is not at all surprising in light of \S\ref{compactsupport}. 

Finally, we mention that \cite{malin} uses another version of duality, that of \textit{Spanier-Whitehead}, to establish the proper homotopy invariance of the equivariant stable homotopy type of the configuration
space $\conf{k}(M)$ for a topological manifold $M$.


\section{Applications: a fundamental sample}\label{homotopytheory}

Configuration spaces intervene in fundamental ways in many fields of mathematics and are used in the solution of important and beautiful problems.

\subsection{Coincidences of maps}\label{coincidences}
Applications of configuration spaces to combinatorial topology are numerous.
An elementary application of the cohomology computations of \S\ref{homology} and \S\ref{cohomology}, and a precursor of many things to come, is the following slick proof of the Borsuk-Ulam theorem which we also include for its attractiveness. Start with a continuous function $f: S^n\rightarrow\bbr^n$. The classical Borsuk-Ulam theorem asserts that there is an antipodal pair $(p,-p)$ in $S^n$ such that $f(p) = f(-p)$. To show this, consider the extended mapping 
$F: S^n\rightarrow\bbr^n\times\bbr^n, F(p) = (f(p),f(-p))$, and suppose no antipodal pair as above exists. Then $F$ has image in $\conf{2} (\bbr^n)$, and it is $\bbz_2$-equivariant, so it induces a map between orbit spaces
$$\widehat F: \bbr P^n\longrightarrow C_2(\bbr^n)\simeq\bbr P^{n-1}$$
This induced map on cohomology with $\bbf_2$ coefficients gives an impossibility (using the cohomology ring structure of projective space), contradicting the claim that $f(p)\neq f(-p)$ for all $p\in S^n$. 
This circle of ideas has obvious generalizations. One can consider a connected Hausdorff space $X$ with free cyclic action by $\bbz_p$, $p$ a prime, and study conditions on maps $f: X\rightarrow Y$, $Y$ a manifold, such that there exists a point $x\in X$ with $f(x) = f(\sigma^ix)$ for some $i\neq 0$.
This question is addressed in \cite{fredlusk} using configuration spaces, with extensions in \cite{blz}.
 
The best-known application of configuration spaces to equivariant topology is related to the problem of equipartitions by convex bodies and to the resolution of the Nandakumar and Ramana Rao conjecture \cite{nrr, jelic}. The conjecture states the following:
\textit{Every convex polygon P in the plane can be partitioned
into any prescribed number n of convex pieces that have
equal area and equal perimeter.}
It was verified by the authors themselves for $n=2$ using the intermediate value theorem,
and then more generally for $n=2^k$. For general $n$,
Karasev, Hubbard and Aronov \cite{kha} showed that the conjecture, and its natural generalization to higher dimension, would follow from the non-existence
of a $\mathfrak S_n$-equivariant map out of configuration spaces, just like the Borsuk-Ulam earlier argument. They then proved that
this map did not exist for $n$ a prime power, thus proving the conjecture under
that condition.
More precisely, they consider the linear subspace
$W_n :=\{(x_1,\cdots, x_n)\in\bbr^n\ |  x_1+x_2+\cdots + x_n=0\}$.
The symmetric group $\mathfrak S_n$ acts on $W_n$, and on its unit sphere $S(W_n)$, by 
permutation of coordinates. The truth of the N-RR conjecture in partitioning convex body in $\bbr^d$
now reduces to the non-existence of
a $\mathfrak S_n$-equivariant map
\begin{equation}\label{themap}
f: \conf{n}(\bbr^d)\lrar S(W_n^{\oplus d-1})
\end{equation}
In verifying this conjecture for $n$ a prime power, \cite{kha}
establish similar results regarding equipartitions with respect to continuous functionals and absolutely continuous measures on convex bodies. These include a generalization of the ham-sandwich theorem to arbitrary number of convex pieces confirming a conjecture of Kaneko and Kano, a similar generalization of perfect partitions of a cake and its icing, and a generalization of the Gromov-Borsuk-Ulam theorem for convex sets in the model spaces of constant curvature.

The N-RR conjecture was alternatively approached by Blagovic and Ziegler \cite{blagozieg} who
give a simpler proof using equivariant obstruction theory, as developed by Tom Dieck. The key element is the construction of an explicit $\mathfrak S_n$-equivariant $(d-1)(n-1)$-dimensional cell complex that is a $\mathfrak S_n$-equivariant deformation retract of $\conf{n}(\bbr^d)$ (see \S\ref{cellularmodels}).  The obstruction technique approach has the advantage of providing a converse result, namely there exists a nontrivial equivariant map $f : \conf{n}(\bbr^d)\lrar S(W_n^{\oplus d-1})$ if and only if
$n$ is not a prime power, $n,d\geq 2$. Therefore, for $n$ a prime power, no such map exits, and this yields to the complete resolution of the conjecture, stated below as a theorem. 

\begin{quote}\textit{{\bf Theorem 8.1} \cite{blagozieg, kha}: Given a convex body $K$ in $\bbr^d$, a prime $p$ and a positive integer $k$, it is possible
to partition $K$ into $n = p^k$ convex bodies with equal $d$-dimensional volumes and equal $(d-1)$-
dimensional surface areas.}\end{quote}
\vskip 5pt

A different proof of his conjecture, using the Fadell-Husseini index, and an extension of the coincidence results related to the Borsuk-Ulam theorem are in \cite{blz}.

\subsection{Link homotopy} Using classical homotopy theory and the Fadell-Neuwirth fibrations, it is possible to construct invariants of links, living in homotopy groups of spheres \cite{massey}. A $3$-link of spheres in Euclidean space is an embedding of pairwise distinct spheres
$$L: S^{p_1}\cup S^{p_2}\cup S^{p_3}\lrar \bbr^m\ \hbox{or}\ S^m\ \ \ \ ,\ \ \ \ 
p_i<m-2$$
The map $L$ is continuous and $L(S^{p_i})\cap L(S^{p_j}) = \emptyset$ for $i\neq j$. 
In 1954, Milnor \cite{milnor} introduced the notion of link homotopy to classify circle links in three space.
Note that in codimension at least $3$ (i.e. $m-p_i>2$, $i=1,2,3$), the theories for links in $\bbr^m$ and $S^m$ are equivalent.

Any $3$-link $L$ as above induces a map
$$\tilde L: S^{p_1}\times S^{p_2}\times S^{p_3}\lrar\conf{3}(\bbr^m)\ \ ,\ \ 
\tilde L(x_1,x_2,x_3) = (L(x_1), L(x_2), L(x_3))$$
The homotopy class of $\tilde L$ is invariant under link homotopies of $L$ (\cite{massey, koshorke}). Since $\conf{3}(\bbr^m)$ is $m-2$-connected, the restriction of $\tilde L$ to each sphere factor is null-homotopic, and so by restricting to any product of two spheres, we get induced maps $S^{p_i+p_j}=S^{p_i}\times S^{p_j}/S^{p_i}\vee S^{p_j}$ (the smash product) into $\conf{2}(\bbr^m)\simeq S^{m-1}$ (the homotopy class of such a map is a $2$-link invariant). To get the $3$-link invariant, we factor $\tilde L$ up to homotopy through the quotient $W(p_1,p_2,p_3):={S^{p_1}\times S^{p_2}\times S^{p_3}/ S^{p_1}\vee S^{p_2}\vee S^{p_3}}$, $p_i\geq 1$. It can be checked that the set of based homotopy classes of maps
$[W(p_1,p_2,p_3),X]$ has a natural group structure, and is isomorphic to
$\pi_{p_1+p_2}(X)\oplus \pi_{p_1+p_3}(X)\oplus \pi_{p_2+p_3}(X)\oplus \pi_{p_1+p_2+p_3}(X)$. The map $\tilde L$ defines therefore a class $\gamma (L)\in\pi_{p_1+p_2+p_3}(\conf{3}(\bbr^m))$.
To get a homotopy class of a sphere, consider the combined Fadell-Neuwirth projections
$p: \conf{3}(\bbr^m)\lrar (\conf{2}(\bbr^m))^3$, sending $p(x,y,z) = ((x,y), (x,z), (y,z))$. This is not a fibration, but has homotopy fiber $F_m$ that is $2m-4$-connected. Moreover Massey shows that $p_*(\gamma(L))=0$. By the short exact sequence of homotopy groups for this homotopy fibration, with fiber $F_m$, there is a class $\beta (L)\in\pi_{p_1+p_2+p_3}(F_m)$ that lifts $\gamma (L)$. It can be checked as well that the groups $\pi_{*}(F_m)$ are isomorphic to $\pi_{*}(S^{2m-3})$ if $*< 3m-5$, and since $p_1+p_2+p_3<3m-5$, $\beta (L)$ lives in $\pi_{p_1+p_2+p_3}(S^{2m-3})$. This is Massey's $3$-link invariant. An obvious necessary condition that two such links $L$ and $L'$ should be link homotopic is that $\beta (L) = \beta ( L')$ and that each $2$-component sublink of $L$ should be link homotopic to the corresponding $2$-component sublink of $L'$. Massey shows that under certain restrictions on the dimensions $p_1,p_2,p_3$ and $m$, this necessary condition is also sufficient. 

\subsection{Embedding theory}\label{embeddingtheory} 
Configuration spaces are most naturally associated with embeddings. Indeed, if $f: M\rightarrow N$ is an embedding, it defines a map $\conf{k}(f): \conf{k}(M)\rightarrow\conf{k}(N)$, for every $k\geq 1$. Restrictions at the level of cohomology can lead to obstructions for the embedding of $M$ in $N$.

\noindent{\bf 8.3.1.} A beautiful and powerful necessary and sufficient condition for the PL-embeddings of finite simplicial complexes in Euclidean space was given surprisingly early in 1933 by VanKampen, anticipating by a few years the development of cohomology. We describe this \textit{obstruction} based on \cite{freedteich}. Write the \textit{discretized} configuration space \cite{abrams} (or also the \textit{simplicial configuration space})
$$D_k(X) = \bigcup_{{\sigma_i\ \hbox{\tiny pairwise disjoint}}\atop
{\hbox{\tiny closed cells in}\ X}}\sigma_1 \times\cdots \times \sigma_k$$
In the case $k=2$ for instance, $D_2(X)$ is the set of all pairs $(x_1,x_2)\in X^2$, with $x_1,x_2$ lying in disjoint closed simplices of $X$. This is a deformation retract of $\conf{2}(X)$ \cite{hu}, which implies the amusing little observation of \cite{copeland} that for $X$ a \textit{finite simplicial} complex, and $\Sigma$ the \textit{unreduced} suspension, there is a homotopy equivalence
$\displaystyle\Sigma\conf{2}(X) \simeq \conf{2}(\Sigma X)$.

For $k\geq 3$, $D_k(X)$ has different homotopy type than $\conf{k}(X)$ in general. This is however a simplicial complex, in fact the largest such complex that is contained in the product $X^k$ minus its thick diagonal. This complex has been used fundamentally to study configuration spaces of graphs $\conf{k}(\Gamma)$, and Abrams spells out in this case an explicit criterion as to when the two configuration spaces are equivalent \cite{abrams}. We mention (another) amusing result of Ummel \cite{ummel} who proves that $D_2(\Gamma)$ is a closed topological surface if and only if $\Gamma=K_5$ or $K_{3,3}$. 

If $C_*(X)$ is a cellular chain complex of a CW-complex $X$, and if $G$ acts freely on $X$, and $M$ is a $\bbz [G]$-module, then one defines
$H_G^q(X;M) := H_q(\hbox{Hom}_{\bbz [G]}(C_*(X),M)$. Note that if $G$ acts freely on $X$ and trivially on $M$, then $H_G^q(X;M)=H^q(X/G;M)$. Let $K$ be an $n$-dimensionaly complex, $n\geq 1$, and let $f: K\rightarrow\bbr^{2n}$ be a PL immersion which one can assume to be such that $f(\sigma)\cap f(\tau)=\emptyset$ if $\dim (\sigma) + \dim (\tau )\leq 2n-1$, for two open cells $\sigma$ and $\tau$ of $X$, while in the top dimension, the cells intersect each other transversally in at most a finite number of double points. Define the cochain $o_f$ on $D_2(K)$ by 
$o_f(\sigma\times\tau):= f(\sigma)\cdot f(\tau)\in\bbz$ (this is the intersection number, assuming a fixed orientation in $\bbr^{2n}$). The generator of $\bbz_2$ acts on $D_2(K)$ via $\iota (x,y)=(y,x)$, and one has
$$(o_f\circ \iota)(\sigma\times\tau) = o_f(\tau\times\sigma) = f(\tau)\cdot f(\sigma) = (-1)^nf(\sigma)\cdot f(\tau)$$
We can consider now $M=\bbz$ as a $\bbz [\bbz_2]$-module (the \textit{group ring}) with action of $\bbz_2$ on $\bbz$ by multiplication by $(-1)^n$ (where $n$ is half the ambiant dimension $\bbr^{2n}$).
The following theorem is attributed to Van-Kampen, Shapiro and Wu in the 1950's.

\begin{quote}\textit{{\bf Theorem 8.2} \cite{freedteich}: Given $K$ an $n$-dimensional simplicial complex, and a PL-immersion $f:K\rightarrow\bbr^{2n}$, the cohomology class $o(K)$ of $o_f$ is an element of $H_{\bbz/2}^{2n}(D_2(K),\bbz)$, it is independent of $f$, and when $n\geq 3$, its vanishing is a necessary and sufficient condition for the existence of an embedding $K\hookrightarrow\bbr^{2n}$.}
\end{quote}
\vskip 5pt

A detailed proof, with the relevant references, is in \cite{freedteich} who also show that this obstruction fails in dimension $2$ by exhibiting a two-dimensional simplicial complex $K$ whose $o_K$ is trivial but which does not admit a PL embedding in $\bbr^4$. Note that in the case $n=1$, Sarkaria shows that this obstruction provides a necessary and sufficient condition in that dimension, and is thus equivalent to Kuratowski's subgraph condition.

\noindent{\bf 8.3.2.} Haefliger  and Wu were the first to investigate the embeddings of compact differentiable manifolds in Euclidean spaces using their induced equivariant maps out of configuration spaces \cite{haefliger}. Suppose $f: M\hookrightarrow \bbr^m$ is a given embedding of a manifold $M$, then $f$ induces a map $\tilde f: \conf{2}(M)\rightarrow S^{m-1}$ sending $\displaystyle (x,y)\longmapsto {f(x)-f(y)\over |f(x)-f(y)|}$. This map is clearly equivariant with respect to the involution that interchanges the factors of $\conf{2}(M)$ and the antipodal map of $S^{m-1}$. Also, an isotopy $f_t$, $t\in [0,1]$, of two embeddings $f_0,f_1$ of $M$ in $\bbr^m$ induces an equivariant homotopy $\tilde f_t$. The following theorem is at the start of vast generalizations by work of Goodwillie, Klein, Weiss, Williams, and collaborators (more below).
Consider the correspondence that associates to an isotopy class of a differentiable embedding $f$ the equivariant homotopy class of the map $\tilde f$ above.

\begin{quote}\textit{{\bf Theorem 8.3} \cite{haefliger}: Let $M$ be $n$ dimensional compact differentiable manifold. Consider the Haefliger-Wu correspondence which associates to an isotopy class of a differentiable embedding $f: M\hookrightarrow\bbr^m$ the equivariant cohomotopy class $\tilde f$ defined above. Then this correspondence is surjective if $2m\geq 3(n+1)$ and bijective if $2m > 3(n+1)$.}
\end{quote}
\vskip 5pt

In particular,  there exists a differentiable embedding of $M$ in $\bbr^m$ provided $2m\geq  3(n + 1)$.
Computability of the cohomotopy class is possible in certain situations and leads to explicit conditions on the embeddability of projective spaces for example \cite{handel}. The idea is that, since $\tilde f$ above is $\bbz_2$-equivariant, one should consider the induced quotient map $C_2(M)\rightarrow\bbr P^{m-1}$ and then translate those conditions into the vanishing of the pullback of the first Stiefel-Whitney class of the double cover $\conf{2}(M)\rightarrow C_2(M)$.

\noindent{\sc Example}: Let
$n\geq 4$. There exists a unique isotopy class of embeddings of $\bbp^n$ (complex projective space) in $\bbr^{4n}$. There exists just two isotopy classes of embeddings of $\bbp^n$ in $\bbr^{4n-1}$ \cite{yasui}.

\noindent{\bf 8.3.3.} Embedding theory takes a sharp turn with the introduction of the Goodwillie-Weiss calculus \cite{aroneturchin, weiss}.
As in the Haefliger-Wu theory, the starting point is the space of immersions, which is quite well understood by Smale-Hirsch theory, then constructing interpolating spaces $T_k\hbox{Emb}(M,N)$, similar to the truncation of a ``Taylor series'', which enjoy some homotopy invariance property and ``converge'' to $\hbox{Emb}(M,N)$ through a tower
$$\hbox{Emb}(M,N)\rightarrow T_\infty\rightarrow\cdots\rightarrow T_k\rightarrow\cdots\rightarrow T_1=\hbox{Imm}(M,N)$$
It is a deep theorem of Goodwillie and Klein that the $k$-th approximation $ev_k: \hbox{Emb}(M,N)\rightarrow T_k(\hbox{Emb}(M,N))$ is $(3-n+(m+1)(n-m-2)$-connected, if $\dim M\leq \dim N-3$
(this excludes $m=1$ and $n=3$, so this excludes knot theory).
Haefliger’s theory of embeddings in
the stable range, which we already discussed, appears as a calculation by second-order Taylor approximation \cite{weiss}. 
The main connection with configuration spaces is that the homotopy fiber of the map $T_k\hbox{Emb}(M,N)\rightarrow T_{k-1}\hbox{Emb}(M,N)$ is a bundle over configuration space $C_k(M)$. That's not all, in fact, the wording ``calculus'' for this theory is quite well chosen since maps $\bbr\rightarrow\bbr$ have as analogy functors  $F(-)$ from $\mathcal O^{op}(M)\rightarrow \hbox{Top}$, where $\mathcal O(M)$ is the category of open sets of $M$ with inclusions. These functors send isotopy equivalences to homotopy equivalences, and the analog of continuity would state that for all series of embeddings $U_0\subset U_1\subset U_2\subset\cdots$ in $\mathcal O(M)$, the natural map $F(\bigcup_kU_k)\rightarrow \hbox{holim}_kF(U_k)$ is a homotopy equivalence. 
The analogs of real polynomials become the polynomial functors $T_k$, and the \textit{linear} part is $\hbox{Imm}(-,N)$ which has the characteristic property that for all
$U_0,U_1\in\mathcal{O}(M)$, the natural map
$$F(U_0\cup U_1)\longrightarrow \hbox{holim} (F(U_0)\leftarrow F(U_0\cap U_1)\rightarrow F(U_1))$$
is a weak equivalence.
Since polynomials are determined by their value at $0$ and $k$ distinct points, the polynomial functors are determined in a suitable sense by their value on $\emptyset$ and on $\leq k$ disjoint disks of the dimension of $M$ \cite{weiss}.  

Central to the theory is therefore understanding embeddings of a disjoint union of disks $\hbox{Emb}(\bigsqcup D^m,\bbr^n)$. 
The main result there is that if $L(\bbr^m,\bbr^n)$ is the space of injective linear maps from $\bbr^m$ to $\bbr^n$,  then there is a homotopy equivalence
$$\xymatrix{\hbox{Emb}( \bigsqcup_k D^m,\bbr^n)\ar[r]^-\cong& \conf{k}(\bbr^n)\times L(\bbr^m,\bbr^n)^k}$$
defined by evaluating an embedding at the centers of the discs, and also differentiating
at the centers of the discs.
The current state-of-the-art of embedding theory relates the space of embeddings $\hbox{Emb}(M,N)$ to ``derived mapping spaces'', framed configuration spaces and operad theory \cite{fresse}. 


\subsection{Marked curves and Hurwitz spaces}\label{geometry}
Configuration spaces appear in algebraic geometry as branched points (or Weirstrass points) of algebraic curves, and for special families of curves, they determine entirely the isomorphism class of the curve, up to an automorphism. More precisely, fix $p\geq 2$, $B = \{b_1,\ldots, b_n\} \in C_n(\bbc)$, and consider the  ``$p$-elliptic'' curve which is the projective compactification of the curve having affine equation 
\begin{equation}\label{curve}
S_B = \{(x,y)\in\bbc^2\ | \ y^p = (x-b_1)\cdots (y-b_n), b_i\in B\}
\end{equation}
The case $p=2$ corresponds to hyperelliptic curves.
 Notice that each surface
$S_B$ is a smooth affine complex curve and the projection 
$S_B\longrightarrow\bbc\ ,\ (x,y)\mapsto x$ is a $p$-fold branched cover that extends to the projectivized curve to give a branched covering over the Riemann sphere
$\bar\pi : \bar{S}_B\rightarrow\bbp^1$. This means that $\pi$ is a regular $p$-covering over $\bbc\setminus B$, and over each $B$, the leaves of the covering come together in a way codified by a monodromy map. The Riemann-Hurwitz formula for the Euler characteristic gives a relation between $p, n$ and the genus of $\bar{S}_B$.
A pleasant account of this very classical material is in \cite{bott}.  If we write $E = \{(x,y,B)\ |\ (x,y)\in S_B\}$, then the projection $(x,y,B)\mapsto B$ is a surface bundle over
$C_n(\bbc)$. The study of this single $S$-bundle is already very interesting, with connections to representations of braid groups (see {\bf 2.2.3}) and geometric structures on moduli spaces of Riemann surfaces \cite{mcmullen}. A very readable account is in \cite{salter}.

The branched covers $\bar\pi : \bar{S}_B\rightarrow\bbp^1$ totally determine the isomorphism type of the $p$-elliptic curve \eqref{curve}. If we focus on the case of hyperelliptic curves ($p=2$), with genus $g$, then the covering is branched over $2g+2$ points of $\bbp^1$. These Weirstrass points determine the curve uniquely up to an action of $Aut(\bbp^1)=PSL_2(\bbc)$. So if we denote by $M_{0,2g+2}:= \conf{2g+2}(\bbp^1)/Aut(\bbp^1)=\conf{2g+1}(\bbc)/\hbox{Aff}(\bbc)$ and by 
$\widetilde{M}_{0,2g+2}:= C_{2g+2}(\bbp^1)/Aut(\bbp^1)$, 
the map
$$\xymatrix{h: \widetilde{M}_{0,2g+2}\ar[r]^\cong& \mathcal H_g\subset \mathcal M_g}$$
given by taking the isomorphism class of the projectivized double cover branched at the marked points, is an isomorphism of varieties. 
 Here of course
$\mathcal H_g$ is the moduli space of hyperelliptic curves sitting inside $\mathcal M_g$  the moduli space of smooth curves of genus $g$, the variety whose points are in one-to-one correspondence with isomorphism classes of smooth curves of genus $g\geq 2$.
For $g=2$, $\mathcal H_2=\mathcal M_2$, for $g=3$ the image of $h$ has codimension one, and for $g\geq 4$ the image has a higher codimension and isn’t a divisor.
In all cases, configuration spaces are an essential ingredient in the study of the algebraic topology of the moduli space of hyperelliptic curves. For example, an application of the scanning (electrostatic) map \eqref{scanningmap} to these moduli is as follows: the map \eqref{scanningmap} descends to a map  $\widetilde{M}_{0,n}\rightarrow \left(\Omega^2_{n-1}S^2\right)/S^1$, into the based loop space of all degree $n-1$ maps, quotiented by the circle acting on functions by rotating the codomain. As in the case of the scanning map, this map is an isomorphism through a range of degrees. The homology of the loop space (modulo $S^1$) is entirely torsion \cite{westerland}. This recovers the surprising but well-known result that $\widetilde{M}_{0,n}$ has the rational homology of a point, i.e. $H_p(\widetilde{M}_{0,n};\bbq )=0$, if $p>0$.

The study of the moduli spaces $M_{g,n}$ of genus $g$ curves with $n$ marked points, is an important and hard problem in topology. These spaces have complicated and
mostly unknown rational cohomology. Considering the action of the symmetric group
$\mathfrak S_n$ by permuting marked points, this cohomology becomes a $\mathfrak S_n$-representation. Most recent results of \cite{gadishhainaut} shed light on the representations associated to $H^{n+*}_c(M_{2,n},\bbq)$.
Since $M_{g,n}$ is an algebraic variety, mixed Hodge theory produces a weight filtration
$gr_0^W=F_0\subset F_1\subset\cdots\subset F_i=H^i_c(M_{2,n};\bbq)$, and \cite{gadishhainaut} produce explicit subrepresentations of $F_0$, for small $i$, involving a number of copies of the representations
$sgn_n\otimes H_{even}(\conf{n}(S^3)/SU (2))$, where $sgn_n$ is the sign representation. 

\noindent{\bf 8.4.1.} Branched covers (and configuration spaces) are intimately related to \textit{Hurwitz spaces} and Hurwitz covers. Let $\pi: C\rightarrow\bbp^1$ be a non-constant holomorphic map into the Riemann sphere. Then it is a $d$ branched cover over $n$ branched points $b_i\in B\subset \bbp^1$. If we assume the cover is \textit{simple}, meaning that $\pi^{-1}(b_i)$ has
cardinality $d-1$ exactly, then
$n=2d+2g-2$ by the Riemann-Hurwitz formula (in particular $n$ is always even).
Two simple coverings $\pi: X\rightarrow\bbp^1$ and $\pi': Y\rightarrow\bbp^1$ are
equivalent means there is an isomorphism (of algebraic curves) $f: X\rightarrow Y$
over $\bbp^1$ (i.e. $\pi\circ f = \pi'$).  The Hurwitz scheme $\hbox{Hur}_{n,d}$
is the space of all equivalence classes of simple coverings of $\bbp^1$ (for a fixed $n$ and degree $d$). The spaces $H_{n,d}$ are connected (Hurwitz's) and in fact, irreducible smooth varieties (as proven by Fulton and Severi).
The projection $\hbox{Hur}_{n,d}\rightarrow C_n(S^2)$ is a covering whose degree is known as the \textit{Hurwitz number}. 
Currently, very little is known about the homology of the simple Hurwitz spaces, not even for small degrees. 

In applications \cite{bianchi, ellenberg}, one does not limit himself to simple coverings but considers  $(G,c)$-branched coverings $\pi: C\rightarrow \bbp^1$, where $G$ is a finite group and $c\subset G$ a union of conjugacy classes. Such a map $\pi$ is by definition, away from a set of $n$ points in $\bbp^1$, a regular covering with Galois group $G$, with the monodromy around these points prescribed to lie in $c$. The moduli of such maps is denoted Hur$^c_{G,n}$, and here too, the map which carries $\pi$ to its branch locus in $C_n(\bbp^1)$ is a covering map. These spaces have connected components and homology in a range that is in part computed in \cite{ellenberg}, with important applications to \textit{arithmetics}. The upshot consists in using these computations to produce an upper bound for the number of extensions with bounded discriminant and fixed Galois group of a rational function field $\bbf_q(t)$. The approach in \cite{ellenberg} is to restrict to coverings over $C_n(\bbc)\subset C_n(\bbp^1)$, whose monodromy is a Braid group representation $L_c$, and whose homology can be computed as $H_*(B_n,L_c) = H_*(\widetilde{C_*^{FN}}(\bbc)\otimes_{\bbz [B_n]}L_c)$, where $\widetilde{C_*^{FN}}(\bbc )$ the cellular chain complex of the universal cover of
$\conf{n}(\bbc )$, obtained by lifting the Fox-Neuwirth/Fuks cells (see section {\bf 3.3}).
More recently, Bianchi has probed deep into the geometry of Hurwitz spaces and used it to give a surprising new proof of the old-standing Mumford conjecture on the rational cohomology of $\mathcal M_g$ \cite{bianchi}. 

\subsection{The $N$-Body problem and free loop spaces}\label{nbody} In non-linear analysis, $\conf{N}(\bbr^3)$ is the natural topological receptacle for the $N$-body problem (see \cite{diacu} for an interesting history of this topic).  In \cite{bahrab}, Bahri and Rabinowitz studied Hamiltonian systems of $3$-body type of the form
\begin{equation}\label{bodytype}
m_i\ddot{q}_i + {\partial V\over\partial q_i}(t,q) = 0\ \ \ ,\ \ \ 1\leq i\leq 3
\end{equation}
for $q=(q_1,q_2,q_3)\in\conf{3}(\bbr^d)$ and
where the interaction potential $V:\bbr\times\conf{3}(\bbr^d)\rightarrow\bbr$ has the form
$\displaystyle V = \sum_{i,j=1\atop i\neq j}^3V_{ij}(t,q_i-q_j)$,
with a list of appropriate conditions including the usual ones:
$V_{ij}(t,q)\rightarrow -\infty$ as $q\rightarrow 0$, uniformly in $t$, and
${\partial V_{ij}\over\partial q_k}(t,q)\rightarrow 0$ as $|q|\rightarrow \infty$ \cite{riahi}.
They prove that the function corresponding to \eqref{bodytype} has an unbounded sequence of critical values in the subspace $L$ of the Sobolev space $E=W^{1,2}_T(\bbr,(\bbr^d)^3)$ of $T$-periodic maps, corresponding to non-collusion orbits.  It should  be noted that the usual symmetry assumptions $V_{ij} = V_{ji}$ was not made. The point of interest is that the space $L$ has the same homotopy type as the space of free loops $L\conf{3}(\bbr^d):=\hbox{Map}(S^1,\conf{3}(\bbr^d))$,
and the basic topological fact  needed in \cite{bahrab} to reach their conclusion was to show that the relative groups $H_*(L\conf{3}(\bbr^d), L\conf{2}(\bbr^d);\bbq )$ grow unbounded. 

For more general problems of $N$-body type, $N\geq 4$, obtaining an unbounded sequence of critical values requires showing an inequality of the form 
$$\hbox{rank} H_q(L\conf{N}(\bbr^d)) > f(q) + g(q)\hbox{rank}H_q (L\conf{N-1}(\bbr^d))$$
for $q$ sufficiently large, and $f,g$ are polynomials
 \cite{riahi}. This program was carried out in the book of Fadell and Husseini \cite{fh} using ``RPT-models'', and was actually the main motivation for writing that book. The author in \cite{benhamouda} treats the case of three points in $\bbr^d$, for $d\geq 3$, using a much more streamlined approach. The idea is to observe that $\conf{3}(\bbr^d)$ is formal over $\bbz$ for all $d\geq 2$, and so over any field $\bbf$,
 $$H_*(LX,\bbf) \cong HH_*(H^*(\conf{3}(\bbr^d),\bbf)$$
 where $HH_*$ stands for Hochschild homology. This is now amenable to explicit calculations. 

\section{Compactified configuration spaces}\label{compactificationsection}

\subsection{Fulton-MacPherson and Axelrod-Singer}\label{compactification} For any locally compact and Hausdorff space $X$, it is possible to compactify $\conf{k}(X)$ by taking its one point compactification or, in case $X$ is compact, by taking $X^k$. Both compactifications exhibit $\conf{k}(X)$ as an open dense subspace, but neither of them preserves the homology or homotopy type of the configuration space. Let's say property 1 is satisfied when this homotopy type is preserved. 
Moreover, in algebraic geometry one usually prefers to compactify a noncompact variety $M$ so that it becomes the complement of a divisor with normal crossings, meaning that the compactification $V$ contains $M$ as an open subset, so that the complement $D:=V\setminus M$ has every irreducible component smooth and any number of components of $D$ intersect transversely. Let's call this property 2.
Fulton and MacPherson manage to produce a compactification of $\conf{k}(M)$ which enjoys this second property when $M$ is a smooth complex projective variety \cite{fm}. 
Axelrod and Singer \cite{axelsing} adapted the work of Fulton and MacPherson \cite{fm}
from the algebro-geometric setting to the differential-geometric setting. This turns out to satisfy property 1. Their compactification is written conveniently as $\conf{k}[M]$ and it is simply defined to be the closure of the image of a suitable embedding
$$\conf{k}(M)\hookrightarrow M^k\times\prod_{S\subset \{1,\ldots, k\}\atop |S|\geq 2} Bl(M^S,\Delta_S)$$
Here $M^S$ is the space maps from $S$ to $M$, a finite product of $M$’s, and $\Delta_S$ is the thin
diagonal of all equal entries in $M^S$ (see \cite{cattaneo}). Here $Bl(X, Y )$ denotes the differential-geometric blowup of $X$ along $Y$, which consists in replacing $Y$ by the sphere bundle of its normal bundle in $X$. A good way to think of $\conf{k}[M]$ is as a manifold with boundary (and corners), with interior $\conf{k}(M)$, \textit{homeomorphic} to $M^k$ with an open neighborhood of the fat diagonal removed. 
As alluded to, while A.S. retains the homotopy type of the configuration space (which is its interior), F.M. does not in general. For example, 
$\conf{2}[X]$ is homeomorphic to the complement of the open tubular neighborhood of the diagonal in $X^2$, which indeed has the homotopy type of $\conf{2}(X)$, while the F.M. compactification of $\conf{2}(X)$ is $X\times X$ if $\dim_\bbc X=1$ (\footnote{Thanking Sasha Voronov for discussing this point. 
}).

\noindent{\bf 9.1.1.} The Axelrod-Singer compactification enjoys further remarkable properties: it has the natural structure of a manifold with corners, its boundary is conveniently stratified and it has equivariant functorial properties under embeddings. If $M$ is compact, then $\conf{k}[M]$ is compact. Axelrod and Singer used these compactifications to define invariants of three-manifolds coming from Chern–Simons theory and in so doing opened up a vast area of research (see \S\ref{integrals}).  
We summarize the main result   \cite{cattaneo}.

\begin{quote}\textit{{\bf Theorem 9.1}: 
Let $M$ be a (smooth) compact manifold. Then $\conf{k}[M]$ are (smooth) manifolds with corners, and
all the projections $\conf{k}(M)\rightarrow \conf{k-q}(M)$ extend to smooth projections on the
corresponding compactified spaces.}
\end{quote}

The boundaries of $\conf{k}[M]$ correspond to the “collision” of at least two of the
$k$ points of $M$. Boundaries are the union of different strata corresponding
to the different ways in which all the points may collide, keeping track of directions and relative rates of collisions. 
In other words, three points colliding at the same time gives a different point in the boundary than two colliding, then the third joining them. This boundary can be defined in terms of parentheses or trees. 
One calls hidden faces those corresponding to subsets S with $|S|\geq 3$ and principal faces those for which $|S|=2$. 
This structure defines a \textit{decomposition} of $\conf{k}[M]$ into a collection of closed faces (or strata) of various dimensions whose intersections are again strata. When $M=\bbr$, the components of $\conf{k}[\bbr]$ are homeomorphic to the \textit{associahedron}, a classical object from homotopy theory, while components of $\conf{k}[S^1]$ are homeomorphic to $S^1\times W_k$,  $W_k$ being the \textit{cyclohedron} (see \cite{bott}).


\noindent{\bf 9.1.2.} Kontsevich made similar constructions at about the same time as Fulton-MacPherson and Axelrod-Singer, and used the real variant to prove the absence of divergences in perturbative Chern-Simons theory \cite{kontsevich}. 
Applications of F.M. include work of Beilinson and Ginsburg who relate $\conf{k}[X]$ to the geometry of moduli spaces of holomorphic spaces of
principal bundles on a projective smooth curve $X$ of genus $g\geq 2$. 
Gaiffi generalized A.S. to arbitrary hyperplane arrangements over the real numbers, giving a  description of the category of strata using the language of blow-ups of posets.   Perhaps the most popular construction of configuration space compactifications  is the one given by Sinha \cite{dev2} (see also \cite{kontsevich}, \S3.3.1). He first constructs $\conf{k}[\bbr^n]$ as the closure of a natural inclusion 
$$\alpha_k : \conf{k}(\bbr^n)\hookrightarrow A_k[\bbr^n]:=(\bbr^n)^k\times (S^{n-1})^{\left[{k\atop 2}\right]}\times [0,\infty]^{\left[{k\atop 3}\right]}$$ Now, if $M$ is smoothly embedded in $\bbr^n$, $\conf{k}(M)$ embeds in $\conf{k}(\bbr^n)$, and so one defines $\conf{k}[M]$ to be the closure of $\alpha_k\left(\conf{k}(M)\right)$ in $A_k[\bbr^n]$.
Sinha's construction affords maps and boundary conditions needed for applications to knot theory \cite{dev3} or geometric problems \cite{mccleary}. It has a simplicial variant equipped with projection maps and diagonal maps which satisfy \textit{cosimplicial identities}. This is used by Sinha to give a  cosimplicial space model for the space of knots in $M$ \cite{dev3}.

\noindent{\bf 9.1.3.} (Affine action and operads) An important and useful property of compactified configurations in Euclidean space, after moding out by translations and scaling, is that they give rise to operads, as originally observed Getzler-Jones \cite{getjon}.  More precisely, one compactifies the quotient space $\conf{k}(\bbr^n)/G_n$, where $G_n = \{\phi: x\mapsto ax+b\ |\ a\in\bbr^{>0}, b\in\bbr^n\}\subset \hbox{Aff}(\bbr^n)$ is the subgroup of affine transformations acting freely on the configuration space. This action is by translation and dilation.  For example, any configuration of two points in $\conf{2}(\bbr^n)$ can be translated first so that the center of mass is at the origin, then dilated so that the points are at distance $1$ from the origin. This gives a diffeomorphism $\conf{2}(\bbr^n)/G_n\cong S^{n-1}$. Generally,  the quotient $\conf{k}(\bbr^n)/G_n$ is a smooth manifold of dimension $nk-(n+1)$, not always compact. Since  $G_n=\bbr^{>0}\ltimes\bbr^n$ is contractible, this quotient is homotopy equivalent to $\conf{k}(\bbr^n)$. By taking its closure again in a product of ``blowups'', one obtains a compactification $\hbox{\bf FM}_n(k)$ of $\conf{k}(\bbr^n)/G_n$, giving yet another compactified model \cite{idrissi2}. One has that $\FM_n(0)$ and $\FM_n(1)$ are singletons, while $\FM_n(2)\cong S^{n-1}$. 
The $\FM$ model is also conveniently stratified, with strata indexed by trees on $ k$ leaves. There are operations $o_i$ of grafting trees $\hbox{\bf FM}_n(k)\times \hbox{\bf FM}_n(\ell)\rightarrow \hbox{\bf FM}_n(k+\ell -1)$, one for each $1\leq i\leq k$, which are associative and compatible with permutations. They give rise to an operad $\FM_n$ which is weakly equivalent to the little disks operad \cite{pascal, paolo0}. This means that there exists a zigzag of operadic morphisms
$D_n\leftarrow\cdots \rightarrow \FM_n$ which are weak homotopy equivalences on each component.
Markl \cite{markl} extends this construction to arbitrary manifolds $M$ so that his compactification of $\conf{k}(M)$ is a ``partial module'' over $\FM_n$, $n=\dim M$. 
A similar construction of the compactified $\FM_n$ is in \cite{kontsevich2}, without the operad structure.

The operadic structure of configuration space compactifications plays a key role in knot theory \cite{dev3}, in the proof of the real formality of the little disks operad \cite{pascal}, in the proof of Theorem 4.3 of \S\ref{maymilgram} \cite{paolo0}, and in constructing a real model for configuration spaces on compact manifolds, with or without boundary \cite{idrissi, campos2} (see \S\ref{models}). A potent idea in this context is the fact that if $M$ \textit{is framed},
i.e. if one can coherently identify the tangent space at every point of $M$ with $\bbr^n$, $n=\dim M$, 
then the $\conf{k}[M]$ assemble to form a \textit{right module} over the operad $\FM_n$ (application to factorization homology in {\bf 4.2.2}). Let us mention that in \cite{merkulov}, various operads of compactified configuration spaces are discussed in the category of smooth manifolds with corners and their complexes of fundamental chains. 


\subsection{Knot theory and configuration space integrals}\label{integrals}  Main references for this section are the  excellent accounts by Bott \cite{bott}, Volic \cite{volic, volic2} and Alvarez \cite{alvarez}, see also \cite{alshuler, cattaneo}. Configuration spaces are used to study knot invariants, that is the zeroth cohomology of the space $\mathcal K$ of embeddings of $S^1$ into $\mathbb R^3$ i.e. $H^0(\hbox{Emb}(S^1, \bbr^3)) = \hbox{Hom}(\bbz[\pi_0(\hbox{Emb}(S^1, \bbr^3)], \bbz )$). 
This study originated with physicists, in perturbative Chern-Simons theory \cite{guada, axelsing}, and was pioneered by Bott and Taubes, Witten, Bar-Natan and Kontsevich. Given a compact Lie group $G$, a compact, oriented $3$-manifold $M$, a link $L\subset M$, and for each
component of $L$ a representation of $G$, this theory associates topological invariants to these data.  Guadagnini et al. define their invariants in the case $M=S^3$, $L\neq\emptyset$, using propagators and Feynman diagrams \cite{guada}. This approach was then elaborated upon by Bar-Natan (see \cite{alshuler} and references therein). The case $M=S^3$, $L=\emptyset$ was treated
by Axelrod and Singer. A common feature of all these works is the Feynman diagram expansion familiar in perturbative quantum field theory \cite{alshuler}.  

A more algebraic topological and quite successful approach to construct invariants of embeddings $S^{2k-1}\hookrightarrow\bbr^{2k+1}$, $k>1$, appeared in work of Bott and Taubes, paving the way for extensive research and applications in that direction \cite{bt}. The authors were motivated by their desire to understand Kontsevich's \textit{Fundamental Theorem
of Finite Type Invariants} \cite{kontsevich2} which roughly states
that for every knot one can compute an integral, now called the Kontsevich integral, which is a universal Vassiliev invariant, meaning that every Vassiliev invariant\footnote{A knot invariant $V$ is finite type $k$ (or Vassiliev of
type $k$) if it vanishes on singular knots with $k + 1$ self-intersections \cite{volic}.} can be obtained from it by an appropriate evaluation. 
The guiding idea of the work of Bott and Taubes was that the familiar linking number of two-component links, given by the
Gauss integral, should be adaptable to give an invariant (or a family of invariants) of knots. 
The ensuing developments led to many results which can be summarized as follows \cite{volic}. Let $\mathcal K_m$ be the space of smooth embeddings of $S^1$ in $\bbr^m$ and set $\mathcal K_3=\mathcal K$.

\begin{quote}\textit{{\bf Theorem 9.2} \cite{bt, cattaneo}: Bott-Taubes configuration space integrals combine to yield nontrivial cohomology
classes of $\mathcal K_m$. For $\mathcal K$, they represent a universal finite type knot invariant.}
\end{quote}
\vskip 5pt

The construction of non-trivial cohomology classes for higher dimensional knot spaces (i.e embeddings of higher dimensional spheres), using the same ideas, was carried out by \cite{cattaneo}. To see what these integrals are, and how configuration spaces intervene, start with the evaluation map $$ev(k): \conf{k}(S^1)\times\mathcal K\rightarrow \conf{k}(\bbr^3)\ \ ,\ \  ev_k((x_1,\ldots, x_k),f) = (f(x_1),\ldots, f(x_k))$$
One then takes pushouts of cohomology classes from $\conf{k}(\bbr^3)$ (as in Haefliger-Hu theory \S\ref{embeddingtheory}), and  ``integrates them over the fiber'' to obtain pushout differential forms on $\mathcal K$. The problem that is  encountered here, and in other related constructions for 3-manifold invariants, is the convergence of the above integrals, a nontrivial fact since the tautological forms are not compactly supported. The elegant solution to this problem relies on extending the map $ev(k)$ to a map $ev[k]$ between the compactifications of configuration spaces on which the tautological forms extend as ``smooth'' forms. One then considers the diagram of extended maps
$$\xymatrix{\mathcal K&\conf{k}[S^1]\times\mathcal K\ar[r]^{\ \ \tiny ev[k]}\ar[l]_{\pi\ \ \ \ \  }&\conf{k}[\bbr^3]}$$
By pulling back cohomology classes from $\conf{k}[\bbr^3]$, followed by a pushforward map $\pi_*: \Omega^*(\conf{k}[S^1]\times\mathcal K)\rightarrow\Omega^{*-k}(\mathcal K)$ (integration over the fiber) on cocycles, one obtains forms on $\mathcal K$. 
Since the boundary $\partial\conf{k}[S^1]$ is naturally stratified, as explained in \S\ref{compactification}, one needs to understand how the different strata (also called faces) contribute to this pushforward. It is shown that for all ``hidden'' faces (or higher codimension faces), the integral vanishes, but on the faces determined by two points colliding' at a time (so-called principal faces), it generally did not. To cancel this term out, Bott and Taubes introduced a novel construction which consists in compactifying the configuration spaces of points in $\bbr^3$ with some points lying on the image of the given embedding $K$ of $S^1$. 
We have no space to explain the theory, but we choose to illustrate this idea through a well-explained example from \cite{alvarez, bt}. Take $\alpha = \tau_{13}\wedge\tau_{24}$, where $\tau_{ij}=ev^*(\alpha_{ij})$ is the pullback of the tautological form \eqref{alphaij} (Theorem 2.3, \S\ref{cohomology}). This is a $4$-form on $\conf{4}[S^1]\times\mathcal K$ which one integrates along the four dimensional fibers of $\pi$ to get a $0$-form on $\mathcal K$. Its derivative is non-zero along the principal faces. Define $\conf{4,3}[\bbr^3,K]$ to be the ordered configuration space of $4$ points, $3$ of which are restricted to lie on the image of $K$. There is a manifold compactification of the space of such configurations, and maps as before:
$\widetilde{ev}: \conf{4,3}[S^1,\mathcal K]\rightarrow\conf{4}[\bbr^3]$ and
$\widetilde{\pi}: \conf{4,3}[S^1,\mathcal K]\rightarrow \mathcal K$. Write $\widetilde{\tau_{ij}} = {\widetilde{ev}}^*(\alpha_{ij})$, then the following result holds (\cite{bott}, Theorem 1, \cite{alvarez}, Theorem 2.3.9). 

\begin{quote}\textit{{\bf Theorem 9.3} \cite{guada, bt}: ${1\over 4}\pi_*(\tau_{13}\wedge\tau_{24})- {1\over 3}\widetilde{\pi_*}(\widetilde{\tau_{14}}\wedge \widetilde{\tau_{24}}\wedge \widetilde{\tau_{34}})$ is a (Vassiliev) knot invariant of type $2$.}
\end{quote}
\vskip 5pt

Building on these ideas, one can prove many more results. D. Thurston generalized the Bott-Taubes construction to obtain all finite-type invariants this way, thus validating this approach to the fundamental Theorem of finite-type invariants \cite{volic2}.
Graph homology also leads to invariants for framed knots \cite{alshuler} and invariants of immersions and higher embeddings  (see \S\ref{graphhom}). 
Further extensions of these ideas to invariants of links and 3-manifolds by means of ``graph
configurations'' is in work of Lescop \cite{lescop}. 
Finally, a purely homotopy theoretic approach to configuration space integrals is developed in \cite{koytcheff}.

\subsection{Graph cohomology} \label{graphhom} The discussion above shows the importance of integrating forms over the strata of the compactified configuration space $\conf{k}[\bbr^n]$. This is at the core of the idea of constructing ``graphical models'' for the de Rham cohomology of configuration spaces. That graph complexes may model the de Rham cohomology of configuration spaces of points is due to Kontsevich \cite{kontsevich3} as a key ingredient of his ground breaking proof of the rational formality of the little n-disks operad. This construction is explained in \cite{campos2} and we follow their presentation. 
The main idea is to build a Hopf cooperad $Graphs_n$ to connect $H^*(\FM_n)$ with the \textit{piecewise semi-algebraic forms} on $\conf{k}[M]$, denoted $\Omega^*_{PA}(\FM_n)$, through quasi-isomorphisms of 
differential graded-commutative algebra (in fact of homotopy Hopf cooperads) 
$$H^*(\FM_n) \longleftarrow \hbox{Graphs}_n\longrightarrow \Omega^*_{PA}(\FM_n)$$
Here $\FM_n$ is a defined in \S\ref{compactification}, the middle graph complex $\hbox{Graphs}_n$ (described in \cite{campos2} for example) is denoted in Nlab\cite{nlabgraphs} by the
``Graph complex of n-point Feynman diagrams for Chern-Simons theory''.
As we know, $H^*(\FM_n)$ is the cohomology of the configuration spaces $\coprod_k\conf{k}(\bbr^n)$. The semi-algebraic structure of $\FM_n$ is explained in \cite{pascal, camposwill}, and the need for resorting to these piecewise semi-algebraic forms is
because the 
projections $\FM_n(k + l)\rightarrow \FM_n(k)$ are not submersions in general (so not convenient for working with usual deRham forms), but are semi-algebraic bundles, in some well-defined sense \cite{pascal}.
The morphism Graphs$_n\rightarrow H^*(\FM_n)$ is given by sending an edge between $i$ and $j$ to the tautological generator $\alpha_{ij}$, and any graph with internal vertices to zero. The other morphism
$\hbox{Graphs}_n\longrightarrow \Omega^*_{PA}(\FM_n)$ is even more intricate and is obtained using integrals by ``regarding a graph as a Feynman diagram for Chern-Simons theory on $\bbr^n$ and sending it to its corresponding Feynman amplitude, namely to the configuration space-integral of the wedge product of Chern-Simons propagators associated to the edges, regarding Feynman amplitudes as differential forms on configuration spaces of points'' \cite{nlabgraphs}. Both morphisms, as mentioned, are quasi-isomorphisms.

The idea of constructing graph complexes turns out to be deep and powerful, with many emulations. In \cite{cattaneo}, the authors construct a cochain map $D_n\rightarrow \Omega^*(\mathcal K_n)$  
between a certain diagram complex $D_n$ generalizing trivalent diagrams of \cite{bt} and  the deRham complex of $\mathcal K_n$, where $\mathcal K_n$ is the space of knots in $\bbr^n, n>3$ (see \S\ref{integrals}). 
They use this map to show that spaces of knots have cohomology in arbitrarily high degrees. 
Other applications include the real invariance of configuration spaces on manifolds as discussed in \S\ref{models}. The proof of \cite{idrissi} relies on constructing a ``labeled'' graph complex Graphs$_R(*)$ 
that connects $\mathcal G_A$ (the Lambrechts-Stanley model \S\ref{models}) on the one hand,  and the piecewise semi-algebraic deRham complex $\Omega^*_{PA}(\conf{*}[M])$ on the other.  A nice discussion of graph cohomology is in \cite{idrissi2}.

\section{Many variants}\label{variants}

Configuration spaces of points can come in many variants, often labeled in the literature as ``generalized'' or ``colored" configuration spaces. In \S\ref{labeled}, we have already discussed one important such variant consisting of ``labeled particles''. 

\noindent{\bf 10.1}. A most natural generalization comes from the theory of hyperplane arrangements.
More precisely,  since $\conf{k}(X)$ is obtained from $X^k$ by removing the fat diagonal, one can immediately think of taking an intermediate space which we can write $\conf{k}^\ell (X)$, $2\leq \ell\leq k$, consisting of the complement of all diagonal subspaces of the form
$$\Delta_{i_1,\ldots, i_\ell}(X) := \{(x_1,\ldots, x_k)\in X^k\ |\ x_{i_1} = x_{i_2}=\cdots = x_{i_\ell}\}$$
This space is dubbed the \textbf{no-$\ell$ equal configuration space} of points \cite{bjowel}.  These spaces interpolate between $\conf{k}(X)$ when $\ell=2$, and $X^k$. The symmetric group acts on this space by permuting coordinates, we can similarly write $C_k^\ell (X)$ the orbit space, and we have the sequence of open embeddings
$$C_k(X)=:C_k^2(X)\subset \cdots\subset
C_k^k(X)\subset \sp{k}(X)$$
The space $\conf{k}^\ell(\bbr^n)$ is generally well-understood, and as in the case $\ell=2$ (Theorem 1.1, \S\ref{homology}) its homology is torsion free.
The homology and cohomology ring of $\conf{k}^\ell(\bbr^n)$, for $n\geq 2$, is given in \cite{dobtur}, using earlier ideas of Y. Baryshnikov, where it is related to the $k$-non overlapping disks, in the same way configuration spaces relate to the little disks. Dobrinskaya and Turchin do not get an operad in this case, but a bimodule over the operad of little disks. They give a complete (co)homological answer which is a fluid generalization of Theorem 2.1, \S\ref{homology}. A rational model is known in the case $n=2$ (S. Yuzvinski), and based on this model, some Massey products are shown to be non-trivial for $\ell= 3, k\geq 7$ (M. Miller), thus showing that the space is non-formal, unlike the case of the classical configuration spaces $\ell=2$ (Arnold, \S\ref{cohomology}). 
The unordered no-$\ell$ equal configurations are studied in \cite{kalsai} where their connectivity and fundamental groups and some homology groups are determined for finite simplicial complexes. General results on $\conf{k}^\ell(\bbr^n)$ and $C_k^\ell(\bbr^n)$ are relevant to the study of $\ell$-immersions \cite{ssv} or to \textit{linear decision trees} as in work of Bj\"orner and Lovasz. Further generalizations to configuration spaces that are allowed to collide in clusters or with multiplicity are described as
\textit{polychromatic configuration spaces} in \cite{kosar}, or as  \textit{spaces of $0$-cycles} in \cite{fww}.

\noindent{\bf 10.2}. \textit{Chromatic configuration spaces} form the next family of interesting and natural extensions of the configuration spaces of points. They have been first introduced as the complement of ``graphic subspace arrangements'' and only relatively recently have they been systematically studied from the algebraic topological point of view \cite{radmila, bmp, eastwood}. In \cite{radmila}, they were dubbed ``graph configuration spaces'', and in \cite{bmp}, ``partial configuration spaces''. To define these spaces, we let $G$ be a simple
graph (no loops and no multiple edges),  $V(G) = \{v_1,\ldots, v_n\}$ denote its set of vertices, and $E(G)$ the  set of edges of $G$.  The chromatic configuration space of $X$ is defined to be
\begin{equation}\label{configG}\conf{G} (X)=\{(x_1,\cdots,x_n)\in X^{|V|}\;|\; x_i\neq x_j \;\text{if}\;\{i,j\}\in E(G)\}
\end{equation}
Clearly, different labeling of vertices induces homeomorphic spaces. When the graph is complete $G=K_n$ (i.e. any two vertices are adjacent), one recovers the classical configuration space of pairwise distinct points, i.e.  $\conf{K_n} (X ) = \conf{n}(X)$. Note that $\conf{G}(X)$ can be viewed as the space of colorings of the graph $G$, having colors in $X$. The following ``categorification'' result is the most attractive in the theory \cite{eastwood}. Eastwood and Huggett formulate it for $M=\bbc P^{n-1}$, and the statement below is in \cite{walid}, with a different proof.

\begin{quote}\textit{{\bf Theorem 10.1 \cite{eastwood}} (see also \cite{walid}): 
Let $M$ be an $n$-dimensional topological manifold (with or without boundary), let $G$ be a simple finite connected graph on the vertex set $V$, and let $ch (G,t)$ the chromatic polynomial of $G$. Then\ 
$$\chi (\conf{G}(M)) = (-1)^{n|V|}ch(G, (-1)^n\chi (M))$$} \end{quote}

In \cite{radmila}, the authors addressed the problem of computing the homology of $\conf{G}(X)$ for general $X$. They constructed an explicit graph complex for that purpose, generalizing the complex earlier given by Bendersky and Gitler (see {\bf 6.1.3}). In the case $X=\bbr^n$, $\conf{G}(\bbr^n)$ is a complement of a subspace arrangement, and these have been studied under the name of ``graphic arrangements''. Longueville and Schultz have computed the cohomology ring, as a consequence of their general study of the cohomology of complements of $n$-arrangements, and an aesthetic derivation of this computation has been given by Bockstedt and Minuz \cite{bm}, in the spirit of the Arnold-Cohen computation. From the stable homotopy point of view, these spaces split after one suspension as a bouquet of spheres (in particular, their homology is torsion-free), with the number of the spheres given by the positive part of the coefficients of the chromatic polynomial of $G$ (see Part II, Theorem \ref{poinc}). This is essentially extracted from combined work of Goresky-MacPherson and Orlik-Solomon. Remarkably,
the Betti number
$\beta_{(|V|-k)(n-1)}$ of
$\conf{G}(\bbr^n)$
is the number of
``spanning forests on k
trees with no broken
cycles'' of $G$, and the end result is a direct generalization of \eqref{singlesplit}.
Note that \textit{generalized configuration spaces} coming from partitions  \cite{petersen} coincide with chromatic configuration spaces as well. 

A special case of a graph configuration space is the \textit{cyclic configuration} space studied by Farber and Tabachnikov \cite{farber}, in relation to Billiard-type problems. This is precisely $\conf{C_k}(X)$, where $C_k$ is the cyclic graph with $k$ vertices. The spaces $\conf{C_k}(X)$ consist of tuples $(x_1,\ldots, x_k)\in X^k$ such that $x_i\neq x_{i+1}$, indexes taken modulo $k$. The question to be answered in that paper is this: how many periodic billiard trajectories are there
in a smooth strictly convex domain in $\bbr^{m+1}$? The authors give lower bounds, and their approach involves the explicit computation of the integral cohomology ring structure of $\conf{C_k}(\bbr^{m})$. This is done using the exact same approach inaugurated by Totaro (see \S\ref{spectral}) which consists in analyzing the Leray spectral sequence of the inclusion $\conf{C_k}(\bbr^{m+1})\hookrightarrow (\bbr^{m+1})^k$.

\noindent{\bf 10.3}. \textit{Orbit configuration spaces} form another family of configuration spaces introduced in the work of Dung \cite{dung} and Xicotencatl \cite{xico}. Let $G$ be a group, which we assume finite, acting properly on $X$, and define
$$\conf{n}^{G}(X) = \{(x_1,\ldots, x_n)\in X^n\ |\ Gx_i\cap Gx_j=\emptyset\ ,
\ \hbox{if}\ i\neq j\}$$
where $Gx$ is the orbit of $x$. If the action of $G$ is trivial (i.e. fixes every point), we recover $\conf{n}(X)$. F. Cohen, M. Xicotencatl have studied these spaces, and more extensively by Bibby and Gaddish in the context of smooth manifolds and varieties, whereby spectral sequences have been analyzed and representation stability results obtained (see \cite{bibby}, references therein and related work by the authors). The Euler characteristic of this space is computed in \cite{walid} in terms of the orbit stratification of the action of $G$ on $X$.

\noindent{\bf 10.4}. Last but not least,  \textit{configurations of hard balls} or \textit{thick configuration spaces} are among the most well-studied models of matter in statistical mechanics where phase transition turns out to depend on the change of topology of the underlying configuration space. To define the space in question, let $\mathcal B$ be a bounded region in $\bbr^d$, then
$\hbox{Conf}_n(\mathcal B, r)$ is the space of  $n$-tuples of non-overlapping balls of radius $r$ in $\mathcal B$. If $\mathcal B$ has boundary, then the balls have to avoid the boundary. One is interested here in understanding when the topology changes if $n$ is fixed and $r$ is varying.
Expanding the particles to have positive thickness complicates the topology of the underlying configuration space significantly. However, there is real interest in studying these spaces. Hard disks systems are often considered prototypes for simple fluids. A set of hard disks inhabiting a bounded area may be regarded as a model thermodynamic system. From this standpoint, a topological property such as connectivity of the configuration space is a fundamental concern. 
One can prove sharp results using Morse theory \cite{carlsson, bbk} or posets and nerve theorem \cite{alpert}. The theory has some suprising and remarkable results like this one from \cite{carlsson} which applies in dimension two (hard disks): consider five disks
in the unit square $[0,1]^2$. Then the topology of $\conf{5}(\mathcal B,r)$ changes at least $20$ times as
the disk radius varies and for radius $0.1686 < r < 0.1692$, the configuration space has the first betti number $\beta_1$ = 2176.

\clearpage


.\vskip 200pt
\centerline{\bf\Huge Part II:}
\vskip 20pt

\centerline{\bf\Huge Configuration Spaces}
\vskip 5pt

\centerline{\bf\Huge and the Chromatic Polynomial}


\vskip 40pt

\begin{quote}
We study the chromatic configuration space $\conf{\Gamma}(\bbr^N)$ associated to a simple finite graph $\Gamma$. This is the complement of the so-called graphic subspace arrangement associated to $\Gamma$.
Using poset topology, we show  that the Poincar\'e polynomial of the chromatic configuration space is the reciprocal of the chromatic polynomial of $\Gamma$ (with signs).
We further show that these spaces split after a single suspension as a wedge of spheres, the number of wedge summands being given in terms of the coefficients of the chromatic polynomial. This splitting is deduced from the description of the homology generators in terms of ``forests of spanning trees with no-broken cycles (NBC)''. This description generalizes the theory of classical configuration spaces. As a good application, we deduce the homology of spaces of configurations consisting of ``$n$ moving objects in $\bbr^N$, distinct or not, each avoiding a given subset of $r$  fixed obstacles''.\footnote{\noindent{\sc Acknowledgment}: The author is grateful to Moez Bouzouita for many discussions related to this chaper, and to In\`es Saihi for her support. He is grateful to Pavle Blagojevi\'c for helpful comments. } 
\end{quote}

\clearpage

\part{Configuration Spaces and the Chromatic Polynomial}
\setcounter{section}{0}
\section{Introduction}

Chromatic configuration spaces appear to have been first introduced as the complement of ``graphic arrangements'' in \cite{gz}, and investigated in the more general algebraic topological context more recently in \cite{radmila, bmp, bm, eastwood, walid, zakharov}. These spaces offer both a natural and elegant extension of the pairwise distinct point configuration spaces studied in the first part of this user's guide. This extension is not ``esoteric'', it has some real applications to billiard-type problems \cite{farber} or to the study of moving objects in $\bbr^N$ avoiding some fixed set of obstacles as discussed in this work. 

We first review the definition. Let $\Gamma = (V(\G), E(\Gamma ))$ be an abstract graph on vertices labeled $v_1,\ldots, v_m$. The set of vertices is conveniently written $V(\G) = \{1,\ldots, m\}$, where $i$ refers to vertex $v_i$. An element of $E(\G)$ is an edge of the form $\{v_i,v_j\}$ (or $\{i,j\})$, $i\neq j$.  Two vertices are said to be \textit{adjacent} if they form an edge. All graphs in this paper will be simple, meaning they have no loops and no multiple edges. We write $|\Gamma|:=|E(\Gamma)|$ the number of edges of $\Gamma$ and call it, as is customary, the \textit{size} of the graph. The number of vertices is $|V|=|V(\G)|$, and this will be $m$.

Let $X$ be a path-connected topological space. The chromatic configuration space associated to a simple graph $\G$ was defined in Part I \S\ref{variants}, as follows
\begin{equation}\conf{\G} (X)=\{(x_1,\cdots,x_{|V|})\in X^{|V|}\;|\; x_i\neq x_j \;\text{if}\;\{i,j\}\in E(\G)\}
\end{equation}
This is the complement in $X^{|V|}$ of some diagonal subspaces where $x_i=x_j$ if $\{i,j\}$ is an edge. Different labeling of the vertices produce homeomorphic spaces. When the graph is complete $\G=K_m$ (i.e. any two vertices are adjacent), then  $$\conf{K_m}(X) = \conf{m} (X )$$ is the space we studied extensively in Part I.
Note that if $\Gamma'$ is obtained from $\Gamma$ by removing an edge (keeping the same number of vertices), then $\conf{\Gamma} (X )$ is an open subspace of
$\conf{\Gamma'} (X)$. In particular, $\conf{K_m}(X)$ is a dense open subspace of $\conf{\G}(X)$, where $\Gamma$ has $m$ vertices. 

All graphs $\Gamma$ in this paper are assumed to be connected, there is no loss of generality in assuming them to be so since for disjoint graphs $\Gamma_1$, $\Gamma_2$,
\begin{equation*}\label{disjoint}
\conf{\Gamma_1\sqcup \Gamma_2}(X)\cong
\conf{\Gamma_1}(X)\times\conf{\Gamma_2}(X)
\end{equation*}
If $\G$ is a connected simple graph on $m$ vertices, we write $\chi_\Gamma (\lambda)$ its chromatic polynomial (see \S\ref{linearterm}).
The  Poincar\'e polynomial for $\conf{\G}$ is on the other hand written as
$$P_t(\conf{\Gamma}(\bbr^N)) := \sum_{i\geq 0}\hbox{rank} (H_i(\conf{\G} (\bbr^N ),\bbz )t^i$$
with $i$-th betti number 
$\beta_i := \hbox{rank} (H_i(\conf{\G} (\bbr^N ),\bbz)$. Our main objective is to prove the following beautiful result.

\begin{theorem}\label{poinc} Let $\Gamma$ be a connected simple graph on $m$ vertices with chromatic polynomial $\chi_\G$, and $N\geq 2$. The homology of $\conf{\Gamma}(\bbr^N)$ is torsion free and its Poincar\'e polynomial is given as follows\
\begin{eqnarray*}
P_t(\conf{\Gamma}(\bbr^N)) &=& (-1)^mt^{m(N-1)}\chi_\Gamma\left(-t^{(1-N)}\right)
\end{eqnarray*}
\end{theorem}

The author knows of no reference to this result stated explicitly as is in the literature. As pointed out to us by Victor Reiner, this result can however be obtained by combining  known results of Goresky-MacPherson and Orlik-Terao. More precisely, \cite{orlikterao} (see \S2.4) prove this same result in the complex codimension $1$ case,  and \cite{goreskymacpherson} (Part III, Chapter 4, Theorem B) shows that the Poincar\'e polynomial of the corresponding chromatic $c$-arrangement has the same form, with powers of $t$ appended. Here $c$-arrangements are subspace arrangements in which every intersection of the subspaces has codimension equal to a multiple of some fixed $c$, so $c=1$ corresponds to real hyperplane arrangements
(see \S\ref{bondlat}, \cite{schaper} Definition 2 or \cite{goreskymacpherson}). Since a chromatic configuration space $\conf{\Gamma} (\bbr^N)$, $|V(\Gamma )|=k$, is
the complement in $(\bbr^N)^k$ of codimension $N$ subspaces defined by $x_i=x_j$
if $\{i,j\}$ is an edge, it is an example of a $c$-arrangement with $c=N$, and the stated result follows. 

Our approach to Theorem \ref{poinc}\footnote{All numbered references in the text are from PartII. A result from Part I will be preceded by an explicit "Part I". } uses ``Poset Topology'', as presented in \cite{wachs} for example, and is completely self-contained. It still relies on the Goresky-MacPherson formula for the homology of the complement of an arrangement. It uses the combinatorics of the poset of the arrangement (the so-called \textit{bond poset}), its Mobius function, and some interesting relations between the coefficients of the chromatic polynomial. Our treatment of this Theorem here nicely highlights the connection between configuration spaces and graph invariants. 

Theorem \ref{poinc} can be formulated in slightly different ways. Whitney's broken cycle theorem \cite{whitney} expresses the chromatic polynomial as an alternating polynomial
\begin{equation}\label{whitney1}
\chi_\Gamma (\lambda) = \sum_{i=1}^m(-1)^{m-i}a_i(\Gamma )\lambda^i
\end{equation}
where the coefficient $a_i(\Gamma )\geq 0$ counts the number of spanning forests $F$ of $\Gamma$ that have exactly $m-i$ edges and that contain \textit{no broken cycles} (i.e. ``NBC''). 
We recall that a forest $F$ is spanning in $\G$ if $F$ is a forest and $V(F)=V(\G)$. The definition of NBC is given shortly below. When $i=1$, a subgraph with $m-1$ edges is a spanning tree, and the term $a_1(\Gamma )$ therefore counts all NBC spanning trees. The term $a_m(\Gamma)=1$ always (\cite{read}, Theorem 8).

\begin{corollary}\label{poinc2} Let $\Gamma$ be a connected simple graph on $m$ vertices with chromatic polynomial $\eqref{whitney1}$, and $N\geq 2$. Then
\begin{eqnarray*}
P_t(\conf{\Gamma}(\bbr^N)) &=& \sum_{i=1}^ma_i(\Gamma)t^{(m-i)(N-1)}
\end{eqnarray*}
\end{corollary} 

This expression of $P_t$ seems optimal, nonetheless we can further reformulate it as follows. We recall that a \textit{cycle} in a graph is a closed path in a graph with no repeated vertices except the beginning and ending vertices. Fix an ordering on the edges of the graph $\Gamma$. An NBC (i.e. a \textit{broken cycle or circuit}) in $\Gamma$ is obtained by taking the edges of any cycle in the graph and removing its largest edge, in the given ordering. An NBC forest in $\Gamma$ is any set of edges $\mathcal F\subset E$ that make up a forest and that does not contain a broken cycle. In other words, $\mathcal F$ is a collection of disconnected subtrees of $\Gamma$, and for every edge $e\in E\setminus \mathcal F$, if $\mathcal F\cup e$ contains a cycle, then $e$ is not the largest edge of that cycle.
As in \cite{jenssen}, let $\mathcal F^{nbc}(\Gamma)$ be the set of all NBC forests in $\Gamma$ (including the empty set), and write $|\mathcal F|$  the size of $\mathcal F\subset E$, i.e. the number of edges in the forest. The following is a consequence of Whitney's description of $\chi_\Gamma$ (see \cite{jenssen}, Theorem 1.4). The result does not depend on the ordering of the edges of $\Gamma$.

\begin{corollary}\label{coro} For $N\geq 2$,  $\Gamma$ a simple graph, 
$\displaystyle P_t(\conf{\Gamma}(\bbr^N))=\sum_{\F\in\mathcal F^{nbc}(\Gamma)}t^{|\F |(N-1)}$.
\end{corollary}

The above corollary motivates Theorem \ref{generators} next, which describes precisely the generators in homology corresponding to these NBC-forests.

\bre\label{properties1} 
The first non-zero betti number of $\conf{\Gamma}(\bbr^N)$ is
$\ds \beta_{N-1}= |\Gamma |$ the size of the graph (i.e. its number of edges), which is also well-known to be the value of $a_{m-1}(\Gamma )$ (\cite{read}, Theorem 11). The top rank is
 $\beta_{(m-1)(N-1)} = a_1(\Gamma )$, and this is $a_1(\Gamma)$ the positive part of the linear coefficient of $\chi_\Gamma$.
The second non-zero betti number is
\begin{equation}\label{beta2}
    \ds \beta_{2(N-1)} = {|\Gamma|\choose 2} - s_3(\Gamma),
\end{equation}
where $s_3(\Gamma)$ is the number of ``three cycles'' of the graph (see \S\ref{proofmain}).  This last formula for $\beta_{2(N-1)}$ illustrates the important fact that cycles in the graph introduce relations, a well-known fact in the theory \cite{bm, longueville}. This is discussed in \S\ref{stable}.
\ere

\bex Consider the ``diamond graph'' $\Gamma$ which is the planar, undirected graph with $m=4$ vertices and $5$ edges. Its chromatic polynomial is $\chi_\Gamma (\lambda) = \lambda (\lambda -1)(\lambda - 2)^2 = \lambda^4-5\lambda^3+8\lambda^2-4\lambda$, therefore the Poincar\'e series of the associated chromatic configuration space is
$$P_t(\conf{\Gamma} (\bbr^N)) = 4t^{3(N-1)} + 8t^{2(N-1)} + 5t^{N-1} + 1$$
Here the top coefficient $\beta_{3(N-1)}=4$   is the number of NBC spanning trees which, for this diamond graph, is half the total of all spanning trees.
\eex

\section{Geometry and Stable Splitting} 

What Part II also accomplishes is to give the explicit construction of a basis of torsion-free generators for the homology using planetary systems as done in the complete graph case (see \S\ref{labeling} and of course Part I, \S\ref{homology}), and a proof of the stable splitting of $\conf{\G}(\bbr^N)$ (see \S\ref{stable}). We explain the ideas in this section.

A \textit{toric homology class} of $X$ is the \textit{injective} image of the orientation class of an embedded torus $(S^{k})^r\hookrightarrow X$. The homology generators for $\conf{\Gamma}(\bbr^N)$ turn out to be all toric. If $\Gamma=T$ is a tree, the configuration space itself is a product of spheres (Proposition \ref{treedecomposition}) and the claim is trivial. 
In the general case, we explain in \S\ref{labeling} how every  subtree $T$, in an edge ordered graph $\Gamma$, gives rise to a map $(S^{N-1})^{|T|}\hookrightarrow\conf{ \Gamma}(\bbr^N)$, with $|T|=|V(T)|-1$ the size of the tree. Similarly every forest $\mathcal F$ on disjoint trees $T_1,\ldots, T_k$ of $\Gamma$ (i.e. \textit{a forest in $\G$}) gives rise to an embedded torus 
\begin{eqnarray}\label{toric}
\alpha_{ F}&:& (S^{N-1})^{|T_1|}\times\cdots\times (S^{N-1})^{|T_k|}\
\xymatrix{\ar@{^(->}[rr]&&\conf{\Gamma}(\bbr^N)}
\end{eqnarray}
and a non-zero toric homology class  
\begin{equation}\label{classF}
[ F] = \alpha_{ F_*}\left[(S^{N-1})^{\sum |T_i|}\right]\in H_{(N-1)(\sum |T_i|)}(\conf{\Gamma} (\bbr^N))
\end{equation}
That is $[ F]$ is the image under $\alpha_{ F_*}$ of the orientation class of the torus. 
A \textit{spanning forest} $ F$ is a collection of disjoint trees $T_1,T_2,\ldots, T_k$ of $\Gamma$ such that the sets of vertices $V(T_i)$, for $1\leq i\leq k$, form a set partition of $V(\Gamma)$, in other words, $\sum |T_i|=m-k$, $m=|V(\Gamma)|$.  

Let $\F^{nbc}_k(\Gamma)\subset\mathcal F^{nbc}(\Gamma)$ be the set of spanning forests in $\G$ with $k$ component trees and  no broken cycles (i.e. NBC). Note that some of the trees can be reduced to a single vertex. 
The following is a natural generalization of the classical theory \cite{fredbible, paolo1, dev} from the complete graphs to arbitrary simple graphs, and it is proven in Sections \ref{labeling}--\ref{stable}.

\begin{theorem}\label{generators} Let $\Gamma$ be a finite simple connected graph with $m$ vertices and edge set $E$ with a chosen linear ordering, and $N\geq 2$. The graded homology of $\conf{\Gamma}(\bbr^N)$ is generated by all toric homology classes $[ F]$ associated to forests $ F$ on $m$ vertices in $\Gamma$, with basis generators the forests with no broken cycles. The basis elements of $H_{(m-k)(N-1)}(\conf{\Gamma} (\bbr^N),\bbz )$ are given by the classes
$[F]$ corresponding to spanning forests with no broken cycles and $k$ component trees, i.e. for $F\in\mathcal F_k^{nbc}(\Gamma)$.
\end{theorem}

The extreme cases $k=m=|V(\Gamma)|$ and $k=1$ are worth being discussed. When $k=m$, $\mathcal F_m^{abc}(\Gamma)$ is reduced to a singleton corresponding to the forest having each component consisting of a vertex. This contributes the single generating class in $H_0$. When $k=1$, we get the top homology generator in degree $(m-1)(N-1)$. An element in $\mathcal F_1^{nbc}(\Gamma)$ is necessarily a spanning tree, and the number of such NBC trees is the positive linear term of the chromatic polynomial $a_1(\Gamma)$. This term enters fundamentally in all calculations and is discussed in \S\ref{linearterm}.

As is implicit in this Theorem, cycles in the graph impose relations among generators, and we describe those in \S\ref{stable} in the form of ``generalized Jacobi relations''.


\vskip 5pt\noindent{1.1. {\bf Stable Splitting.}} 
Many spaces have the homology of a wedge of spheres but are stably not such a wedge (eg. complex projective space). Complements of subspace or hyperplane arrangements tend to decompose stably. Exceptionally, some arrangement complements split unstably into a bouquet of spheres \cite{gt}, but in general one needs at least one suspension \cite{schaper}.  For two spaces $X$ and $Y$, write $X\simeq_sY$ if $\Sigma X\simeq \Sigma Y$, where $\Sigma$ is unreduced suspension. For a based space $X$, we write $X^{\vee n}$ the $n$-fold wedge of $X$ with itself, so for example $X^{\vee 3}=X\vee X\vee X$. 
As is customary, we denote by $X_+$ the disjoint union of a space $X$ with a basepoint. 

\begin{theorem}\label{main}
Let $\Gamma$ be a finite simple graph, $N\geq 2$ and $m=|V(\Gamma )|$. If $\Gamma= T$ is a tree,
then $\conf{ T}(\bbr^N)\simeq \left(S^{N-1}\right)^{m-1}$ (a torus). In general, we have a stable splitting 
\begin{equation}\label{formula}\displaystyle\conf{\Gamma} (\bbr^N)_+\simeq_s 
\bigvee_{i=1}^m\left(S^{(m-i)(N-1)}\right)^{\vee\ a_i(\Gamma)}
\end{equation}
where $a_i(\G)$ are the Whitney coefficients \eqref{whitney1}.
\end{theorem}

The $+$ (disjoint basepoint) in the splitting comes from the fact that when $i=m$, we get the term $S^0$ on the righthand side, thus an extra basepoint $+$ must be accounted for on the lefthand side.

This result is proven in \S\ref{stable}. We must compare it with the literature. A general splitting result for $c$-arrangements with $c=N\geq 2$ is given in \cite{schaper}. The proof consists in analyzing the Ziegler-\v{Z}ivaljevic model \cite{zz} and applying Spanier-Whitehead duality. It is not explicit to the point of giving spheres, their dimensions, or the number of wedge summands. Our approach has intersection points with Shaper's work \cite{schaper}, but it is more explicit and follows directly from
the description of homology generators as toric classes.



\section{Literature review and Preliminary Results}\label{prelim}

We review the current literature on the chromatic configuration spaces. In \cite{radmila}, the authors addressed the problem of computing the homology of $\conf{\G}(X)$ for a general simple graph $\Gamma$, and constructed an explicit chain complex for that purpose. Based on methods of Totaro and Gorinov, \cite{zakharov} gives the rational cohomology ring for $\conf{\G}(X)$ when $X$ is a smooth compact algebraic variety. In \cite{bmp}, the authors gave an extensive study of $\conf{\Gamma}(X)$ and of its fundamental group when $X=\Sigma_g$ is a closed Riemann surface.
In the case $X=\bbr^1$, $\conf{\G}(\bbr)$ is a complement of a hyperplane arrangement which
 has been studied in \cite{orlik, rota} under the name of ``graphic arrangements''. In the case $X=\bbr^N$, the cohomology rings of our spaces are described explicitly in \cite{longueville} as an application of their general study of geometric lattices and complements of $c$-arrangements. A fairly aesthetic but different derivation of this same computation is given in \cite{bm}.  In both cases, the answer is given as a quotient algebra by an ideal of relations, without an explicit formula for the betti numbers.

Let $G = (V(\G), E(\G))$ and $H = (V(H), E(H))$ be two simple and undirected graphs. A graph homomorphism $f : \G \rightarrow H$
is a function from $V(\G)$ to $V(H)$ that takes edges to edges, and so induces a map $E(\G)\rightarrow E(H)$. Formally, $\{u,v\} \in E(\G)$ implies $\{f(u),f(v)\} \in E(H)$, for all pairs of vertices $u, v$ in $V(\G)$.
In particular, morphisms respect edge incidence. Note that given two graphs, there may not be any morphisms between them, for example we cannot map the triangle graph $K_3$ to the edge graph $K_2$.

Let's define $\bf\mathcal G$ to be the category whose objects are undirected  simple graphs that are vertex labeled and whose morphisms are the graph homomorphisms (see for example \cite{droz}).
Let $\Psi$ be the map taking a graph
$\G$ to $\conf{\Gamma} (X)$. To make it a functor, we need define $\Psi$ on homomorphisms so that $\Psi (f\circ g) = \Psi (f)\circ \Psi(g)$. We give such a construction and it is contravariant.

Let $f: \G\rightarrow H$ be a graph morphism, and let
$(x_1,\ldots, x_m)$ be an element of
$\conf{H} (X)$, $m = |V(H)|$.
Let $V(H) = \{v^H_1,\ldots, v^H_m\}$ and $V(\G) = \{v^\G_1,\ldots, v^\G_n\}$.
Define
$\Psi(f)(x_1,\ldots, x_m)$ to be the tuple
$(y_1,\ldots, y_n)\in X^{|V(\G)|}$ where
$y_j=x_i$ if $v^G_j\in f^{-1}(v^H_i)$.
It is fitting to look at the examples given in Figure \ref{morphisms}.

\begin{figure}[htb]
    \centering
    \includegraphics[scale=0.6]{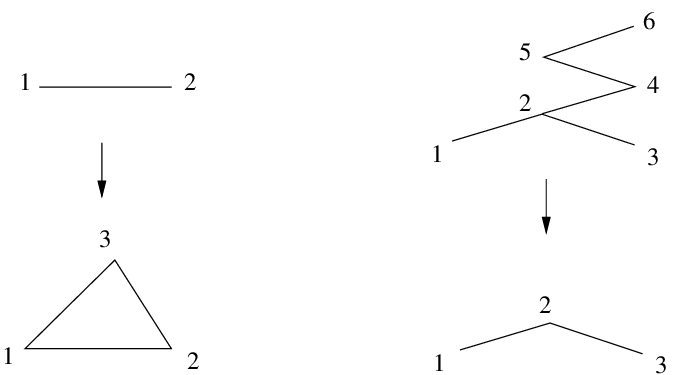}
    \caption{The first morphism $f_1$ is an inclusion of graphs. In this case $\Psi(f_1)(x_1,x_2,x_3) = (x_1,x_2)$. The second morphism $f_2$ maps all edges but $\{1,2\}$ to $\{2,3\}$. Here
    $\Psi (f_2)(x_1,x_2,x_3) = (x_1,x_2,x_3,x_3,x_2,x_3)$.}\label{morphisms}
\end{figure}

Alternatively, the better description of $\Psi$ is as follows. Using $X$-colorings, think of $(x_1,\ldots, x_m)\in\conf{H}(X)$ as a morphism from the vertex labeled $H$ into $X$ such that any two adjacent vertices take different values. Composition with $\G\rightarrow H$ defines the $X$-coloring $\G\rightarrow X$ which is taken to be the element
$\Psi(f)(x_1,\ldots, x_m)\in\conf{\Gamma} (X)$.

\ble The map $\Psi$ defines a contravariant functor
${\mathcal G}\lrar {\bf Top}$.
\ele

\begin{proof} Let $f: \G\rightarrow H$. The map $\Psi (f)$ is well-defined (i.e. its  image lies in $\conf{\Gamma}(X)$), and composition holds $\Psi (f\circ g) = \Psi (f)\circ \Psi (g)$. 
\end{proof}

\bre\label{mappi} (The Case of Subgraphs). The special case of $f: H\hookrightarrow \G$ a subgraph embedding yields an induced map
\begin{equation}\label{pigh}\pi_\G^H : \conf{\Gamma}(X)\rightarrow\conf{H}(X)
\end{equation}
This map is our induced functorial map after relabeling.  If $H$ is a spanning subgraph (i.e. if $V(H)=V(\G)$), then we can use the same vertex labeling for both $\G$ and $H$ and define $\pi_\G^H = \Psi (i)$ where $i: H\hookrightarrow\G$ is the subgraph inclusion. 
If $H$ has a proper vertex set $V(H)\neq V(G)$, then we must relabel the vertices of $H$ so they become
$v_1^H,v_2^H,\ldots, v_k^H$ consistently with the labeling of $\G$ in such a way that increasing indices of the new labeling from $1$ to $k$ corresponds to increasing indices of labeled vertices in $V(H)\subset V(\G)$. For example, if $V(\G) = \{v_1,v_2,\ldots, v_8\}$ and $V(H) = \{v_2,v_3,v_5,v_7\}\subset V(\G)$, then we relabel as follows
$$v_1^H=v_2\ ,\ 
v_2^H=v_3\ ,\ 
v_3^H=v_5\ ,\ 
v_4^H=v_7
$$
Once this is done, we can describe $f: H\hookrightarrow G$ by the map which on vertices is $v_1^H\rightarrow v_2$, $v_2^H\rightarrow v_3$, $v_3^H\rightarrow v_5$ and $v_4^H\rightarrow v_7$, and is the obvious map on edges. We then set $\pi_\G^H = \Psi(f)$ as in \eqref{pigh}.
\ere

\bre When $H$ is spanning in $G$, the map $\pi_\G^H$ in \eqref{pigh} is an open embedding. 
\ere

\vskip 5pt\noindent{2.1. {\bf Euler Characteristic.}}  The following section is at the origin of this work. The main observation is that the Euler characteristic of $\conf{\Gamma}(X)$ is related to the chromatic polynomial of $\Gamma$ \cite{eastwood, walid}. Let's write $\chi_\Gamma (t)$ the chromatic polynomial for the graph $\Gamma$. Recall that this polynomial is uniquely determined by the property that $\chi_\Gamma (k)$ is the number of proper colorings of $\Gamma$ using $k$ colors. The following result has been obtained in \cite{walid} using the additivity of the Euler characteristic with compact supports.

\begin{theorem}\label{chiconf} Let $\G$ be a simple graph. Then
$$\chi (\conf{\Gamma} (\bbr^N)) = (-1)^{N|V|}\chi_\Gamma ((-1)^N)$$
\end{theorem}

\bex We can illustrate Theorem \ref{chiconf} with an example. Let $\Gamma=K_3$. Then $\conf{K_3}(\bbr^N)$
is the configuration space
$\conf{3} (\bbr^N) = \{(x,y,z)\in (\bbr^N)^3\ |\ x\neq y, y\neq z, x\neq z\}$. This fibers over $\conf{2} (\bbr^N)$ with fiber homotopic to $S^{N-1}\vee S^{N-1}$. Since $\conf{2} (\bbr^N)\simeq S^{N-1}$, we see that
$\chi (\conf{K_3} (\bbr^N)) = \chi (S^{N-1})\cdot \chi (S^{N-1}\vee S^{N-1}) = (1+(-1)^{N-1})(1+2(-1)^{N-1})$. This  coincides with our formula $$(-1)^{3|V|}\chi_{K_3}(t)=
(-1)^{3|V|}t(t-1)(t-2)$$
when $t = (-1)^N$ as asserted.
\eex
 
For fixed $X$, the homology groups
$H_*(\conf{\Gamma}(X);\bbz)$ form an invariant of the graph. When $M$ is a manifold, $H^*(M,R)$ is a projective $R$-module and $\Gamma$ is simple, \cite{radmila} construct a cohomology spectral sequence with $E_1$-term an explicit graph complex  converging, after a duality regrading, to the homology of
$\conf{\Gamma}(M)$ with coefficients in $R$.
In the rational case, most recent work of Zakharov \cite{zakharov} gives a generalization of
the Kriz-Totaro model for
$\conf{\G}(M)$ when $M$ a smooth proper algebraic variety over $\bbc$. Here also, the rational homotopy type has a cdga model given by the quotient of $H^*(M^{|V(G)|})\otimes \Lambda\langle \Delta_{i,j}\rangle$ by an explicit ideal of relations, where $\Delta_{(i,j)}$ are classes of dimension $2\dim_\bbc M-1$ over all ordered pairs $(i,j)$, with $\{i,j\}\in E(V)$. The ideal of relations is derived from the cohomology ring in the Euclidean case, that is for $\conf{\G}(\bbr^N)$, where $N=2\dim_\bbc M$.
This cohomology ring, have been computed explicitly by \cite{bm, longueville} as reproduced below.

\begin{theorem}\label{coho} (Corollary 5.6 of \cite{longueville}) Let $\mathcal A = \{H_1,\ldots, H_m\}$ be a
$c$-arrangement in a real vector space $W$, and $M_{\mathcal A}$ its complement. Then the integral cohomology ring of the complement
$M_{\mathcal A}$ has the presentation
$$0\lrar I\lrar \Lambda\mathbb Z^m\lrar H^*(M_{\mathcal A})\lrar 0$$
if $n$ is even, $\Lambda$ denoting the exterior algebra, and
$$\xymatrix{0\ar[r]& I\ar[r]& S{\mathbb Z^m} \ar[r]^{\pi}& H^*(M_{\mathcal A})\ar[r]& 0}$$
if $n$ is odd, $S$ denoting the symmetric algebra, and
$\pi (e_i)\in H^{c-1}(M_{\mathcal A})$ for the canonical basis
$\{e_1,\ldots, e_m\}$ of $\mathbb Z^m$. The ideal $I$ of relations is generated by
\begin{equation}\label{relation}\sum_{i=0}^k(-1)^i\epsilon (a_0,\ldots, \hat a_i,\ldots a_k)
e_{a_0}\wedge \cdots\wedge \hat e_{a_i}\wedge\cdots\wedge e_{a_k}
\end{equation}
for all minimal dependent sets $\{H_{a_0},\ldots, H_{a_k}\}$ and some well
determined sign $\epsilon$.
\end{theorem}

The dimension of all the generators is $c-1$, and the ideal of relations comes from minimal dependent sets $\{H_0,..., H_r\}$. 
A minimal dependent set corresponds to a minimal cycle. The cohomology relation \eqref{relation} is dubbed the generalized Arnold relation, and it is phrased in terms of cycles in the graph \cite{bm}.
This is a satisfactory result of course, albeit not explicit when it comes to computing ranks. Our approach uncovers many new features not directly apparent from Theorem \ref{coho}.

\section{Three examples}\label{trees}

As a good starting point, we illustrate Theorem \ref{poinc} on three basic cases: trees,  complete graphs and  cyclic graphs. 

\vskip 5pt
\noindent{\bf 4.1. Trees}. It would be convenient here to introduce a partial ordering on the vertices of a rooted tree $T$ with vertices labeled $v_1,\ldots, v_k$, and root $v_1$. The partial ordering is defined by setting $v_i\prec v_j$ if the length of a path from $v_i$ to $v_1$ is smaller than the length of a path from $v_j$ to $v_1$. Note if $v_i\prec v_j$, then $v_j$ covers $v_i$ if and only if $v_j$ and $v_i$ are adjacent. Also, in this partial ordering, any vertex other than $v_1$ has a unique predecessor.

\bpr\label{treedecomposition} 
let $T$ be a tree on $m$ vertices, $m\geq 2$. There is a homeomorphism
\begin{equation}\label{homeo}
\conf{ T}(\bbr^N)\cong \bbr^N\times (\bbr^N\setminus \{0\})^{m-1}
\end{equation}
and consequently,
$\conf{T}(\bbr^N)\simeq (S^{N-1})^{|T|}$. In particular 
$$P_t(\conf{T}(\bbr^N)) = (1+t^{N-1})^{m-1}$$
\epr

\begin{proof}
This is verified directly as follows. We will root the tree at its lowest vertex $v_1$ (for ease, the vertices are labeled from $1$ to $m$) and partially order them as above. For every vertex labeled $m>1$, let $m^*$ be its (unique) predecessor in the partial ordering $\prec$. In other words, $v_{m^*}$ and $v_m$ are adjacent, and $v_{m^*}\prec v_m$. Consider the map
\begin{eqnarray}\Psi_T: \conf{ T}(\bbr^N)&\lrar& \bbr^N\times (\bbr^N - \{0\})^{m-1}\label{themapII}\\
(x_1,\ldots, x_{m})&\longmapsto&(x_1, x_2-x_{2^*}, \ldots, x_{m}-x_{m^*})\nonumber
\end{eqnarray}
This is a homeomorphism, with inverse
$\Psi^{-1}_T: (y_1,\ldots, y_m)\mapsto 
(y_1, y_2^*, y_3^*,\ldots, y_m^*)$, where
$$y_i^* = y_i + y_{i_1} + y_{i_2} + \cdots + y_{i_r}$$ if $v_1 \prec v_{i_1}\prec v_{i_2}\prec\ldots\prec v_{i_r} = v_{i}$ is the full chain from the root $v_1$ to $v_i$ (see Fig. \ref{homeograph}).
\end{proof}

\bex The homeomorphism \eqref{themapII} is given explicitly in the case of the line graph $L_m$, with edges of the form $\{i, i+1\}$, $0\leq i\leq m-1$, by 
$$(x_1,x_2,\ldots, x_m)\longmapsto 
(x_1,x_2-x_1,x_3-x_2,\ldots, x_m-x_{m-1})$$

Another example is the tree in Fig. \ref{homeograph}. Its chromatic configuration is homeomorphic to $\bbr^N\times (\bbr^N\setminus\{0\})^5$ via the map $(x_1,\ldots, x_6)\mapsto (x_1, x_2-x_6, x_3-x_6, x_4-x_2, x_5-x_2, x_6-x_1).$ The inverse $\bbr^N\times (\bbr^N\setminus\{0\})^5\rightarrow\conf{\Gamma}(\bbr^N)$
    sends
    $$(y_1,\ldots, y_6)\ \longmapsto\ (y_1, y_1+y_6+y_2, y_1+y_6+y_3, y_1+y_6+y_2+y_4, y_1+y_6+y_2+y_5, y_1+y_6).$$
    \begin{figure}[htb]
\centering
\includegraphics[scale=0.6]{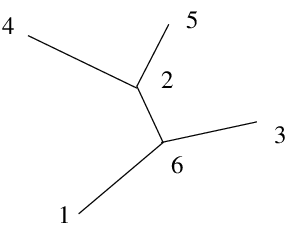}
    \caption{For this tree $T$, $\conf{T}(\bbr^N)\simeq (S^{N-1})^5$. 
    The partial ordering on the vertices is given by
    $v_1\prec v_6\prec v_3\ ,\ v_6\prec v_2\prec v_5\ ,\ v_2\prec v_4$.
    }\label{homeograph}
\end{figure}
\eex

\bre\label{betaT} (The map $\beta_T$). For a tree $T$ rooted at $v_1$, we write $\beta_T$ the composite of a toric embedding
$(S^{N-1})^{|T|}\hookrightarrow 
\bbr^N\times (\bbr^N)^{|T|}$ with 
the inverse of the map \eqref{themapII}
\begin{equation}\label{themapbetaT}
\xymatrix{\beta_T: (S^{N-1})^{|T|}\ar@{^(->}[r]&
\bbr^N\times (\bbr^N\setminus\{0\})^{|T|}\ar[r]^{\ \ \ \ \ \  \Psi_T^{-1}}& \conf{T}(\bbr^N)}\end{equation}
This map depends on the labeling of the vertices as described explicitly earlier.  Its image is a torus $C_T$ in $(\bbr^N)^{|T|}$ which is a deformation retract of $\conf{T}(\bbr^N)$. This image has the following pictorial description: the  elements of $C_T$ are of the form $(x_1=0,x_2,\ldots, x_{|T|})$, and under the partial ordering $\prec$ defined at the start of this section, if $v_j$ covers $v_i$, then in $C_T$, the entry $x_j$ lies in a sphere of radius $1$ centered at $x_i$.  
\ere

\bre The homology of $\conf{T}(\bbr^N)$ is a tensor product of $|T|$ copies of $H_*(S^{N-1})$. By adding an edge with a termination point, we add a new tensor factor, and there are no relations. The relations in homology come only from the cycles in the graph (see \S\ref{stable}).
\ere

\vskip 5pt
\noindent{\bf 4.2. Complete Graph}. The Poincar\'e series in this case is given in \eqref{poincseries}. This is recovered as follows.
The chromatic polynomial for $K_m$ is the falling factorial $\chi_{K_m}(t) = t(t-1)\cdots (t - m+1)$.
It is direct to check Theorem \ref{main} on this formula
\begin{eqnarray*}
(-1)^mt^{m(N-1)}\chi_\Gamma(-t^{(1-N)})&=&
(-1)^mt^{m(N-1)}\left(-t^{1-N}(-t^{1-N}-1)\cdots (-t^{1-N} - m+1)\right)\\
&=&(-1)^mt^{m(N-1)}(-1)^mt^{1-N}(t^{1-N}+1)\cdots (t^{1-N}+m-1)\\
&=&(1+t^{N-1})(2+t^{N-1})\cdots (1+(m-1)t^{N-1})\\
&=& \prod_{k=1}^{m-1}(1+kt^{N-1})\\
&=&
P_t(\conf{K_m}(\bbr^N)) 
\end{eqnarray*}

We can also verify the properties in Remark \ref{properties1}. We have
$\ds \beta_{N-1}(\conf{m} (\bbr^N))={m(m-1)\over 2}$ as expected, since there are these many edges. On the other hand,
\begin{eqnarray*}\label{calc2}
\beta_{2(N-1)}(\conf{m} (\bbr^N))&=&\begin{bmatrix} m\\ m-2\end{bmatrix}= {(3m-1)m(m-1)(m-2)\over 24}\nonumber
\end{eqnarray*}
and this is precisely $\ds {|K_m|\choose 2} = {{m\choose 2}\over 2} - s_3(K_m)$, with $\ds s_3(K_m) = {m\choose 3}$.

\vskip 5pt
\noindent{\bf 4.3. Cyclic graph}. Let $C_m$ be the cyclic graph on $m$ vertices. The so-called
\textit{cyclic configuration space}
$$\conf{C_m} (X) = \{(x_1,\ldots, x_m)\ |\
x_i\neq x_{i+1}, 1\leq i\leq m \ \hbox{and}\ x_{m+1}=x_1\}$$
has been extensively studied by M. Farber and S. Tabachnikov \cite{farber} in connection with the problem of finding upper bounds to the number of periodic
trajectories of high dimensional billiard problems (they denoted this space $G(X,m)$). Using  spectral sequences, they derived the Poincar\'e series (Proposition 2.2 of \cite{farber})
\begin{equation}\label{cyclic2} 
\displaystyle P_t(\conf{C_m}(\bbr^N)) = (t^{N-1}+1)^m - t^{(m-1)(N-1)} - t^{m(N-1)}
\end{equation}
This is however recovered immediately from Theorem \ref{poinc}. The chromatic polynomial for the cyclic graph is $\chi_{C_m}(\lambda) = (\lambda -1)^m + (-1)^m(\lambda -1)$, and a straightforward verification, as in \S4.2, confirms the Farber-Tabachnikov computation.


\section{Poset Topology}\label{poset}

Poset topology, developed by H. Whitney, G. Rota, A. Bjorner, R.P. Stanley, M. Wachs and many others, has had remarkable applications to algebraic topology, ranging from classical subjects like subspace arrangements to the more recent field of toplogical data analysis. In this section we collect some of the basic notions that we need and refer to \cite{bjorner, wachs} for details.

A poset $P$ is any finite set with a partial ordering $\leq$. A poset $(P,\leq)$ is bounded if there exist a top element $\hat 1$ and a bottom element $\hat 0$ such that
$\hat 0\leq x\leq \hat 1$ for all $x\in P$.
All posets in this work will be bounded. The proper part of a bounded poset is
$\mathring{P}=(\hat 0,\hat 1) = \{x\in P,\ \hat 0<x<\hat 1\}$, and so $P= \mathring{P}\cup \{\hat 0,\hat 1\}$. A lattice is a (bounded) poset where any two elements have a join (or greatest lower bound) and a meet (or least upper bound).

A bounded poset is
\textit{ranked} if all maximal chains have the same length.

Below, all posets $(P,\leq)$ will be written $P$ for short.

Every poset has a Mobius function $\mu (= \mu_P )$ defined recursively on closed intervals of $P$ as follows:
$\mu (x, x) = 1$, for all $x \in P$ and
$\mu (x, y) = -\sum_{x\leq z < y}\mu (x, z)$, for all
$x < y \in P$. For a bounded poset $P$, one defines the Mobius invariant (\cite{wachs}, \S1.2)
\begin{equation}\label{mobiustotal}\mu (P) := \mu_P (\hat 0, \hat 1)
\end{equation}

The ``order complex'' $\Delta (P)$ of the poset (also called ``flag complex'') is the abstract simplicial complex with vertices the elements of $P$, and a subset $S\subset P$ forms a face if and only if all elements of $S$ are pairwise comparable. Each simplex then corresponds to a chain (i.e. a totally ordered subset) of $P$, so that $\Delta(P)$ is the union of all chains
$\{ \{i_1,\ldots, i_k\},
i_1< i_2<\cdots < i_k\}$ in $P$. 

 The order complex construction defines a functor $\Delta : \hbox{\bf Pos}\ \lrar\ \hbox{\bf SCpx}$ from our category of posets to the category of simplicial complexes.
Using the notation of \cite{bjorner}, we will write $|P|$ for the geometric realization of $\Delta (P)$. Topological invariants of $P$ or $\Delta (P)$ (e.g. Euler characteristic, betti numbers, homotopy type) are meant to be those of $|P|$. The following is very relevant (\cite{wachs}, Proposition 1.2.6).

\bpr\label{hall} (Hall)
Let $P$ be a bounded poset with minimum $\hat 0$ and
maximum $\hat 1$. Then
$$\mu (\hat{0}, \hat{1}) = \tilde\chi (\Delta (P \setminus \{\hat 0, \hat 1\}))
= \tilde\chi (\mathring{P})$$
where $\tilde \chi$ denotes the reduced Euler characteristic.
\epr

The reduced Euler characteristic $\tilde\chi (\Delta )$ of a simplicial complex $\Delta$ is
$-1 + f_0 - f_1 + f_2 - + \cdots$,
where $f_i$ denotes the number of faces of $\Delta$ of dimension $i$
(i.e., faces with $i + 1$ elements/vertices).

\vskip 5pt\noindent{4.1. {\bf Shellability.}} 
 A $d$-dimensional simplicial complex is called \textit{pure} if its maximal simplices (i.e facets) all have dimension $d$. For instance, the order complex of a ranked poset is pure.  

A pure simplicial complex is \textit{shellable} if its facets can be ordered $(F_1,\ldots, F_k)$ (referred to as \textit{shelling order}) so that $(F_1\cup\cdots\cup F_{i-1})\cap F_i$ is a non-empty union of $(n-1)$-dimensional faces, for every $i$.

\bex Any $1$-dimensional connected simplicial complex is shellable.
\eex

Another very useful result in poset topology is that a shellable simplicial complex has the homotopy type of a wedge of spheres. This is standard, see \cite{wachs}, \cite{bjorner} (Theorem 7.9.1) or (\cite{delucchi}, Theorem 2). 

\begin{theorem}\label{shellable}
Let $P = \mathring{P}\cup \{\hat 0,\hat 1\}$ be a bounded and ranked poset ($\mathring{P}$ is its proper part), and suppose that $|\mathring{P}|$ is shellable. Then $|\mathring{P}|$ is a wedge of $(-1)^d\mu (P)$ spheres of dimension $d=rk(P)-2$.
\end{theorem}

\begin{proof}   Since $\Delta (\mathring{P})$ is shellable, it is a wedge of spheres of the same dimension. If $d$ is the dimension of the spheres, then the Euler characteristic of $\Delta (\mathring{P})$ must be $1+(-1)^dk$, where $k$ is the number of spheres in the wedge. By Hall's result (Proposition \ref{hall}), $\chi (| \mathring{P}|) = \mu (P) + 1$, so that
$\mu (P)= (-1)^dk$, as is claimed.
\end{proof}

The next statement is even more explicit and is needed for the proof of Lemma \ref{fromPtoL}.
In the shelling order, we call a facet $F_i$ a \textit{capping facet} if when added, it introduces a cycle. The following is explained in (\cite{kennedy}, Chapter 2).

\begin{proposition}\label{capping} If an $n$-dimensional pure simplicial complex is shellable, then it is homotopy equivalent to a wedge of $n$-spheres. The number of $n$-spheres is exactly the number of capping facets (if there is none, the complex is contractible).
\end{proposition}

A primary example of a shellable complex is the order complex of a \textit{geometric} lattice (see \cite{stanley}, Definition 3.9, and discussion therein).  A special class of geometric arrangements are $c$-arrangements $\mathcal A= \{A_1,\ldots, A_m\}$ whose intersection poset $L(\mathcal A)$ has a unique minimal element $\hat 0 = \bbr^N$.  It is a semi-lattice, by construction, but becomes a lattice when $L(\mathcal A)$ has a maximal element $\hat 1$ (i.e. if it is central). The following result is stated and proved in (\cite{goreskymacpherson}, Part III, chapter 4).

\bpr\label{geometric} 
The intersection lattice of a central $c$-arrangement is geometric.
\epr

This fact has been used  in \cite{longueville} in their derivation of the cohomology of $c$-arrangements. It is well-known that geometric lattices are shellable (\cite{bjorner}, 7.6.3), and this is a result we will use everywhere in this chapter.

\section{The Bond Lattice}\label{bondlat}

We will fix throughout the Euclidean dimension to be $N$.
We view again $\conf{\Gamma} (\bbr^N)$ as a complement of an affine hyperplane arrangement $\mathcal A$ in $\bbr^{mN}$,
$|V(\Gamma )|=m$. This arrangement is made out of the linear subspaces
\begin{equation}\label{aij}
A_{ij} = \{(x_1,\ldots, x_m)\in \left(\bbr^N\right)^{m}\ |\ x_i= x_j\ \hbox{if}\ \{i,j\}\in E(\Gamma )\}\end{equation}
so that
$$\conf{\Gamma} (\bbr^N) = \bbr^{mN}-\bigcup_{\{i,j\}\in E(\Gamma)}A_{ij}$$
This is a $c$-arrangement with $c=N$, which is central since $\bigcap_{\{i,j\}\in E(\Gamma)} A_{ij}\neq\emptyset$.

When $G=K_m$ is the complete graph on $m$ vertices, $\conf{K_m} (X) = \conf{m} (X)$ is the classical configuration space of tuples of pairwise distinct points,
and we refer to the corresponding graph arrangement as the ``braid arrangement''. All graph arrangements are sub-arrangements of the braid arrangement. It is customary to write $V_{\mathcal A}:=\bigcup_{\{i,j\}\in E(\Gamma)} A_{ij}$, the so-called \textit{link} of the arrangement, which is the union of all its subspaces (more about this in \S\ref{naturality}).

The intersection lattice for $\mathcal A=\{A_{ij}\}_{\{i,j\}\in E(\Gamma)}$ in \eqref{aij} is written  $L_\G$. Its components are the various intersections of the $A_{ij}$'s, ordered by reverse inclusion, and the rank is given by codimension. 
This is a \textit{bounded} lattice with $\hat 0=\bbr^{mN}$ and $\hat 1 = \bigcap A_{ij}\neq\emptyset$. The ``atoms''
are the $A_{ij}$'s which are, by definition, the elements that cover $\hat 0$. 

\bde The lattice of intersections $L_\G$ of the graphic subspace arrangement associated to $\G$ is called the \textit{chromatic  lattice}. This turns out to correspond to the ``bond lattice of the graph'' as discussed next.
\ede

A fundamental object associated to a graph $\Gamma$ is its \textit{bond partition}  which is a sublattice of the partition lattice. Given a graph $\Gamma$ with vertex set $V(\Gamma)$ on $m$ vertices, a ``connected partition'' or ``a bond partition'' $B$ of $\Gamma$ is any set partition of $V(\Gamma )$, written $B=B_1|B_2|\cdots |B_k$, where the $B_i$'s are blocks assumed to be the vertices of a connected induced subgraph $\Gamma_i$  of $\Gamma$\footnote{A subgraph of $\G$ is induced if for every $i,j\in V(H)$, if $\{i,j\}\in E(\G)$, then $\{i,j\}\in E(H)$.}. For simplicity, we will view $B_i$ as both a block or a subset of $V(\Gamma)$, depending on the context, and write $B_i=V(\Gamma_i)$. 
The integer $k$, $1\leq k\leq |V(\Gamma )|$, is the \textit{length} of the partition which we write as $|B|=k$. In more technical jargon, the Bond Lattice is \textit{the upper ideal of the partition lattice generated by the edges of the graph}.

\bex\label{exampleline5}  Consider the line graph $L_5$ on $5$ vertices labeled $1,2,\ldots, 5$.
\begin{figure}[htb]
\begin{center}
\begin{tikzpicture}[scale=0.6]
\draw (-2,1)-- (0,1)-- (2,1)-- (4,1)--(6,1);
\draw [fill=black] (-2,1) circle (2.5pt);
\draw[color=black] (-2.03,0.41) node {$1$};
\draw [fill=black] (0,1) circle (2.5pt);
\draw[color=black] (-0.01,0.43) node {$2$};
\draw [fill=black] (2,1) circle (2.5pt);
\draw[color=black] (1.99,0.43) node {$3$};
\draw [fill=black] (4,1) circle (2.5pt);
\draw[color=black] (3.99,0.43) node {$4$};\draw [fill=black] (6,1) circle (2.5pt);
\draw[color=black] (5.99,0.43) node {$5$};
\end{tikzpicture}
\end{center}
\end{figure}

The bond partitions of length $3$ of $L_5$ are listed lexicographically as follows:
$$1|2|345\ ,\ 1|5|234\ ,\ 4|5|123 \ ,\ 1|23|45\ ,\  3|12|45\ ,\ 5|12|34
$$
For the stellar graph $St_5$, whose central vertex is labeled $v_1$, the bond partitions of length $3$ are given as follows
$$2|3|145\ ,\ 2|4|135\ ,\ 2|5|134 \ ,\ 3|5|124\ ,\  3|4|125\ ,\ 4|5|123
$$
The total number of connected partitions is the same but the structure of blocks is different (see Corollary \ref{combinatorial}).
\eex

The set of connected partitions
of $\Gamma$ having length $k$ is denoted by $\mathcal B_k(\Gamma)$, and we write $\mathcal B(\Gamma )=\bigcup_{1\leq k\leq m}{\mathcal B_k}(\Gamma)$. By coarsening the partitions we get a partial ordering, so that $\mathcal B$ becomes a poset, in fact, it forms a lattice. An explicit example for the square graph is depicted in Fig. \ref{bondC4}. 
This bond lattice also has many other names in the literature: the ``lattice of contractions'' or the ``lattice of connected components'' \cite{thatte}.
Notice that the bond lattice of
the complete graph is precisely the partition lattice since any induced subgraph of the complete
graph is connected. 

\begin{figure}[htb]
    \centering
\includegraphics[scale=0.7]{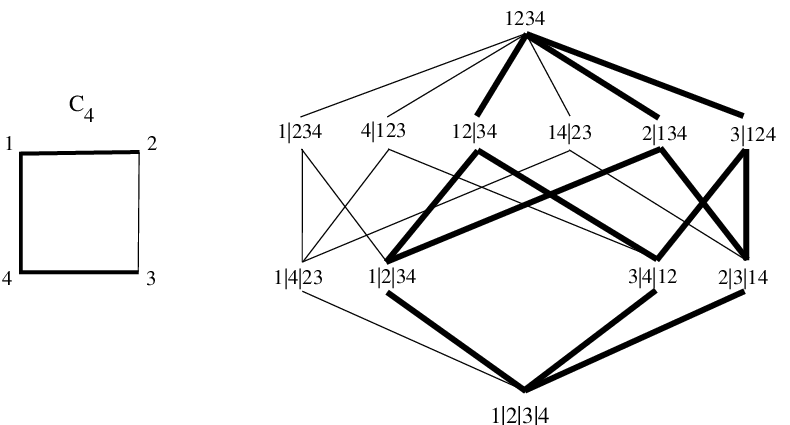}
    \caption{The bond lattice for $C_4$. The top element is $\hat 1 = (1234)$, which corresponds in the intersection lattice to the thin diagonal, and bottom element $\hat 0 = 1|2|3|4$ which corresponds to $\bbr^{4N}$. The subposet in bold is the  bond lattice of the three-edge subgraph (also shown in bold) of $C_4$, and it is a  sublattice.}\label{bondC4}
\end{figure}

The Bond lattice uniquely determines the graph, up to isomorphism. 
It has important combinatorial properties \cite{thatte}, and its characteristic polynomial is related to the chromatic
polynomial of the graph (see \eqref{rotaform}).

\ble\label{bondlat2} The bond lattice $\Pi_\Gamma$ of a simple graph $\Gamma$ is isomorphic to the intersection lattice of the associated graphic arrangement $ L_\G$ in $\bbr^N$ \eqref{aij}. It is a geometric lattice.
\ele

\begin{proof} This is straightforward. Let $m=|V(\Gamma )|$ as always. The isomorphism $\Pi_\Gamma\lrar  L(A)$ sends a bond partition 
$B_1|\ldots |B_k$, where
$B_j = \{j_1,\ldots, j_{r_j}\}$, $\bigsqcup B_i = \{1,\ldots,m\}$, to the subspace
$$\{(x_1,\ldots, x_n)\ |\ x_{j_1}=\ldots = x_{j_{r_j}}\ ,\ \forall j\}$$
In particular $1|2|\cdots |m\mapsto \hat 0 = (\bbr^N)^m$ and 
$12\cdots m\mapsto \{(x,\ldots, x), x\in\bbr^N\}$ the thin diagonal. This is a bijection between sets. Since reverse inclusion for the intersection semi-lattice corresponds to coarsering of partitions, the bijection preserves orders, and the posets are isomorphic. Finally, the chromatic arrangement is a $c$-arrangement, it is a lattice and it is geometric by Proposition \ref{geometric} (See also \cite{tsuchiya}, Corollary 3.2).
\end{proof}

Another key property of bond posets is stated next. This turns out to be essential in the homology analysis of Section \ref{proofmain}. We recall that a subinterval of an interval in a poset consists of $x<y$ in the interval and all elements in between. A proper interval of a bounded poset is an interval not containing $\hat 0$ and $\hat 1$. The following is simple inspection.

\ble\label{subinterval} Let $H$ be a spanning subgraph of a finite simple graph $\G$. Moreover, the length of a chain from $x$ to $y$ in $\Pi_H$ is the same as in $\Pi_\G$ (i.e. if $y$ covers $x$ in $\Pi_H$, it must also cover $x \in \Pi_\G$).
\ele

\begin{definition}\label{bondproperty}\rm Associated to $x\in\Pi_\Gamma$ are three positive values: the \textit{rank}, the 
\textit{length} and the \textit{dimension}.
\begin{enumerate}[label=(\roman*)]
\item The length $\ell (x)$ of $x\in\Pi_\Gamma$ is $k$ if   $x = B_1|\ldots |B_k\in \Pi_\Gamma$, with
$B_i\subset V(\Gamma)$.
\item  The rank $\rho (x)$ of $x\in\Pi_\Gamma$ is the length of a maximal chain from $\hat 0$ to $x$. Clearly $\rho (\hat 0)=0$, $\rho (\hat 1)=m-1$, where $m= |V(\Gamma )|$, and in general, if
$x = B_1|B_2|\cdots |B_k$,
\begin{equation}\label{length}
\rho (x) = m-k 
\end{equation}
\item The rank and the length are related by the equation
$\rho (x) = m-\ell (x)$.
\item We can also associate a dimension to a bond element $B$. By Lemma \ref{bondlat2}, $B$ corresponds to an affine subspace in the intersection lattice, and $\dim B$ is the dimension of that subspace. We have
$$\dim B =\ell (B)N$$
\end{enumerate}
\end{definition}


\vskip 5pt\noindent{5.1. {\bf Combinatorics}}  The cardinality of the bond partitions of length $k$, $\mathcal B_k(\Gamma )$, is not known in general and it is not clear how to write such a formula for a general graph. We give the computation for trees in Lemma \ref{combinatorial} and for the cyclic graphs in Lemma \ref{count}.
But we first indicate how to write a bond partition uniquely: label the vertices $1,2,\ldots, m$ of $\Gamma$.  The blocks of the partition are written from smaller to bigger size blocks, and blocks of the same size appear from left to right ordered by the size of the smallest element in each block.

In the case of trees, the cardinality of $\mathcal B_k(T)$ corresponds to the number of ways of splitting the tree into $k$-disjoint subtrees (i.e no common vertices). The following combinatorial count is an immediate consequence of \eqref{homeo}.

\begin{corollary}\label{combinatorial} Given a tree $T$ on $m$ vertices, the number of ways to divide it into $k$ disjoint subtrees is independent of the tree, and equals $\displaystyle |\mathcal B_k(T)| = {m-1\choose m-k}$. 
\end{corollary}

\begin{proof}
Proposition \ref{treedecomposition} gives that the coefficient of $(t^{N-1})^{m-k}$ in the Poincar\'e series, for $1\leq k\leq m-1$, is $\displaystyle {m-1\choose m-k}$. This must correspond to $|\mathcal B_k|$ according to Corollary \ref{poinc2}, since $a_1(B)=1$ for  $B\in\mathcal B_k$. Indeed, each block of $B$ corresponds to a tree, and $a_1$ for trees is always $1$ \cite{eisenberg}. 
\end{proof}

For cyclic graphs, the count is given as follows.

\ble\label{count} Let $\mathcal B_k(C_m)$ be the set of all bond partitions of the cyclic graph $C_m$, $k\geq 2$. Then $\displaystyle |\mathcal B_k(C_m)| = {m\choose k}$.
\ele

\begin{proof}
Let's say that a bond partition $B_1|B_2|\cdots |B_k$ contains the block $D$ if $D\subset B_i$ for some $i$. So for example, $15$ is contained in $2|3|451$. When going from $L_m$, with vertices $\{1,2,\ldots, m\}$, to $C_m$, we add an edge from $m$ to $1$, and so we add some extra bond partitions containing a block having $1m$ in it (see \eqref{unique} for the case $m=5$). These extra partitions, which are not partitions for $L_m$, are all the bond partitions having the block $1m$, then the block $1m(m-1)$ or $12m$, etc.

Since the size of a block $D$ containing $1m$ is at least $2$ and at most $m-k+1$, we will have to count, for every $2\leq j\leq m-k+1$, the number of bond partitions of length $k$ containing a given block $D$ of size $j$. This number is $\displaystyle {m-j-1\choose k-2}$ since once we remove the block $D$, we are left with a path graph having $m-j$ vertices, and from those, we are taking all bond partitions of length $k-1$ (remember, we already counted the block $D$). Since there are exactly $j-1$ blocks $D$ of size $j$ containing $1m$, the final count must therefore be
\begin{eqnarray}
|\mathcal B_k(C_m)|&=&
{m-1\choose k-1} +
\sum_{j=2}^{m-k+1}(j-1) {m-j-1\choose k-2}\label{count1}
\end{eqnarray}
We can invoke the following identity between binomial coefficients
\begin{equation}\label{count2}
{m-1\choose k} = \sum_{j=2}^{m-k+1}(j-1) {m-j-1\choose k-2}
\end{equation}
which can be proved by induction on $m$, using the identity
$\displaystyle {m+1\choose k} ={m\choose k} + {m\choose k-1}$.
Replacing in \eqref{count1}, we get
$\displaystyle |\mathcal B_k(C_m)|=
{m-1\choose k-1} + {m-1\choose k} = {m\choose k}$, and this completes the proof.
\end{proof}

\bex
$L_5$ has $6={4\choose 2}$ bond partitions of length $3$. We add an extra edge $\{1,5\}$ to $L_5$ to get the cyclic graph $C_5$, and thus we get the extra length 3 bond partitions
\begin{equation}\label{unique}
2|15|34\ ,\ 4|15|23\ ,\ 2|3|145\ ,\ 3|4|125
\end{equation}
There are indeed in total $10 = {5\choose 3}$ such partitions for $C_5$.
\eex


\section{The Chromatic Polynomial and its Linear Term}\label{linearterm}

Following \cite{dong}, let $\Gamma$ be a graph, and $\lambda\in\bbn$. A mapping $f: V(G)\rightarrow \{1,2,\ldots,\lambda\}$ is called a $\lambda$-colouring of $\Gamma$ if $f(i)\neq f(j)$ whenever $\{i,j\}\in E(\Gamma )$. The number of distinct $\lambda$-colourings of $\Gamma$ is denoted by $\chi_\Gamma (\lambda)$, and this is a polynomial in $\lambda$ (the chromatic polynomial) of degree the number of vertices of the graph.
As is well-known, the chromatic polynomials are always alternating in sign.

\bex\label{chromatic} We will need the following examples :
\begin{itemize}
\item If $\Gamma = O_m$ is the trivial graph on $m$-vertices, then $\chi_{O_m}(\lambda) = \lambda^m$.
\item If $\Gamma = L_m$ is the line graph with $m$ vertices, then $\chi_{L_m}(\lambda ) = \lambda(\lambda -1)^{m-1}$.
\item If $\Gamma = K_m$ is the complete graph on $m$-vertices, then $\chi_{K_m}(\lambda)$ is the falling factorial
\begin{eqnarray*}
\chi_{K_m}(\lambda) &=& \lambda (\lambda - 1)\cdots (\lambda -m+1)
=\sum_{k=1}^ms(m,k)\lambda^k
\end{eqnarray*}
where $s(m,k)$ are the (signed) Stirling numbers of the first kind.
\item If $\Gamma = C_m$ is the cycle graph, then
$\chi_{C_m}(\lambda) = (\lambda -1)^m + (-1)^m(\lambda -1)$.
\item The wedge of two connected graphs $G_1\vee G_2$ is their union having only one vertex in common, and
\begin{equation}\label{wedge}
\chi_{G_1\vee G_2}(t) = {1\over t}\chi_{G_1}(t)\cdot\chi_{G_2}(t)
\end{equation}
\end{itemize}
\eex


Whitney's broken cycle theorem \cite{whitney} expresses the chromatic polynomial as an alternating polynomial
\begin{equation}\label{whitneyform}
\chi_\Gamma (\lambda) = \sum_{i=1}^m(-1)^{m-i}a_i(\Gamma )\lambda^i
\end{equation}
where the coefficient $a_i(\Gamma )\geq 0$ counts the number of spanning subgraphs of $\Gamma$ that have exactly $m-i$ edges and that contain no broken cycles. It is clear that $a_m=1$ and $a_{m-1}$ is the number of edges. The term $a_1(\Gamma )$ counts all such $m - 1$-subsets of edges which, being acyclic, correspond to spanning trees. In particular, we have the following characterization (see \cite{benson} for general discussion).

\bpr\label{whitney}\cite{whitney}
For every undirected, connected graph $\Gamma$, the coefficient of the $\lambda$-term in the chromatic polynomial $\chi_\Gamma(\lambda)$ corresponds, up to sign, to the number of NBC spanning trees.
\epr

An important property of $a_1$ is listed below.

\begin{theorem}\label{properties} (Eisenberg \cite{eisenberg})  $\Gamma$ is a tree if and only if $a_1(\Gamma )=1$.
\end{theorem}

Next, we describe the relationship of $a_1(\Gamma)$ to orientations of graphs. An orientation of a graph $\Gamma = (V,E)$ is an assignment of a direction (i.e. arrow) to each edge $\{i,j\}$, denoted by $i\rightarrow j$ or $j\rightarrow i$, as the case may be.
An orientation of $\Gamma$ is said to be acyclic if it has no directed cycles. A vertex $v_0$ of $\Gamma$ is a source of an acyclic orientation if all arrows emanate from $v_0$. Fix $v_0\in V(\Gamma)$ and write $A(\Gamma, v_0)$ the set of all acyclic orientations of $\Gamma$ whose \textit{unique} source is $v_0$. We write $a_1(\Gamma)$ the cardinality of $A(\Gamma, v_0)$. This number does not depend on the choice of $v_0$, nor does it depend on whether $v_0$ is a sink or a source.

\begin{theorem} (\cite{gz}, Theorem 7.3).
Let $|A(\Gamma, v_0)|$ be the number of all acyclic orientations with a unique sink (or source) $v_0$. Then
$|A(G, v_0)| = |a_1|$, and this number is independent of the choice of $v_0$.
\end{theorem}

The reference \cite{gebhard} gives three independent proofs of this fact.
Theorem 4.1 of \cite{benson} lists all the several interpretations of $a_1(\Gamma)$.

\bex Let $C_4$ be the square graph. Its chromatic polynomial is $\chi_{C_4}(\lambda) = \lambda^4-4\lambda^3+6\lambda^2-3\lambda$, so
$a_1=3$. The acyclic orientations of $C_4$ with a single source are displayed in Fig. \ref{acyclicfig}
\begin{figure}[htb]
    \centering
\includegraphics[scale=0.6]{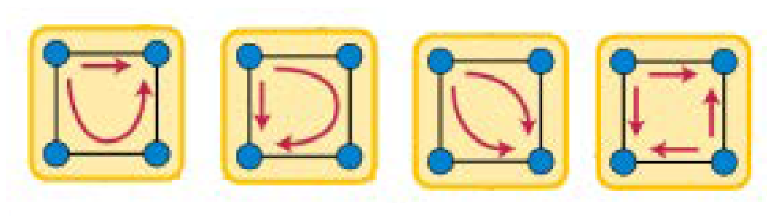}
    \caption{All the acyclic orientations of $C_4$ with the top left vertex $v_0$ being a source. There are only $a_1=3$ orientations where $v_0$ is the ``unique'' source. There are in total $14$ acyclic orientations of $C_4$. Figure extracted from Wikipedia.}
    \label{acyclicfig}
\end{figure}
\eex

The relevance of $a_1(\Gamma)$ to our work is embodied in the following Lemma which will be generalized in Proposition \ref{wedgebond} for all forests.

\ble\label{wedgespheres} Let $\Gamma$ be a simple graph with $m$ vertices, and let $\mathring\Pi_\Gamma$ be the proper part of the bond lattice.
Then
$\displaystyle |\mathring\Pi_\Gamma |\ \simeq \bigvee^{a_1(\Gamma)} S^{m-3}$.
\ele

\begin{proof} The pure dimension of $\Delta (\mathring\Pi_\Gamma )$ is $|\Pi_\Gamma|-2=m-3$, so this gives directly the dimension of the spheres.
Next, it is a beautiful result of Rota \cite{rota} that the characteristic polynomial of $\Pi_\Gamma $ coincides with the chromatic polynomial (for a short topological proof of this fact using stratifications, see \cite{walid}). More precisely, we have that
\begin{equation}\label{rotaform}
\chi_\Gamma (\lambda ) = \sum_{x\in\Pi_\Gamma} \mu (\hat 0,x)\lambda^{\ell (x)}
\end{equation}
where again $\mu = \mu_{\Pi_\G}$ is the Mobius function of the bond poset, and $\ell (x)$ is the number of blocks of $x$ (see \S\ref{poset}).  The extreme terms in this formula are when $x=\hat 0$, $\ell (\hat 0) =n$ (one block for every vertex), while $\mu (\hat 0,\hat 0)=1=a_m$. When $x$ is an atom, $\ell (x)=m-1$. There are as many atoms as there are edges, and so the coefficient of $\lambda^{m-1}$ in Rota's formula \eqref{rotaform} is $- m$, in agreement with \eqref{whitneyform}.

Since $\hat 1$ is the unique single block partition, one has that $\ell (\hat 1)=1$. By comparing Whitney's and Rota's formulas, we see that the Mobius invariant \eqref{mobiustotal}
\begin{equation}\label{key}
\mu (\Pi_\Gamma ) := \mu (\hat 0,\hat 1) = (-1)^{m-1}a_1(\Gamma )
\end{equation}
By Theorem \ref{shellable}, it follows that the number of spheres in the decomposition is
$$(-1)^{|\hat 1|} \mu (\Pi_\Gamma )
= (-1)^{m-1}(-1)^{m-1}a_1(\Gamma ) = a_1(\Gamma )$$
as claimed.
\end{proof}

\bex\label{partitiona1}  If $\Gamma = K_m$, then $\Pi_\Gamma$ is the partition lattice of $\{1,\ldots, n\}$, and it is an old result of Stanley (see \cite{wachs}, \S3.2.2) that $$|\mathring\Pi_\Gamma |\simeq\bigvee^{(m-1)!}S^{m-3}$$
The computation $(m-1)!$ above is a special case of the computation of the Mobius function of the partition lattice given as follows (see \cite{bender}): if $\sigma$ is a partition of $m$, with blocks $\sigma_i$, $1\leq i\leq k$, $\hat 0 = 1|2|\cdots |m$, then $\mu (\hat 0,\sigma) = (-1)^{k}(|\sigma_1|-1)\cdots (|\sigma_k|-1)!$ 
\eex


\vskip 5pt
\noindent{6.1. {\bf A relation between the Whitney coefficients.}}
Consider the Whitney coefficients $a_k(\Gamma)$ of the chromatic polynomial \eqref{whitneyform}, $1\leq k\leq m=|V(\Gamma)|$. If $B=B_1|B_2|\cdots |B_k\in\mathcal B_k(\Gamma)$ is a bond partition, and $\Gamma_i$ the induced (connected) subgraph whose vertex set is the block $B_i$, define the product $$a_1(B) := a_1(\Gamma_1)\ldots a_1(\Gamma_k)$$ The following is a key ingredient in the proof of Theorem \ref{main}.

\begin{proposition}\label{aktoa1}
For $1\leq k\leq m$,\ 
$\displaystyle a_k(\Gamma) = \sum_{B\in\mathcal B_k}a_1(B)$.
\end{proposition}

\begin{proof} Remember that $a_k(\Gamma)$ counts the number of spanning forests of $\Gamma$ with size $m-k$ that do not contain any broken cycle with respect to the ordering \cite{whitney, erey}. 
A spanning forest has $m-k$ edges if and only if it has $k$ components. It follows that any such spanning forest is of the form $T_1\sqcup\cdots\sqcup T_k$, with $T_i\subset \Gamma$ a tree.   By setting $B_i=V(T_i)$, we get a set partition of $V(\Gamma)$ and a bond partition $B=B_1|\ldots |B_k\in\mathcal B_k(\Gamma)$. The subgraph $\Gamma_i\subset \Gamma$ is the induced subgraph on vertices $V(T_i)$. The collection of subgraphs $\Gamma_1,\ldots, \Gamma_k$ in uniquely determined by $B=B_1|\ldots |B_k$. The point is that any choice of disjoint trees $(T'_1,\ldots, T'_k)$, where $T'_i$ in an NBC spanning tree of $\Gamma_i$, produces an NBC spanning forest of $\Gamma$ of size $m-k$. We therefore have to count those
for every $B\in\mathcal B_k(\Gamma)$.
But the number of NBC spanning trees on $\Gamma_i$ is $a_1(\Gamma_i)$ as already asserted in Proposition \ref{whitney}, and so there is in total $a_1(B):=a_1(\Gamma_1)\cdots a_1(\Gamma_k)$ NBC spanning forests associated to every $B$ of length $k$. Summing over all $B$ we get $a_k(\Gamma)$.
\end{proof}


\section{The Poincar\'e Polynomial}\label{proofmain}

We derive the homology of $\conf{\Gamma} (\bbr^N)$ from its bond lattice. The method is standard using the Goresky-MacPherson formula but the details are specific. Let's again denote by $\mathcal A=\{A_{e}\}_{e\in E(\Gamma)}$ the graphical arrangement defining $\conf{\Gamma}(\bbr^N)$ (see \S\ref{bondlat}).
We need some notation: for $x$ in the lattice of intersections of $L(\mathcal A)$, we denote by $(\hat 0,x)$ the full subposet consisting of elements $\{y\in\mathcal L(\mathcal A)\ \ |\ \hat 0< y<x\}$ (i.e. the \textit{lower} open interval). An element $x\in L(\mathcal A)$ represents a linear subspace, of dimension $\dim x$ (a multiple of $N$). In this lattice, $\hat 0 = (\bbr^N)^{m}$, with $\dim \hat 0 = mN$, and $\hat 1$ is the intersecction of all $A_{ij}$.  

There is now a remarkable
formula of Goresky and MacPherson's relating  the lower intervals of the intersection lattice to the homology of the complement of the subspace arrangement (\cite{wachs}, Theorem 1.3.8, \cite{goreskymacpherson}). In the case of the complement of a graphic arrangement associated to a graph $\Gamma$, $|V(\Gamma )|=m$, this formula takes the form
\begin{equation}\label{gmformula}\tilde H^i(\conf{\Gamma} (\bbr^N);\bbz)\cong
\oplus_{x\in\mathcal L(\mathcal A)\setminus \{\hat 0\}}\tilde H_{mN-i-\dim x - 2}(\Delta\mathcal (\hat 0,x);\bbz)
\end{equation}
In this formula $\tilde H_{-1}(\emptyset)\cong\bbz$. Since we also need to establish a naturality result about this isomorphism, we review in the Appendix how a homotopy theoretc proof of this formula is obtained \cite{zz}.

\bex We can illustrate the formula \eqref{gmformula} for the complement
$\conf{2}(\bbr^N)= (\bbr^N)^2\setminus \{(x,x)\}$.
This space is of the homotopy type of $S^{N-1}$. The lattice of intersection consists of $\hat 0$ and a single atom (the diagonal) of codimension $N$. There is one only empty interval and
$$H^{N-1}(\conf{2} (\bbr^N))\cong \tilde H_{-1}(\emptyset)\cong\bbz$$
For $\conf{3}(\bbr^N)$ (and $\Gamma = K_3$), the lattice of intersections consists of $\hat 0 = \bbr^{3N}$, $\hat 1=\{(x,x,x)\}$ the diagonal, and three atoms. We use the formula
\eqref{gmformula} with $m=3$.
Each atom $x_i$ has dimension $2N$, while $\hat 1$ has dimension $N$. Each atom $x$ contributes an empty interval $(\hat 0,x)$, and so the group
$H_{3N-i-2N-2}(\emptyset)
= \tilde H_{-1}(\emptyset) = \bbz$ when $i=N-1$. Since there are $3$ atoms, we get that
$H^{N-1}(\conf{3}(\bbr^n))\cong\bbz^3$. The interval $(\hat 0,\hat 1)$ consists of three points (the atoms), and it contributes the group
$H^{2(N-1)}(\conf{3}(\bbr^N))\cong\tilde H_{3N-(2N-2)-N-2})(\hat 0,\hat 1) = {\tilde H}_0(\hat 0,\hat 1)\cong\bbz^2$. 
This is of course what we already know from \eqref{calc2}.
\eex

Our next objective is thus to understand the order complexes of the lower open intervals $(\hat 0,x)$ for the bond lattice of a graph.
Since $ L(A)$ is geometric,  all intervals are geometric, 
and by Theorem \ref{shellable}, the order complexes of all subintervals are wedges of spheres.
 As a first consequence of being geometric, all lower intervals of the bond lattice have the homology of wedges of spheres, and hence, by the Goresky-MacPherson formula, $H_*(\conf{\Gamma} (\bbr^N),\bbz )$ is torsion free.

As suggested by P. Blagovic, the next two results can be derived by using analog results for the partition lattice in (\cite{stanley2}, p.319) and a lemma on the homotopy type of general intervals of product posets (\cite{walker}, Theorem 6.1). We find it more convenient however, for us and for the reader, to give the direct argument. 

Since each $(\hat 0,x)$ is a wedge of spheres, we need to compute its number of summands, in terms of $x$. To do that, we prove a fundamental decomposition property. Recall that the product of two posets $(P_1,\leq)$ and $(P_2,\leq)$ is the
poset whose underlying set is $P_1\times P_2$ and partial ordering
$$(x_1,y_1)\leq (x_2,y_2)\ \Longleftrightarrow\ x_1\leq x_2\ \hbox{and}\ y_1\leq y_2$$
As before, if $B_i$ is a bond for $\G$, we will write $\Gamma_i$ the induced subgraph with vertices $B_i=V(\Gamma_i)$.

\ble If $x = B_1|B_2|\cdots |B_k\in\Pi_\Gamma$, and $\Gamma_i\subset\Gamma$ is the full connected subgraph of $\G$ with vertices $B_i$,
then the interval $[\hat 0,x]$ in $\Pi_\G$ is isomorphic to the product of the bond posets $[\hat 0,x]\cong  \Pi_{\Gamma_1}\times \cdots \times \Pi_{\Gamma_k}$.
\ele

\begin{proof} 
An element in $[\hat 0,x]$
is necessarily of the form 
\begin{equation}\label{element}\underbrace{B_{i^1_1}|\cdots |B_{i^1_{r_1}}} |\underbrace{B_{i^2_1}|\cdots |B_{i^2_{r_2}}} |\cdots\
\underbrace{|B_{i^k_1}|\cdots |B_{i^k_{r_k}}} 
\end{equation}
where each grouping  $|B_{i^j_1}|\cdots |B_{i^j_{r_j}}|$ is a sub-partition of $B_j$.
The reason for this is because no block can contain elements of two different subgraphs $\Gamma_i$ and $\Gamma_j$ simultaneously, since the graphs are disjoint. We then have a well-defined and bijective set map between the posets $[\hat 0,x]$ and $\Pi_{\G_1}\times\cdots\times\Pi_{\G_k}$ sending \eqref{element} to 
$\ds (B_{i^1_1}|\cdots |B_{i^1_{r_1}}) \times (B_{i^2_1}|\cdots |B_{i^2_{r_2}})\times\cdots\times
(B_{i^k_1}|\cdots |B_{i^k_{r_k}})$.
This is a poset map because the only way we can coarsen the partition in \eqref{element} is by coarsening each grouping individually.
\end{proof}

\bpr\label{wedgebond}
$\displaystyle 
\Delta (\hat 0,B_1|\cdots |B_k)\simeq
\bigvee^{a_1(\Gamma_1)\cdots a_1(\Gamma_k)}S^{m-k-2}$.
\epr  

\begin{proof} Write $\mu_\Gamma = \mu_{\Pi_\G}$ the Mobius function \eqref{mobiustotal}.
 By Proposition 1.2.1 of \cite{wachs}, $\mu$ is multiplicative on a product of posets, and by an iterate use of \eqref{key}, we obtain
\begin{equation*}\label{muinterval}\mu_\G (\hat 0,x) =
\prod \mu_{\G_i} (\hat 0, B_i) = \prod (-1^{k_i-1})a_1(\Gamma_i) = (-1)^{m-k}\prod a_1(\Gamma_i)
\end{equation*}
On the other hand, 
each $x=B_1|\ldots, |B_k$ has length $k$ and rank $m-k$ (Definition \ref{bondproperty}). According to  Theorem \ref{shellable}, the interval
$(\hat 0,x)$ breaks down into a wedge of spheres of dimension $m-k-2$.
Let $r$ be the number of such spheres $S^{m-k-2}$, then by the proof of that same theorem,
$\mu (\hat 0,x) = (-1)^{m-k-2}r = (-1)^{m-k}r$. By equating with the above expression for $\mu_\G (\hat 0,x)$, we get that $r = \prod a_1(\Gamma_i)$. This proves the claim.
\end{proof}

\begin{proof}(of Theorem \ref{poinc}).
We apply the Goresky-Macpherson formula now that we know the homology of every interval (Proposition \ref{wedgebond}). We recall that the bond poset is the intersection lattice of the graph arrangement (Lemma \ref{bondlat2}). If $x=B_1|\ldots |B_k$, then the corresponding linear subspace $x$ has dimension $kN$.
In the formula \eqref{gmformula}, a class in degree $m-k-2$ contributes to the homology of $\conf(\bbr^N,\Gamma)$ a class in degree $i$ (remember the shift of degrees in the formula) where upon replacing
$$m-k-2 = mN - i - kN -2\ \Longrightarrow\ i = mN-m -kN+k = (m-k)(N-1)$$
There are $\prod a_1(\Gamma_i)$ such classes, for each bond partition $x = B_1|\cdots |B_k\in\mathcal \mathcal B(\Gamma )$. 
The homology is then explicitly given as follows
\begin{eqnarray*}
P_t(\conf{\Gamma}(\bbr^N)) &=& \sum_{1\leq k\leq m}\left(\sum_{B\in\mathcal B_k}a_1(B)\right)t^{(m-k)(N-1)} \\
&=&\sum_{1\leq k\leq m} a_k(B) t^{(m-k)(N-1)}\ \ \ \ ,\ \ \ \hbox{by Proposition \ref{aktoa1}} 
\end{eqnarray*}
This proves Theorem \ref{poinc}
\end{proof}

\begin{proof} (of Formula \eqref{beta2}). We derive the rank $\ds \beta_{2(N-1)} = {|\Gamma|\choose 2} - s_3(\Gamma)$. The block partitions that contribute to $\beta_{2(N-1)}$ must have length $m-2$, $m = |V(\Gamma )|$. These partitions are of the form
$|ij|rs|-$, $\{i,j\}\cap \{r,s\}=\emptyset$, where $-$ means a sequence of one blocks, or
$|ijk|-$, where $\{i,j\},\{r,s\}\in E(\Gamma )$ are edges. Now, any pair of disjoint edges give a class in degree $2(N-1)$. Similarly, a pair of adjacent edges gives rise to at least one such class. Let $\Gamma_{ijk}$ be the subgraph whose vertices are $i,j,k$. If it is a line graph, then $a_1(\Gamma_{ijk})=1$, and the block $|ijk|-$ contributes a single class. If $\Gamma_{ijk}$ is a triangle, then $a_1(\Gamma_{ijk})=(3-1)!=2$.
Let's count a class for every pair of edges. If three edges make up a triangle, we get three such classes, but they only account for one block, the block $|ijk|-$, which only contributes $a_1(\Gamma_{ijk})=2$ classes, not $3$. There is one too many classes for every triangle. This gives the count.
\end{proof}

\bco For $N>2$, $\conf{\Gamma} (\bbr^N )$ is $N-2$-connected.
\eco

\begin{proof}
Since the first non-zero betti number in positive degree is $\beta_{N-1}$, and since $\conf{\Gamma} (\bbr^N)$ is of the homotopy type of a CW complex, it is then enough to show that this space is simply connected. But if $N > 2$, $\conf{\Gamma} (\bbr^N )$ is obtained from $\bbr^{mN}$ (a simply connected space) by removing closed submanifolds of codimension strictly greater than $2$. The fundamental group of the complement agrees with the fundamental group of $(\bbr^n)^{|V(\Gamma)}|$ which is trivial. 
\end{proof}


\section{Edge Labelings and Toric Classes}\label{labeling}

The main construction of \cite{paolo1} is to associate to every rooted tree in $K_m$, whose vertices have increasing labels along any path starting from the root (see \S 8.1),  a homology class in $\conf{m}(\bbr^N)$. We describe here a similar construction that works for any simple graph. This requires the labeling of edges.

\begin{theorem}\label{factorizing} Let $\Gamma$ be a simple graph on a vertex set $V(\G)=\{1,\ldots, m\}$, with a linear ordering on its edges. Let $T$ be a spanning tree of $\Gamma$. Then there is a map $\alpha_T$ and a homotopy commutative diagram
\begin{equation}\label{commutative}
\xymatrix{\conf{m}(\bbr^N)\ar@{^(->}[rr]^{\pi_K^\G}&&\conf{\Gamma}(\bbr^N)\ar[d]^{\pi_\G^T}\\
    (S^{N-1})^{|T|}\ar[rr]^{\beta_T} \ar[u]^{\alpha_T}&&\conf{T}(\bbr^N) }
\end{equation}
where $\beta_T$ is the homotopy equivalence of \eqref{themapbetaT},
and both $\pi_K^\G$ and $\pi_\G^T$ are the natural functorial morphisms (see Remark \ref{mappi}).
\end{theorem}

As pointed out already, both maps $\pi_K^\G$ and $\pi_\G^T$ are open embeddings. The Theorem can be stated for non spanning trees, but it is necessary to keep track of the labeling of vertices in this case for $T$, and this an extra technicality that we do not need to consider for now. As indicated in Remark \ref{mappi}, in the case of spanning trees, the vertices in $T$ and $\G$
have the same labeling. The diagram says that we can construct a homotopy equivalent map to $\beta_T$ that factors through the configuration space once we fix a labeling of the edges. This section gives the construction and the proof of this fact. 

Fix a linear ordering on the edges of $\G$ and let $T$ be a spanning tree of length $|T|= m-1$, with edges 
$e_1 < e_2 < \ldots < e_{m-1}$. The only condition we add is that $e_1$ is an edge incident to the root $v_1$. 
Write $e^-$ and $e^+$ the vertices incident to an edge $e\in E(\Gamma)$, where $e^-, e^+\in \{1,\ldots, m\}$. The way the $+$ and the $-$ are assigned is as follows:
\begin{itemize}
    \item Start with the smallest edge $e_1$ and label its vertices $e_1^+,e_1^-$. The vertex $e_1^-$ is $v_1$ by construction.
    \item If $e_r$ is incident to $e_1$, then set $e_r^-$ the common vertex with $e_1$, that is $e_r^-$ is either $e_1^-$ or $e_1^+$. Obviously $e_r^+$ is the other vertex of $e_r$.
    \item Continue iteratively. Given $e_s$ with $s<r$ in the linear ordering, if $e_r$ is incident to $e_s$, then $e_r^-=e_s^-$ or $e_s^+$ at the point of incidence, and $e_r^+$ is the other vertex. Figure \ref{planetary1} explains the construction via an example.
\end{itemize}

The idea now is to use this setup to construct a toric embedding corresponding to this spanning tree.   Fig. \ref{planetary1} on the left illustrates the construction for that particular tree. We spell out the details.

\begin{figure}[htb]
    \centering    \includegraphics[scale=0.7]{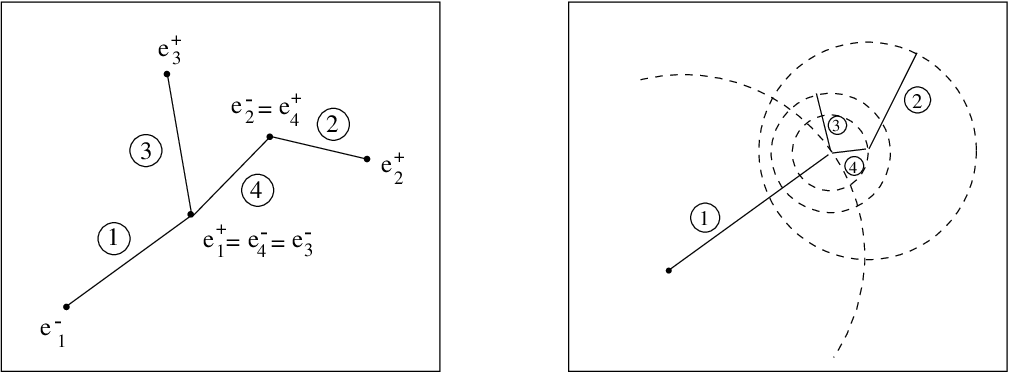}
    \caption{The tree $T$ has ordered edges $\raisebox{.5pt}{\textcircled{\raisebox{-.9pt} {$1$}}}, \ldots, 
    \raisebox{.5pt}{\textcircled{\raisebox{-.9pt} {$4$}}}
        $ (equivalently $e_1,e_2,e_3$ and $e_4$). The tore associated with this tree in $\conf{5}(\bbr^N)$ is drawn on the right. The tore is the product of $4$ spheres (the dashed circles give a snapshot of these spheres). The radii of the embedded spheres are $1, {1\over 2}, {1\over 4}$ and ${1\over 2^3} = {1\over 8}$. See Example \ref{planetexample}.}\label{planetary1}
\end{figure}

\vskip 5pt
\noindent{\sc Setup}: Note that edges have a linear ordering and vertices have labels and a partial ordering. Both will be used in the proof. Let $V(\Gamma) = \{1,\ldots, m\}$ and $e\in E(\G)$. If $e_t^\pm=v_i$ (or $i\in V(\Gamma)$), we write 
$\ell (e_t^\pm) = i$.  More explicitly, if vertex $v_i$ is incident to $e_t$, labeled  $e_t^+$ there, then $\ell (e_t^+)=i$. As above, we will fix $\ell (e_1^-)=1$, this means the lowest edge $e_1$ is rooted at $v_1$.

Given a spanning tree $T\subset\G$, we construct an embedded torus $C_T\subset \conf{m}(\bbr^N)$, $C_T\cong (S^{N-1})^{m-1}$, as the locus of all configurations $(x_1,\ldots, x_m)$ obtained as follows:
\begin{itemize}
\item $x_{\ell (e_1^-)}=x_1=0$ is at the origin. 
\item  $x_{\ell (e_1^+)}$ is any point on a sphere $S_1$ of radius $1$ centered at the origin.
\item Now suppose $e_t$ is adjacent to $e_1$, with $e_t^-=e_1^+$. Then $x_{\ell (e_t^+)}$ can take any position on a sphere centered at a point $x_{\ell (e_t^-)} = x_{\ell (e_1^+)}\in S_1$, of radius ${1\over 2^{t-1}}$. If $e_t$ is adjacent to $e_1$ so that $e_t^-=e_1^-$ (the root), then $x_{\ell (e_t^+)}$ is placed on a sphere centered at $e_1^-$ of radius $1\over {2^{t-1}}$.
\item Clearly, we proceed recursively. If the locus of $x_{\ell (e_s^+)}$ and $x_{\ell (e_s^-)}$ are determined, look at the next edge in the linear ordering $e_r$ that is adjacent to $e_s$. If $e_r^- = e_s^{\pm}$, then $x_{\ell (e_r^+)}$ must be on a sphere centered at some $x_{\ell (e_r^-)}$ with radius ${1\over 2^{r-1}}$.
\item We continue this construction edge by edge as we move along paths away from the origin, until we exhaust the tree.  
\item Once we exhaust the entire tree, the subspace $C_T$ of all such configurations
$(x_1=0,\ldots, x_m)$ that we constructed is in $\conf{m}(\bbr^N)$ and is homeomorphic to a product of spheres  $(S^{N-1})^{|T|}$. 
\end{itemize}

The next example gives an illustration of this construction. 

\bex\label{planetexample}
Figure \ref{planetary1} describes the torus obtained from a tree $T$ with $|T|=4$. The ordering of the edges is given by the encircled numbers. We can assume that as vertices of $T$, 
$$\ell (e_1^-)=1\ ,\ \ell (e_1^+)=\ell (e_3^-)=\ell (e_4^-)=3,\ \ \ell(e_3^+)=4\ ,\ \ell (e_2^-)=\ell (e_4^+)=2\ \ ,\ \ \ell (e_2^+)=5$$
Applying the construction described above, we obtain a product of $4$ spheres, the one with largest radius $1$ is centered at the origin $x_1=0$. Then $x_3$ is on that sphere, while $x_4$ is on any sphere of radius ${1\over 2^2}$ centered at some $x_3$,  $x_2$ is on any sphere of radius ${1\over 2^3}$ centered at $x_3$ as well, and finally $x_5$ is on any sphere centered at some $x_2$ of radius $1\over 2$. This product of spheres is of course lying in $\conf{5}(\bbr^N)$ which is itself lying inside
$$\conf{T}(\bbr^N) = \{(x_1,\ldots, x_5)\ | x_1\neq x_3, x_3\neq x_2, x_3\neq x_4, x_2\neq x_5\}$$
The picture on the right side of Figure \ref{planetary1} gives a snapshopt of $C_T$ which consists of all $(x_1=0,x_2,x_3,x_4,x_5)$ with each $x_i$ being on the dashes sphere corresponding to $\ell (e_t^+)=i$ for some $t$.
\eex

\begin{proof} (of Theorem \ref{factorizing}) The map $\alpha_T$ is of course the embedding we just constructed of $(S^{N-1})^{|T|}$ whose image is $C_T\subset \conf{m}(\bbr^N)$. We are left to explain why the diagram homotopy commutes. Both maps $\beta_T$ and $\pi\circ\pi_K^\G\circ\alpha_T$  are constructed similarly. The first one uses a labeling of the vertices and their partial ordering, and the second uses a labeling of the edges. These are related as follows: if $e$ is an edge, with vertices $e^-$ and $e^+$, and if $\ell (e^+)=i$, then $e^+ = v_i$ (by definition) and $e^-=v_{i^*}$ (notation in the proof of Proposition \ref{treedecomposition}). The map $\beta_T$ is homotopic to this composite via the homotopy which consists in opening up the radii of the embedded spheres that we carefully constructed for $\alpha_T$ from radii ${1\over 2^t}$ to unit radius. The homotopy can be done in steps by linearly expanding the radii as we move away from the root (this should be self-explanatory and no details will be provided). This homotopy is well-defined, since when the points of a configuration coalesce during homotopy, they are allowed to do so in $\conf{T}(\bbr^N)$. See Example \ref{homotopy2}.
\end{proof}

\bex\label{homotopy2} We go through the example of the linear graph $L_4$ in details to explain Theorem \ref{factorizing}. The idea is simple and attractive, and is illustrated in Fig. \ref{homotopy}. The vertices are labeled $v_1,v_2,v_3,v_4$ (or $1,2,3,4$ as shown), with partial ordering
$v_1\prec v_3\prec v_4\prec v_2$. The tree is rooted at $v_1$. Figure \ref{homotopy} on the left shows the embedding via $\alpha_T$ which factorizes through the configuration space $\conf{4}(\bbr^N)$. The spheres are of different radii so that the points in the configuration never meet. The linear ordering on the edges showing as $\raisebox{.5pt}{\textcircled{\raisebox{-.9pt} {$1$}}}, \raisebox{.5pt}{\textcircled{\raisebox{-.9pt} {$2$}}}, \raisebox{.5pt}{\textcircled{\raisebox{-.9pt} {$3$}}}$ in the figure determines the radii of the spheres which in this case are $1\over 2$ for the sphere where $x_4$ rotates around a fixed $x_3$, and a sphere of radius $1\over 4$ consisting of the locus of an $x_2$ rotating around a fixed $x_4$. Now ``inflate'' the spheres as shown, step by step as we move away from the root, so at the end all spheres have unit radius. The final result is depicted on the right, and this is the image of the map $\beta_T$. In this image, $x_1$ and $x_3$ cannot coincide, $x_3$ and $x_4$ cannot as well, and also $x_4$ and $x_2$, however $x_1$ and $x_4$ can touch, so can $x_2$ and $x_3$. This is consistent with the conditions for defining $\conf{T}(\bbr^N)$, and the homotopy lies in this subspace.
 \eex

\begin{figure}[htb]
    \centering    \includegraphics[scale=0.5]{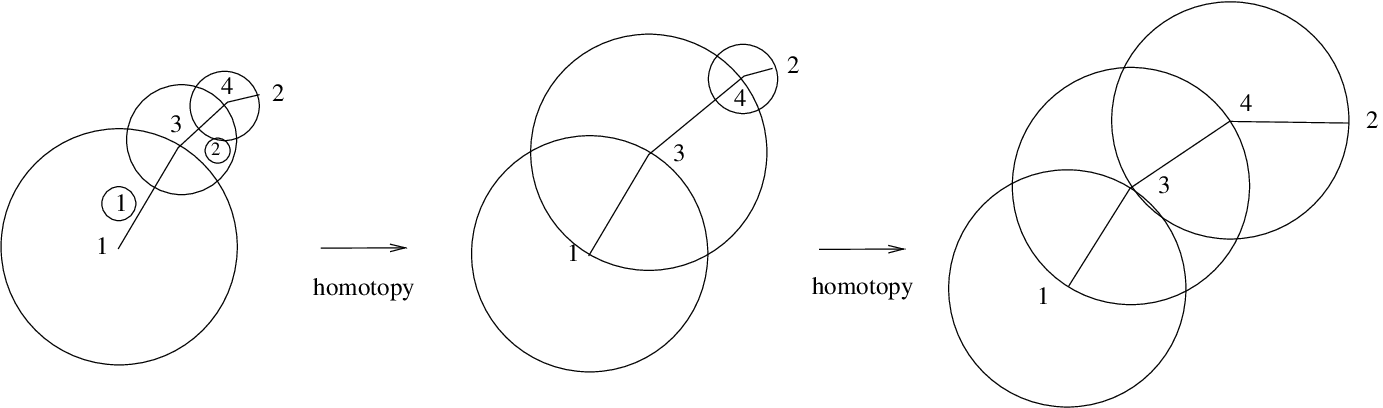}
    \caption{The image of $\gamma_T=\pi\circ\iota\circ\alpha_T$ in $(\bbr^N)^4$ is depicted on the left hand side, while the image of $\beta_T$ is on the right. These depictions show snapshots of spheres. The homotopy between both maps inflates the radii of the embedded spheres from radii $1
    \over 2$ and then $1\over 4$ to unit radius.}\label{homotopy}
\end{figure}

Similar embeddings of tori can be constructed from forests, and the analogue of Theorem \ref{factorizing} can be stated in this case as well.
Fix a graph $\Gamma$ with vertices labeled in $\{1,\ldots, m\}$ and let $\mathcal F$ be a forest on $k$ disjoint trees. Then  we have a
homeomorphism and a homotopy equivalence
\begin{eqnarray*}
    \conf{\mathcal F}(\bbr^N)&\cong&\conf{T_1}(\bbr^N)\times\cdots\times\conf{T_k}(\bbr^N)\\
&\simeq&(S^{N-1})^{|T_1|}\times \cdots\times (S^{N-1})^{|T_k|}
\end{eqnarray*}
For a spanning forest $\mathcal F$ in $\G$, $\conf{\G}(\bbr^N)$ is an open subset of $\conf{\mathcal F}(\bbr^N)$, and there is an embedded
torus $C_\mathcal F\cong (S^{N-1})^{|T_1|+\cdots +|T_k|} = (S^{N-1})^{m-k}$ inside $\conf{\G}(\bbr^N)$ which is a deformation retract of 
$\conf{\mathcal F}(\bbr^N)$.
In other words, there is a homotopy commutative diagram
\begin{equation}\label{diagramforest}
\xymatrix{ 
&\conf{|T_1|+\cdots + |T_k|+k}(\bbr^N)\ar[r]^{\ \ \ \ \ \ \pi^\G_K}&\conf{\Gamma}(\bbr^N)\ar[d] \\
&(S^{N-1})^{|T_1|}\times\cdots \times (S^{N-1})^{|T_k|}\ar[u]^{\alpha_{\mathcal F}}\ar@{..>}[r]&\conf{\mathcal F}(\bbr^N)}
\end{equation}
whose end result (the bottom map) is a homotopy equivalence.
Here $m=|T_1|+\cdots + |T_k|+k = |V( F)|$. The bottom map is very explicit, constructed tree by tree by embedding the tori for each tree far from each other by rooting them at different points at a distance of $m$ units from each other. Details are omitted.
Figure \ref{planetary2} gives an illustration.

\begin{figure}[htb]
    \centering    \includegraphics[scale=0.7]{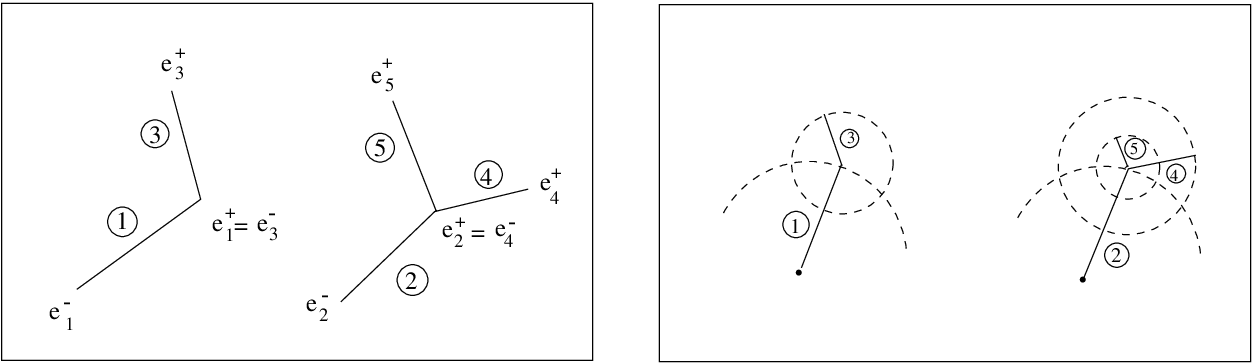}
    \caption{The forest $F $ with two components (left) and its associated torus in $\conf{7}(\bbr^N)$ ($7 = |T_1| + |T_2|+2 = |V(F )|$). One performs the construction for trees (which are rooted), for every component of the forest, with the modification that the roots are placed away from each other (not all are at the origin).}\label{planetary2}
\end{figure}

\vskip 4pt
\noindent{8.1 {\bf Complete Graphs and NBC forests}.}
The classical theory pertaining to when $\Gamma$ is complete does not use the language of edge labeling and broken cycles \cite{paolo1, dev}. The reason, as we uncover here, is that the bases constructed in these references turn out to be a byproduct of a canonical ordering of the edges, so basis generators are automatically NBC. 
We explain below the idea and, in so doing, explain a method to easily construct NBC forests for complete graphs.

As always, start with $\Gamma$ a simple graph on $m$ vertices labeled from $1$ to $m$, and now order the edges as follows. We list all the edges of $\Gamma$ as pairs 
$(i,j)$ with $i<j$. Assume there are $n$ such edges. We then assign the label $n$ to the edge $(i,j)$ with the smallest sum $i+j$ and largest $j$. The next smallest sum is assigned to the second largest label $n-1$, and so on. If two edges $(i_1,j_1)$ and $(i_2,j_2)$ have the same vertex sum $i_1+j_1 = i_2+j_2$, then we assign the smaller edge label to the edge $(i_1,j_1)$ if $j_1> j_2$. We call this labeling of the edges, obtained from a given ordering of the vertices, the NBC ordering.

\bex The square has $4$ edges $(1,2)$, $(2,3)$, $(3,4)$ and $(1,4)$. The labeling of the edges is shown below encircled
\begin{figure}[htb]
    \centering
    \includegraphics[scale=0.5]{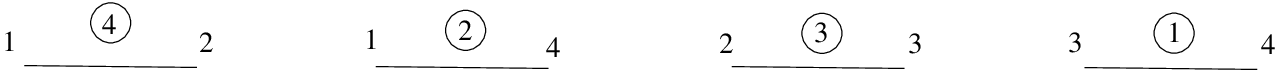}
\end{figure}
\eex

\bde A tree in $\Gamma$ is called {\it increasing} if it is rooted, the root is at its smallest valued vertex, and any path from the root to any of the leaves will go through increasing vertices. 
A $\Gamma$-forest is increasing if all of its trees are increasing. Below is an increasing forest, and next to edges are the NBC labelings.
\begin{figure}[htb]
    \centering
    \includegraphics[scale=0.5]{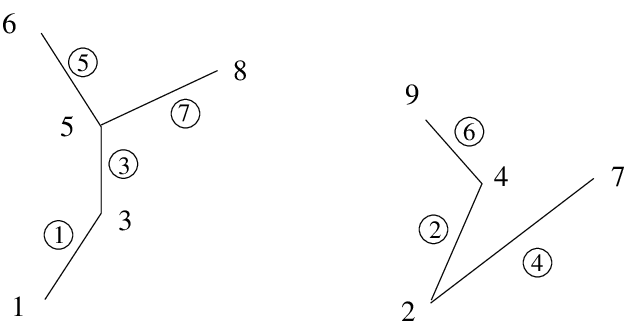}
\end{figure}
\ede

\ble Let $\Gamma$ be a simple graph on $m$ vertices and $n$ edges. We assume the vertices are labeled, and we give the edges the NBC labeling. Then any increasing $\Gamma$-forest must be NBC.
\ele

\begin{proof} Let $T$ be an increasing tree and add an edge there to close a cycle. 
The added edge is of the form $e=\{i_j,i_k\}$ say, with $i_j<i_k$ are two vertices in $T$.  The adjacent vertex $i_t$ in the path going from $i_j$ to $i_k$ cannot be $i_k$, so $i_t<i_k$ (note that we are not saying that $i_j<i_t$ which may or may not be true). Since $i_j+i_t< i_j+i_k$, the added edge $e$ cannot have maximal ordering in the cycle according ot the NBC-ordering, so the cycle is NBC. 
\end{proof}

It turns out that increasing forests give all NBC forests for complete graphs. 

\ble Give edges of $K_m$ the NBC labeling, and choose a tree rooted at its lowest labeled vertex. This tree is NBC if and only if it is increasing.
\ele

\begin{proof} We already know that increasing trees are NBC. Assume that such an NBC-tree is not increasing. So three consecutive vertices along the path from a root to a leaf will not be increasing. We distinguish two cases. (i) assume the ordered vertices (from root to leaf) are $(i_1,i_2,i_3)$ with $i_1<i_2>i_3$. In the complete graph, we can add the edge $\{i_3,i_1\}$ to make a cycle. Note that $i_2+i_1 > i_3+i_1$ and $i_2+i_3>i_1+i_3$ so the sum $i_3+i_1$, which means the NBC labeling of this edge is maximal, and the added edge $\{i_1,i_3\}$ breaks the cycle, which contradicts the fact that the tree is NBC. Similarly (ii) if $(i_1,i_2,i_3)$ are lined up along a path, with $i_1>i_2<i_3$, Here $i_1+i_3$ may not break the cycle made up of the vertices $i_1,i_2$ and $i_3$. This cannot consist of the entire tree, since otherwise it will be rooted at $i_2$ (the smallest labeled vertex). So there is a path in the tree from the root, say $i_k$ to $i_1$, $i_k<i_1$. We're back to the situation (i) for the triple $(i_k,i_1,i_2)$, so the edge $\{i_k,i_2\}$ breaks a cycle (not necessarily a triangular cycle, since between $i_k$ and $i_1$ there could be other vertices). In all cases, if the tree is not increasing, some edge will break some cycle, and the tree is not NBC as claimed.
\end{proof}

\bre \label{completecase}
In the case $\G = K_m$ is complete, and
$\conf{K_m}(\bbr^N)=\conf{m}(\bbr^N)$, a graded basis for the homology is given by the classes $[F ]$ of \textit{increasing spanning forests}. This is a reformulation of the main  computation in \cite{paolo1}. This is the same as giving a basis by NBC- spanning trees, in the induced NBC labeling. This fact will be needed in the proof of Theorem \ref{mainbis}.
\ere


\section{Naturality of the Goresky-MacPherson Isomorphism}\label{naturality}

We here analyze the embedding $\pi_\Gamma: \conf{m}(\bbr^N)\hookrightarrow\conf{\Gamma}(\bbr^N)$ in homology and use our analysis to derive the stable splitting in \S\ref{stable}. Let $\G$ be a simple connected graph on $m$-vertices, $\G\subset K_m$, and $\G$ is taken to be a spanning subgraph of the complete graph $K_m=K$. An interval $(\hat 0,B)$ of $\Pi_\Gamma$ is a sublattice of the corresponding interval in $\Pi_{K_m}$ (Lemma \ref{subinterval}). We will write $\Delta_\Gamma(\hat 0,B)$ the order complex of the interval $(\hat 0,B)$ of  $\Pi_\Gamma$, and similarly $\Delta_K(\hat 0,B)$ the order complex of this interval in $\Pi_{K_m}$. There is a full subcomplex inclusion
$$\tau^B:\xymatrix{\Delta_\Gamma (\hat 0,B)\ar@{^(->}[r]& \Delta_K(\hat 0,B)}$$ 
We know that each interval is a wedge of spheres (Proposition \eqref{wedgebond}). It turns out that $\tau^B$ respects this splitting as we now make precise. More precisely, if $G=K$ or $\G$, then $(\hat 0,B)_{\Pi(G)}$ is a bouquet of spheres of dimension $m-|B|-2$, and the number of such spheres is $a_1^G(B) = a_1(G_1)\cdots a_1(G_k)$, where $G_i$ is the induced graph on vertices $B_i$ in $G$. 

\ble\label{fromPtoL} Let $\Gamma$ be a spanning graph of $K=K_m$, and $B$ a bond partition of $\Gamma$. The inclusion $\tau^B$ is homotopic to an inclusion 
$\bigvee^{a_1^\Gamma(B)}S^{m-|B|-2}\hookrightarrow\bigvee^{a_1^K(B)}S^{m-|B|-2}$ as a wedge summand. In particular, the map $\tau^B_*$ in homology sends generators to generators.
\ele

\begin{proof}
The interval
$(\hat 0,B)$ in $\Pi(\Gamma)$ is obtained from the same interval in $K$ by removing atoms. The subcomplex 
$\Delta_\Gamma (\hat 0,B)$ is a full subcomplex of $\Delta_K(\hat 0,B)$, and moreover, maximal chains have the same length by Lemma \ref{subinterval}. This situation is well illustrated by Fig. \ref{bondC4} for example. This implies that capping facets of $\Delta_\Gamma (\hat 0,B)$ are necessarily capping facets of $\Delta_K(\hat 0,B)$, in particular they are of the same dimension. The claim now follows from Proposition \ref{capping}.
\end{proof}

\bre From the Lemma , $a_1^\Gamma (B)\leq a_1^K(B)$, for any bond partition $B$ of $\G$. It is interesting to see this in a different way. Write $B=B_1|B|\cdots |B_k$, with each $B_i$ corresponding to a subset of $V(\G)=V(K)$. By construction,
$a_1^\Gamma (B) = a_1(\G_1)\ldots a_1(\G_k)$, and $a_1^K(B) = a_1(H_1)\cdots a_1(H_k)$, where $\G_i$ (respectively $H_i)$ is the induced subgraph of $\G$ (respectively of $K$) with vertex set $B_i$. This means that $\G_i$ is a spanning subgraph of $H_i$.
But $a_1(\G_i)$ is the number of acyclic orientations of $\G_i$ with a unique source.  Adding more edges to a graph cannot increase the number of acyclic orientations with a unique source, so that $a_1(\G_i)\leq a_1(H_i)$. This holds for all $1\leq i\leq k$, so holds for the product and $a_1^\G(B)\leq a_1^K(B)$.
\ere

Next is the main result of this section on the functoriality of the Goresky-MacPherson (GM) isomorphism
in \eqref{gmformula}.

\begin{proposition}\label{dualityGM} Let $\G$ be a spanning subgraph of $K_m$. Then the following diagram commutes
$$\xymatrix{\tilde H^i(\conf{\Gamma} (\bbr^N))\ar[rr]^{(\pi_\Gamma^{K})^*}\ar[d]^\cong&&\tilde H^i(\conf{m}(\bbr^N))\ar[d]^\cong\\
\oplus_{B\in\Pi_\Gamma\setminus \{\hat 0\}}\tilde H_{mN-i-\dim B - 2}(\Delta_\Gamma\mathcal (\hat 0,B))\ar[rr]^{\oplus \tau_*^B}&&
\oplus_{B\in\Pi_K\setminus \{\hat 0\}}\tilde H_{mN-i-\dim B - 2}(\Delta_K\mathcal (\hat 0,B))
}$$
where the vertical maps are given by the GM-isomorphism formula.
\end{proposition}

\begin{proof} 
The proof follows from a naturality result for the Alexander duality isomorphism \cite{clarke}, and by inspection of the proof of the GM-formula given in \cite{zz} and discussed in the appendix.

Generally, given a subspace arrangement $\mathcal A = \{A_1,\ldots, A_n\}$, $A_i\subset\bbr^L$, write  $V_{\mathcal A}$ the union of all subspaces $\bigcup_{i=1}^nA_i$, and $M_{\mathcal A}$ its complement in $\bbr^L$ \cite{deshpande}. The compactification of these subspaces in $S^L=\bbr^L\cup\{\infty\}$ are denoted $\hat V_{\mathcal A}$ and $\hat M_{\mathcal A}$ respectively. 
Let $\mathcal B\subset \mathcal A$ be a subcollection of affine subspaces. There is an inclusion of compact sets
$\iota : \hat V_{\mathcal B}\hookrightarrow\hat V_{\mathcal A}$. By Theorem 5.1 of \cite{clarke}, there is a commutative diagram, where $D$ is the Alexander duality operator
\begin{equation}\label{diagram1}\xymatrix{\tilde H_i(\hat V_{\mathcal B})\ar[r]^{\iota^*}\ar[d]^{D}&\tilde H_i(\hat V_{\mathcal A})\ar[d]^{D}\\
\tilde H^{L-i-1}(M_{\mathcal B})\ar[r]&
\tilde H^{L-i-1}(M_{\mathcal A})}
\end{equation}
and the bottom map is the induced cohomology map from the other inclusion $M_{\mathcal A}\hookrightarrow M_{\mathcal B}$. 

Let $L_\infty (\mathcal A)$ be the intersection lattice $L(\mathcal A)$ adjoined an extra maximal element $\hat 1=\infty$, and let
$(\hat 0,\beta)_{\mathcal A}$ be an open interval in this lattice. Since $\mathcal B$ is a sublattice of $\mathcal A$, $\Delta_{\mathcal B}(\hat 0,\beta)$ is a subcomplex of $\Delta_{\mathcal A}(\hat 0,\beta)$, and we have natural maps $\Delta_{\mathcal B}*X\rightarrow \Delta_{\mathcal A}*X$ for any $X$. By the wedge lemma (\cite{zz}, Theorem 2.2,), and as explained in the appendix, we have a homotopy  commutative diagram
\begin{equation}\label{wedgelemma}
    \xymatrix{\hat V_{\mathcal B}\ar@{^(->}[r]\ar[d]^{\simeq}&\hat V_{\mathcal A}\ar[d]^{\simeq}\\
    \bigvee_{\beta\in L_\infty(\mathcal B)\setminus \{\hat 0\}} \Delta_{\mathcal B} (\hat 0,\beta)* S^{d(\beta)}\ar[r]&
    \bigvee_{\beta\in L_\infty(\mathcal A)\setminus \{\hat 0\}} \Delta_{\mathcal A} (\hat 0,\beta)* S^{d(\beta)} }
\end{equation}
where the vertical maps are homotopy equivalences (the wedge decompositions), $d(\beta)$ is the dimension of the subspace $A_\beta$, corresponding to $\beta$ in $L(\mathcal A)$, and $S^{d(\beta)}$ is of course its compactification in $S^N$. 
Note that for atoms $\beta$, $\Delta (\hat 0,\beta)=\emptyset$, and $\emptyset*X=X$.

Putting this all together, we have the extended version of diagram \eqref{diagram1} where all vertical maps are isomorphisms
$$\xymatrix{
\tilde H^{L-i-1}(M_{\mathcal B})\ar[r]\ar[d]^{D^{-1}}&
\tilde H^{L-i-1}(M_{\mathcal A})\ar[d]^{D^{-1}}\\
\tilde H_i(\hat V_{\mathcal B})\ar[r]^{\iota^*}\ar[d]^\cong&\tilde H_i(\hat V_{\mathcal A})\ar[d]^\cong\\
\tilde H_i(\bigvee_{\beta\in L_\infty(\mathcal B)\setminus \{\hat 0\}}(\Delta (\hat 0,\beta)*S^{d(\beta)}))\ar[r]&
\tilde H_i(\bigvee_{\beta\in L_\infty(\mathcal A)\setminus \{\hat 0\}}(\Delta (\hat 0,\beta)*S^{d(\beta)}))
}
$$
The vertical composites are precisely the GM-isomorphisms (see \cite{zz}, Discussion after Theorem 2.2) for the arrangements $\mathcal B$ and $\mathcal A$ respectively, once we write for $L'=L\setminus \{\hat 0\}$
$$\tilde H_i(\bigvee_{\beta\in L'}(\Delta (\hat 0,\beta)*S^{d(\beta)})\cong
\bigoplus_{\beta\in L'}
\tilde H_i(\Delta (\hat 0,\beta)*S^{d(\beta)})
\cong \bigoplus_{\beta\in L'}
\tilde H_{i-d(\beta)-1}(\Delta (\hat 0,\beta))
$$
In our case, we must replace 
$\mathcal A$ and $\mathcal B$ by the chromatic subspace arrangements of $K_m$ and $\Gamma$ respectively, 
$L(\mathcal B) = \Pi_{K_m}$,
$L(\mathcal A) = \Pi_\G$, also replace $L=mN$, $\beta=B$ and $d(\beta)=\dim B$. The claim now  follows immediately.
\end{proof}


\section{The Stable Splitting}\label{stable}

Let $\Gamma$ be a simple connected graph, and $B=B_1|\cdots |B_k$ be a bond partition of $\G$. We will define the set ${\mathcal F}_B$ of all spanning forests $F$ in $\G$ consisting of pairwise disjoint trees $T_1, T_2,\ldots, T_k$ such that $V(T_i)=B_i$. Note that $T_i$ need not be an induced tree on the set $B_i$, or in other words it could be any tree in $\G$ on the vertex set $B_i$.
There are many possible such forests. 

Choose any linear ordering of the edges of $\G$ and use it to construct for any such 
$F\in\mathcal F_B$ the map $\alpha_F$ in the diagram \eqref{diagramforest}. This gives that $\conf{ F}(\bbr^N)$ is a (homotopy) retract of $\conf{\Gamma}(\bbr^N)$, and
$H_*(\conf{F}(\bbr^N))$
is a retract of $H_*(\conf{\Gamma}(\bbr^N))$.  
As in \eqref{classF}, we denote by $[F]$ the image of the top homology generator of $\conf{F}(\bbr^N)$
in $H_{|F|(N-1)}(\conf{ \G}(\bbr^N))\cong\bbz$.

\begin{theorem}\label{mainbis} Choose a linear ordering of the edges of $\Gamma$, $m=|V(\G)|$, $1\leq k\leq m$. Then basis generators of $H_{(m-k)(N-1)}(\conf{\Gamma} (\bbr^N))$ can be chosen to be all classes $[F]$ of spanning forests $F\in\mathcal F_B$ having no-broken cycles, and $B$ any bond partition in $\mathcal B_k(\Gamma)$.
\end{theorem}

\begin{proof}
Any simple $\G$ on $m$ vertices is a subgraph of  $K=K_m$. We make a series of deductions:
\begin{itemize}
    \item Assume $\G$ has $r$ edges. We know that $K_m$ has $\displaystyle {m\choose 2}$ edges. Order the edges of $K_m$ so that the edges of $\G$ are labeled the highest $\displaystyle {m\choose 2}, {m\choose 2}-1, \ldots, {m\choose 2}-r+1$. Extend this ordering to edges of $K$ by assigning the other values $\displaystyle 1,\ldots, {m\choose 2}-r$. The good thing about this ordering is that if a tree in $\G$ is NBC, then it is NBC in $K$ as well.  This is the linear ordering we use.
\item Lemma \ref{fromPtoL} shows that the bottom map in diagram \eqref{dualityGM} is injective, which means that $(\pi_K^\G)^*: H^*(\conf{\Gamma}(\bbr^N))\rightarrow H^*(\conf{m}(\bbr^N))$ is injective, and since all groups are torsion-free, the map $(\pi_K^\G)_*$ in homology is surjective.
\item The image $\alpha_{ F_*}([ F])$ in $\conf{m}(\bbr^N)$ is a generator, since this is a non-trivial class coming from a spanning NBC forest of $K$, and NBC forests of $K$ generate the homology of $\conf{K}(\bbr^N)$ \cite{paolo1}).
 It follows that $(\pi^\Gamma_{K})_*\circ \alpha_{ F_*}([ F])$ is a generator of $H_*(\conf{\Gamma}(\bbr^N))$, for every NBC spanning forest of $\G$.
\end{itemize}

We have just argued that the top homology classes of NBC-forests in $\G$ generate the  graded homology groups
$H_*(\conf{\Gamma}(\bbr^n))$. For  degree $*=(m-k)(N-1)$, these classes come from NBC spanning forests on $k$ components. It remains to show they are independent. But the rank of the homology in that degree is $$\beta_{(m-k)(N-1)}
=\sum_{B=B_1|\ldots |B_k}a_1(B) :=\sum_{B=B_1|\ldots |B_k}a_1(\Gamma_1)\ldots a_1(\Gamma_k)$$ 
where $\G_i$ is the induced subgraph of $\G$ on the vertex set $B_i$ (see \S6.1).
This count is precisely the count of all NBC-spanning forests on $k$ components (Proposition \ref{wedgebond}). These must form a basis and the proof is complete.
\end{proof}

\vskip 5pt\noindent{11.1. {\bf Relations.}}
Cycles $C_m$ inside $\Gamma$ give rise to relations in homology for a simple reason which is that a cycle has $m$ distinct spanning trees, but the number of spanning trees with no-broken cycles is $a_1(C_m)=m-1$. More precisely, the spanning trees in $C_m$ give rise to $m$ distinct generators in $H_{(m-1)(N-1)}(\conf{C_m} (\bbr^N),\bbz )\cong\bbz^{m-1}$, according to \eqref{homeo}, so there must be a relation between them.

\ble Label the edges of $C_m$ successively from $1$ to $m$, with $m\geq 3$. We write $T_{i}$ the tree with successive edges $i,i+1,\ldots, i+m-1$ (modulo $m$), and write $[T_i]\in H_{(m-1)(N-1)}(\conf{C_m} (\bbr^N),\bbz )$ the corresponding homology class. Then 
$$[T_1] + [T_2]+\cdots + [T_m] = 0\ \ \ \ \ \ \hbox{(generalized Jacobi relation)}$$
\ele

\begin{proof}
    By rank condition, there must be a relation $a_1[T_1]+\cdots + a_m[T_m]=0$. The cyclic group $\bbz_m$ with generator $\tau$ acts via $\tau [T_i] = [T_{i+1}]$ on this homology. This means that for all $i,j>0$, $a_1[T_i] + a_2[T_{i+j}]+\cdots +a_m[T_{i+j+m-1}]=0$ (indices are taken modulo $m$). This implies in turn that all of the $a_i$'s must be equal, thus our relation.
\end{proof}


\vskip 5pt\noindent{11.2. {\bf Proof of the stable splitting: Theorem 1.6.}}\label{theproof} 
 Let $(S^j)^n$ be a torus. Its suspension has the form $\Sigma (S^j)^n\simeq S^{jn+1}\vee W$, where $W$ is a wedge of spheres of dimension less than $jn-1$. In other words, the top homology class of the suspended torus is generated by a sphere (i.e. it is a spherical class).

Let $\Gamma$ be a simple graph on $m$ vertices. Every basis generator of $H_{*}(\conf{\Gamma} (\bbr^N))$ comes from the top orientation class $[{ F}]$ of an embedded torus,  $F\in\mathcal F_B$ and $B$ is some bond partititon. The suspension of \eqref{toric} is now a  map
$$\Sigma\alpha_{ F}: S^{i(N-1)+1}\vee W\longrightarrow \Sigma\conf{\Gamma} (\bbr^N)$$
 $i = \sum |T_i|$, $1\leq i\leq k-1$,
which induces a homology monomorphism, and for which the image of the orientation class
$\ds S^{i(N-1)+1}\longrightarrow \Sigma\conf{\Gamma} (\bbr^N)$
is a generator of 
$H_{i(N-1)+1}(\Sigma\conf{\Gamma}(\bbr^N))$. 
 This class is the suspended class $\Sigma [F]$. Write $\gamma_{F}$ the restriction of the map $\Sigma\alpha_{ F}$ to this top dimensional sphere. The wedge of all these maps, over all basis generators, that is over all NBC-spanning forests associated to every bond lattice, is a wedge of maps
$$\bigvee_{ F\in\mathcal F^{nbc}(\Gamma)\atop B\in \mathcal B(\Gamma)}\left(\gamma_{\mathcal F_B}: S^{(m-|B|)(N-1)+1}\longrightarrow\Sigma\conf{\Gamma} (\bbr^N)\right)$$
which induces, by construction, a homology isomorphism in all degrees. By the Whitehead theorem, since the source and target are simply connected for $N\geq 2$, this map is a homology equivalence. The wedge product on the left is precisely the wedge decomposition of Theorem \ref{main}, and the proof is complete.


\section{Application: Configurations with Obstacles}\label{obstacles}


This last section gives a useful application of Theorem \ref{main} and Corollary \ref{poinc2} to the so-called configuration spaces with obstacles. They are explicitly and geometrically defined as follows. 
Let $\{p_1,\ldots, p_r\}$ be pairwise distinct fixed points in a space $X$ which we assume to be, as before, locally compact Hausdorff and connected. These are called \textit{ obstacles}. 
As is common, we use the notation $[r] = \{1,\ldots, r\}$ and $\ds {[r]\choose k}$ all subsets of cardinality $k$ of $[r]$. Now consider the space of tuples $(x_1,\ldots, x_n)\in X^n$ such that some specific entries must be pairwise distinct, and must \textit{avoid} some other specified obstacles among the $p_j$'s.  We formulate those conditions as (*) below:
$$
\begin{array}{cl}
\begin{array}{c}
(*)
\end{array}
&\left\{
 \begin{array}{l}
x_i\neq x_j, \quad\text{ if  $\{i,j\}\in C$}\ \ ,\ \ \hbox{$``C''$ stands for constraints}\cr
 x_k\neq p_s, \quad\text{ if  $(k,s)\in O$\ \ , \ \ \hbox{$``O''$ stands for obstacles}}
 \end{array}
\right.
\end{array}
$$
Here $C$ is some subset of ${[n]\choose 2}$ and $O$ some subset of $[n]\times [r]$.
The space of all such configurations in $X^n$ is called a configuration space with obstacles and is written as $\conf{C,O}(X)$. Schematically, this is \textit{the space of $n$ moving objects in $X$, which may or may not collide, each moving to avoid a prescribed subset of $r$ given obstacles}. See Figure \ref{obstacles}.
\begin{figure}[htb]
    \centering
\includegraphics[scale=0.6]{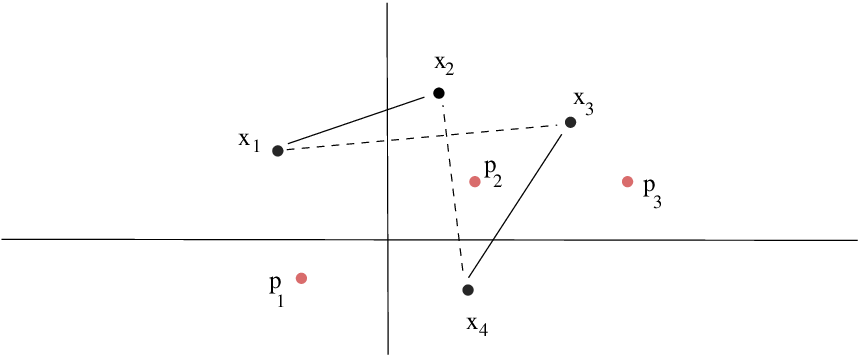}
\caption{A space of obstacles: elements are configurations $(x_1,x_2,x_3,x_4)$ such that all $x_i$'s must avoid the obstacles labeled $\{p_1,p_2,p_3\}$, and such that (in this case) $x_1\neq x_2$, $x_3\neq x_4$ (non-coincidence indicated by solid lines), while $x_2,x_4$, and $x_1,x_3$ can collide (broken lines).}\label{collisionfig}
\end{figure}

To better describe this configuration space, we build the following graph from the data (*): let $\G (n,r)$ be the graph with $V(\G (n,r))=[n+r]$ and edges
\begin{eqnarray*}
    &&E(\G (n,r))=\{\{i,j\}, \{k,s\}, \{n+a,n+b\}\},\\ &&\hbox{where}\ \{i,j\}\in C, (k,s)\in O\ \hbox{and}\ \forall a\neq b\in [r]\}
\end{eqnarray*}
This graph has $K_r$ as a subgraph whose vertices are associated with the obstacles. In $V(\G (n,r))$, the vertices $\{1,\ldots, n\}$ should be viewed as indices for the moving objects $x_i$, and   $\{n+1,\ldots, n+r\}$ are viewed as indices for the obstacles $p_j$. This section is to fully describe the homology of this space. This is given as follows.

\begin{theorem}\label{main2} With $C,O$ and $\G (n,r)$ as above, $N\geq 2$, the configuration space $\conf{C,O}(\bbr^N)$, describing $n\geq 1$ moving objects with $r\geq 1$ obstacles, has torsion free homology and Poincar\'e polynomial
$$P_t(\conf{C,O}(\bbr^N))=
{P_t(\conf{\Gamma(n,r)}(\bbr^N))\over P_t(\conf{K_r}(\bbr^N))}
$$
\end{theorem}

\bex The case of the join of the graphs
$\G (n,r) = K_n*K_r = K_{n+r}$ corresponds to the configuration space $\conf{C,O}(\bbr^N) = \conf{n}(\bbr^N-Q_r)$, where $Q_r$ is a fixed set of $r$ points (see \S\ref{obstacles}). This is the set of all configurations that avoid all the obstacles. The case of the cartesian (or box) product $\Gamma (n,r) = K_2\Box K_r$ is also of particular interest and is discussed as a final example in this section. It describes the set of all configurations $(x_1,\ldots, x_n)$ such that the $x_i$'s are pairwise distinct, and $x_i\neq p_i$, $\forall i$ (note $x_i$ and $p_j$ can coincide if $i\neq j$).
\eex

\vskip 5pt
\noindent{12.1 {\bf Bundle Maps.}}
The starting point is to exhibit a necessary and sufficient condition for projection maps between chromatic configuration spaces to be bundle maps. This gives a sharper version of the well-known fibration theorem of Fadell and Neuwirth \cite{fn}. 

Let $H\subset\Gamma$ be a subgraph of $\Gamma$. There is always an induced map (see Remark \ref{mappi})
\begin{equation}\label{projection}
\begin{array}{cccl}
         \pi_\Gamma^H : & \conf{\Gamma} (M) & \longrightarrow & \conf{H}(M) 
    \end{array}
\end{equation}
This map is ``inclusion followed by projection''. The point we want to make is that it is not generally a bundle projection as the following simple example shows: consider the $L_3$ graph with the middle vertex being labeled $2$, and project onto the trivial subgraph $t_2$ on $\{1,3\}$, i.e. $\pi: \conf{L_3}(X)\rightarrow \conf{t_2}(X) = X^2$ (here we abbreviate notation $\pi = \pi_{L_3}^{t_2}$). Over a point $(x,y)$, $x\neq y$, $\pi^{-1}(x,y) = \{z\in\bbr^N\ |\ z\neq x, z\neq y\}\simeq S^{n-1}\vee S^{n-1}$. Over $(x,x)$,
$\pi^{-1}(x,x) = \{z\in\bbr^N, z\neq x\}\simeq S^{N-1}$. Since the homotopy type of inverse images is not constant, $\pi$ is not a bundle map. The precise issue is this: if a vertex $v$ is connected to two vertices of $H$, and there is no edge between these two vertices, then the projection fails to be a bundle map. This turns out to be the only issue.

\bde A subgraph $H\subset \Gamma$ is a \textit{relatively complete subgraph} of $\Gamma$ if
$\forall i\in V(\Gamma) - V(H)$, if $\{i,r\}, \{i,s\}\in E(\Gamma)$, $r,s\in V(H)$, 
    then $\{r,s\}\in E(H)$. In simple words, $H$ is relatively complete if whenever a vertex of $\Gamma$ is adjacent to two vertices in $H$, then these two vertices must be adjacent in $H$.
\ede

\begin{proposition}\label{fadellneuwirth}
Let $M$ be a connected topological manifold and  $\Gamma$ a simple graph. If $H$ is a relatively complete subgraph of $\Gamma$, then the projection map \eqref{projection} is a bundle map. If $M=\bbr^N$, this projection has a section.
\end{proposition} 

\begin{proof}
The proof that this projection is locally trivial follows the same lines as \cite{fn}, and we refer to the Appendix of \cite{eastwood} for the details. The existence of a section in the case $M=\bbr^N$ follows immediately from the fact that we can continuously adjoin to a given $(x_1,\ldots, x_m)\in \conf{H}(\bbr^N), m=|V(H)|, n=|V(\Gamma)|$, $n-m$ distinct points lying outside the sphere of radius $1+\sum |x_i|$.
The existence of a section can be proven for more general $M$ of course.
\end{proof}

\bco\label{bundletocomplete} If $H$ is a complete subgraph of $\Gamma$, then $\pi_\Gamma^H$ is a bundle projection.
\eco

\vskip 5pt
\noindent{12.2 {\bf Proof of Theorem \ref{main2}.}}
We can now prove the main result of this section.
Write $K=K_r$, the complete graph. By Corollary \ref{bundletocomplete}, the projection $\pi: \conf{\Gamma (n,r)}(\bbr^N)\rightarrow\conf{K}(\bbr^N)$ is a bundle projection with fiber Fib$(\pi)$, where $$\hbox{Fib}(\pi) := \pi^{-1}(p_1,\ldots, p_r)$$ 
By direct inspection, and also by design, we have that
\begin{equation}\label{identify}
\conf{C,O}(\bbr^N) = Fib(\pi)
\end{equation}
This fiber is the complement of a $c$-arrangement in $(\bbr^N)^{|V(\Gamma)|}$. The associated intersection poset is geometric (Proposition \ref{geometric}), therefore shellable.  As in the proof of \S6,
the homology of Fib$(\pi)$ must be concentrated in degrees a multiple of $N-1$.
When $N\geq 3$, the base $\conf{K}(\bbr^N)$ is simply connected and its homology is concentrated in degrees that are a multiple of $N-1$ as well. Consequently, the Serre spectral sequence for the projection $\pi$ must collapse at the $E^2$-term in the simply connected case, i.e. when $N>2$, since there can be no differentials. In that case
$$H_*(\conf{\G (n,r)} (\bbr^N))\cong H_*(Fib(\pi))\otimes H_*(\conf{K}(\bbr^N))$$
The Poincar\'e series multiply so using the identification \eqref{identify}, we obtain the desired claim. We use the same argument for $N= 2$, only that this time the base is not simply connected anymore, and we need argue that the bundle has trivial coefficients. The graph $\G (n, r)$ is a subgraph of $K_{n+r}$, and there is a map of bundles
$$\xymatrix{\conf{n}(\bbr^2-Q_r)\ar[r]\ar[d]&\conf{C,O}(\bbr^2)\ar[d]\\
\conf{K_{n+r}}(\bbr^2)\ar[d]\ar[r]&\conf{\G (n,r)}(\bbr^2)\ar[d]\\
\conf{r}(\bbr^2)\ar[r]^=&\conf{r}(\bbr^2)
}$$
We gather three facts: (i) the middle map is a surjection in homology by Corollary \ref{dualityGM}, (ii) the inclusion of the fiber $\conf{C,O}(\bbr^N)\hookrightarrow \conf{\G (n,r)}(\bbr^N)$ has a retract, and (iii) the left-hand spectral sequence has trivial coefficients (\cite{fredbible}, \S6). Since the homology of the fibers survives to $E_\infty$, the map between their homologies must be surjective as well. The fundamental group $\pi_1(\conf{r}(\bbr^N))$ acts trivially on $H_*(\conf{n}(\bbr^2-Q_r))$, so forcibly must act trivially on any homology class in $\conf{C,O}(\bbr^N)$. The right-hand fibration has trivial coefficients as well, and since the Serre spectral sequence for the lefthand fibration collapses, it must also collapse for the fibration on the right at the indicated $E^2$-term. 
\hfill$\square$

The following corollary is known and recorded as Theorem 7.1 in \cite{vershinin}. 

\bco\label{config1}  Let $Q_r$ be a set of $r$ distinct points in $\bbr^N$. Then
\begin{eqnarray*}
    P_t(\conf{n}(\bbr^N\setminus Q_r)) 
    &=&
(1 + rt^{N-1})(1 + (r + 1)t^{N-1})\cdots (1 + (n + r - 1)t^{N-1})
\end{eqnarray*}
\eco

\begin{proof}
The space  $\conf{n}(\bbr^n\setminus Q_r)$
is a configuration space with obstacles corresponding to 
$$C=\{\{i,j\}, i\neq j, 1\leq i,j\leq n\}\ \ ,\ \ 
O = \{(i,s), i\neq s, 1\leq i\leq n, 1\leq s\leq r\}
$$
In words, these are all moving objects $x_i, 1\leq i\leq n$ which are pairwise distinct, and each $x_i$ avoids all of the $p_s\in Q_r, 1\leq s\leq r$. 
The graph $\Gamma(n,r)$ has all possible edges, and is by construction $K_n*K_r=K_{n+r}$, the join of two graphs. 
According to Theorem \ref{main2}
$$\ds P_t(\conf{n}(\bbr^N\setminus Q_r)) = {P_t(\conf{K_{n+r}}(\bbr^N))\over P_t(\conf{K_r}(\bbr^N))}$$
We now plug in the Poincar\'e series for 
$\conf{K}(\bbr^N)$ (see \S4.2) to obtain the desired claim.
\end{proof}

\vskip 5pt\noindent{10.1 {\bf  Main example.}} 
Let $\zeta = (p_1,\ldots, p_n)\in\conf{n}(X)$ and consider the following configuration space of points
\begin{eqnarray*}
\conf{n,n} (X) &:=&
\conf{n}(X)\cap (\bbr^N\setminus\{p_1\})\times\cdots\times (\bbr^N\setminus \{p_n\})\\
&=&\{(x_1,\ldots, x_n)\ |\ x_i\neq x_j, i\neq j\ \hbox{and}\ x_i\neq p_i, \ \forall\ 1\leq i\leq n\}
\end{eqnarray*}
This is the configuration space of points $(x_1,\ldots, x_n)$ (pairwise distinct) such that each $x_i$ avoids the obstacle $p_i$. In our description of configuration spaces with obstacles,
$\conf{n,n}(X)=\conf{C,O}(X)$ where
$$C=\{\{i,j\}, i\neq j, 1\leq i,j\leq n\}\ \ ,\ \ 
O = \{(i,i), 1\leq i\leq n\}
$$
The associated graph
$\Gamma (n,n)$ is the box product $ K_2\Box  K_n$ (or cartesian product) of the complete graph $K_n$ with $K_2$. This is the graph whose vertices $V(\Gamma ) = V(K_n)\times V(K_2)$ (i.e. $2n$ vertices), and edges between any two such vertices $(v_1,w_1)$ and $(v_2,w_2)$ if $w_1=w_2$ ($v_1\neq v_2$) or $v_1=v_2$ ($w_1\neq w_2$). This is like a ``doubling'' operation (see Fig. \ref{collisionfig} for $n=3$). Theorem \ref{main2} gives  
$$P_t(\conf{n,n} (\bbr^N))=
{P_t(\conf{K_2\Box K_n}(\bbr^N))\over P_t(\conf{K_n}(\bbr^N))}
$$
There is no general formula for the chromatic polynomial of the box products of two graphs, even when they are both complete. 

\bex The case $n=2$. Let $\zeta\in\conf{2}(\bbr^N)$, then $K_2\Box K_2=C_4$ is the square (cyclic) graph and 
\begin{eqnarray*}
P_t(\conf{2,2} (\bbr^N))=
{P_t(\conf{C_4}(\bbr^N))\over P_t(\conf{K_2}(\bbr^N))} &=& 
{3t^{3(N-1)} + 6t^{2(N-1)} + 4t^{N-1} + 1\over 1+t^{N-1}}\ \ \ \ \hbox{(by formula \eqref{cyclic2})}
\\
&=&3t^{2(N-1)} + 3t^{N-1}+1
\end{eqnarray*}
When $n=3$, a depiction of the box product graph is shown in Fig. \ref{K3boxK2}
\begin{figure}[htb]
    \centering
\includegraphics[scale=0.6]{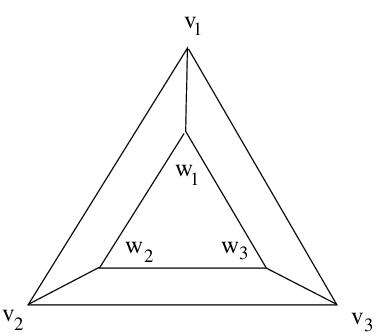}
\caption{The graph $K_2\Box K_3$.}\label{K3boxK2}
\end{figure}
The chromatic polynomial $\chi_{K_n\Box K_2}(\lambda)$ is computed in \cite{pfaff} (Theorem 2). When specialized to the case $n=3$, it gives
\begin{eqnarray*}
 \chi_{K_3\Box K_2}(\lambda ) &=& \lambda (\lambda - 1)^5-4\lambda (\lambda - 1)^4 + 8\lambda (\lambda - 1)^3-9\lambda (\lambda - 1)^2+4\lambda (\lambda - 1)\\
 &=&\lambda^6-9\lambda^5+34\lambda^4-67\lambda^3+67\lambda^2-26\lambda
\end{eqnarray*}
By theorem \ref{main}, with $m=|V(K_3\Box K_2)|=6$,
\begin{eqnarray*}
P_t(\conf{K_3\Box K_2}(\bbr^N)) &=&
t^{6(N-1)}\chi_\Gamma\left(-t^{(1-N)}\right)\\
&=& 
1 + 9t^{(N-1)} + 34t^{2(N-1)}+67t^{3(N-1)}+67t^{4(N-1)} + 26t^{5(N-1)}
\end{eqnarray*}
Note that the coefficient of $\lambda$ in the chromatic polynomial is  $-26$ and so $a_1(K_3\Box K_2)=26$, which is the top betti number $\beta_{5(N-1)}$ of $\conf{K_3\Box K_2}(\bbr^N)$ indeed.
This polynomial is divisible by 
$P_t(\conf{3}(\bbr^N)) = 1+ 3x+2x^2$,
where $x=t^{N-1}$. Using Wolfram-Alpha, we obtain the exact quotient and this is summarized below.
\eex 
\bco Set $x = t^{N-1}$, $N\geq 2$. Then
$$P_t(\conf{K_2\Box K_3}(\bbr^N)) = 1 + 9x + 34x^2+67x^3+67x^4 + 26x^5$$
and
$$P_t(\conf{3,3}(\bbr^N) = 1 + 6x + 14x^2+13x^3$$ 
\eco


\section{Appendix: the Wedge Lemma}

This is a compelling and very useful result in the theory of subspace complements. It is due to \cite{zz} (Theorem 2.2). We discuss it here and explain how it is used to derive diagram \eqref{wedgelemma}. Let $\mathcal A = \{A_1,\ldots, A_n\}$ be a finite affine subspace arrangement in $\bbr^N$ with intersection poset $L(\mathcal A)$ (including $\hat 0=\bbr^N$). It is easy to see that the link $V_{\mathcal A}:=\bigcup_{i\in I}A_i$ is contractible if the arrangement is central (i.e. $\bigcup A_i\neq\emptyset$), and otherwise it is 
$\Delta (P\setminus \{\hat 0\}))$ (note $\hat 1=\emptyset$ is not in this poset) (See \cite{zz}, Theorem 2.1).

One can now compactify all of the affine subspaces in $S^N=\bbr^N\cup\infty$ to obtain a spherical arrangement $\hat{\mathcal A}$ with link $\hat V_{\mathcal A}$. Consider the augmented poset $L_\infty$ which is $L$ adjointed a maximal element $\infty$\footnote{To avoid confusion, the notation $\hat 1=\cap A_i$ is reserved for the intersection of all the flats. In \cite{zz}, $\infty=\hat 1$}
We write $X*Y$ the join of $X$ and $Y$ (the identification space given by the union of all disjoint segments with one end in $X$ and the other in $Y$). The wedge lemma states that there is a homotopy splitting
$$\hat V_{\mathcal A}\simeq
\bigvee_{\beta\in L_\infty(\mathcal A)\setminus \{\hat 0\}} \Delta (\hat 0,\beta)* S^{d(\beta)} $$ 
where the sphere $S^{d(\beta)}$ is the compactification of $A_\beta$ having dimension $d(\beta)$. In particular, if $A_\beta$ is a point, then $S^{d(\beta)}=S^0$. It is worthwhile looking at an example.

\bex Let $\mathcal A = \{A_1,A_2\}$ where $A_1$ is the $x$-axis and $A_2$ the $y$-axis in $\bbr^N=\bbr^2$. They intersect at a point. The link in the sphere $\hat V_{\mathcal A} = \hat A_1\cup\hat A_2$ is the union of two circles with two diameter points in common. This has the homotopy type of $S^1\vee S^1\vee S^1$. Let's verify with the formula. The extended poset $L_\infty$ is depicted below. Note that $\Delta (\hat 0,A_1)=\emptyset = \Delta (\hat 0,A_2)$, while $\Delta (\hat 0,A_{12})=S^0$. We get the two copies of $S^1$ from $\emptyset*S^{\dim (A_i)}$ and the third $S^1$ from $\Delta (\hat 0,A_{12})*S^{\dim (A_{12})}= S^0*S^0 = S^1$.
\begin{figure}[htb]
    \centering
    \includegraphics[scale=0.8]{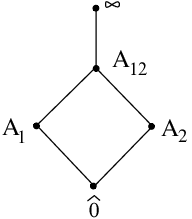}
    \caption{$P=L_\infty (\mathcal A)$ for the intersection poset of the subspace arrangement consisting of the two coordinate axes in $\bbr^2$}
\end{figure}
\eex

To derive the wedge lemma, we need to understand a key idea (and later apply it in our case). The link $\hat V_\mathcal A$ is the colimit of the diagram $D_{\mathcal A}$ made out of spaces $\hat A_p$ and maps $d_{pq}: \hat A_p\rightarrow \hat A_q$, with $p>q$. These maps are inclusions of spheres $S^{d(p)}$ into higher dimensional spheres $S^{d(q)}$. They are homotopically trivial maps. By the homotopy lemma (\cite{zz}, Lemma 1.7), it is possible to replace this diagram by keeping the same spaces but changing the maps, up to homotopy, to constant maps $d'_{pq}: S^
{d(p)}\rightarrow S^{d(q)}$. For this we need to have maps of diagrams as in
$$\xymatrix{\hat A_p\ar[d]^{d_{pq}}\ar[r]^{\alpha_p}&S^{d(p)}\ar[d]^{d'_{pq}}\\
\hat A_q\ar[r]^{\alpha_q}&S^{d(q)}}\ $$
with $p>q$, $\alpha_p$ and $\alpha_q$ homotopy equivalences. To do this, one chooses for every $p<\hat 1$ a point $c_p$ in $A_p\setminus \bigcup_{q>p}A_q$ and considers the map $\alpha_p$ which contracts the complement of a small disc around $c_p\in A_p$, avoiding $\bigcup_{q>p}A_q$, to $\infty$ (common to all spheres). The diagram of spheres $S^{d(p)}, p\in L_{\mathcal A}$ and constant maps $d'_{pq}$ is denoted by $C_{\mathcal A}$.

Now consider a subspace arrangement $\mathcal B\subset\mathcal A$, with intersection lattice $L_{\mathcal B}$. For a flat $A_p\in\mathcal A$, we write $A_p=B_p$ when it is also in $\mathcal B$, and $D_{\mathcal B}$ the smaller diagram of spaces $B_p\hookrightarrow B_  q$, $p>q$ in $L_{\mathcal B}$.
There is an inclusion $\hat V_{\mathcal B}\hookrightarrow\hat V_{\mathcal A}$, and the following diagrams commute
$$\xymatrix{
\hat V_{\mathcal B}\ar[r]&\hat V_{\mathcal A}\\
\hbox{hocolim} D_{\mathcal B}\ar[d]^\simeq\ar[u]_\simeq\ar[r]&\hbox{hocolim}D_{\mathcal A}\ar[d]^\simeq\ar[u]_\simeq\\
\hbox{hocolim} C_{\mathcal B}\ar[d]^\simeq\ar[r]&\hbox{hocolim}C_{\mathcal A}\ar[d]^\simeq\\
\bigvee_{\beta\in L_\infty(\mathcal B)\setminus \{\hat 0\}} \Delta (\hat 0,\beta)* S^{d(\beta)}\ar[r]&
\bigvee_{\beta\in L_\infty(\mathcal A)\setminus \{\hat 0\}} \Delta (\hat 0,\beta)* S^{d(\beta)} 
}$$
The top commuting diagram exists because we have a morphism of diagrams $D_{\mathcal B}\rightarrow D_{\mathcal A}$, and both spaces $\hat V_{\mathcal B}$ and $\hat V_{\mathcal A}$ are colimits of diagram homotopy equivalent to their hocolims. The second commuting diagram exists because the maps
$\beta_p:
\xymatrix{\hat B_p\ar[r]&\hat A_p\ar[r]&S^{d(p)}}$ are homotopy equivalences, $B_p=A_p$ $\alpha_p=\beta_p$, and there is an analog commuting diagram at the level of $D$-diagrams (before passing to hocolim). The middle vertical arrows are homotopy equivalences by the homotopy Lemma (\cite{zz}, Lemma 1.7). Finally we explain the bottom commuting square. The vertical maps are homotopy equivalences and going from the hocolim to the wedge is done via collapsing contractible spaces of the form $\Delta (P_{\leq p})$, for $p\neq\infty$, in the hocolim (see  page 22, \cite{zz}). The contractions for $P=L_\infty(\mathcal A)$ restrict to contractions for $P'=L_\infty (\mathcal B)$, which is the reason for the homotopy commutativity of the diagram.

\clearpage

\clearpage
.\vskip 200pt
\centerline{\bf\Huge Part III: References}
\clearpage

\part{References}

\bibliographystyle{plain[8pt]}

\begin{thebibliography}{doC}

\bibitem{andrade} R. Andrade,
\textit{From manifolds to invariants of En-algebras}, PhD thesis, Massachusetts
Institute of Technology, 2010.
\bibitem{abad} C. A. Abad, 
\textit{Introduction to representations of
braid groups}, 
Revista Colombiana Mat.
{\bf 49}:1 (2015), 1--38.
\bibitem{abrams} A. Abrams, \textit{Configuration spaces of colored graphs,} Geometriae Dedicata {\bf 92} (2002),
pp. 185--194.
\bibitem{abramsghrist}
A. Abrams, R. Ghrist, 
\textit{Finding topology in a factory: configuration spaces}, American Mathematical Monthly {\bf 109} (2002), Issue 2, 140--150.
\bibitem{alpert} H. Alpert, M. Kahle, R. MacPherson,
\textit{Configuration spaces of disks in an infinite strip}, J. Appl. Comput. Topol. {\bf 5} (2021), no.3, 357–390.
\bibitem{alshuler} D. Altschuler, L. Freidel, \textit{Vassiliev knot invariants and Chern-Simons perturbation theory to all orders}, Comm Math Phys {\bf 187} (1997), 261–287.
\bibitem{alvarez} D. Alvarez-Gavela, \textit{Universal Vassiliev invariants via integration on configuration spaces}, Talbot 2015: Applications of operads.
\bibitem{ak} M. Aouina, J.R. Klein, \textit{On the homotopy invariance of configuration spaces},
Algebraic and Geometric Topology {\bf 4}:2 (2004), 813--827. 
\bibitem{arabia} A. Arabia, \textit{Espaces de configuration généralisés.
Espaces topologiques i-acycliques.
Suites spectrales basiques}, M\'emoires de la Soci\'et\'e Math\'ematique de France {\bf 170} (2021).
\bibitem{arone} G. Arone, F. Sarcevic, \textit{The space of r-immersions of a union of discs in $\bbr^n$}. https://arxiv.org/abs/2212.09809
\bibitem{aroneturchin} G. Arone and V. Turchin, \textit{On the rational homology of high-dimensional analogues of spaces of long knots}, Geom. Topol. {\bf 18} (2014), 1261--1322.
\bibitem{artin0} E. Artin, \textit{Theorie der Z\"opfe},  Abh. Math. Semin. Univ. Hambg. {\bf 4} (1925), 47--72. 
\bibitem{artin} E. Artin, \textit{Theory of braids}, Annals of Math. {\bf 48}, No. 1, 101--126.
\bibitem{berceanuachraf} S. Ashraf, B. Berceanu, \textit{Cohomology of 3-points configuration spaces of complex projective spaces}, Advances in Geometry {\bf 14} (2014), 691--718.
\bibitem{axelsing}  S. Axelrod, I.M. Singer, \textit{Chern-Simons perturbation theory II},
J. Diff. Geom. {\bf 39} (1994), no. 1, 173--213.
\bibitem{af}  D. Ayala, J. Francis, \textit{Factorization homology of topological manifolds}, J. Topology {\bf 8}, Issue 4 (2015), 1045--1084.
\bibitem{ah} D. Ayala, R. Hepworth, \textit{Configuration spaces and $\Theta_n$}, Proceedings of the American Mathematical Society
{\bf 142}, No. 7 (2014), 2243--2254.
\bibitem{arnold} V. I. Arnol’d, \textit{The cohomology ring of the colored braid group}, Math. Notes {\bf 5} (1969), no. 2, 138--140.
\bibitem{aj} M. F. Atiyah, J. Jones, \textit{Topological aspects of Yang-Mills theory}, Commun. Math. Phy. {\bf 61} (1978), 97--118.
\bibitem{baez} 
J. C. Baez, M. S. Ody, and W. Richter, \textit{Topological aspects of spin and statistics in nonlinear sigma models}
Journal of Mathematical Physics {\bf 36} (1995) 247--257.
\bibitem{bahrab} A. Bahri, P.H. Rabinowitz, \textit{Periodic solutions of Hamiltonian systems of 3-body type}, Ann. Inst. H. Poincaré Anal. Non Linéaire {\bf 8} (1991), no. 6, 561--649.
\bibitem{bfsv} C. Balteanu, Z. Fiedorowicz, R. Schw\"anzl, and R. Vogt, \textit{Iterated monoidal categories}, Adv.
Math. {\bf 176}:2 (2003), 277--349.
\bibitem{radmila} V. Baranovsky, R. Sazdanovic, \textit{Graph homology and graph configuration spaces}, J. Homotopy Relat. Struct. {\bf 7} (2012), no. 2, 223--235.
\bibitem{bf} K. Barnett and M. Farber, \textit{Topology of Configuration Space of Two Particles on a Graph I}, Algebraic and Geometric Topology {\bf 9} (2009), Issue 1, 593--624.
\bibitem{barysh} Y. Baryshnikov, \textit{Euler Characteristics of Exotic Configuration Spaces}, Séminaire Lotharingien de Combinatoire 84B (2020) Proceedings of the 32nd Conference on Formal Power, Article 20, 12 pp.
\bibitem{bbk} Y. Baryshnikov, P. Bubenik, M. Kahle, \textit{Min-type Morse theory for configuration spaces of hard spheres}, Int. Math. Res. Not. (2014), 2577--2592. 
\bibitem{bgs} A. Beilinson, V. Ginzburg, V. Schechtman, \textit{Koszul duality}, JGP {\bf 5}:3 (1988), 317--350.
\bibitem{beidrin} A. Beilinson, V. Drinfeld, \textit{Chiral Algebras}, Colloquium Publications {\bf 51} (2004).
\bibitem{bellingeri} P. Bellingeri, \textit{On presentations of surface braid groups}, J. Algebra {\bf 274} (2004), 543--563.
\bibitem{benhamouda} W. Ben Hammouda, \textit{Homologie de l'espace des lacets des espaces de configurations de trois points dans $\bbr^n$ et $S^n$}, Topological Methods in Nonlinear  Analysis
Volume {\bf 42}, No. 1, 2013, 181--191.

\bibitem{bender} E. A. Bender and J. R. Goldman, \textit{On the applications of Moebius inversion in combinatorial analysis}, Amer. Math. Monthly {\bf 82}:8 (1975), 789--803.
\bibitem{bg} M. Bendersky, S. Gitler, \textit{The cohomology of certain function spaces}, AMS Transactions 
 {\bf 326} (1), (1991), 423--440. 
 \bibitem{benson} B. Benson, D. Chakrabarty, P. Tetali, \textit{G-parking functions, acyclic orientations and spanning trees},
Discrete Mathematics {\bf 310} (2010) 1340--1353.
\bibitem{bermarpap} B. Berceanu, M. Markl, S. Papadima, \textit{Multiplicative models for configuration spaces of algebraic varieties}, Topology {\bf 44}, Issue 2, (2005), 415--440.
\bibitem{bmp} B. Berceanu, D. A. Macinic, S. Papadima and C.D. Popescu, \textit{On the geometry and topology of partial configuration spaces of Riemann surfaces}, Algebraic \& Geometric Topology {\bf 17} (2017), 1163--1188.
\bibitem{berger} C. Berger, \textit{Combinatorial models for real configuration spaces
and $E_n$-operads}, Contemp. Math. {\bf 202}, Amer. Math. Soc. (1997), 
37--52.
\bibitem{berger2} C. Berger,
\textit{A cellular nerve for higher categories}, Adv. Math. {\bf 169} no.1 (2002), 118--175.
\bibitem{berglund} A. Berglund, \textit{Koszul spaces}, Trans. Amer. Math. Soc. {\bf 366}, n0.9 (2014), 4551--4569.
\bibitem{bezru} R. Bezrukavnikov, \textit{Koszul DG-algebras arising from configuration spaces}, Geometric and Functional Analysis {\bf 4}, no.2 (1994), 119--135.
\bibitem{bianchi} A. Bianchi, \textit{Moduli spaces of Riemann surfaces as Hurwitz spaces}, Adv. Math. {\bf 430} (2023) 62 pages, Article ID 109217.
\bibitem{bibby} C. Bibby and N. Gadish,
\textit{Combinatorics of orbit configuration spaces}, S\'em. Lothar. Combin. 80B (2018), Art. 72.
\bibitem{bl} J. S. Birman, X-S. Lin, \textit{Knot polynomials and Vassiliev’s invariants}, Invent. Math. {\bf 111}(2) (1993), 225--270.
\bibitem{bjorner} A. Bj\"orner, \textit{Homology and shellability of matroids and geometric complexes},  A chapter in ``Matroid Applications'', ed. N. White (1990).

\bibitem{bjorner2} A. Bj\"orner, \textit{On the homology of geometric lattices}, Algebra Universalis {\bf 14} (1982), 107--128.

\bibitem{bjowel} A. Bj\"rner, V. Welker, \textit{The homology of ”k-equal” manifolds and related partition lattices}, Adv. Math.
{\bf 110} (1995), no. 2, 277--313.
\bibitem{blagozieg} P. V. M. Blagojević, G. M. Ziegler, \textit{Convex equipartitions via equivariant obstruction theory}, Israel Journal of Math. {\bf 200} (2014), 49--77.
\bibitem{blz} P. V. M. Blagojević, W. L\"uck, G. M. Ziegler, \textit{Equivariant topology of configuration
spaces}, J. Topology {\bf 8}:2 (2015), 414--456.
\bibitem{blago2} P. Blagojevi\'c, F. R. Cohen, M. C. Crabb, W. L\"uck, G. M. Ziegler,
\textit{Equivariant cohomology of configuration spaces mod 2}, Lecture Notes in Mathematics {\bf 2282}, Springer, Cham, (2021).
\bibitem{bbs} F.J. Bloore, I. Bratley, J.M. Selig, \textit{$SU(n)$ bundle over the configuration space of three identical particles moving on $\bbr^3$}, J. Phys. A. Math. Gen. {\bf 16} (1983), 729--736.
\bibitem{boavida} P. Boavida, M. Weiss,
\textit{Spaces of smooth embeddings and configuration categories},
J. Topol. {\bf 11} (2018), no. 1, 65--143.
\bibitem{bvogt} J. M. Boardman, R.M. Vogt, \textit{Homotopy invariant algebraic structures on topological spaces},
Springer Lecture Notes in Math. {\bf 347} (1973).
\bibitem{bm} M. B\"okstedt, E. Minuz, \textit{Cohomology of generalised configuration spaces of points on $\bbr^r$}. ArXiv:2004.08370.

\bibitem{bodig} C.F. B\"odigheimer, \textit{Stable splittings of mapping spaces},
Springer Lecture Notes in Math. {\bf 1286},  (1987) 174-187.
\bibitem{bodig2} C.F. Bodigheimer, \textit{Gefarbte Konfigurationen: Modelle fur die stabile homotopie
von Eilenberg-MacLane-Raumen}, PhD thesis, University of Heidelberg, 1984.
\bibitem{bc} C.F. Bodigheimer, F. Cohen, \textit{Rational cohomology of configuration spaces}, in: Algebraic
Topology and Transformation Groups, Springer Lectures Notes in Math. {\bf 1361}, 
(1988), 7--13.
\bibitem{bcm} C.F. Bodigheimer, F.R. Cohen, R.J. Milgram,
\textit{Truncated symmetric products and configuration spaces},
Math. Zeit., {\bf 214} (1993), 179--216.
\bibitem{bct} C. F. Bodigheimer, F. Cohen, and L. Taylor, \textit{On the homology of configuration spaces}, Topology {\bf 28} (1) (1989), 111--123.
\bibitem{madsenbodig}  C.F. Bodigheimer, I. Madsen, \textit{Homotopy quotients of mapping spaces and their stable splittings}, Quarterly J. Math. {\bf 39}:4 (1988), 401--409.
\bibitem{bohnenblust} F. Bohnenblust, \textit{The algebraical braid group}, Ann. Math. {\bf 48} (1947), 127--136.
\bibitem{bt} R. Bott, Taubes, \textit{On the self-linking of knots}, J. Math. Phys. {\bf 35}:10 (1994), 5247--5287.
\bibitem{bott} R. Bott, \textit{Configuration spaces and imbedding invariants}, Turkish J. Math.
{\bf 20}, no. 1  (1996), 1--17.
\bibitem{bhmm} C.P. Boyer, J.C. Hurtubise, B.M. Mann, R.J. Milgram, \textit{
The topology of instanton moduli spaces,
I: The Atiyah-Jones conjecture}, Annals of Math. {\bf 137} (1993), 561--609.
\bibitem{ryan} R.D. Budney, \textit{Representations of mapping class groups via topological constructions}, Ph.D. Dissertation,  Cornell University (2002).
\bibitem{browder} W. Browder. \textit{Homology operations and loop spaces}, Illinois J.
Math. {\bf 4} (1960), 347--357.
\bibitem{bw} R.F. Brown, J.H. White, \textit{Homology and Morse theory of third configuration space}, Indiana University Mathematics J. {\bf 30}, no. 4 (1981), 501--512.
Algebr. Geom. Topol. {\bf 19} (1) (2019), 1--30.
\bibitem{calimici} G. Calimici, \textit{The configuration space of the three dimensional Lens space $L(7,2)$ and its model}, Graduate J. Math {\bf 2}, Issue 1 (2017), 29--36.
\bibitem{callegaro} F. Callegaro, \textit{Salvetti complex, spectral sequences and cohomology of Artin groups}, Annales de la facult\'e des sciences de Toulouse,
Tome XXIII, {\bf 2} (2014), 267--296.
\bibitem{filippo}  F. Callegaro and M. Salvetti, \textit{Homology of the family of hyperelliptic curves}, Isra\"el J. of Math.
{\bf 230} (2019), 653--692.
\bibitem{camposwill} 
R. Campos and T. Willwacher, \textit{A model for configuration spaces of points},
Algebraic and Geometric Topology {\bf 23} (2023), 2029--2106.
\bibitem{campos2} 
R. Campos, J. Ducoulombier, N. Idrissi, T. Willwacher, \textit{A model for framed configuration spaces of points}, https://arxiv.org/pdf/1807.08319.pdf
\bibitem{martin}  F Cantero, M Palmer, \textit{On homological stability for configuration spaces on closed background manifolds}, Documenta Mathematica {\bf 20}, 753--805.
\bibitem{mccleary} J. Cantarella, E. Denne, J. McCleary, \textit{Configuration spaces, multijet transversality, and the square-peg problem}, Illinois J. Math (2021).
\bibitem{carlsson} G. Carlsson, J. Gorham, M. Kahle, J. Mason, \textit{Computational topology for configuration spaces of hard disks}, Physical Review E 85, 011303 (2012).
\bibitem{cattaneo}  A.S. Cattaneo, P. Cotta-Ramusino, R. Longoni, \textit{Configuration spaces and Vassiliev classes in any dimension}, Geometry \& Topology {\bf 6} (2002), 949--1000.
\bibitem{church} T. Church, \textit{Homological stability for configuration spaces of manifolds},
{\bf 188} (2012), 465--504.
\bibitem{churchfarb} T. Church, B. Farb, \textit{Representation theory and homological stability}, Advances in Math. {\bf 245} (2013), 250--314.
\bibitem{clarke} B. Clarke, \textit{A note on Alexander’s duality}, Mathematika {\bf 3} (1956), 41--46.
\bibitem{cnossen} B. Cnossen, \textit{Configuration spaces as partial abelian monoids}, Master's thesis, Friedrich-Wilhelms-U. Bonn (2019).
\bibitem{simplicial}
A. A. Cooper, Vin de Silva, R. Sazdanovic,
\textit{On configuration spaces and simplicial complexes}, New York J. Math. {\bf 25} (2019) 723–744.

\bibitem{fredbible} F.R. Cohen, \textit{ The homology of $C_{n+1}$-spaces, $n\geq 0$}; in F. Cohen, T. Lada, P. May: The homology of iterated loop spaces; Springer Lecture Notes Math {\bf 533} (1976), 207--351.
\bibitem{fred2} F.R. Cohen, \textit{Introduction to configuration spaces and their applications}, Lecture Notes Series, Institute for Mathematical Sciences, National University of Singapore, 183--261.
\bibitem{fredanthology} F. R. Cohen, \textit{An anthology of configuration spaces: I and II}, (2012).
\bibitem{fred5} F.R. Cohen, \textit{On configuration spaces, their homology,
and Lie algebras}, Journal of Pure and Applied Algebra {\bf 100} (1995) 19--42.
\bibitem{fredsplitting} F. R. Cohen, \textit{The unstable decomposition and its applications}, Math. Z. {\bf 182} (1983), 553--568.
\bibitem{fredmaytaylor} F.R. Cohen, J.P. May, L.R. Taylor, \textit{Splitting of some more spaces}, Math. Proc. Cambridge Philos. Soc. {\bf 86} (2) (1979), 227--236.
\bibitem{fredrep} F. R. Cohen and L. R. Taylor, \textit{On the representation theory associated to the cohomology of configuration spaces}, In: Algebraic topology (Oaxtepec, 1991). Vol. 146. Contemp. Math. Amer. Math. Soc., Providence (1993), 91--109.
\bibitem{c2m2} F.R. Cohen, R.L. Cohen, B.M. Mann, R.J. Milgram,
\textit{The topology of rational functions and divisors of surfaces},
Acta Math. {\bf 166} (1991), 163--221
\bibitem{cg} F. R. Cohen and S. Gitler, \textit{On loop spaces of configuration spaces}, Trans. AMS {\bf 354} (2002), 1705--1748.
\bibitem{cohtay} F.R. Cohen, L. Taylor, \textit{Computations of Gelfand-Fuks cohomology, the cohomology of function spaces, and the cohomology of configuration spaces},  in Geometric Applications of Homotopy Theory I, Lecture Notes in Math. {\bf 657} (1978), 106-143.
\bibitem{fredlusk} F.R. Cohen, E.L. Lusk, \textit{Coincidence point results for Spaces with Free $\bbz_p$ actions}, Proceedings AMS {\bf 49} (1) (1975), 245--252.
\bibitem{fredwu} F.R. Cohen, J. Wu, \textit{On braid groups and homotopy groups}, Geometry and Topology Monographs {\bf 13} (2008), 169--193.
\bibitem{ralph} R.L. Cohen, \textit{Stable proofs of stable splittings}, Math. Proc. Camb. Phil. Soc. {\bf 88} (1980), 149--151.
\bibitem{copeland} A.H. Copeland, \textit{Deleted products with prescribed homotopy types}, Proceedings AMS {\bf 19}:5 (1968), 1109--1114.
\bibitem{cordova} H. Cordova Bulens, \textit{Rational model of the configuration space of two points in a simply connected closed manifold}, 
Proceedings AMS {\bf 143}, No. 12 (2015), 5437--5453
\bibitem{gaiffi} G. D'Antonio, G. Gaiffi, \textit{Symmetric group actions on the cohomology of configurations in $\Bbb R^d$},
Atti Accad. Naz. Lincei Rend. Lincei Mat. Appl. 21 (2010), no. 3, 235--250.
\bibitem{delucchi}  E. Delucchi, N. Girard, G. Paolini, \textit{Shellability of posets of labeled partitions and arrangements defined by root systems}, Electronic J. Combinatorics {\bf 26}, Issue 4 (2019).
\bibitem{deshpande} P. Deshpande, \textit{The Goresky-MacPherson formula for toric arrangements},  	arXiv:1808.06054.
\bibitem{diacu} F. Diacu, \textit{The solution of the n-body problem},  Mathematical Intelligencer {\bf 18} (1996), 66--70.
\bibitem{dobrinskaya} N. È. Dobrinskaya 
\textit{Configuration spaces of labeled particles and finite Eilenberg-MacLane complexes}, 
 Proceedings of the Steklov Institute of Mathematics {\bf 252} (2006), 30--46.
  \bibitem{natalya} N. Dobrinskaya, \textit{Configuration spaces with labels and loop spaces on $K$-products},
Russian Math. Surveys {\bf 63}:6 (2008), 1141–1143.
 \bibitem{dobtur} N, Dobrinskaya, V Turchin, \textit{Homology of non k-overlapping discs},  Homology, Homotopy \& Appl. {\bf 17} (2015) 261--290.
\bibitem{dold} A.  Dold, \textit{Homology of symmetric products and other functors of complexes}, Annals of Math. {\bf 68} No. 1 (1958), 54--80.
\bibitem{dominguez} C. Domínguez, J. González, P. Landweber, \textit{The integral cohomology of configuration spaces of pairs of points in real projective spaces}, Forum Math. {\bf 25} (2013), 1217--1248.
\bibitem{dong} F. Dong, K.M. Kong, K.L. Teo, \textit{Chromatic polynomials and chromaticity of graphs}, World Scientific (2005).
\bibitem{dror} Dror Bar-Natan, \textit{Vassiliev and Quantum Invariants of Braids}, Geom. Topol. Monogr. {\bf 4} (2002) 143--160.
\bibitem{droz} J-M Droz, \textit{Quillen model structures on the category of graphs}, Homology, homotopy and application {\bf 14} (2) 20212, 265--284.
\bibitem{dung} N. V. Dung, \textit{Homotopy of configuration spaces}, Vietnam Journal of Mathematics {\bf 3O}:1 (2002), 97--102.
\bibitem{eastwood} M. Eastwood, S. Huggett,
\textit{Euler characteristics and chromatic polynomials},  European J. Combin. {\bf 28} (2007), no. 6, 1553--1560.

\bibitem{eisenberg} B. Eisenberg, \textit{Characterization of a tree by means of coefficients of the chromatic polynomial}, Trans. New York Acad. Sci. {\bf 34} (1972), 146--153.

\bibitem{ellenberg} J. S. Ellenberg, A. Venkatesh and C. Westerland, \textit{Homological stability for Hurwitz spaces
and the Cohen-Lenstra conjecture over
function fields}, Ann. of Math. (2) {\bf 183}:3 (2016), 729--786.
\bibitem{etw} J.S.  Ellenberg, T. Tran, C. Westerland, 
\textit{Fox-Neuwirth-Fuks cells, quantum shuffle algebras, and Malle's conjecture for function fields}, https://arxiv.org/abs/1701.04541.

\bibitem{erey} A. Erey, \textit{A broken cycle theorem for the restrained chromatic function}, Turkish J. Math {\bf 43}(1) (2019), 355-360.

\bibitem{ericok} O. B. Ericok, J. K. Mason,
\textit{Quotient maps and configuration spaces of hard disks},  Granular Matter {\bf 24}, number 76 (2022) 24:76.
\bibitem{fn} E.R. Fadell, L. Neuwirth, \textit{Configuration spaces}, Mathematica Scandinavica {\bf 10} (1962), 111--118.
\bibitem{fh} E.R. Fadell, S.Y. Husseini, \textit{Geometry and topology of configuration spaces}, Springer Monographs in Math. (2001).
\bibitem{fh2} E.R. Fadell, S.Y. Husseini, \textit{Configuration spaces on punctured manifolds},  Journal of the Juliusz Schauder Center {\bf 20} (2002), 25--42.
\bibitem{falk} M. Falk, and R. Randell, \textit{The lower central series of a fiber-type arrangement}, Inventiones {\bf 82} (1985), 77--88.
\bibitem{farb} B. Farb, \textit{Representation stability}, Proc. International Congress of Mathematicians, Seoul 2014,Vol. II, 1173-1196.
\bibitem{fww} B. Farb, J. Wolfson, M. Wood, 
\textit{Coincidences between homological densities, predicted by arithmetic}, Advances in Math. {\bf 352} (2019), 670--716.
\bibitem{farber} M. Farber, S. Tabachnikov, \textit{Topology of cyclic configuration spaces and periodic trajectories of multi-dimensional billiards}, Topology {\bf 41}, Issue 3 (2002), 553--589.
\bibitem{fz} E. M. Feichtner, G. M. Ziegler, \textit{The integral cohomology algebras of ordered configuration spaces of spheres}, Documenta Math. {\bf 5} (2000), 115--139.
\bibitem{feltan} Y. F\'elix, D. Tanr\'e,  \textit{The cohomology algebra of unordered configuration spaces}, J. London Math. Soc. {\bf 72}, Issue 2 (2005), 525--544.
\bibitem{ft} Y. F\'elix, J.C. Thomas, \textit{Rational Betti numbers of configuration spaces}, Topology and its Applications {\bf 102} (2000) 139--149.
\bibitem{ft2} Y. F\'elix, J.C. Thomas, \textit{Configuration spaces and Massey products}, Int. Math. Res. Not. {\bf 2004}, no. 33, 1685--1702.
\bibitem{fox} R.H. Fox, L.P. Neuwirth, \textit{The braid groups}, Math. Scand. {\bf 10} (1962), 119--126.
\bibitem{fresse} B. Fresse, \textit{Homotopy of Operads and Grothendieck–Teichmüller Groups: Parts 1 and 2},  Mathematical Surveys and Monographs {\bf 217} (2017).
\bibitem{freedteich}
M.H. Freedman, V.S. Krushkal, P. Teichner, \textit{Van Kampen's embedding obstruction is incomplete for $2$-complexes in $\bbr^4$},
Math. Res. Lett. {\bf 1} (1994), no. 2, 167--176.
\bibitem{fm} W. Fulton, R. MacPherson, \textit{Compactification of configuration spaces}, Annals of Math. {\bf 139} (1994), 183--225.
\bibitem{fuks} D. B Fuks, \textit{Cohomology of the braid group mod2}, 
English translation in Funct. Anal. Appl. {\bf 4} (1970) 143--151.
\bibitem{gadishhainaut} N. Gadish, L. Hainaut, \textit{Configuration spaces on a wedge of spheres and
Hochschild–Pirashvili homology}, 
Ann. H. Lebesgue 7 (2024), 841--902.
\bibitem{gal} S. R. Gal, \textit{Euler characteristic of the configuration space
of a complex}, Colloquium Math. {\bf 89}:1 (2001), 61--67.
\bibitem{getzler} E. Getzler, \textit{Mixed Hodge structures of configuration spaces},
arXiv:alg-geom/9510018 (1995).
\bibitem{gebhard} D. D. Gebhard,
\textit{Sinks in Acyclic orientations of graphs}, Journal of Combinatorial Theory, Series B 80 (2000), 130--146.

\bibitem{getjon} E. Getzler and J. D. S. Jones, \textit{Operads, homotopy algebra and iterated integrals for
double loop spaces} (1994). arXiv:hep-th/9403055
\bibitem{ghrist1} R. Ghrist, \textit{Configuration spaces and braid groups on graphs in robotics}, Knots, braids, and mapping class groups, AMS/IP Stud. Adv. Math. {\bf 24} (2001), 29--40.
\bibitem{ginot} G. Ginot, \textit{Notes on factorization algebras, factorization
homology and applications}, Winter School in Mathematical Physics: Mathematical Aspects of Quantum Field Theory : Les Houches, France. Springer (2015), 429--552.
\bibitem{gs} C. Giusti, D. Sinha, \textit{Fox-Neuwirth cell structures and the cohomology of symmetric groups}, in Configuration Spaces (Geometry, Combinatorics and Topology), Centro di Ricerca Matematica Ermio De Giorgi {\bf 14}, Springer (2012).
\bibitem{gss} C. Giusti, P. Salvatore and D.P. Sinha, \textit{The mod 2 cohomology of symmetric groups as a hopf ring over the steenrod algebra}, J. Topology {\bf 5}1 (2012), 169-198.
\bibitem{meneses} J. González-Meneses, \textit{Basic results on braid groups}, Annales math\'ematiques Blaise Pascal {\bf 18} (2011), 15--59.
\bibitem{goreskymacpherson} 
        M. Goresky, R. MacPherson, \textit{Stratified Morse theory}, Ergebnisse der Mathematik Grenzgebiete. 3. Folge {\bf 14} (1988).
\bibitem{gorjunov} V. V. Gorjunov, \textit{Cohomology of braid groups of series C and D}, Trudy Moskov. Mat.
Obshch. {\bf 42} (1981), 234--242.
\bibitem{guada} E. Guadagnini, M. Martellini and M. Mintchev, Nucl. Phys. {\bf B330} (1990) 575.
\bibitem{gt} J. Grbic, S. Theriault, \textit{The homotopy type of the complement of a coordinate subspace arrangement}, Topology {\bf 46} (2007), 357--396.
\bibitem{gz} C. Greene, T. Zaslavsky, \textit{On the interpretation of Whitney numbers through arrangements of hyperplanes, zonotopes, non-radon partitions, and orientations of graphs}, Transactions AMS {\bf 280} (1983),  97–-126.
\bibitem{guest} 
M. Guest, \textit{The topology of the space of rational curves on a toric variety}, Acta Math.
{\bf 174} (1995), 119–145.
\bibitem{haefliger} A. Haefliger, \textit{Plongements diff\'erentiables dans le domaine stable}, Comment. Math. Helv. {\bf 37} (1962), 155--176.
\bibitem{hainaut} L. Hainaut, \textit{The Euler characteristic of configuration spaces}, Bulletin of the Belgian Math. Soc. Simon Stevin {\bf 29}:1 (2022), 87--97.
\bibitem{handel} D. Handel, \textit{An embedding theorem for real projective spaces}, Topdogy {\bf 7}, (1968), 125--130.
\bibitem{hatcher} A. Hatcher, \textit{Algebraic Topology}, Cambridge University Press, Cambridge (2002).
\bibitem{hersh} P. Hersh, V. Reiner, \textit{Representation stability for cohomology of configuration spaces in $\bbr^d$}, Int. Math. Res. Not. IMRN (2017), no. 5, 1433--1486.
\bibitem{hu} S.T. Hu, \textit{Isotopy invariants of topological spaces}, Proc. Roy. Soc. London. Ser. A {\bf 255} (1960), 314--421.
\bibitem{huang} Y. Huang, \textit{Cohomology of configuration spaces on punctured
varieties}, arXiv:2011.07153 (2020).
\bibitem{hyde} T. Hyde, \textit{Polynomial Factorization Statistics and point configurations in $\bbr^3$}, International Math. Research Notices {\bf 24} (2020), 10154–10179.
\bibitem{idrissi} N. Idrissi, \textit{The Lambrechts–Stanley model of configuration spaces},  Invent. Math. {\bf 216} (2019), 1--68.
\bibitem{idrissi2} N. Idrissi, \textit{Real homotopy of configuration spaces}, Peccot Lecture, Coll\`ege de France (2020). Springer International Publishing, {\bf 2303}, 2022, Lecture Notes in Mathematics,
\bibitem{imbo} T. D. Imbo, C.S. Imbo , E.C.G. Sudarshan, \textit{Identical particles, exotic statistics and braid groups}, Physics letters B 4 {\bf 234} (1990) no. 1,2.
\bibitem{jelic} M. Jelić, \textit{Methods of equivariant topology in two nice discrete geometry problems}, Graduate J. Math. {\bf 111} (2016), 18--27.
\bibitem{jenssen} M. Jenssen, V. Patel, G. Regts, \textit{Improved bounds for the zeros of the chromatic polynomial via Whitney's Broken Circuit Theorem}, J. Combinatorial Theory, Series B
{\bf 169} (2024), 233--252.
\bibitem{malchiodi}  A. Jevnikar, S. Kallel. A. Malchiodi, \textit{A topological join construction and the Toda system on compact surfaces of arbitrary genus}, Anal. PDE {\bf 8} no.8 (2015), 1963 - 2027.
\bibitem{joyal} A. Joyal, \textit{Quasi-categories and Kan complexes}, J. Pure Appl. Algebra {\bf 175} 
    (2002), 207--222.
\bibitem{ramon} D. Kahrobaei, R. Flores, M. Noce, M.E. Habeeb, C. Battarbee, \textit{Applications of group theory in cryptography—post-quantum group-based cryptography}, 
AMS Math. Surveys Monogr., {\bf 278} (2024) 2024, xvii+141 pp. 
\bibitem{kallel1} S. Kallel, \textit{Symmetric products, duality and homological dimension of configuration spaces}, Geometry \& Topology Monographs {\bf 13} (2008), 499--527.
\bibitem{kallel2} S. Kallel, \textit{Spaces of particles on manifolds and generalized Poincaré Dualities},  Quarterly J. Math. {\bf 52}:1 (2001), 45--70.
\bibitem{kalsai} S. Kallel, I. Saihi, \textit{Homotopy Groups of Diagonal Complements}, Algebr. Geom. Topol.{\bf 16} (2016) 2949--2980.
\bibitem{walid} S. Kallel, W. Taamallah, \textit{Combinatorial invariants of stratifiable spaces}, preprint.
\bibitem{kashi} T. Kashiwabara, \textit{On the homotopy type of configuration complexes}, AMS Contemp. Math. {\bf 146}  (1993), 159--170.
\bibitem{kha} R. Karasev, A. Hubard, B. Aronov, \textit{Convex equipartitions: the spicy chicken theorem}, Geometriae Dedicata {\bf 170}:1 (2014), 263--279.
\bibitem{kennedy} C. Kennedy, \textit{The directed forest complex of Cayley graphs}, Masters thesis, Boise state university (2020).
\bibitem{knudsen} B. Knudsen, \textit{Configuration spaces in algebraic topology}, arXiv:1803.11165.
\bibitem{knudsen2} B. Knudsen, \textit{Betti numbers and stability for configuration spaces via factorization homology}, Algebr. Geom. Topol. {\bf 17} (2017), no. 5, 3137--3187.
\bibitem{kohno1} T. Kohno, \textit{Loop spaces of configuration spaces and finite type invariants}, Invariants of knots and 3-manifolds (Kyoto, 2001), Geom. Topol. Monogr. {\bf 4} (2002), 143--160.
\bibitem{kohno2} T. Kohno, \textit{S\'erie de Poincar\'e-Koszul associée aux groupes de tresses pures}, Inventiones mathematicae {\bf 82} (1985), 57--76.
\bibitem{kontsevich} M. Kontsevich, \textit{Operads and motives in deformation quantization}, Lett. Math. Phys. {\bf 48} (1999), 35–72.
\bibitem{kontsevich2} M. Kontsevich, \textit{Feynman diagrams and low-dimensional topology}, First European Congress of Mathematics, Vol. II (Paris, 1992), Progr. Math. {\bf 120} (1994), 97--121.
\bibitem{kontsevich3} M. Kontsevich, \textit{Deformation quantization of Poisson manifold}, Lett. Math. Phys. {\bf 66} (2003) 157–216.
\bibitem{kosar} N. Kosar, \textit{Cohomology of polychromatic configuration Spaces of Euclidean Space}, https://arxiv.org/abs/1612.02773.
\bibitem{koytcheff} R. Koytcheff. \textit{A homotopy-theoretic view of Bott–Taubes integrals and knot spaces}, Algebr. Geom. Topol. {\bf 9}(3) (2009), 1467--1501.
\bibitem{koshorke} U. Koshorke, \textit{Higher order invariants for higher dimensional link maps}, Springer Lecture Notes in Math. {\bf 1172} (1984), 116--128.
\bibitem{kriz} I. Kriz, \textit{On the rational homotopy type of configuration spaces}, Ann. of Math. {\bf 139} (1994), 227--237.
\bibitem{pascal} P. Lambrechts and I. Volic, \textit{Formality of the little N-disks operad}, Memoirs AMS vol.{\bf 230}, no. 1079 (2014).
\bibitem{lamstan} P. Lambrechts, D. Stanley, \textit{A remarkable DG-module model for configuration spaces}, Algeb. Geom. Topol. {\bf 8} no.2 (2008), 1191--1222.
\bibitem{lehrer} G.I. Lehrer, L. Solomon, \textit{On the action of the symmetric group on the cohomology of the complement of its reflecting hyperplanes},
J. Algebra {\bf 104} (1986), no. 2, 410--424.
\bibitem{leinmyr} J.M. Leinaas, J. Myrheim, \textit{On the theory of identical particles}, Nuovo Cim. 37B, 1 (1977).
\bibitem{lehseg} G.I. Lehrer, G.B. Segal, \textit{Homology stability for classical regular semisimple varieties},  Math Z. {\bf 236} (2001), 251--290.
\bibitem{lescop} C. Lescop, \textit{Invariants of links and 3–manifolds from graph
configurations}, https://arxiv.org/pdf/2001.09929
\bibitem{levitt} N. Levitt, \textit{Spaces of arcs and configuration spaces of manifolds}, Topology {\bf 34} (1995), 217--230.
\bibitem{lv} J.L. Loday, B. Vallette, \textit{Algebraic operads}, Grundlehren der Math. Wiss.  {\bf 346}, Springer.
\bibitem{ml} P. Loffler, J. Milgram, \textit{The structure of deleted symmetric products}, Contemp. Math. {\bf 78} (1988), 415--424.
\bibitem{longueville} M. de Longueville, C.A. Shultz,
\textit{The cohomology rings of complements of subspace arrangements}, Math. Ann. {\bf 319} (2001), 625–64
\bibitem{ls} R. Longoni, P. Salvatore, \textit{Configuration spaces are not homotopy invariant}, Topology {\bf 44} (2005) 375--380.
\bibitem{loojenga}  
E. Looijenga, \textit{Torelli group action on the configuration space of a surface}, Journal of Topology and Analysis {\bf 15} (2023), No. 01,  215--222. 
\bibitem{lowen} H. Lowen, \textit{Fun with hard spheres}, Statistical physics and spatial statistics {\bf 554} (1999), Springer Lecture Notes in Phys. 295--331.
\bibitem{lurie} J. Lurie, \textit{Higher algebra}, September 2017 version.
\bibitem{malin} C. Malin, \textit{An elementary proof of the homotopy invariance of stabilized configuration spaces}, Proc. Amer. Math. Soc. {\bf 151}:8 (2023), 3635--3644.
\bibitem{sawiki} T. Maciazek, A. Sawicki, \textit{Non-abelian Quantum Statistics on Graphs},
Communications in Mathematical Physics {\bf 371}, Issue 3, 921--973.
\bibitem{massey1} W.S. Massey, \textit{The homotopy type of certain configuration spaces}, Bol. Sociedad Mat. Mexicana {\bf 37} (1992), 355--365.
\bibitem{massey} W.S. Massey, \textit{Homotopy classification of $3$-component links of codimension greater than $2$}, Topology and its Applications {\bf 34} (1990) 269--300.
\bibitem{may1} J.P. May, \textit{The geometry of iterated loop spaces}, Springer lectures notes in Math. {\bf 271} (1972).
\bibitem{may2} J.P. May, \textit{Infinite loop space theory}, Bulletin of the AMS 
{\bf 83}, Number 4 (1977), 456--494.
\bibitem{markl} M. Markl, \textit{ A compactification of the real configuration space as an operadic completion}, Journal of Algebra {\bf 215}:1 (1999), 185--204.
\bibitem{mcduff} D. McDuff, \textit{Configuration spaces of positive and negative particles}, Topology {\bf 14} (1975), 91--107.
\bibitem{segaldusa} 
D. McDuff, G. Segal, \textit{Homology fibrations and the group-completion theorem},  Inventiones Math. {\bf 31} (1976), 279--284.
\bibitem{cluresmith} J. McClure and J. Smith, \textit{Multivariable cochain operations and little n-cubes}, J. Amer. Math.
Soc. {\bf 16} (2003), 681--704.
\bibitem{mcmullen} C.T. McMullen, \textit{Braid groups and Hodge theory}, Math. Ann. {\bf 355}, no. 3 (2013),  893--946.
\bibitem{paolomedina}  A. M. Medina-Mardones, A. Pizzi, P. Salvatore,  \textit{Multisimplicial chains and configuration spaces}, Journal of Homotopy and Related Structures
{\bf 19} (2024), 275–296. 
\bibitem{merkulov} S.A. Merkulov, \textit{operads, configuration spaces and quantization},  Bulletin of the Brazilian Mathematical Society, New Series {\bf 42} (2011), 683--781.
\bibitem{miller} J. Miller,
\textit{Nonabelian Poincaré duality after stabilizing}, Trans. Amer. Math. Soc. {\bf 367} (2015), 1969--1991.  
\bibitem{milnor} J. Milnor, \textit{Link groups}, Ann. of Math. {\bf 59} (1954) 177--195.
\bibitem{morgan} J.W. Morgan, \textit{The algebraic topology of smooth, algebraic varieties}, Publ. Math. I.H.E.S. {\bf 48} (1978), 177-204.
\bibitem{morton} H.R. Morton, \textit{Symmetric products of the circle}, Proc. Cambridge Philos. Soc. {\bf 63} (1967), 349--352.
\bibitem{muiagt} J. Hubbuck, N. H. V. Hung, L. Schwartz, Preface, proceedings of the School and Conference in Algebraic Topology (Hanoi, 9–20 August 2004).  Geometry \& Topology Monographs {\bf 11} (2007). 
\bibitem{myershishurs} 
D.J. Myers, H. Sati, U. Schreiber, \textit{Topological Quantum Gates in Homotopy Type Theory}, https://arxiv.org/abs/2303.02382
\bibitem{nakamura} T. Nakamura, \textit{On Cohomology operations}, Japanese J. of Math {\bf 33}, 93--145. 
\bibitem{nakaoka} M. Nakaoka, \textit{Homology of the infinite symmetric group}, Ann. of Math. (2) {\bf 73} (1961), 229--257. 
\bibitem{nrr} R. Nandakumar, N. Ramana Rao, \textit{`Fair’ partitions of polygons – an introduction}, arXiv:0812.2241 (2008).
\bibitem{napolitano} F. Napolitano, \textit{On the cohomology of configuration spaces on surfaces}, J. Lond. Math. Soc. (2) {\bf 68} (2003), 477--492.
\bibitem{nlabconfigs} https://ncatlab.org/nlab/show/configuration+space+of+points
\bibitem{nlabgraphs} https://ncatlab.org/nlab/show/graph+complex
\bibitem{okuyama} S Okuyama, \textit{The space of intervals in a Euclidean space}, Algebr. Geom. Topol. {\bf 5} (2005) 1555--1572.
\bibitem{orlik} P. Orlik,
L. Solomon,  \textit{Combinatorics and topology of complements of
hyperplanes}, Invent. Math. {\bf 56} (1980), no. 2, 167--189.
\bibitem{orlikterao} P. Orlik, Terao, \textit{Arrangements of hyperplanes}, Springer Grundlehren der mathematischen Wissenschaften (1992). 
\bibitem{petersen} D. Petersen, \textit{Cohomology of generalized configuration spaces}, Compositio Math, {\bf 156}, Issue 2 (2020), 251 - 298.
\bibitem{pfaff}  T. Pfaff, J. Walker
\textit{The chromatic polynomial of $P_2\times P_n$ and $C_3\times P_n$}, Missouri J. of Math. Sciences (2008), 169--177.
\bibitem{priddy} S.B. Priddy, \textit{Koszul resolutions}, Trans. Amer. Math. Soc. {\bf 152} (1970) 39--60.
\bibitem{oscar} O. Randall-Williams, Appendix by Quoc P. Ho,  \textit{Configuration spaces as commutative monoids}, ArXiv:2306.02345.
\bibitem{randell} R. Randell, \textit{The fundamental group of the complement of a union of complex hyperplanes}, Invent. Math. {\bf 69} (1982) 103--108.

\bibitem{read} R.C. Read, \textit{An introduction to chromatic polynomials}, J. Combin. Theory {\bf 4} (1968), 52--71.

\bibitem{riahi} H. Riahi, \textit{Periodic orbits of n-body type problems: the fixed period case}, AMS Transactions {\bf 347} (1995), 4663--4685.
\bibitem{rivera} M. Rivera, \textit{Adams’ cobar construction revisited}, Documenta Mathematica {\bf 27} (2022) 1213--1223.

\bibitem{rota} G. Rota, \textit{On the Foundations of Combinatorial Theory: I. Theory of M\"obius Inversion.} Z. Wahrscheinlichkeitstheorie {\bf 2} (1964), 340--368.

\bibitem{roth} F. Roth, 
\textit{On the Category of Euclidean Configuration Spaces
and associated Fibrations}, Geometry \& Topology Monographs {\bf 13} (2008), 447-461.
\bibitem{salter} N. Salter, B. Tshishiku, \textit{Surface bundles in topology, algebraic geometry, and group theory}. Notices Amer. Math. Soc. {\bf 67} (2020), no. 2, 146--154.
\bibitem{paolo1} P. Salvatore, \textit{The homotopy type of Euclidean configuration spaces}, 
Proceedings of the 20th Winter School "Geometry and Physics", Publisher: Circolo Matematico di Palermo(Palermo) Serie II {\bf 66} (2001), 161--164.
\bibitem{paolo0} P. Salvatore, \textit{Configuration spaces with summable labels}, Progress in Mathematics {\bf 196}
(2001), Birkhauser.
\bibitem{paolo2} P. Salvatore,  \textit{Non-formality of planar configuration spaces in characteristic 2},
International Math. Research Notices {\bf 10} (2020), 3100--3129.
\bibitem{paolo3} P. Salvatore \textit{Configuration spaces on the sphere and higher loop spaces}, Math. Zeitschrift  {\bf 248} (2004), 527–540.
\bibitem{samelson} H. Samelson, \textit{A Connection Between the Whitehead and the Pontryagin Product}, American J. of Math. {\bf 75}: 4 (1953), 744--752.
\bibitem{ss} H. Sati, U. Schreiber, \textit{Differential cohomotopy implies intersecting brane observables via configuration spaces and chord diagrams}, Advances in Theoretical and Mathematical Physics
Vol. {\bf 26}, No 4 (2022), 957--1051.
\bibitem{ss2} H. Sati, U. Schreiber, \textit{Anyonic defect branes and conformal blocks in twisted equivariant differential (TED) K-theory},
 Reviews in Mathematical Physics {\bf 35}, No. 06 (2023).
\bibitem{schaper} Ch. Schaper, \textit{Suspensions of affine arrangements}, Math. Ann. {\bf 309} (1997), 463--473.
\bibitem{schiessl} C. Schiessl, \textit{Integral cohomology of configuration spaces of the sphere}, Homology, Homotopy and Applications {\bf 21} (1) (2019), 283--302.
\bibitem{ssv} 
B. Schreiner, F. Šarčević, I. Volić, \textit{Low stages of the Taylor tower for r-immersions}, Involve {\bf 13}:1 (2020), 51–75 
\bibitem{segal} G. Segal, \textit{Configuration-spaces and iterated loop-spaces}, Inventiones {\bf 21} (1973), 213--221.
\bibitem{segalacta} G. Segal, \textit{The topology of rational functions}, Acta Math. 
{\bf 143} (1979), 39--72.
\bibitem{shimakawa} K. Shimakawa,
\textit{Configuration spaces with partially summable labels and homology theories},
Math. J. Okayama Univ. {\bf 43} (2001), 43--72.
\bibitem{shimakawa2} K Shimakawa, \textit{Labeled configuration spaces and group completions},  Forum Math. {\bf 19}:2 (2007), 353-364.
\bibitem{dev} D. Sinha, \textit{The non-equivariant homology of the little disks operad}, SMF Séminaires et Congrès {\bf 26} (2013).
\bibitem{dev2} D. Sinha, \textit{Manifold-theoretic compactifications of configuration spaces}, Selecta math. {\bf 10} (2004), 391 – 428.
\bibitem{dev3} D. Sinha, \textit{Operads and knot spaces},
Journal of the American Mathematical Society {\bf 19}:2 (2006), 461--486. 
\bibitem{snaith} V. Snaith, \textit{A stable decomposition of $\Omega^n\Sigma^n X$}, Journal of the London Mathematical Soc. {\bf 7} (1974), 577-- 583.
\bibitem{sohail} T. Sohail, \textit{Cohomology of configuration spaces of complex projective spaces}, Czech. Math. J. {\bf 60}(135) (2010), 411--422.
\bibitem{souriau} J.M. Souriau, \textit{Structure des syst\`emes dynamiques} 1970 (Paris: Dunod), 383--92.
\bibitem{stanley} R. P. Stanley, \textit{An Introduction to hyperplane arrangements}, Geometric combinatorics, IAS/Park City Math. Ser. {\bf 13}, American Math. Soc., Providence, RI (2007), 389--496.
\bibitem{stanley2} R. P. Stanley, \textit{Enumerative combinatorics, Volume I}, 2nd ed, Cambridge Studies in Advanced Mathematics {\bf 49}, Cambridge University Press (2012). 
\bibitem{st} N.P.S. Strickland, P. R. Turner, \textit{Rational Morava E-theory
and $DS^0$}, Topology {\bf 36} (1997), no. 1, 137--151.
\bibitem{tamaki} D. Tamaki, \textit{Cellular stratified spaces},
 \textit{Combinatorial and toric homotopy}, 305–435, Lect. Notes Ser. Inst. Math. Sci. Natl. Univ. Singap., {\bf 35}, World Sci. Publ., Hackensack, NJ, (2018).
 \bibitem{taylor} L. R. Taylor, \textit{Fibrations, cofibrations and related results}, online notes, January 13, 2012.
 \bibitem{thatte} B.D. Thatte, \textit{
The connected partition lattice of a graph and the reconstruction conjecture}, J Graph Theory (2019) 1--22.
\bibitem{totaro} B. Totaro, \textit{Configuration spaces of algebraic varieties}, Topology {\bf 35} (1996), 1057--1067.
\bibitem{totaro2} B. Totaro, \textit{The integral cohomology of the Hilbert scheme of two points}, Forum Math. Sigma {\bf 4} (2016), Paper No. e8, 20 pp.

\bibitem{tsuchiya} M. Tsuchiya, \textit{On bond lattices of graphs}, Chinese Journal of Mathematics {\bf 20}:3,  287--299.

\bibitem{ummel} B.R. Ummel, \textit{Some examples relating the deleted product ot imbeddability}, Proceeedings AMS {\bf 31}:1 (1972), 307--311.
\bibitem{vainshtein} F. V. Vainshtein, \textit{The cohomology of braid groups}, Funktsional. Anal. i Prilozhen. {\bf 12}:2 (1978), 72--73.
\bibitem{vassiliev} V.A. Vassiliev, \textit{Complements of discriminants of smooth maps: topology and
applications}, translated from the Russian by B. Goldfarb. Translations of Mathematical
Monographs {\bf 98}, AMS (1992).
\bibitem{vershinin} V. Vershinin, \textit{Braid groups and loop spaces}, Russian Mathematical Surveys {\bf 54}:2 (1999), 273--350.
\bibitem{volic} I. Volić, \textit{Configuration space integrals and the topology of knot and link spaces}, Morfismos {\bf 17}:2, 2013, 1--56.
\bibitem{volic2} I. Volić, \textit{A survey of Bott-Taubes integration}, J. Knot Theory Ramifications {\bf 16}:1 (2007), 1--42.
\bibitem{voronov} S. Voronov, \textit{Rational homotopy theory}, Encyclopedia of Mathematical Physics, Ed.2.
\bibitem{wachs} M.L. Wachs, \textit{Poset topology: tools and applications}, Geometric combinatorics, IAS/Park City Math. Ser.13, Amer. Math. Soc. (2007), 497--615.
\bibitem{wada} M. Wada, \textit{Group invariants of links}, Topology {\bf 31} no 2 (1992), 399--406.
\bibitem{walker} J.W. Walker, \textit{Canonical homeomorphisms of posets}, European Journal of Combinatorics {\bf 9}, Issue 2, 97--107. 
\bibitem{wang} J.H. Wang, \textit{On the braid groups for $\bbr P^2$}, J. Pure and Applied Algebra {\bf 166} (2002) 203--227.
\bibitem{weiss} M. Weiss, \textit{Calculus of embeddings}, Bulletin AMS {\bf 33}:2 (1996), 177--187.
\bibitem{westerland} C. Westerland, \textit{Configuration spaces in topology and geometry}, Australian Math. Soc. Gazette {\bf 38}:5 (2011),  279--283.
\bibitem{westerland2} C. Westerland, \textit{Stable splittings of surface mapping spaces}, {\bf 153}, Issue 15 (2006), 2834--2865.
\bibitem{whitney}  H. Whitney, \textit{A logical expansion in mathematics}, Bull. Amer. Math. Soc. {\bf 38} (1932), 572--579.
\bibitem{Whitney2} H. Whitney, \emph{Congruent graphs and the connectivity of graphs}. American J. Math. {\bf 54} (1): 150–168, (1932)
\bibitem{wilshire} J.D. Wiltshire-Gordon, \textit{
Models for configuration space in a simplicial complex}, Colloq. Math. {\bf 155} (2019), no. 1, 127--139.
\bibitem{wilson} J. Wilson, \textit{A brief introduction to representation stability}, OberwolfachWorkshop Workshop (Jan 2018).
\bibitem{witten} E. Witten, \textit{Quantum field theory and the Jones polynomial}, Commun. Math.
Phys. {\bf 121} (1989), 351–399.
\bibitem{wu}  Jie Wu, \textit{On the homology of configuration spaces $C((M, M_0)\times \bbr^n; X)$}, Math. Z. {\bf 22} (1998), 235--248.
\bibitem{xico} M.A. Xicotencatl, \textit{Orbit configuration spaces, infinitesimal braid relations in homology and equivariant loop spaces}, Ph.D. Thesis, University of Rochester (1997).
\bibitem{yasui}  T. Yasui, \textit{The reduced symmetric product of a complex projective space and the embedding problem}, Hiroshima Math. J. {\bf 1} (1971), 27--40.
\bibitem{stylian} S. Zanos, \textit{Méthodes de scindements homologiques en topologie et en géométrie}, Th\`ese Universit\'e de Lille 2009.
\bibitem{zariski} O. Zariski, \textit{The topological discriminant of Riemann surface of genus p}, American J. of Math. {\bf 59} (1937), 335--358.
\bibitem{zakharov} A.S. Zakharov, \textit{Rational Homotopy type of complements of submanifold arrangements}, https://arxiv.org/abs/2211.05033
\bibitem{zhang} A. Y. Zhang, \textit{Quillen homology of spectral Lie algebras with application to mod p homology of labeled configuration spaces}, https://arxiv.org/abs/2110.08428
\bibitem{zz} G. M. Ziegler and R. T. Živaljevic, \textit{Homotopy types of subspace arrangements via diagrams of spaces},
Math. Ann. {\bf 295}:3 (1993), 527--54.
\bibitem{zou} F. Zou, \textit{A geometric approach to equivariant factorization homology and nonabelian Poincare duality}, Math. Z. {\bf 303}:4 (2023), publication number 98. 

\end{thebibliography}

\end{document}